\documentclass[11pt]{article}
\usepackage{xr}
\RequirePackage[OT1]{fontenc}
\RequirePackage{amsthm,amsmath}
\RequirePackage[numbers]{natbib}
\RequirePackage[colorlinks,citecolor=blue,urlcolor=blue]{hyperref}
\RequirePackage{yk}
\usepackage{comment}
\usepackage{amsthm,amsmath,amssymb,enumerate,mathrsfs,multirow}
\usepackage{color}
\usepackage{dsfont}
\usepackage{mathtools}
\usepackage[noend]{algpseudocode}
\usepackage{mathtools}
\usepackage{caption}
\usepackage{subcaption}
\usepackage[plain,noend]{algorithm2e}

\usepackage{hyperref}
\hypersetup{
    colorlinks,%
    citecolor=blue,%
    filecolor=red,%
    linkcolor=red,%
    urlcolor=black,
}
\usepackage[top=1in, bottom=1in, left=1in, right=1in]{geometry}
\usepackage{tikz}
\usetikzlibrary{shapes}
\usepackage[skins,theorems]{tcolorbox}
\tcbset{highlight math style={enhanced,
  colframe=red,colback=white,arc=0pt,boxrule=1pt}}

\newcommand{\indep}{\perp\hspace{-.5em}\perp}

\RequirePackage{xcolor}

\theoremstyle{definition}

\theoremstyle{remark}

\newtcbtheorem{mytheo}{Theorem}%
{colback=green!5,colframe=green!35!black,fonttitle=\bfseries}{th}

\newtcbtheorem{mylemma}{Lemma}%
{colback=blue!5,colframe=purple!35!black,fonttitle=\bfseries}{th}

\newtcbtheorem{myassmp}{Assumption}%
{colback=green!5,colframe=black!35!black,fonttitle=\bfseries}{th}

\newcommand{\boeta}{\boldsymbol\eta}

\newdimen\AAdi%
\newbox\AAbo%
\def\AAk#1#2{\setbox\AAbo=\hbox{#2}\AAdi=\wd\AAbo\kern#1\AAdi{}}%

\newcounter{rcnt}[section]

\def\argmin{\mathop{\rm argmin}}

\title{Estimation of a score-explained non-randomized treatment effect in fixed and high dimensions}
\author{Debarghya Mukherjee, Moulinath Banerjee and Ya'acov Ritov \\\\ Department of Statistics, University of Michigan}
\begin{document}
\maketitle
%
%
%
%
%
%


\begin{abstract}
Non-randomized treatment effect models are widely used for the assessment of treatment effects in various fields and in particular social science disciplines like political science, psychometry, psychology. More specifically, these are situations where treatment is assigned to an individual based on some of their characteristics (e.g. scholarship is allocated based on merit or antihypertensive treatments are allocated based on blood pressure level) instead of being allocated randomly, as is the case, for example, in randomized clinical trials. Popular methods that have been largely employed till date for estimation of such treatment effects suffer from slow rates of convergence (i.e. slower than $\sqrt{n}$). In this paper, we present a new model coined SCENTS: Score Explained Non-Randomized Treatment Systems, and a corresponding method that allows estimation of the treatment effect at $\sqrt{n}$ rate in the presence of fairly general forms of confoundedness, when the `score' variable on whose basis treatment is assigned can be explained via certain feature measurements of the individuals under study. We show that our estimator is asymptotically normal in general and semi-parametrically efficient under normal errors. We further extend our analysis to high dimensional
covariates and propose a $\sqrt n$ consistent and asymptotically normal estimator based on a de-biasing procedure. Our analysis for the high dimensional incarnation can be readily extended to analyze partial linear models in the presence of noisy variables corresponding to the non-linear part of the model, where the noise can be correlated with the variables corresponding to the linear part. We analyze two real datasets via our method and compare our results with those obtained by using previous approaches. We conclude this paper with a discussion on some possible extensions of our approach. 
\end{abstract}



\section{Introduction and background}
\label{sec:intro}
\noindent
Estimation of treatment effect under non-random treatment allocation has been extensively studied in the statistics, biomedicine and econometrics literatures. As an introduction to the idea, imagine that a scholarship granting agency tests a group of high school students and assigns scholarship to those whose scores are above some pre-determined cutoff (w.l.o.g. 0, after centering). Of interest is to determine whether the scholarship has any tangible outcome on future academic performance.   Letting $Y_i$ be the score of student $i$ in a subsequent semester, we can write down a model of the type:
\begin{equation}
\label{eq:main}
Y_i = \alpha_0 \mathds{1}_{Q_i \ge 0} + X_i^{\top}\beta_0 + \nu_i \,
\end{equation}
where $X_i$ is a covariate vector including demographic information on the students, $Q_i$ is the centered score on the test, and $\nu_i$ is a residual term. The parameter $\alpha_0$ represents the effect of the treatment: scholarship.  If the hypothesis $\alpha_0 = 0$ is rejected by a statistical test, one concludes that the scholarship has a significant impact on subsequent academic score. A simple multiple regression cannot be performed to estimate the treatment $\alpha_0$ owing to the possible dependence between $Q$ and $\nu$. Similar examples are also prevalent in biomedicine, especially to assess the efficacy (or possible side effects) of existing drugs (for example, to quantify the side effects of the statin prescribed if some threshold related to cholesterol level is exceeded or anti-depressant prescribed when the Hamilton's score is higher than some threshold).  

It is instructive to take a brief moment to compare and contrast this setup with traditional randomized clinical trials. In both, we estimate the effect of a treatment which applies to a subset of the participants. The main difference is that the allocation of treatment in the latter case is independent of the covariates and error terms in the model, while in the former case, the applied treatment is non-trivially correlated with the covariate and error terms and is therefore endogenous. In our example, the score $Q_i$ on the assessment test, based on which the treatment (scholarship when $Q_i > 0$) is applied, is not only correlated with the indicator of treatment  and the background covariates $X_i$, but \emph{also} the unmeasured sources of variation (say, native abilities of the individual not captured by the observed covariates)  contained in the residual $\nu_i$. More meritorious students are more likely to get the scholarship and are also more likely to have higher values of $Y$, above what is predicted by the scholarship effect and their demographics! This typically translates to what is called the `endogeneity assumption': $E(\nu_i |Q_i) \ne 0$, and hence the model in equation \eqref{eq:main} cannot be estimated using simple linear regression of $Y$ on $(1(Q >0), X)$ unlike the randomized trial framework. This model was initially studied under the name of regression discontinuity design in (\citep{thistlethwaite1960regression}, \citep{campbell1963experimental}), who analyzed data from a national merit competition (see \citep{holland1957} and \citep{thistlethwaite1959effects}) and have since found varied applications, for example \citep{stadthaus1972comparison}, \citep{erickson1972evaluating} in education,  \citep{lohr1972historical} in health related research, \citep{berk1983capitalizing} in social research, \citep{finkelstein1996clinical} in epidemiology, to name a few. 

Most of the analysis in such jump-effect models, till date, is based off of a local analysis, since observations from a small neighborhood of the cut-off $Q$ are considered as almost free from endogeneity and only those are used to get an approximately unbiased estimate of the treatment effect. This is the main reason why these methods only provide an estimation of local average of the treatment effect around the cut-off. 
However, in many practical scenarios, the interest is \emph{not solely} on the jump at the discontinuity but also on the \emph{global effect} of the treatment, which also requires assessing treatment effectivity for individuals that lie substantially away from the discontinuity.  If we consider for example a prestigious program, most of the investment is in students who are clearly above the threshold, and the social implications of the program are far from being restricted to those near it. Thus, measuring the effect of the jump should ideally keep the large implication in mind, whilst in the standard regression discontinuous design analysis, it is only the local phenomenon around the jump that is the object of focus. The aim of this manuscript is to propose and analyze a model which captures this overall phenomenon. 

To highlight our contribution and contrast it with the extant literature on the estimation of treatment effect in presence of endogeneity, let's first consider one of the standard approaches, as exemplified by the model below:
 \begin{equation}
     \label{eq:CCT}
      Y_i = (\alpha_0 + b_0 Q_i)\mathds{1}_{Q_i < \tau_0} + (\alpha_1 + b_1 Q_i)\mathds{1}_{Q_i \ge \tau_0} + X_i^{\top}\beta_0 + \nu_i \,,
 \end{equation}
 where $Q_i$ is the score variable which determines treatment via the known threshold $\tau_0$, and $X_i$'s are background covariates. The parameter of interest is $\alpha_1 - \alpha_0$ which encodes the effect of the treatment. Generally, only weak assumptions are made on the conditional expectation of $E(Y \mid Q, X)$ to encode possible endogeneity.  As the main takeaway from the model is that both the intercept and the slope of $Q$ change at the cut-off, the traditional approach for the estimation of the treatment effect is as follows: Select a possibly (data-driven) bandwidth $h_n$ and look at all the observations $Y_i's$ for which $|Q_i - \tau_0| \le h_n$. Now run a weighted local polynomial regression on these observations to estimate $\alpha_1 - \alpha_0$. As this approach effectively uses $O_p(n h_n)$ observations for estimation purposes, the rate of convergence of the estimator is \emph{slower that $\sqrt{n}$} (typically $\sqrt{nh_n}$, where the bandwidth $h_n \to 0$ is chosen by standard bias-variance trade-off), see, for example \cite{calonico2019regression}. However, the precision with which the treatment-effect is estimated can be improved in situations where it is possible to exploit the relationship between the score $Q$ and the measured covariates $X$ on the individuals of interest. In such cases, we argue below that it becomes possible to use \emph{all our samples} to estimate the treatment effect at $\sqrt{n}$ rate. \\

\noindent
\textbf{Our main contribution: } We assume that the linear equation \eqref{eq:main} can be augmented by a second equation that explains the score $Q$ through the background covariates $X$ to obtain:
 \begin{align}
\label{eq:main_1}
& Y_i = \alpha_0 \mathds{1}_{Q_i \ge 0} + X_i^{\top}\beta_0 + \nu_i \notag \\
& Q_i = g(X_i) + \eta_i \,.
\end{align}
The effect of unmeasured covariates that can affect both $Q$ and $Y$ is encoded by the mean 0 error vector $(\nu_i, \eta_i)$ on which we place no parametric assumptions, they may be dependent and $b(\eta) \equiv E(\nu|\eta) \ne 0.$  In fact, one can (and indeed this may be necessary under certain circumstances) generalize the second equation and write: 
 \begin{align}
\label{eq:main_2}
& Y_i = \alpha_0 \mathds{1}_{Q_i \ge 0} + X_i^{\top}\beta_0 + \nu_i \notag \\
& Q_i = g(Z_i) + \eta_i \,.
\end{align}
where $X$ and $Z$ may be identical or completely disjoint or may share several covariates. The intuition behind this generalization is that it is quite possible that some extra covariates are available while measuring $Y$ but not while measuring $Q$ or vice-versa. The updated equation \eqref{eq:main_2} takes care of that. 

At first glance, the augmentation of the second equation in \eqref{eq:main_2} may suggest that $Z$ is being used as an instrumental variable. However, this is not necessarily the case in our model since the \emph{exclusion restriction} \citep{lousdal2018introduction} that is critical in the instrumental variables approach, namely that the instrument should influence $Y$ only through $Q$ and not directly, is no longer satisfied (especially when $X = Z$, which is typically the case for many real applications).   Consequently, the two-stage least squares procedure \citep{angrist1995two} typically employed in the instrumental variable literature is \emph{invalid} in our problem. We emphasize that, in contrast to instrumental variable regression, we \emph{do not need} to deploy a completely new set of covariates to explain $Q$, rather using at least part of the already available background information $X$ that influences $Y$ to also explain the score $Q$. This is typically the case in many real life examples, e.g. in the scholarship example, the background information, that explains the future academic performance of a student (Y) also influences their scholarship test score $Q$.

In this paper, we analyze a simpler version of the above model (equation \eqref{eq:main_2}) for technical simplicity: 
\begin{align}
\label{eq:main_3}
& Y_i = \alpha_0 \mathds{1}_{Q_i \ge 0} + X_i^{\top}\beta_0 + \nu_i \notag \\
& Q_i = Z_i^{\top}\gamma_0 + \eta_i \,.
\end{align}
where we assume the error $(\nu, \eta) \indep (X, Z)$, $(X, Z)$ can be arbitrarily related, and $\nu$ is \emph{correlated} with $\eta$. However our methods can be extended to the more general model \eqref{eq:main_2} as discussed in Subsection \ref{sec:extension_non-param}. \footnote{Note that if $Z \indep X$, then one may use standard techniques from instrumental variable regression. However, this fails to hold in most our the intended applications.} 

We next discuss how this augmentation helps us obtain a $\sqrt{n}$-consistent estimator. \emph{If we observe $\eta$}, writing $\nu_i = b(\eta_i) + \eps_i$, our model would reduce to a simple partial linear model (see e.g. \citep{yatchew1997elementary}, \citep{bhattacharya1997semiparametric}, \citep{muller2015partial} and references therein) with $(\alpha_0, \beta_0)$ being the parametric component and the unknown mean function $b$ being the nonparametric component: 
\begin{equation}
\label{eq:plm}
 Y_i = \alpha_0 \mathds{1}_{Q_i \ge 0} + X_i^{\top}\beta_0 + b(\eta_i) + \eps_i \,.
\end{equation}
Hence, we could acquire a $\sqrt{n}$ consistent estimator of $\alpha_0$ following the standard analysis of the partial linear model \emph{provided $\bbE(\var(\mathds{1}(Q > 0) \mid \eta)) > 0)$} (for more details about semiparametric efficient estimation in partial linear models, see \citep{robinson1988root} or \citep{schick1986asymptotically}). Indeed, without this assumption it is easy to see that there is an identifiability problem and the effect of $\alpha_0$ cannot be separated from the effect of $b$ in the above model. 

As \emph{we don't observe $\eta$}, we cannot use the standard partial linear model analysis proposed in \citep{robinson1988root} or \citep{schick1986asymptotically} to estimate $\alpha_0$. We can nevertheless approximate it by \eqref{eq:main_3}, using the residuals obtained via regressing $Q$ on $Z$, i.e. plug-in the estimates of $\eta_i$'s in equation \eqref{eq:plm} and treat the resulting equation as an \emph{approximate partial linear model}.  Indeed, this idea lies at the heart of our method. 
However, a naive replacement of $\eta$ by $\hat \eta$ (obtained by regressing $Q$ on $Z$) is not sufficient, as the approximation error of $\eta$ is $\hat \eta - \eta = O_p(n^{-1/2})$, whilst we need the approximation error to be of the order $o_p(n^{-1/2})$ for $\sqrt{n}$-consistent estimation of $\alpha_0$. Hence, some more fine-tuning is needed to remove the bias. We elaborate our method of estimation in Section \ref{sec:est_method}. 

The assumption $\gamma_0 \ne 0$ is critical to our analysis, because if not, 
$Q_i = \eta_i$ and equation \eqref{eq:plm} becomes: 
$$
Y_i = \alpha_0 \mathds{1}_{Q_i \ge 0} + X_i^{\top}\beta_0 + b(Q_i) + \eps_i \,.
$$
It is  now no longer possible to estimate $\alpha_0$ at $\sqrt{n}$ rate as it is hard to separate the first and third term of the RHS. Thus, the augmenting equation $Q = Z^{\top}\gamma_0 + \eta$ with $\gamma_0 \ne 0$ is what prevents the variable corresponding to the parametric component of interest from becoming a measurable function of the variable corresponding to the non-parametric component of the model, and enables converting the estimation problem to an \emph{approximate} partial linear model. 
As noted earlier, for $\alpha_0$ to be estimated as $\sqrt{n}$ rate, it is not necessary for $Q$ to have a linear relation with $Z$, in fact any non-linear parametric relation like $Q = g(Z) + \eta$ for some known link function $g$ will work as long as $\bbE\left(\var(\mathds{1}(Q > 0) \mid \eta)\right) > 0$. Also the assumption $(\nu, \eta) \indep (X, Z)$ is not necessary, all we need is that $\bbE[\eta \mid Z] = 0$ and $\bbE[\eps \mid X, Z, \eta] = 0$ for $\sqrt{n}$-consistent estimation of $\alpha_0$. 

There is also a degree of similarity between our model and \emph{triangular simultaneous equation models}, which are well-studied in the economics literature. Interested readers may take a look at \citep{newey1999nonparametric}, or \citep{pinkse2000nonparametric} and references therein for more details. However, there are two key differences: (1) The triangular simultaneous equation models generally assume a smooth link function between $Y_i$ and $(X_i, Z_i, Q_i)$, whilst our model presents an inherent discontinuity. Hence our work \emph{cannot be derived} from analyses of the triangular simultaneous equation model. (2) The direct influence of X on Y would ruin the identification in these non-separable models. 

Estimation of the treatment effect via partial linear model in the context of regression discontinuity design has also been mentioned in \citep{porter2003estimation} 
through the lens of the following model: 
$$
Y_i = \alpha_0 \mathds{1}_{Q_i > 0} + b(Q_i) + \eps_i, \ \ \  \bbE[\eps \mid Q, \mathds{1}_{Q>0}] = 0 \,.
$$
The author argued that one \emph{cannot estimate} $\alpha_0$ at $\sqrt{n}$ rate as $S_i = \mathds{1}_{Q_i > 0}$ can be completely explained by $Q_i$, i.e. $\bbE[\var(S \mid Q)] = \bbE\left[(S - \bbE[S\mid Q])^2\right] = 0$ (see the discussion after equation \eqref{eq:plm}). The main difference between our model and that of \citep{porter2003estimation} is the usage of the background information $X$ (or $(X, Z)$) which influences both $Q$ and $Y$. This, on one hand, prevents us from using $X$ as simple instruments, and on the other hand ensures that $S$ is not completely explained by $\eta$ (the random variable corresponding non-parametric component $b$) at it also depends on $X$, which enables us to estimate the parametric component at $\sqrt{n}$ rate. 

Recently \citep{angrist2015wanna} presented an idea where one may exploit background covariates to estimate the treatment effect at $\sqrt{n}$ rate under the assumption $\bbE[Y \mid Q, X] = \bbE[Y \mid X]$, i.e. the relation between $Q$ and $Y$ can be explained fully by $X$, which permits them to perform simple OLS method to obtain a $\sqrt{n}$-consistent estimator of the treatment effect. However, their assumption is not satisfied in our case: the main component of our model is the unobserved $\eta$ that quantifies the \emph{innate} ability of an individual (i.e. intelligence), which, along with $X$'s, influences both $Q$ and $Y$ (in-fact it might even be argued that $\eta$ is as important as $X$, especially in the scholarship example like one presented in \citep{angrist2015wanna}, as innate abilities plays a primary role in the performance of an individual in the scholarship test and beyond).  
\\\\
\noindent
{\bf Extension to high dimensional setup: }We have also extended our analysis to the high dimensional model, where we assume both the dimension of $X$ and $Z$ are much larger than the available sample size. Analysis of treatment effect in presence of high dimensional covariates is a relatively new topic. For example,  \cite{belloni2014inference}  proposed a debiased approach for inference on the treatment effect in presence of high dimensional controls, and, recently, \cite{qiu2021inference} studied heterogeneous treatment effect in the presence of high dimensional covariates under the availability of instrumental variables; \cite{guo2021estimation} proposed a modeling strategy for personalized medicine in presence of high dimensional covariates and  \cite{arai2019causal} presented an approach to estimate the parameters of a standard regression discontinuity design in high dimensions.

Following the previous discussion, our model in the high dimensional setup can be viewed as a variant of the high dimensional partial linear model where we do not observe the random variable corresponding to the non-linear part, but instead a noisy version of it, with the noise being correlated with both the random variables corresponding to the linear and non-linear parts. Several efforts have been made to estimate both the linear and the non-linear part of a high dimensional partial linear model, of which we mention a few here. \cite{lian2019projected}  used a projected spline estimator along with the $\ell_1$ penalty to estimate the parametric and non-parametric parts of a high dimensional partial linear model; \cite{lv2017debiased} extended the analysis to the distributed setup; \cite{muller2015partial} used the $\ell_1$ penalty for the linear part and a smoothing penalty (i.e. penalty involving the double derivative of the non-parametric function) for the non-linear part; \cite{han2017adaptive} proposed an adaptive approach based on \cite{honore1997pairwise}'s pairwise difference based method which does not require knowledge of any function class a-priori;  \cite{zhu2019high} proposed a method to de-bias the estimator of the parameter corresponding to the linear part which facilitates inference for any particular co-ordinate or some linear combination of the parameters corresponding to the linear part. To the best of our knowledge, the closest paper that studies a similar situation (i.e. partial linear model with erroneous observations corresponding to the non-linear part) is \cite{zhu2017nonasymptotic}. Our analysis is different from that of \cite{zhu2017nonasymptotic} for the following reasons: (i) our basic estimation procedure  (elaborated in the subsequent sections) is different as our model is quite different from that considered in \cite{zhu2017nonasymptotic} and (ii) we provide a de-biased estimate of the treatment effect $\alpha_0$ and establish asymptotic normality for purposes of inference, whereas \cite{zhu2017nonasymptotic} only calculate the rate of the estimation error in their model.   
\\\\
\noindent
To summarize, our work makes the following contributions:
\begin{enumerate}
\item It is able to take care of endogeneity among the errors in a general manner and provide a $\sqrt{n}$-consistent estimator of the treatment effect. Indeed, the main feature of our approach lies in modeling $Q$ itself in terms of covariates up to error terms, which enables \emph{the use of the entirety of data available} and \emph{not just the observations in a small vicinity of the boundary defined by the $Q$ threshold}, on which existing approaches are typically based.
\item Our estimate achieves semiparametric efficiency under an appropriate submodel. 
\item Our method does not depend on tuning parameters, in the sense that the use of tuning parameters to estimate the treatment effect is
\emph{secondary}. As will be seen in Section \ref{sec:est_method}, we do require tuning parameter specifications for non-parametric estimation of $b(\eta)$, but as long as those parameters satisfy some minimal conditions, our  estimate of $\alpha_0$ -- in terms of both rate of convergence and asymptotic distribution-- does not depend on it.
\item Our analysis of SCENT in presence of high dimensional covariates appears to be the first systematic attempt to deal with de-biasing in the high dimensional partial linear model where we observe a noisy version of the random variable corresponding to the non-linear part, and the noise is correlated with both the random variables corresponding to the linear and the non-linear parts. 
\end{enumerate}

\noindent
\textbf{Organization of the paper: } In Section \ref{sec:est_method} we describe the estimation procedure. Section \ref{sec:theory} provides the theoretical results along with a brief outline of the proof of asymptotic normality of our estimator. In Section \ref{sec:real_data}, we provide analyses of two real data examples using our method, as well as comparisons to previous methods. In Section \ref{sec:conclusion} we present some possible future research directions based on this work. In particular, we discuss the scenario where the treatment effect, which is assumed to be constant for the model studied in this paper, can depend on the innate ability $\eta$: in other words, what happens when $\alpha_0$ is replaced by $\alpha(\eta)$, a generalization that may be warranted in certain applications. We point out that $\sqrt{n}$-consistent estimates of the integrated treatment effect can be obtained in this situation as well.  Rigorous proofs of the main results are established in Section \ref{sec:proof_main_thm} - Section \ref{sec:proof_prop} of the Supplementary document. In Section \ref{sec:additional_results} of the Supplementary document, we present proofs of auxiliary results that are required to prove the main results.  In Section \ref{sec:spline_details} of the Supplementary document, we provide a few details on the spline estimation techniques used in our analysis and in Section \ref{sec:main_algo}  of the Supplementary document, an algorithm integrating the main steps of the proposed methodology for the ease of implementation. 
\\\\
\noindent
\textbf{Notation: }For any matrix $A$, we denote by $A_{i, *}$, the $i^{th}$ row of $A$ and by $A_{*, j}$ the $j^{th}$ column of $A$. Both $a_n \lesssim_P b_n$ and $a_n  = O_p(b_n)$ bear the same meaning, i.e. $a_n/b_n$ is a tight (random) sequence. Also, for two non-negative sequences $\{a_n\}$ and $\{b_n\}$, we denote by $a_n \gg b_n$ (respectively $a_n \ll b_n$), the conditions that $\liminf_n a_n/b_n \to \infty$ (respectively $\limsup_{n} a_n/b_n \to 0$).  For any random variable (or object) $X$, we denote by $\cF(X)$, the sigma-field generated by $X$. For two random variables $X$ and $Y$, $X \indep Y$ indicates that $X$ and $Y$ are independent, $X \perp Y$ denotes that they are uncorrelated. 
\section{Estimation procedure for $\alpha_0$}
\label{sec:est_method}
\noindent
As mentioned in Section \ref{sec:intro}, our model for estimating the treatment effect can be written as: 
\begin{align*}
Y_i & = \alpha_0 \mathds{1}_{Z_i^{\top}\gamma_0 + \eta_i \ge 0} + X_i^{\top}\beta_0 + \nu_i \,.
\end{align*}
where $\nu_i$ and $\eta_i$ are correlated. Defining $b(\eta_i) = \bbE(\nu_i \mid \eta_i)$, write $\nu_i = b(\eta_i) + \epsilon_i$ where $\epsilon_i \perp \eta_i$. Using this we can rewrite our model as: 
\begin{align}
\label{eq:new1}Y_i & = \alpha_0 \mathds{1}_{Z_i^{\top}\gamma_0 + \eta_i \ge 0} + X_i^{\top}\beta_0 + b(\eta_i) + \epsilon_i \,.
\end{align}
We first divide the whole data in three (almost) equal parts, say $\cD_1, \cD_2, \cD_3$. Henceforth for simplicity we assume each $\cD_i$ has $n/3$ observations. Denote the dimension of $X$ and $Z$ by $p_1$ and $p_2$ respectively. The first two data sets are used to obtain (consistent) estimates of several nuisance parameters which are then plugged into the equations corresponding to the third data set from which an estimator of the treatment effect is constructed. The data splitting technique makes the theoretical analysis more tractable as one can use independence among the three subsamples to our benefit. Furthermore, by rotating the samples (to be elaborated below), we obtain three asymptotically independent and identically distributed estimates which are then averaged to produce a final estimate that takes advantage of the full sample size.  While we believe that the estimator obtained without data-splitting achieves the same asymptotic variance, a point that is corroborated by simulation studies (not reported in the manuscript), an analysis of this estimator would be incredibly tedious with minimal further insight. 
\\\\
\noindent
We, first, estimate $\gamma_0$ from $\cD_1$ via standard least squares regression of $Z$ on $Q$: 
$$
\hat \gamma_n = (Z^{\top}Z)^{-1}Z^{\top}Q \,.
$$
Using $\hat \gamma_n$, we expand our first model equation as :  
%
\begin{align}
Y_i & = \alpha_0 S_i + X_i^{\top}\beta_0 + b(\hat \eta_i) + (\eta_i - \hat \eta_i)b'(\hat \eta_i) + R_{1, i} + \epsilon_i  \notag \\
\label{eq:exp_1} & = \alpha_0 S_i + X_i^{\top}\beta_0 + b(\hat \eta_i) + b'(\hat \eta_i)Z_i^T(\hat \gamma_n - \gamma_0) + R_{1, i} + \epsilon_i \,.
\end{align}
where $S_i = \mathds{1}_{Q_i > 0}$, $\hat \eta_i  = Q_i - Z_i^{\top}\hat \gamma_n$ and $R_{1, i} = (\hat \eta_i - \eta_i)^2 b''(\tilde \eta_i)/2$ is the residual obtained from a two-step Taylor expansion with $\tilde \eta_i$ lying between $\eta_i$ and $\hat \eta_i$. It should be pointed out that if $\eta_i$ were known, our model (equation \eqref{eq:new1}) would reduce to a simple partial linear model and estimation of $\alpha_0$ would become straight-forward. As we \emph{don't} observe $\eta_i$, but rather use $\hat \eta_i$ as its proxy, the corresponding approximation error needs careful handling, as we need to show that the estimator does not inherit any resulting bias. Indeed, this is one of the core technical challenges of this paper, and explained in the subsequent development. 

Note that the function $b'$ in equation \eqref{eq:exp_1} is unknown. Since $\hat \gamma_n = \gamma_0 + O_p(n^{-1/2})$ (which, in turn, implies $\hat \eta_i = \eta_i + O_p(n^{-1/2})$), as long as we have a consistent estimate of $b'$, say $\hat b^{'}$,  the approximation error $(b'(\hat \eta_i) - \hat b'(\hat \eta_i))Z_i(\hat \gamma_n - \gamma_0)$ is $o_p(n^{-1/2})$ and therefore asymptotically negligible. 

We now elaborate how we use B-spline basis to estimate $b'$ from $\cD_2$. Equation \eqref{eq:new1} can be rewritten as: 
\begin{align}
\label{eq:spline1}
Y_i & = \alpha_0 S_i + X_i^{\top}\beta_0 + b(\hat \eta_i) + \epsilon_i + R_{i}\,.
\end{align}
where $R_{i} = b(\eta_i)  - b(\hat \eta_i)$. Ignoring the remainder term (which is shown to be asymptotically negligible under some mild smoothness condition on $b$ to be specified later), we estimate $b'$ via a B-spline basis from equation \eqref{eq:spline1}. The theory of spline approximation is mostly explored for estimating compactly supported non-parametric regression functions. Our errors $\eta$ are, of course, assumed to be unbounded as otherwise the problem would become artificial. However, there are certain technical issues with estimating $b'$ on the entire support (see Remark \ref{rem:spline_entire}), and to circumvent that we restrict ourselves to a compact (but arbitrary) support $[-\tau, \tau]$, i.e. we consider those observations for which $|\hat \eta_i| \le \tau$. We then use a cubic B-spline basis appropriately scaled to the interval of interest with equispaced knots to estimate $b'$ \footnote{While we work with cubic splines, one may certainly use higher degree polynomials. However, in practice, it has been observed that cubic spline works really well in most of the scenarios}.  Notationally, we use $K-1$ knots to divide $[-\tau, \tau]$ into $K$ intervals of length $2\tau/K$ where $K = K_n$ increases with $n$ at an appropriate rate (see remark \ref{rem:K} below), giving us in total $K+3$ spline basis functions. For any $x$, we use the notation $\tilde N_K(x) \in \bbR^{(K+3)}$ to denote the vector of scaled B-spline basis functions $\{\tilde N_{K,j}\}_{j=1}^{K+3}$ evaluated at $x$\footnote{A brief discussion on B-spline basis and scaled B-spline basis is presented in Section \ref{sec:spline_details} of the Supplementary document for the ease of the reader.}. Using these basis functions we further expand equation \eqref{eq:spline1} as follows: 
\begin{equation}
\label{eq:basis_exp}
Y_i = \alpha_0 S_i + X_i^{\top}\beta_0 + \tilde N_K(\hat \eta_i)^{\top}\omega_{b, \infty, n} + \epsilon_i + R_{i} + T_{i}\,.
\end{equation}
for all those observations with $|\hat \eta_i| \le \tau$, where $T_{i}$ is the spline approximation error, and $T_i = b(\hat \eta_i) - N_K(\hat \eta_i)^{\top}\omega_{b, \infty, n}$ and $\omega_{b, \infty, n}$ is the (population) parameter defined as: 
$$
\omega_{b, \infty, n} = \argmin_{\omega \in \bbR^{(K+3)}} \sup_{|x| \le \tau} \left|b(x) - \tilde N_K^{\top}(x)\omega\right| \,.
$$
Suppose we have $n_{2} \le n/3$ observations in $\cD_2$ with $|\hat \eta_i| \le \tau$. Denote by $\bY$ the vector of all the corresponding $n_2$ responses, by $\bX \in \bbR^{(n_2, p_1)}$ the covariate matrix, and by $\tilde{\bN}_K \in \bbR^{(n_2, K+3)}$ the approximation matrix with rows $\tilde{\bN}_{K, i*} = \tilde N_K(\hat \eta_i)$. Regressing $Y$ on $(\bS, \bX, \tilde{\bN}_K)$ we estimate $\omega_{b, \infty, n}$ (details can be found in the proof of Proposition \ref{thm:spline_consistency}) and set $\hat b'(x) = \nabla \tilde{N}_K(x)^{\top}\hat \omega_{b, \infty, n}$ where $\nabla \tilde{N}_K(x)$ is the vector of derivates of each of the co-ordinates of $\tilde{N}_K(x)$. The following theorem establishes consistency of our estimator of $b'$ (proof can be found in Section \ref{sec:proof_prop} of the supplementary document):  
\begin{proposition}
\label{thm:spline_consistency}
Under Assumptions \ref{assm:independence} - \ref{assm:moment} (elaborated in Section \ref{sec:theory}) we have: 
$$
\sup_{|x| \le \tau} \left|b'(x) - \hat b'(x)\right| = \sup_{|x| \le \tau} \left|b'(x) - \nabla \tilde N_K(x)^{\top}\hat \omega_{b, \infty, n}\right| = o_p(1) \,. 
$$
\end{proposition}
\begin{remark}
\label{rem:K}
Henceforth, we choose $K \equiv K_n$ such that $n^{1/8} \ll K \ll n^{1/3}$ to control the approximation errors of certain non-parametric functions (including $b(\eta)$) that appear in our analysis via the B-spline basis. However the bounds can be improved in presence of additional derivations of the non-parametric functions involved.  
\end{remark}
\vspace{0.1in}
\noindent
The final (key) step involves $\alpha_0$ from $\cD_3$. Suppose there are $n_3 \le n/3$ observations in $\cD_3$ with $|\hat \eta_i| \le \tau$. Replacing $b'$ by $\hat b'$ obtained from $\cD_2$ in equation \eqref{eq:exp_1} we obtain: 
\begin{align}
\label{eq:new3} Y_i & = \alpha_0 S_i + X_i^{\top}\beta_0 + b(\hat \eta_i) + \hat b'(\hat \eta_i)Z^{\top}_i (\hat \gamma_n - \gamma_0) + R_{1, i} + R_{2, i} + \epsilon_i \notag\\ 
& \triangleq \alpha_0 S_i + X_i^{\top}\beta_0 + b(\hat \eta_i) + \tilde Z^{\top}_i (\hat \gamma_n - \gamma_0) + R_{1, i} + R_{2, i} + \epsilon_i \,, \\
Q_i - Z_i^{\top} \hat \gamma_n & = - Z_i^{\top}(\hat \gamma_n - \gamma_0) + \eta_i \notag \,.
\end{align}
where we define $\tilde Z_i = \hat b'(\hat \eta_i) Z_i$, the residual term $R_{2, i}$ as: 
$$
R_{2, i} = \left(b'(\hat \eta_i) - \hat b'(\hat \eta_i)\right)Z_i^{\top}\left(\hat \gamma_n - \gamma_0\right) =  \left(b'(\hat \eta_i) - \hat b'(\hat \eta_i)\right)\left(\hat \eta_i - \eta_i\right)  \,.
$$
and $R_{1, i}$ is same as in equation \eqref{eq:exp_1}. An inspection of the first part of Equation \eqref{eq:new3} shows that up to the remainder terms $\left\{\left(R_{1, i}, R_{2, i}\right)\right\}_{i=1}^{n_3}$, our model is a partial linear model with parameters $(\alpha_0, \beta_0, b)$. These remainder terms are asymptotically negligible as shown in the proof of our main theorem (Theorem \ref{thm:main_thm}). We estimate $\alpha_0$ from equation \eqref{eq:new3} using standard techniques for the partial linear model which, again, involve approximating the function $b$ with the same B-spline basis as before:  
\begin{align}
Y_i & = \alpha_0 S_i + X_i^{\top}\beta_0 + \tilde Z_i^{\top}(\hat \gamma_n - \gamma_0)+ b(\hat \eta_i) + R_{1, i} + R_{2, i} + \eps_i \notag \\
\label{eq:spline_exp_1} & =  \alpha_0 S_i + X_i^{\top}\beta_0 + \tilde Z_i^{\top}(\hat \gamma_n - \gamma_0)+ \tilde N_K(\hat \eta_i)^{\top}\omega_{b, \infty, n} + R_{1, i} + R_{2, i} + R_{3, i} + \eps_i \,.
\end{align}
\newline
where $R_{3, i}$ is the spline approximation error, i.e. $R_{3, i} = b(\hat \eta_i) - \tilde N_K(\hat \eta_i)^{\top}\omega_{b, \infty, n}$. Combining equation \eqref{eq:spline_exp_1} along with the second equation of \eqref{eq:new3} we formulate the following linear model equation: 
\begin{align}
\label{eq:j}
\begin{pmatrix}
\bY \\
\bQ - \bZ\hat \gamma_n
\end{pmatrix}
& =
\begin{pmatrix}
\bS & \bX & \tilde \bZ  & \tilde \bN_K\\
0 & 0 & -\bZ & 0
\end{pmatrix}
\begin{pmatrix}
\alpha_0 \\
\beta_0 \\
\hat \gamma_n - \gamma_0  \\
\omega_{b, \infty, n}
\end{pmatrix}
+ 
\begin{pmatrix}
\bR \\ 0 
\end{pmatrix}
+
\begin{pmatrix}
\epsilon \\
\eta
\end{pmatrix} \notag \\
& = \begin{pmatrix}
\bW & \tilde \bN_{K, a}
\end{pmatrix}
\begin{pmatrix}
\theta_0 \\
\omega_{b, \infty, n}
\end{pmatrix}
+ 
\begin{pmatrix}
\bR \\ 0 
\end{pmatrix}
+
\begin{pmatrix}
\epsilon \\
\eta
\end{pmatrix}
\end{align}
where $\theta_0 = (\alpha_0, \beta_0, \hat \gamma_n - \gamma_0)$ , 
$$
\bW = \begin{bmatrix}
\bW_1 \\ \bW_2
\end{bmatrix} \ \ \text{ with } \ \ 
\bW_1 = \begin{pmatrix}
\bS & \bX & \tilde \bZ 
\end{pmatrix} , \ \ \ 
\bW_2 = \begin{pmatrix}
0 & 0 & -\bZ
\end{pmatrix} \,,
$$
where $\bW_1 \in \bbR^{(n_3, 1+p_1+p_2)}, \bW_2 \in \bbR^{(n/3, 1+p_1+p_2)}$ (as we are using all the observations in $\cD_3$ in the second regression equation of \eqref{eq:new3}), $\bR$ is the vector of the sum of three residuals $R_{1, i}, R_{2, i}, R_{3, i}$ mentioned in equation \eqref{eq:spline_exp_1}, and the matrix $\tilde \bN_{K, a}$ (read $\tilde \bN_K$ appended) has the form: 
$$
\tilde \bN_{K, a}  = \begin{bmatrix} \tilde \bN_K \\ 0 \end{bmatrix} \,.
$$
From equation \eqref{eq:j}, we estimate $\theta_0$ via ordinary least squares methods: 
$$
\hat \theta = \left(\bW^{\top}\proj^{\perp}_{\tilde \bN_{K, a}}\bW\right)^{-1}\bW^{\top}\proj^{\perp}_{\tilde \bN_{K, a}}\begin{pmatrix}
\bY \\
\bQ - \bZ\hat \gamma_n
\end{pmatrix}
$$
and set $\hat \alpha$ as the first co-ordinate of $\hat \theta$, where for any matrix $A$, we define $\proj_A$ as the projection matrix on the columns space of $A$.
Note that because of data-splitting, $\hat \gamma_n$, $\hat b'$ and $\cD_3$ are mutually independent, which provides a significant technical advantage in dealing with the asymptotics of our estimator. 

Finally, we apply the same methodology on permutations of the three sets of data: i.e. we estimate $\hat \gamma_n$ from $\cD_2$, $\hat b'$ from $\cD_3$, $\hat \alpha \in \cD_1$ and $\hat \gamma_n$ from $\cD_3$, $\hat b'$ from $\cD_1$, $\hat \alpha \in \cD_2$. Denote by $\hat \alpha_i$, the estimator $\hat \alpha$  estimated from $\cD_i$ for $1 \le i \le 3$. Our final estimate is then $\overline{\hat \alpha} = (1/3) \sum_i \alpha_i$.

\begin{remark}
\label{rem:double_est_b}
In our estimation procedure, we effectively estimate the conditional mean function $b(\cdot)$ twice: once while estimating $b'$ from $\cD_2$ and again while estimating $\alpha_0$ from $\cD_3$. Note that, the second re-estimation of $b$ is quite critical (i.e. we cannot \emph{use} the estimator of $b$ obtained from $\cD_2$) due to presence of higher order bias (slower than $n^{-1/2}$). 
\end{remark}

\begin{remark}
\label{rem:spline_entire}
As described in this section, we only use observations for which $|\hat \eta_i| \le \tau$, losing some efficiency in the process. One way to circumvent this issue is to use a sequence $\{\tau_n\}$ slowly increasing to $\infty$ and considering all those observations for which $|\hat \eta| \le \tau_n$. Although, this will acquire efficiency in the limit, we need a stronger set of assumptions to make it work: for starters, we need conditions on the decay of the density of $\eta$, we need it not to vanish  anywhere in $[-\tau_n, \tau_n]$, and knowledge of the rate at which $\min_{|x| \le \tau_n} f_\eta(x)$ approaches $0$. We also need stronger conditions on some conditional expectation functions (i.e. conditional expectation of $\begin{pmatrix} S & \bX & \bZ\end{pmatrix}$ given $\eta + a^{\top}Z$ for some vector $a$), e.g., bounded derivatives in both co-ordinates over the entire space (see Lemma \ref{lem:g}  of the Supplementary document.) With more technical nuances we believe our method can be extended to the entire real line by using a growing interval, but from a purely statistical angle, it will not bring anything insightful to the methodology that we propose here.  
\end{remark}

\section{Analysis for fixed dimensional covariates}
\label{sec:theory}
In this section we presents our main theorems with broad outline of the proofs. Details are provided in the Supplementary document. To establish the theory, we need the following assumptions:

%

%
%

\begin{assumption}
\label{assm:independence}
The errors $(\eta, \nu)$ are independent of the distribution of $(X ,Z)$ and have zero expectation. 
\end{assumption}

\begin{assumption}
\label{assm:b} 
The distribution of $(\eta, \nu)$ satisfies the following conditions: 
\begin{enumerate}[i)]
\item The density of $\eta$, denoted by $f_{\eta}$, is continuously differentiable and both $f_{\eta}$ and its derivatives are uniformly bounded. 
\item The conditional mean function $b(\eta) = \bbE[\nu \mid \eta]$ is $3$ times differentiable with $b''$ and $b'''$ are uniformly bounded over real line.  
\item The variance function $\sigma^2(\eta) = \var(\eps \mid \eta)$ is uniformly bounded from above. 
\item There exists some $\xi > 0$ such that: $\min_{|x| \le \tau + \xi} f_\eta(x) > 0$. 
\end{enumerate}
\end{assumption}

%
%
%

\begin{assumption}
\label{assm:eigen}
Define the matrices $\Omega$ and $\Omega^*$ as:  
\begin{align*}
\Omega & =  \bbE\left[\var\left(\begin{bmatrix} S & X & Z b'(\eta)\end{bmatrix}\left| \right. \eta\right)\right] + \var\left(\begin{bmatrix} 0 & 0 & Z \end{bmatrix}\right)\\
\Omega^* & = \bbE\left[\sigma^2(\eta)\var\left(\begin{bmatrix} S & X & Z b'(\eta)\end{bmatrix}\mid \eta\right)\right] + \var(\eta)\var\left(\begin{bmatrix} 0 & 0 & Z \end{bmatrix}\right) \,.
\end{align*} 
Then, the minimum eigenvalues of $\Omega$ and $\Omega^*$ are strictly positive. 
\end{assumption}

\begin{assumption}
\label{assm:moment}
The distribution of $(X, Z)$ satisfies the following conditions: 
\begin{enumerate}[i)]
\item $(X, Z)$ has bounded continuous density function and have zero expectation. 
\item The first four moments of $X, Z$ are finite. 
\end{enumerate}
\end{assumption}

\noindent
\begin{remark}
\label{rem:assm_discussion}
Assumption \ref{assm:b} provides a low-level assumption on the smoothness of the density of $\eta$, the conditional mean function $b(\eta)$ and the conditional variance profile $\sigma^2(\eta)$, which is required for the standard asymptotic analysis of the partial linear model. Assumption \ref{assm:eigen}, again is a standard assumption in partial linear model literature. It is essential for the asymptotic normality of the treatment effect as the asymptotic variance of our estimator is a function of these variances. If this assumption is violated, then the asymptotic variance will be infinite and that estimation at $\sqrt{n}$ is not possible. (Note that if $\gamma_0 = 0$, then Assumption \ref{assm:eigen} is violated.) As our method does not use all observations, but a fraction depending on the interval $[-\tau, \tau]$, our limiting variance comprises of the following truncated versions of $\Omega$ and $\Omega^*$: 
\begin{align*}
\Omega_{\tau} & =  \bbE\left[\var\left(\begin{bmatrix} S & X & Z b'(\eta)\end{bmatrix}\mid \eta\right)\mathds{1}_{|\eta| \le \tau}\right] + \var\left(\begin{bmatrix} 0 & 0 & Z \end{bmatrix}\right)\,,\\
\Omega^*_{\tau} & = \bbE\left[\sigma^2(\eta)\var\left(\begin{bmatrix} S & X & Z b'(\eta)\end{bmatrix}\mid \eta\right)\mathds{1}_{|\eta| \le \tau}\right] + \var(\eta)\var\left(\begin{bmatrix} 0 & 0 & Z \end{bmatrix}\right)\,.
\end{align*} 
It is immediate that if $\tau \to \infty$, then $\Omega_{\tau} \to \Omega$ and $\Omega^*_{\tau} \to \Omega^*$. Hence, in light of Assumption \ref{assm:eigen}, by continuity, the minimum eigenvalues of $\Omega_{\tau}$ and $\Omega^*_{\tau}$ are also positive for all large $\tau$. 
\end{remark}

\vspace{0.1in}

\noindent
\subsection{Asymptotic normality} 
\vspace{0.1in}
\noindent
We now state the main result of the paper: 
\begin{theorem}
\label{thm:main_thm}
Consider the estimates obtained at the end of the previous section. Under assumptions \ref{assm:independence}-\ref{assm:moment}: 
$$
\sqrt{n}\left(\hat \alpha - \alpha_0\right) \overset{\mathscr{L}}{\implies} \cN(0, 3e_1^{\top}\Omega_{\tau}^{-1}\Omega^*_{\tau}\Omega_{\tau}^{-1}e_1)\,,
$$
whilst
$$
\sqrt{n}\left(\bar{\hat \alpha} - \alpha_0\right) \overset{\mathscr{L}}{\implies} \cN(0, e_1^{\top}\Omega_{\tau}^{-1}\Omega^*_{\tau}\Omega_{\tau}^{-1}e_1)
$$
\end{theorem}
\vspace{0.2in}
\noindent
{\bf Sketch of proof: }We present a high-level outline of the key steps of the proof, deferring all technical details to Subsection \ref{sec:proof_main_thm} of the Supplementary document. From \eqref{eq:j}, on a set of probability approaching 1, our estimator can be written as: 
\begin{align*}
\hat \alpha & = e_1^{\top}\left(\bW^{\top}\proj^{\perp}_{\tilde \bN_{K, a}}\bW\right)^{-1}\bW^{\top}\proj^{\perp}_{\tilde \bN_{K, a}}\begin{pmatrix}
\bY \\
\bQ - \bZ\hat \gamma_n
\end{pmatrix} \\
& =  \alpha_0 + e_1^{\top}\left(\bW^{\top}\proj^{\perp}_{\tilde \bN_{K, a}}\bW\right)^{-1}\bW^{\top}\proj^{\perp}_{\tilde \bN_{K, a}}\left[\begin{pmatrix}
\bR \\ 0
\end{pmatrix} + 
\begin{pmatrix}
\eps \\ \eta
\end{pmatrix}
\right]\,.
\end{align*}
This implies: 
\begin{equation}
\label{eq:main_exp_1}
\sqrt{n}\left(\hat \alpha - \alpha_0\right) = 
 e_1^{\top}\left(\frac{\bW^{\top}\proj^{\perp}_{\tilde \bN_{K, a}}\bW}{n}\right)^{-1}\frac{\bW^{\top}\proj^{\perp}_{\tilde \bN_{K, a}}}{\sqrt{n}}\left[\begin{pmatrix}
\bR \\ 0
\end{pmatrix} + 
\begin{pmatrix}
\eps \\ \eta
\end{pmatrix}
\right]
\end{equation}
which is our main estimating equation. We next outline the key steps of our proof.
\\\\
\noindent

{\bf Step 1: } First show that: 
$$
\frac{\bW^{\top}\proj^{\perp}_{\tilde \bN_{K, a}}\bW}{n} \overset{P}{\longrightarrow} \frac13 \Omega_{\tau} \,.
$$
where $\Omega_{\tau}$ is as defined in Remark \ref{rem:assm_discussion}.

{\bf Step 2: }Next, establish the following asymptotic linear expansion: 
$$
\frac{\bW^{\top}\proj^{\perp}_{\tilde \bN_{K, a}}}{\sqrt{n}}\begin{pmatrix}
\beps \\ \eta
\end{pmatrix} = \frac{1}{\sqrt{n}}\sum_{i=1}^{n/3}\varphi\left(X_i, Z_i, \eta_i, \nu_i\right)\ + o_p(1) \,.
$$
for some influence function $\varphi$. 

{\bf Step 3: } Apply central limit theorem to obtain: 
$$
\frac{1}{\sqrt{n}}\sum_{i=1}^{n/3}\varphi\left(X_i, Z_i, \eta_i, \nu_i\right) \overset{\mathscr{L}}{\implies} \cN(0, \frac13\Omega^*_{\tau}) \,.
$$
where $\Omega^*_{\tau}$ is as defined in Remark \ref{rem:assm_discussion}. 

{\bf Step 4: }Finally ensure that `residual term'  is asymptotically negligible: 
$$
\frac{\bW^{\top}\proj^{\perp}_{\bN_{K, a}}}{\sqrt{n}}\begin{pmatrix}
\bR \\ 0
\end{pmatrix}  \overset{P}{\longrightarrow} 0 \,.
$$
Now, combining above steps we can conclude: 
\begin{align*}
\sqrt{n}(\hat \alpha - \alpha_0) &= e_1^{\top}\left(\frac13 \Omega_{\tau}\right)^{-1}\frac{1}{\sqrt{n}}\sum_{i=1}^{n/3}\varphi\left(X_i, Z_i, \eta_i, \nu_i\right) + o_p(1) \\
& \overset{\mathscr{L}}{\implies} \cN\left(0, 3e_1^{\top}\Omega_{\tau}^{-1}\Omega^*_{\tau}\Omega^{-1}_{\tau}e_1\right)  
\end{align*}
where the leading term only depends on the observations in $\cD_3$ and is consequently independent of $\cD_1, \cD_2$. Finally, rotating the dataset and taking average of the $\hat \alpha$'s we further conclude: 
$$
\sqrt{n}(\bar{\hat \alpha} - \alpha_0) \overset{\mathscr{L}}{\implies} \cN\left(0, e_1^{\top}\Omega_{\tau}^{-1}\Omega^*_{\tau}\Omega^{-1}_{\tau}e_1\right)  \,.
$$

\subsection{Semi-parametric efficiency}
We further show that our estimator is semi-parametrically efficient under certain restrictions. As our estimator is based on least square approach, it can not be shown to be efficient unless the error $\eps$ is normal. We prove the following theorem (proof can be found in Section \ref{sec:sem_eff} of the supplementary document): 
\begin{theorem}
\label{thm:sem_eff}
Suppose the model is the following: 
$$
Y_i = \alpha_0 \mathds{1}_{Z_i^{\top}\gamma_0 + \eta > 0} + X_i^{\top}\beta_0 + b(\eta_i) + \eps_i \,.
$$
where $\eps_i \sim \cN(0,\tau^2) \indep \eta_i$. Then our estimator of $\alpha_0$ is semi-parametrically efficient, i.e. its asymptotic variance $\sigma_0^2  = \tau^2 e_1^{\top}\Omega^{-1}e_1$ attains the semi-parametric information bound for this model.   
\end{theorem}

\begin{remark}
\label{rem:WLS}
The assumption of the normality of $\eps$ is necessary to establish semiparametric efficiency for least squares type methods, but the assumption of homoskedasticity is essential only if we use ordinary least squares method. One may easily take care of heteroskedasticity by using weighted least squares instead. The first step towards that direction is to approximate the variance profile $\sigma(\eta)$ using $\hat \sigma(\hat \eta)$ for some non-parametric estimate $\hat \sigma(\cdot)$ of $\sigma(\cdot)$. Then defining $\bD \in \bbR^{(n_3 + n/3) \times (n_3 + n/3)}$ to be the diagonal matrix with first $n/3$ diagonal entries being $\hat \sigma(\hat \eta_i)$'s (i.e. for all those $\hat \eta_i$'s such that $|\hat \eta_i| \le \tau$) and last $n/3$ diagonal entries  being $\hat \sigma_{\eta}$'s  (an estimate of variance of $\eta$) we estimate the treatment effect as: 
$$
\hat \alpha = e_1^{\top} \left(\bW^{\top}\bD^{-1/2}\proj^{\perp}_{\bD^{-1/2}\tilde \bN_{K, a}}\bD^{-1/2}\bW\right)^{-1}\bW^{\top}\bD^{-1/2}\proj^{\perp}_{\bD^{-1/2} \tilde \bN_{K, a}}\bD^{-1/2}\begin{pmatrix}
\bY \\
\bQ - \bZ\hat \gamma_n
\end{pmatrix}
$$
A more tedious analysis establishes that this estimator is asymptotic normal and semi-parametrically efficient under the error structure: $\nu_i = b(\eta_i) + \sigma(\eta_i)\eps_i$, where $\eps_i \sim \cN(0, 1)$. As this does not add anything of significance to the core idea of the paper, we confine ourselves to use OLS instead of WLS for ease of presentation.  
\end{remark}

\section{Analysis for high dimensional covariates}
\label{sec:high_dim_rdd}
\newcommand{\mathcolorbox}[2]{\colorbox{#1}{$\displaystyle #2$}}
As highlighted in our analysis for fixed dimensional covariates in the previous section, our model is an example of a partial linear model, where the observations corresponding to the non-linear part are noisy and the noise is correlated with the covariates corresponding to the linear part. Recall that our is model is defined as: 
\begin{align}
    Y_i & = \alpha_0 \mathds{1}_{Q_i \ge 0} + X_i^{\top}\beta_0 + b(\eta_i) + \eps_i \notag \\
    \label{eq:model_main_eq} Q_i & = Z_i^{\top}\gamma_0 + \eta_i \,.
\end{align}
with $\bbE[\eta \mid Z] = 0$ and $\bbE[\eps \mid \eta, X, Z] = 0$. In this section, we assume that the dimension of ($X, Z$) is larger than the sample size. More specifically, denoting $p_1 := \dim(X)$ and $p_2 := \dim(Z)$, we assume that $p_1 \wedge p_2 \gg n$. As before, we only observe $(Y_i, X_i, Z_i, Q_i)$ but not $\eta_i$. As in the case for almost all high dimensional statistical analysis, we assume that both $\beta_0$ and $\gamma_0$ are sparse vectors (with $\|\beta_0\|_0 = s_\beta$ and $\|\gamma_0\|_0 = s_\gamma$ where $s_\beta \vee s_\gamma \ll n$) to enable consistent estimation of the treatment effect $\alpha_0$. We will quantify the precise assumptions needed for our theory later. 
%
The core estimation procedure in this high dimensional regime is similar to that for its fixed dimension counterpart, but with certain changes to take care of the effect of high dimensional covariates. As before, we divide the whole data $\cD$ into three parts $\cD_1, \cD_2, \cD_3$ (with $n_i$ data in $\cD_i$, where $n_1 \sim n_2 \sim n_3 \sim n/3$). From $\cD_1$, we estimate $\gamma_0$ by performing LASSO regression of $Q$ on $Z$: 
$$
\hat \gamma = \argmin_{\gamma \in \reals^{p_2}} \left[\frac{1}{2n_3}\left\|\bQ  - \bZ\gamma\right\|^2 + \lambda \left\|\gamma \right\|_1\right]\,,
$$
where $\lambda \asymp \sqrt{\log{p_2}/n}$ and set $\hat \eta_i = Q_i - Z_i^{\top}\hat \gamma$ for all observations in $\cD_2 \cup \cD_3$. This estimator of $\gamma_0$ is consistent and rate optimal under the restricted eigenvalue assumption (henceforth RE) on $\bZ$. Just like in our analysis for fixed dimension, we only consider all the $\hat \eta$'s (both in $\cD_2$ and $\cD_3$) such that $|\hat \eta| \le \tau$ and ignore the remaining data from $\cD_2$ and $\cD_3$. Therefore, the rest of our analysis on $\cD_2$ and $\cD_3$ are solely based on those observation for which $|\hat \eta| \le \tau$. Using this approximation of the unknown $\eta_i$'s, we expand the first equation of \eqref{eq:model_main_eq} as: 
\begin{equation}
\label{eq:main_break_1}
Y_i = \alpha_0 \mathds{1}_{Q_i \ge 0} + X_i^{\top}\beta_0 + b(\hat \eta_i) + \underbrace{(\eta_i - \hat \eta_i)}_{Z_i^{\top}(\hat \gamma_n - \gamma_0)}b'(\hat\eta_i) + (\eta_i - \hat \eta_i)^2 b''(\tilde \eta_i) + \eps_i\,.
\end{equation}
The logic behind this two step Taylor expansion is similar to that for our fixed dimensional analysis as articulated in Section \ref{sec:est_method}.
We next estimate $b'$ using $\cD_2$ using a different expansion of \eqref{eq:model_main_eq}. Equation \eqref{eq:main_break_1} will be used later for estimating $\alpha_0$ based on the data in $\cD_3$. To estimate $b'$ (using the observations in $\cD_2$) we expand the first equation of \eqref{eq:model_main_eq} as: 
\begin{equation}
\label{eq:break_d2}
Y_i  = \alpha_0\mathds{1}_{Q_i \ge 0} + X_i^{\top}\beta_0 + b(\hat \eta_i) + R_i + \eps_i\,.
\end{equation}
The above equation is obtained simply by replacing $b(\eta_i)$ in the first equation of 
\eqref{eq:model_main_eq} with $b(\hat \eta_i)$ giving the residual term $R_i = b(\hat \eta_i) - b(\eta_i)$. We next invoke techniques from the analysis of high dimensional partial linear models to estimate $b'$ using the above equation. Replacing $b'(\hat \eta_i)$ by $\hat b'(\hat \eta_i)$ in equation \eqref{eq:main_break_1} we obtain the following representation for the observations in $\cD_3$: 
\begin{align*}
Y_i & = \alpha_0 \mathds{1}_{Q_i \ge 0} + X_i^{\top}\beta_0 + b(\hat \eta_i) +  Z_i^{\top}(\hat \gamma_n - \gamma_0)\hat b'(\hat\eta_i) \\
& \qquad \qquad +  Z_i^{\top}(\hat \gamma_n - \gamma_0)\left(b'(\hat \eta_i) - \hat b'(\hat \eta_i)\right) + (\eta_i - \hat \eta_i)^2 b''(\tilde \eta_i) + \eps_i\,.
\end{align*}
Further, an approximation of $b(\hat \eta_i)$ by B-spline basis yields: 
\begin{equation}
    \label{eq:exp_v2}
    Y_i = \alpha_0 \mathds{1}_{Q_i \ge 0} + X_i^{\top}\beta_0 + \bN_k(\hat \eta_i)^{\top}\omega_b + Z_i^{\top}(\hat \gamma_n - \gamma_0)\hat b'(\hat\eta_i) + \tilde R_i + \eps_i
\end{equation}
where $\bN_k(\eta)$ consists of the B-spline basis functions evaluated at $\hat \eta_i$. The residual term $\tilde R_i$ can be decomposed as $R_{1, i} + R_{2, i} + R_{3, i}$ where the individual residuals are: 
\begin{align*}
   R_{1 ,i} & = b(\hat \eta_i) - \bN_k(\hat \eta_i)^{\top}\omega_b   \\
    R_{2, i} & = Z_i^{\top}(\hat \gamma_n - \gamma_0)\left(b'(\hat \eta_i) - \hat b'(\hat \eta_i)\right)  \\
  R_{3, i} & = (\eta_i - \hat \eta_i)^2 b''(\tilde \eta_i) \,.
\end{align*}
The first residual is the B-spline approximation error of $b$ while the second is the product of two different error terms: (i) the error in estimation of $\gamma_0$ from $\cD_3$ and (ii) the error in estimation of $b'$ from $\cD_2$ and the last is the Taylor approximation error. We use equation \eqref{eq:exp_v2} as our main estimating equation for $\alpha_0$. Before delving into the estimation procedure we introduce some notation: 
\begin{enumerate}
\item $\breve Z_i$  will be used to denote $Z_i \hat b'(\hat \eta_i)$. 
\item The vector $W$ will be used to denote the random vector $(S, X, Zb'(\eta))$. 
\item $\breve W$ will be used to denote the random vector $(S, X, \breve Z)$ and the $j^{th}$ element of $W$ (resp. $\check W$) will be denoted as $W_j$ (resp. $\check W_j$). 
\item $m_j(\eta)$ will be used to denote $\bbE[W_j \mid \eta]$, for all $0 \le j \le 1 + p_1 + p_2$ where $W_0 = Y$. 
\item $\check m_j(\hat \eta)$ will be used to denote $\bbE[W_j \mid \hat \eta]$, for all $0 \le j \le 1 + p_1 + p_2$ where $W_0 = Y$. 
\item $\widehat{\tilde W}$ will be used to denote $W - \bbE[W \mid \hat \eta]$. 
\item $\tilde W$ will be used to denote the random vector $W - \bbE[W \mid \eta]$. 
\item $\breve \bW, \hat{\tilde \bW}, \tilde \bW$ will be used to denote the matrix version of $\breve W, \hat{\tilde W}, \tilde W$ respectively by concatenating all of the $n_3$ observations of $\cD_3$ row-wise. For any matrix $\bA$, we use $\bA_{*, i}$ to denote the $i^{th}$ row of $\bA$, and $\bA_{*, j}$ to denote the $j^{th}$ column of $\bA$. 
\item $\theta_0 \equiv \theta_{0, n} = (\alpha_0, \beta_0, \hat \gamma_n - \gamma_0)$. 
\item $\bN_k$  will be used to denote the $n \times K$ matrix, whose $i^{th}$ row consists of $K$ B-spline bases evaluated at $\hat \eta_i$. 
\item For any vector $v$ (or matrix $\bA$), the notation $v^{\perp}$ (resp. $\bA^{\perp}$) is used to denote $P^{\perp}_{\bN_k}v$ (resp. $P^{\perp}_{\bN_k}A$). 
\end{enumerate}
Using the above notation, equation \eqref{eq:exp_v2} can be rewritten (in matrix form concatenating all the observations along the rows) as:  
$$
\bY = \breve \bW \theta_0 + \bN_k\omega_b + \tilde \bR + \beps \,.
$$ 
Projecting out the effect of $\bN_k$ from both sides yields: 
$$
P^{\perp}_{\bN_k}\bY = P^{\perp}_{\bN_k}\breve \bW \theta_0 + P^{\perp}_{\bN_k}\tilde \bR + P^{\perp}_{\bN_k}\beps \,.
$$
Henceforth we use $\perp$ notation as superscript to denote that the vector/matrix is pre-multiplied by $P^{\perp}_{\bN_k}$. Our estimate of $\alpha_0$ is defined as follows: 
$$
\hat \alpha = \frac{\left(\bY^{\perp} - \breve \bW^{\perp}_{-1}\hat \theta_{-1,Y}\right)^{\top}\left(\bS^{\perp} - \breve \bW^{\perp}_{-1}\hat \theta_{-1, S}\right)}{\left(\bS^{\perp} - \breve \bW^{\perp}_{-1}\hat \theta_{-1, S}\right)^{\top}\left(\bS^{\perp} - \breve \bW^{\perp}_{-1}\hat \theta_{-1, S}\right)}
$$
where $\hat \theta_{-1, Y}$ denotes the LASSO estimator obtained by regressing $\bY$ on $\breve \bW_{-1}$, i.e. 
\begin{equation}
\label{eq:lasso_est_Y}
\hat \theta_{-1, Y} = \argmin_{\theta} \left[\frac{1}{2n_3}\left\|\bY - \breve \bW_{-1}\theta\right\|^2 + \lambda_0 \left\|\theta\right\|_1\right]
\end{equation}
and $\hat \theta_{-1, S}$ denotes the LASSO estimator obtained by regressing $\bS$ on $\breve \bW_{-1}$, i.e. 
\begin{equation}
\label{eq:lasso_est_S}
\hat \theta_{-1, S} = \argmin_{\theta} \left[\frac{1}{2n_3}\left\|\bS - \breve \bW_{-1}\theta\right\|^2 + \lambda_1 \left\|\theta\right\|_1\right]
\end{equation}
for some appropriate choice of $\lambda_0, \lambda_1$ to be specified later. Here $n_3$ is the number of observations in $\cD_3$ for which $|\hat \eta_i| \le \tau$. It is immediate that $n_3 \asymp n/3 \asymp n$, i.e. the number of observations in $\cD_3$ with $|\hat \eta_i| \le \tau$ is of the order $n$, the total number of observations. Henceforth, we will ignore this difference and state all our results in terms of $n$. We show that $\sqrt{n_3}\left(\hat \alpha - \alpha_0\right)$ is asymptotically normal under certain assumptions which are stated below:  
\begin{assumption}[Smoothness of $b$]
\label{assm:smooth_b}
The function $b$ in equation \eqref{eq:model_main_eq} is assumed to be $\upsilon \ge 3$ times differentiable with bounded derivates. 
\end{assumption}

\begin{assumption}
The density $f_\eta$ of $\eta$ is bounded, continuously differentiable and lower bounded by some $f_- > 0$ on $[-\tau - \xi, \tau + \xi]$ for some small $\xi > 0$. 
\end{assumption}

\begin{assumption}[Smoothness of conditional expectation]
\label{assm:smoothness_conditional}
Define the function $g_j(a, t)$ as: 
$$
g_j(a, t) = \bbE\left[W_j \mid \eta + a^{\top}Z = t\right]
$$ 
for $0 \le j \le 1+p_1+p_2$. Assume for any fixed $\|a\| \le r$, the collection $\{g_j(a, \cdot)\}_{1 \le j \le p}$ belongs to a function class $\Sigma(\alpha, \bl)$ for some $\alpha \ge 3$ and $\bl \in \reals^{\lfloor \alpha \rfloor}$, where $\Sigma(\alpha, \bl)$ is defined as the collection of all the functions $f$, which are $\lfloor \alpha \rfloor$ times differentiable and the $\lfloor \alpha \rfloor^{th}$ derivative satisfies: 
$$
\left|f^{\lfloor \alpha \rfloor}(x) - f^{\lfloor \alpha \rfloor}(y)\right| \le \bl_{\lfloor \alpha \rfloor}\left|x - y\right|^{\alpha - \lfloor \alpha \rfloor}
$$
and $\|f^{(i)}\|_\infty \le \bl_i$ for all $1 \le i \le \lfloor \alpha \rfloor - 1$. Moreover, this $\bl$ is uniform over all $\|a\| \le r$ where $r$ is independent of the underlying dimension. 
\end{assumption}
\begin{remark}
The conditional mean functions $m_j$ and $\check m_j$ are special cases of $g_j$ as $\check m_j(t) = g_j(\hat \gamma_n - \gamma_0, t)$ and $m_j(t) = g_j(0, t)$. 
\end{remark}

\begin{assumption}[Sub-gaussianity]
\label{assm:sg}
Assume that for $0 \le j \le p_1 + p_2$, $\tilde W_{j}, \hat{\tilde{W}}_{j}, m_j(\eta), \check m_j(\hat \eta), \eta$ and $\eps$ are subgaussian random-variables with subgaussianity constant uniformly bounded by $\sigma_W$. Furthermore assume that $\var(W_j \mid \hat \eta)$ is uniformly (over $j$) bounded on $[-\tau, \tau]$. 
\end{assumption}
\noindent
For our next assumption, we define a covariance matrix $\Sigma$ and two vectors $\theta^*_Y \in \reals^{p_1 + p_2}$ and $\theta^*_S \in \reals^{p_1 + p_2}$ which play a crucial role in our analysis: 
\begin{align}
\Sigma_\tau & = \bbE_\eta\left[\left(W - \bbE[W \mid \eta]\right)\left(W - \bbE[W \mid \eta]\right)^{\top}\mathds{1}_{|\eta| \le \tau}\right] \notag \\
& = \bbE[\tilde W \tilde W^{\top}\mathds{1}_{|\eta| \le \tau}] \notag \\ 
& = \bbE_\eta\left[\var(W \mid \eta)\mathds{1}_{|\eta| \le \tau}\right] \notag \\
& = \bbE_\eta\left[\var\left(\begin{pmatrix}S & X & b'(\eta) Z\end{pmatrix} \mid \eta\right)\mathds{1}_{|\eta| \le \tau}\right] \notag \\ \notag \\
\label{eq:def_theta_s} \theta^*_S & = \argmin_\delta \bbE\left[\left(\tilde S - \tilde W_{-1}^{\top}\delta\right)^2\mathds{1}_{|\eta| \le \tau}\right] = \left(\Sigma_{\tau_{-1, -1}}\right)^{-1}\bbE[\tilde W_{-1}\tilde S\mathds{1}_{|\eta| \le \tau}] \\
\label{eq:def_theta_y} \theta^*_Y & = \argmin_\delta \bbE\left[\left(\tilde Y - \tilde W_{-1}^{\top}\delta\right)^2\mathds{1}_{|\eta| \le \tau}\right] = \left(\Sigma_{\tau_{-1, -1}}\right)^{-1}\bbE[\tilde W_{-1}\tilde Y\mathds{1}_{|\eta| \le \tau}] \,.
\end{align}
The matrix $\Sigma_\tau$ is the expectation of conditional variance of $W$ given $\eta$ on the set $|\eta|\le \tau$, which arises in the estimation of the linear part of a partial linear model. This matrix can be thought as high dimensional analogue of $\Omega_\tau$ defined in our analysis for fixed dimensional covariates. $\theta_Y^*$ (resp. $\hat \theta^*_S$) is the best linear estimator in the population for regressing $\tilde Y$ (resp. $\tilde S$) on $\tilde W$. In the definition of our estimator $\hat \alpha$, we have regressed $\bY^{\perp}$ on $\breve \bW_{-1}^{\perp}$ and $\bS^{\perp}$ on $\breve \bW_{-1}^{\perp}$. Intuitively speaking, projecting out the column space of $\bN_k(\hat \eta)$ is asymptotically equivalent to centering around the conditional expectation with respect to $\eta$. Therefore, it is expected that $\hat \theta_{-1, Y}$ (respectively $\hat \theta_{-1, S}$) should be asymptotically consistent for regressing $\theta_Y^*$ (respectively $\hat \theta^*_S)$ under certain sparsity assumption and Restricted eigenvalue (RE) condition on the covariate matrix $\breve \bW_{-1}^{\perp}$. We now state our assumptions on $\Sigma, \theta^*_Y, \theta^*_S$: 


\begin{assumption}[Asymptotic variance]
\label{assm:sigma}
Assume that there exists $C_{\min} > 0$ and $C_{\max} < \infty$ such that: 
$$
C_{\min} \le \lambda_{\min}\left(\Sigma_\tau\right) \le \lambda_{\max}\left(\Sigma_\tau\right) \le C_{\max} < \infty \,.
$$
Furthermore, define
$$
\sigma^2_{n, 1} = \bbE\left[\eps^2 \left(\tilde S - \tilde W_{-1}^{\top}\theta^*_S\right)^2\mathds{1}_{|\eta| \le \tau}\right] 
$$
and assume that: 
$$
C_{\min} \le \liminf_{n} \sigma^2_{n, 1} \le \limsup_n \sigma^2_{n, 1}\le C_{\max} < \infty \,.
$$
Furthermore assume that: 
$$
\limsup_{n \to \infty} \bbE\left[\left(\tilde S - \tilde W_{ -1}^{\top}\theta^*_S\right)^{2 + \xi} \right] < \infty \,.
$$
for some small $\xi > 0$. 
\end{assumption}
\begin{assumption}[Sparsity]
\label{assm:sparsity_RE}
Assume there exists $s_0, s_1 > 0$ such that $\|\theta^*_Y\|_0 \le s_0$, $\|\theta^*_S\|_0 \le s_1$ and $s_i \log{(p_1 \wedge p_2)}/n \to 0$\,. 
\end{assumption}

\begin{remark}
Although it is apparent from equation \eqref{eq:lasso_est_Y} and \eqref{eq:lasso_est_S} that RE condition is needed on $\breve \bW_{-1}^{\perp}$ for the consistency and rate optimality of LASSO estimates $\hat \theta_{-1, S}$ and $\theta^*_{-1, 1}$, from the sub-gaussianity of $\tilde W$, we have that $\tilde \bW_{-1}$ satisfies RE condition with high probability (see \cite{rudelson2012reconstruction}). We prove in Proposition \ref{prop:RE} that this implies $\breve \bW_{-1}^{\perp}$ also satisfies RE condition with high probability and consequently $\hat \theta_{-1, S}$ and $\hat \theta_{-1, Y}$ will be rate efficient. 
\end{remark}

\begin{remark}
From our original model equation \ref{eq:model_main_eq}, we have: 
$$
\tilde Y = Y - \bbE[Y \mid \eta] = \tilde S \alpha_0 + \tilde X \beta_0 \,,
$$
which implies
$$
\theta^*_Y = \left(\Sigma_{\tau_{-1, -1}}\right)^{-1}\bbE[\tilde W_{-1}\tilde Y\mathds{1}_{|\eta| \le\tau}]  = \theta^*_S\alpha_0 + \theta_{0, -1}^*
$$
where $\theta_{0, -1}^* = (\beta_0, 0)$. Note that $\theta^*_{0, -1}$ is already sparse with $\left\|\theta_{0, -1}^*\right\|_0 = s_\beta$ (recall that we have defined $s_\beta = \|\beta_0\|_0$). So sparsity assumption on $\theta^*_S$ automatically ensures sparsity of $\theta^*_Y$, in other words we have $s_0 \le s_1 + s_\beta$. 
\end{remark}

\noindent
Under the above assumptions we prove the following theorem: 
\begin{theorem}
\label{thm:main_hd}
Under the above assumptions and certain conditions on the sparsities: $s_0, s_1, s_\beta, s_\gamma$ (see subsection \ref{sec:discussion_sparsity_rate} for detailed discussion), we have: 
$$
\frac{\sigma^2_{n, 2}}{\sigma_{n, 1}}\sqrt{n_3}\left(\hat \alpha - \alpha_0\right) \overset{\mathscr{L}}{\implies} \cN(0, 1) \,.
$$
where the values of $\sigma_{n, 1}$ is as defined in Assumption \ref{assm:sigma} and $\sigma_{n,2 }$ is: 
\begin{align*}
\sigma_{n, 2} & = \sqrt{\bbE\left[\left(\tilde S - \tilde W_{-1}\theta^*_S\right)^2\mathds{1}_{|\eta| \le \tau}\right]} = \frac{1}{\sqrt{\left(\Sigma_\tau^{-1}\right)_{1, 1}}}\,.
\end{align*} 
\end{theorem}
\begin{remark}
Similar to our analysis for the fixed dimensional covariates, we can gain efficiency (in terms of asymptotic variance) here as well, by rotating the datasets and taking average of the estimates of $\alpha_0$). However, as the analysis is already quite involved, we do not pursue this extension here. 
\end{remark}
\noindent
{\bf A roadmap of the proof: }We now present a basic roadmap of the proof of Theorem \ref{thm:main_hd} for the ease of the readers, the details can be found in Appendix \ref{sec:details_hd} and in the supplementary document. As immediate from the definition, the estimation of $\hat \alpha$ consists of three key steps: 
\begin{enumerate}
\item LASSO regression using $\bY^{\perp}$ on $\breve \bW_{-1}^{\perp}$.
\item LASSO regression using $\bS^{\perp}$ on  $\breve \bW_{-1}^{\perp}$.
\item Finally, regression of the residual of the first LASSO regression on the residual of the second LASSO regression. 
\end{enumerate}
Recall that the last $p_2$ columns of $\breve \bW$ involves $\hat b'(\hat \eta)$ as a coefficient of $Z$, which, asymptotically should be close to $b'(\eta)$, i.e. the random vector $\breve W$ should be asymptotically close to the random vector $W$. Consequently, if we consider $\breve W^{\perp}$, it should be asymptotically close to $\tilde W$, as projecting out the span of the b-spline bases evaluated at $\hat \eta$ is expected to be asymptotically equivalent to centering around $\eta$. Using this intuition, we expect that our LASSO estimates $\hat \theta_{-1, Y}$ (first LASSO) should converge to $\theta^*_Y$ and the estimator $\hat \theta_{-1, S}$ (second LASSO) should converge to $\theta^*_S$ under our assumptions. This is what we establish using Lemma \ref{lem:choice_lambda} and Proposition \ref{prop:RE}. Lemma \ref{lem:sq_rate}, which establishes the \emph{asymptotic closeness} of $\breve \bW^{\perp}$ to $\tilde \bW$ via the following approximation: 
$$
\breve \bW^{\perp} \longrightarrow \bW^{\perp} \longrightarrow \widehat{\tilde{\bW}} \longrightarrow \tilde \bW \,,
$$
is an auxiliary lemma that is used in the proof of Lemma \ref{lem:choice_lambda}. Furthermore, for the optimal rate of LASSO estimator, we need restricted eigenvalue condition of the covariate matrix $\breve \bW_{-1}^{\perp}$. In Assumption \ref{assm:sparsity_RE} we have assumed the matrix $\tilde \bW_{-1}$ satisfied RE condition. In Proposition \ref{prop:RE} we establish that, if $\tilde \bW_{-1}$ satisfies RE, then $\breve \bW_{-1}^{\perp}$ also satisfies it with high probability. To provide further insights of the proof, we expand $\sqrt{n}(\hat \alpha - \alpha_0)$ as follows: 
\allowdisplaybreaks
\begin{align*}
    \hat \alpha & = \frac{\left(\bY^{\perp} - \breve \bW^{\perp}_{-1}\hat \theta_{-1,Y}\right)^{\top}\left(\bS^{\perp} - \breve \bW^{\perp}_{-1}\hat \theta_{-1, S}\right)}{\left(\bS^{\perp} - \breve \bW^{\perp}_{-1}\hat \theta_{-1, S}\right)^{\top}\left(\bS^{\perp} - \breve \bW^{\perp}_{-1}\hat \theta_{-1, S}\right)} \notag \\
    & = \frac{\left(\bS^{\perp}\alpha_0  - \breve \bW^{\perp}_{-1}\hat \theta_{-1, S}\alpha_0 + \breve \bW^{\perp}_{-1}\hat \theta_{-1, S}\alpha_0 - \breve \bW^{\perp}_{-1}\hat \theta_{-1,Y} + \breve \bW^{\perp}_{-1}\theta_{0, -1} + \bR^{\perp} + \beps^{\perp} \right)^{\top}\left(\bS^{\perp} - \tilde \bW^{\perp}_{-1}\hat \theta_{-1, S}\right)}{\left(\bS^{\perp} - \tilde \bW^{\perp}_{-1}\hat \theta_{-1, S}\right)^{\top}\left(\bS^{\perp} - \tilde \bW^{\perp}_{-1}\hat \theta_{-1, S}\right)} \notag \\
    &= \alpha_0 + \frac{\left(\theta_{0,-1} + \hat \theta_{-1, S}\alpha_0 - \hat \theta_{-1, Y} \right)^{\top}\breve \bW_{-1}^{\perp^{\top}}\left(\bS^{\perp} - \breve \bW^{\perp}_{-1}\hat \theta_{-1, S}\right)}{\left(\bS^{\perp} - \breve \bW^{\perp}_{-1}\hat \theta_{-1, S}\right)^{\top}\left(\bS^{\perp} - \breve \bW^{\perp}_{-1}\hat \theta_{-1, S}\right)} \notag \\
    & \qquad \qquad + \frac{\beps^{\perp^{\top}}\left(\bS^{\perp} - \breve \bW^{\perp}_{-1}\hat \theta_{-1, S}\right)}{\left(\bS^{\perp} - \breve \bW^{\perp}_{-1}\hat \theta_{-1, S}\right)^{\top}\left(\bS^{\perp} - \breve \bW^{\perp}_{-1}\hat \theta_{-1, S}\right)} + \frac{\tilde \bR^{\perp^{\top}}\left(\bS^{\perp} - \breve \bW^{\perp}_{-1}\hat \theta_{-1, S}\right)}{\left(\bS^{\perp} - \breve \bW^{\perp}_{-1}\hat \theta_{-1, S}\right)^{\top}\left(\bS^{\perp} - \breve \bW^{\perp}_{-1}\hat \theta_{-1, S}\right)}
\end{align*}
This implies we have: 
\allowdisplaybreaks
\begin{align}
    \sqrt{n}\left(\hat \alpha - \alpha_0\right) & = \frac{\frac{1}{\sqrt{n}} \left(\theta_{0,-1} + \hat \theta_{-1, S}\alpha_0 - \hat \theta_{-1, Y} \right)^{\top}\breve \bW_{-1}^{\perp^{\top}}\left(\bS^{\perp} - \breve \bW^{\perp}_{-1}\hat \theta_{-1, S}\right)}{\frac1n \left(\bS^{\perp} - \breve \bW^{\perp}_{-1}\hat \theta_{-1, S}\right)^{\top}\left(\bS^{\perp} - \breve \bW^{\perp}_{-1}\hat \theta_{-1, S}\right)} \notag \\
    & \qquad \qquad + \frac{\frac{1}{\sqrt{n}}\tilde \bR^{\perp^{\top}}\left(\bS^{\perp} - \breve \bW^{\perp}_{-1}\hat \theta_{-1, S}\right)}{\frac1n \left(\bS^{\perp} - \breve \bW^{\perp}_{-1}\hat \theta_{-1, S}\right)^{\top}\left(\bS^{\perp} - \breve \bW^{\perp}_{-1}\hat \theta_{-1, S}\right)} \notag \\
    \label{eq:debias_1}  & \qquad \qquad \qquad \qquad + \frac{\frac{1}{\sqrt{n}}\beps^{\perp^{\top}}\left(\bS^{\perp} - \breve \bW^{\perp}_{-1}\hat \theta_{-1, S}\right)}{\frac1n \left(\bS^{\perp} - \breve \bW^{\perp}_{-1}\hat \theta_{-1, S}\right)^{\top}\left(\bS^{\perp} - \breve \bW^{\perp}_{-1}\hat \theta_{-1, S}\right)}
\end{align}
Using the consistency of the lasso estimates, we first establish the stability of the common denominator of equation \eqref{eq:debias_1}. More specifically we show in  Proposition \ref{prop:denom_conv} that: 
$$
\frac{1}{n}\left\|\bS^{\perp} - \breve \bW^{\perp}\hat \theta_{-1, S}\right\|^2 = \bbE\left[\left(\tilde S - \tilde W_{-1}^{\top}\theta^*_S\right)^2\mathds{1}_{|\eta| \le \tau}\right] + o_p(1) = \frac{1}{\left(\Sigma_\tau^{-1}\right)_{1, 1}} + o_p(1)\,.
$$
By Assumption \ref{assm:sigma},  we conclude that the common denominator of equation \eqref{eq:debias_1} is $O_p(1)$. The next step is to show that the numerators of the first two terms of the RHS of equation \eqref{eq:debias_1} is $o_p(1)$. The basic intuition behind this asymptotic negligibility is that the numerator of the first term largely is basically product of the prediction error of the two LASSO regressions, and consequently $o_p(1)$, even after scaling by $\sqrt{n}$ under certain condition on sparsity required for debiased LASSO. The numerator of the second term of the RHS of equation \eqref{eq:debias_1} is the inner product the residuals of original model equation \eqref{eq:exp_v2} and the lasso residuals of regression $\bS^{\perp}$ on $\breve \bW_{-1}^{\perp}$. As we have the already established the lasso residuals stabilizes, the asymptotic negligibility of this numerator primarily stems from the asymptotic negligibility of the residual vector. Details can be found in the Appendix. The final step is to show that the numerator of the third term of equation \eqref{eq:debias_1} is asymptotically normal which follows from an application of Lindeberg's central limit theorem. This completes the roadmap of the proof.

\section{Real data analysis}
\label{sec:real_data}
In this section we illustrate our method by analyzing two real datasets. We divide our analysis into two subsections, one for each dataset. We first present a brief description of the data, then present our analysis and compare our results with the existing one.

\subsection{Effect of Islamic party on women's education in Turkey} 
\vspace{0.1in}
\noindent
In this subsection we study the effect of Islamic party rule in Turkey on women empowerment in terms of their high school education. In the 1994 municipality elections, Islamic parties won several municipal mayor seats in Turkey. We are interested in investigating whether this winning had any effect on the education of women, i.e. to determine, statistically, whether the concern that Islamic control may be inimical towards gender equality is supported by the data. The dataset we analyze here was collected by Turkish Statistical Institute and was first analyzed in \cite{meyersson2014islamic}. Since then it has been used by several authors,  appearing for example, as one of the core illustrations in \cite{cattaneo2019practical}\footnote{We have downloaded the dataset from \url{https://github.com/rdpackages-replication/CIT_2019_CUP/blob/master/CIT_2019_Cambridge_polecon.csv}.}. The dataset consists of $n = 2629$ rows where the rows represent municipalities, the units of our analysis. The main target/response variable $Y$ is the percentage of women in the 15-20 year age-group who were recorded to have completed their high school education in the 2000 census. As mentioned in \cite{meyersson2014islamic}, this is the group that should have been most affected by the decisions made by the winners of the 1994 election. 
The treatment determining variable $Q$ is the difference in vote share between the largest Islamic party (i.e. the Islamic party which got maximum votes among all Islamic parties) and the largest secular party (i.e. the non-Islamic party which got the maximum votes among all non-Islamic parties).  Hence, the cutoff is $0$: i.e. if $Q_i > 0$, then the $i^{th}$ municipality unit elected an Islamic party and if $Q_i < 0$, a secular party. The description of the available co-variates is presented in Table \ref{tab:description_islamic} of the supplementary document. For $X$ and $Z$, we use all the co-variates presented in that table except \texttt{i89} because almost $1/3$'rd of the observations (729 many) are missing for this variable. We estimate $\alpha_0$ by $\overline{\hat \alpha}$ as described in Section \ref{sec:est_method}. To test: 
$$
H_0: \alpha_0 = 0 \ \ \ \ vs \ \ \ \ H_1: \alpha_0 \neq 0 
$$
we construct an Efron-bootstrap based confidence interval over 500 iterations. We present our finding in Table \ref{tab:summary_islamic}. From the table, it is clear that, we don't have enough evidence to reject $H_0$ at $5\%$ level as the confidence interval contains $0$. Hence we conclude that, there is no significant effect of Islamic ruling party on women's education. 

We next compare our result to that of \cite{meyersson2014islamic} and \cite{cattaneo2019practical}. \cite{meyersson2014islamic} implemented a simpler model for RDD: 
$$
Y_i = \beta_0 + \alpha_0 S_i + f(Q_i) + \eps_i \,.
$$
with $f$ being a polynomial function and only those observations were used where $Q_i \in (-h,h)$ for some optimal choice of the bandwidth $h$ (chosen according to the prescription of \cite{imbens2012optimal}). The authors found that Islamic party rule has a significantly positive effect on women's education at $1\%$ level test. On the other hand, \cite{cattaneo2019practical} implemented the model based on \eqref{eq:CCT}. We replicate their result using the R-package \emph{rdrobust} as advocated in \cite{cattaneo2019practical}. The function \texttt{rdrobust} inside the package \emph{rdrobust}, takes input $Y, Q, X$ and splits out three different types of estimate of $\alpha_0$ (along with $95\%$ confidence interval), namely \texttt{conventional, bias-corrected} and \texttt{Robust}. As mentioned in \cite{calonico2020package}, \texttt{conventional} presents when the conventional RD estimates (i.e. solving equation \eqref{eq:CCT} via local polynomial regression) with conventional standard errors, \texttt{bias-corrected} implies bias-corrected RD estimate with conventional standard errors and \texttt{robust} indicates bias-corrected RD estimates with robust standard errors (see \cite{calonico2014robust}). We use the parameters \texttt{kernel} = 'triangular', \texttt{scaleregul} = 1, \texttt{p} = 1, \texttt{bwselect} = 'mserd'  of rdrobust function to run the analysis. As evident from Table \ref{tab:summary_islamic_cck}, all these estimates reject $H_0$ at $5\%$ level stipulating a strictly positive effect of Islamic party on the education of women, while our method fails to reject the null at the same level.   

\begin{table}
\caption{Summary Statistics of data on Islamic party based on our method \label{tab:summary_islamic}}
\centering
\fbox{%
 \begin{tabular}{*{2}{l}} 
 \hline\hline
 Point Estimate & 0.4071513  \\
  \hline
 Bootstrap mean. & 0.5760144\\
  \hline
 Bootstrap s.e. & 0.48115 \\
 \hline
 Bootstrap $95\%$ C.I.  & $(-0.4234894 , 1.41942)$  \\[1ex] 
 \hline
\end{tabular}}
\end{table}

\noindent
\begin{table}
\centering
\caption{Summary Statistics of data on Islamic party based on \cite{cattaneo2019practical} \label{tab:summary_islamic_cck}}
\fbox{%
 \begin{tabular}{*{4}{l}} 
 \hline\hline
Name of methods & Coeff & CI Lower & CI Upper \\
\hline \hline
Conventional & 3.005951 & 0.9622239 &   5.049678\\
  \hline
Bias-Corrected & 3.204837 & 1.1611103 & 5.248564 \\
  \hline
Robust     & 3.204837 & 0.8266720 & 5.583003 \\
 \hline\hline
\end{tabular}}
\end{table}

\subsection{Effect of probation on subsequent GPA} 
\vspace{0.1in}
\noindent
We next analyze an educational dataset, originally collected and analyzed in \cite{lindo2010ability}\footnote{We have collected the dataset from \url{https://www.openicpsr.org/openicpsr/project/113751/version/V1/view;jsessionid=A6C09FD5CD7DB8E18EAA77B75BD893B2}.}where we investigate whether putting students on academic probation due to grades below a pre-determined cutoff has any effect on their subsequent GPA. The data are based on students from 3 independent campuses of a large Canadian university -- a major campus and other satellite campuses. The acceptance rate in the major campus is around $55\%$ and in the satellite campuses around $77\%$. The data were collected over $8$ cohorts of students till the end of the 2005 academic year. To observe the students for at least two years, only those who entered the university prior to the beginning of the 2004 academic year were considered. After being put on academic probation in their first year, some students left the university. We, therefore, only have access to GPA during the second year for those students who stayed. Thus, our $Y$ variable is the GPA of the first academic term in the second year and the treatment $S$ is whether the student was put on probation. The treatment determining variable $Q$ is the difference between the first year GPA and the cutoff for academic probation: if $Q < 0$, the student is put on academic probation, otherwise not. The covariates we consider here ($X = Z$) are presented in Table \ref{tab:covariates_LSO} of the supplementary document. 
In Table \ref{tab:summary_gpa} we summarize our result.  It is immediate from the bootstrap confidence interval from Table \ref{tab:summary_gpa} that for testing $H_0: \alpha = 0$ vs $H_1: \alpha \neq 0$, we have enough evidence to Reject $H_0$ at the $5\%$ level and conclude that the students who are put on academic probation and continue with their education, tend to improve their performance (note that the estimated $\alpha_0$ is positive) in the subsequent academic year. This makes sense, as the students who did not leave university after being put on academic probation must have a strong incentive to work harder so that they are not expelled from the university. Our findings are in harmony with those obtained in \cite{lindo2010ability}, where the author also found the treatment effect to be significant at the $5\%$ level.

\begin{remark}
\label{rem:bootstrap}
Note that in Table \ref{tab:summary_islamic} and Table \ref{tab:summary_gpa}, we presented bootstrap confidence intervals for the treatment effect instead of asymptotic confidence  intervals. This is because consistent estimation of the the asymptotic variance of our estimator is not straightforward. Recall that from Theorem \ref{thm:main_thm}, the asymptotic variance of our estimator is $e_1^{\top}\Omega_{\tau}^{-1}\Omega_{\tau}^* \Omega_{\tau}^{-1}e_1$. As mentioned in Step 1 of the sketch of the proof of Theorem \ref{thm:main_thm}, $\bW^{\top}\proj^{\perp}_{\tilde \bN_{K, a}}\bW/n$ is a consistent estimator of $\Omega_{\tau}$ but it is hard to estimate $\Omega^*_{\tau}$ consistently, which forces us to resort to the bootstrap confidence interval. 
\end{remark}

\begin{table}
\caption{Summary Statistics of data on GPA data \label{tab:summary_gpa}}
\centering
\fbox{%
 \begin{tabular}{*{2}{l}} 
 \hline\hline
 Point Estimate &  0.2733371  \\
  \hline
 Bootstrap mean. & 0.2817404 \\
  \hline
 Bootstrap s.e. & 0.016672 \\
 \hline
 Bootstrap $95\%$ C.I.  & $(0.248588, 0.3158934)$  \\[1ex] 
 \hline
\end{tabular}}
\end{table}

\section{Conclusion and possible extensions}
\label{sec:conclusion}
In this paper, we proposed a new approach to estimate an non-randomized treatment effect at the $\sqrt{n}$ rate and showed that under homoscedastic normal errors our method is semiparametrically efficient. We also pointed out in Remark \ref{rem:WLS} that one can use weighted least squares instead of ordinary least squares to take care of heterogenous errors. However, the normality assumption is necessary for semiparametric efficiency as we use least squares for estimating the treatment effect. We now discuss some natural extensions of our models that are worth analyzing as potential future research problems.   
\subsection{Non-constant treatment effect} 
\label{sec:extension_non-param}
\noindent
Consider the following extension of our model with non-constant treatment effect: 
\begin{align}
\label{eq:m1} Y_i &= \alpha(\eta_i)\mathds{1}_{Q_i > 0} + X_i^{\top}\beta_0 + \nu_i =  \alpha(\eta_i)\mathds{1}_{Q_i > 0} + X_i^{\top}\beta_0 + b(\eta_i) + \eps_i\\
\label{eq:m2} Q_i &= Z_i^{\top}\gamma_0 + \eta_i \,.
\end{align}
where as before $b(\eta_i) = \bbE[\nu \mid \eta_i]$. This generalization assumes that the response of a treated candidate with higher abilities is boosted in comparison to another treated candidate with lower ability. As an example, a more capable student upon entering into a prestigious grad school, will most likely get a better mentor, resulting in an amplification of their academic prowess. The random variable $\alpha(\eta)$ represents the \emph{conditional treatment effect} as: 
$$
\bbE[Y \mid X, \eta, Q > 0] - \bbE[Y \mid X, \eta, Q < 0] = \alpha(\eta) \,.
$$
i.e. conditioning on $(X, \eta)$ (background information and innate ability), $\alpha(\eta)$ quantifies the difference between the responses of treated and untreated samples. 
It can be shown that the estimator proposed in our manuscript based on the model with constant $\alpha_0$, estimates, in the newly proposed model, a \emph{weighted average} of $\alpha(\eta)$, i.e. $\bbE[p(\eta)\alpha(\eta)]/\bbE[p(\eta)]$ where $p(\cdot)$ is a non-negative integrable function depending on other parameters. This weight function is basically a function of the information for $\alpha$ in each $\eta$ - stratum. So our current method is applicable when the parameter of interest is a \emph{weighted average treatment effect}. 

However if the parameter of interest is \emph{unweighted average treatment effect} $\bbE[\alpha(\eta)]$, then our method can be slightly modified as follows to yield a $\sqrt{n}$-consistent estimator: 
\begin{enumerate}
\item Split the data into three (almost) equal parts say $\cD_1, \cD_2, \cD_3$. 
\item From $\cD_1$, impute $\hat \eta$ from equation \eqref{eq:m2} by regressing $Q$ on $X$. 
\item Estimate $\alpha(\cdot)$ and $b(\cdot)$ (and their derivatives) from $\cD_2$ using a B-spline series expansion.
\item Note that equation \eqref{eq:m1} can be expanded as: 
\begin{align*}
Y_i & = \alpha(\eta_i)\mathds{1}_{Q_i > 0} + X_i^{\top}\beta_0 + b(\eta_i) + \eps_i \\
& = \alpha(\hat \eta_i)\mathds{1}_{Q_i > 0} + b(\hat \eta_i) + X_i^{\top}\beta_0+ (\hat \gamma - \gamma_0)^{\top}Z_i\left(\alpha'(\hat \eta_i)\mathds{1}_{Q_i > 0} + b'(\hat \eta_i)\right) + \eps_i + R_i  \\
& \approx \alpha(\hat \eta_i)\mathds{1}_{Q_i > 0} + b(\hat \eta_i) + X_i^{\top}\beta_0+ (\hat \gamma - \gamma_0)^{\top}Z_i\left(\hat \alpha'(\hat \eta_i)\mathds{1}_{Q_i > 0} + \hat b'(\hat \eta_i)\right) + \eps_i + R_i
\end{align*}
where in the last line we replaced $\alpha'$ and $b'$ by their estimates obtained from $\cD_2$ in the previous step. Finally, we can re-estimate $\alpha$ from the above equation via a non-parametric method (i.e. expanding via a B-spline basis and regressing $Y$ on an appropriate set of covariates). Our final estimator becomes: 
$$
\hat \bbE[\alpha(\eta)] = \frac{1}{n}\sum_{i=1}^n \hat \alpha(\hat \eta_i) \,.
$$
\end{enumerate}
An analysis similar to that in our paper indicates that this estimator has $\sqrt{n}$ rate of convergence and is asymptotically normal. However whether this is semi-parametrically efficient is not clear at this moment and a potentially interesting problem for future research. 

\subsection{Bootstrap consistency}
As noted in Remark \ref{rem:bootstrap}, we use a bootstrap confidence interval instead of the asymptotic one owing to the intricate form of the asymptotic variance of our estimator, which makes it hard to estimate from the data. Therefore, an immediate question of interest is to investigate whether the bootstrap is consistent under our model assumptions. Although empirical evidence suggests that this is the case, a rigorous theoretical undertaking is essential to establish the claim.

\appendix
\section{Proof of Theorem 4.10}
\label{sec:details_hd}
\begin{proof}
For the ease of notation, we assume $\dim(X) = p_1 \asymp \dim(Z) = p_2 \asymp p$. One can extend our proof quite easily for general $p_1, p_2$ (i.e. when they are not of same order) with a careful booking for the dimension factor. The entire proof is quite long and tedious, therefore in this appendix we will state the key steps and provide proofs of the main parts. Proofs of all supplementary lemmas and propositions can be found in Appendix.
\\\\
\noindent
Recall that our estimation procedure consists of three parts: 
\begin{itemize}
\item Estimate $\gamma_0$ from $\cD_1$. 
\item Estimate $b'$ from $\cD_2$. 
\item Estimate $\alpha_0$ from $\cD_3$. 
\end{itemize}

\subsection{Estimation of $\hat \gamma_0$ from $\cD_1$} This part is the easiest among all the three parts. We estimate $\gamma_0$ by doing LASSO of $\bY$ in $\bZ$. Note that, by sub-gaussianity assumption, $\bZ$ satisfies restricted eigenvalue condition with high probability. Therefore, by standard lasso calculation with the tuning parameter $\lambda \asymp \sqrt{\log{p}/n}$ we have: 
\begin{align*}
\left\| \hat \gamma - \gamma_0 \right\|_2^2 \lesssim_\bbP \frac{s_\gamma \log{p}}{n} \,, \\
\left\| \hat \gamma - \gamma_0 \right\|_1 \lesssim_\bbP s_\gamma \sqrt{\frac{\log{p}}{n}} \,.
\end{align*} 
We will use this estimator in the subsequent analysis.

\subsection{Estimation of $b'$ from $\cD_2$}
In this subsection, we present the estimation error for $b'$.  An in our fixed dimensional analysis, we expand the model equation as: 
$$
\bY = \bX\beta_0 + \bN_k(\hat \eta)\omega_b + \bR_1 + \bR_2 + \beps \,. 
$$
where $\bN_K(\hat \eta) \in \reals^{n \times K}$ with $\bN_K(\hat \eta)_{i,j} = \tilde N_{K, j}(\hat \eta_i)$, $N_{K, j}$ being the scaled $K^{th}$ B-spline basis (see Section \ref{sec:spline_details} of the supplementary document for details). Here $\omega_b$ is the coefficient of \emph{best B-spline approximator} of $b$ with respect to $K$ basis. Our aim is to estimate $\omega_b$ from the data, as upon obtaining $\hat \omega_b$ one can define $\hat b'(t) = \nabla N_K(t)^{\top}\hat \omega_b$, where $\nabla \tilde N_K(t)$ is the vector of derivatives of B-spline basis. Therefore the estimation error $\hat b'$ can be bounded as: 
$$
\sup_{|t|\le \tau} \left|b'(t) - \hat b'(t)\right| \le \sup_{|t| \le \tau} \left\|\nabla \tilde N_k(t)\right\| \times \left\|\hat \omega_b - \omega_b\right\| + \sup_{|t| \le \tau}\left|b'(t) - \tilde N_K(t)^{\top}\omega_b\right| \,. 
$$ 
The second term is the B-spline approximation error, which is bounded by the order of $K^{-(\upsilon - 1)}$ (see Theorem \ref{thm:spline_bias} of the supplementary document). For the first term or estimation error, as mentioned in Section \ref{sec:spline_details}, we have $\sup_{|t| \le \tau} \|\nabla \tilde N_k(t)\| \lesssim K^{3/2}$. This further implies: 
$$
\sup_{|t|\le \tau} \left|b'(t) - \hat b'(t)\right| \lesssim K^{3/2}\left\|\hat \omega_b - \omega_b\right\| + K^{-(\upsilon - 1)}\,,
$$
and consequently, the estimation error $\hat b'$ completely depends on the estimation error of $\hat \omega_b$. In the proposition we present a bound on the estimation error on $\hat \omega_b$ (and consequently on $\hat b'$): 
\begin{proposition}
\label{prop:b_hat_est}
Under the Assumptions stated in Section \ref{sec:high_dim_rdd}, we have: 
$$
\left\|\hat \omega_b - \omega_b\right\| \lesssim_\bbP \sqrt{K}s_\beta \left(\sqrt{\frac{s_\gamma \log{p}}{n}} + K^{-2\upsilon}\right) + K^{-\upsilon} + \sqrt{\frac{s_\gamma \log{p}}{n}} +  \sqrt{\frac{K}{n}}
$$
and consequently: 
$$
\sup_{|t| \le \tau}\left|b'(t) - \hat b'(t)\right| \lesssim_\bbP \left(\frac{n}{s^2_\beta s_\gamma \log{p}}\right)^{\frac{-(\upsilon - \frac32)}{2\upsilon + 1}} \,.
$$
\end{proposition}

\begin{proof}[Proof of Proposition \ref{prop:b_hat_est}]
The method is, first we estimate $\beta_0$ using LASSO as follows: 
$$
\hat \beta = \argmin_\beta\left\{\frac{1}{2n} \left\|P_{\bN_k}^{\perp}\left(\bY - \bX\beta\right)\right\|^2 + \lambda \|\beta\|_1\right\}
$$
Next we estimate $\omega_b$ is: 
\begin{align}
\label{eq:def_omega_b_hat}
\hat \omega_b = \left(\frac{\bN_k^{\top}\bN_k}{n}\right)^{-1}\frac{\bN_k^{\top}\left(\bY - \bX\hat \beta\right)}{n}
\end{align}
Therefore, to establish a bound on $\|\hat \omega_b - \omega_b\|$, a bound on the estimation error $\hat \beta - \beta_0$ is necessary, which is established in the the following lemma:  
\begin{lemma}
\label{lem:beta_est_d2}
Under our assumptions, we have the the following bound on the estimation error of $\beta_0$: 
$$
\left\|\hat \beta - \beta_0\right\|_2^2 \lesssim_\bbP s_\beta \left( \sqrt{\frac{s_\gamma \log{p}}{n}} + K^{-2\upsilon}\right)^2 \lesssim_\bbP \frac{s_\beta s_\gamma \log{p}}{n} + K^{-4\upsilon} \triangleq r_{n, \beta_0}^2 \,,
$$
where $K$ is the number of B-spline basis used for extending $b$. 
\end{lemma}

\begin{proof}[Proof of Lemma \ref{lem:beta_est_d2}]
By the basic LASSO inequality, we have: 
\begin{align*}
\frac{1}{2n} \left\|P_{\bN_k}^{\perp}\left(\bY - \bX\hat \beta\right)\right\|^2 + \lambda \|\hat \beta\|_1 & \le \frac{1}{2n} \left\|P_{\bN_k}^{\perp}\left(\bY - \bX\beta_0\right)\right\|^2 + \lambda \|\beta_0\|_1
\end{align*}
Some algebraic manipulation yields: 
\begin{align*}
\frac{1}{2n}\left\|P_{\bN_k}^{\perp}\bX\left(\hat \beta - \beta_0\right)\right\|^2 + \lambda \|\hat \beta\|_1 & \le \frac{1}{n}\left\|\left(\beps + \bR\right)^{\top}P_{\bN_k}^{\perp}\bX\right\|_\infty\left\|\hat \beta - \beta_0\right\|_1 + \lambda \|\beta_0\|_1
\end{align*}
The matrix $\bbP_{\bN_K}^{\perp}\bX$ satisfies RE condition with high probability due to Assumption \ref{assm:sg} and Proposition \ref{prop:RE}. For the optimal value of $\lambda$, we need a bound on $(1/n)\|\left(\beps + \bR\right)^{\top}P_{\bN_k}^{\perp}\bX\|_\infty$, for which we bound $(1/n)\|\beps^{\top}P_{\bN_k}^{\perp}\bX\|_\infty$ and $(1/n)\|\bR^{\top}P_{\bN_k}^{\perp}\bX\|_\infty$ separately. To bound the term with $\beps$, we use the sub-gaussian concentration inequality: 
\begin{align*}
\bbP\left(\frac{1}{n}\left|\beps^{\top}P_{\bN_k}^{\perp}\bX_{*, j}\right| > t \mid \sigma(X, Z, \eta, \cD_1)\right) & \le 2\exp{\left(-C\frac{nt^2}{\frac{1}{n}\bX_{*, j}^{\top}P_{\bN_k}^{\perp}\bX_{*, j}}\right)} \\
& \le 2\exp{\left(-C\frac{nt^2}{\frac{1}{n}\bX_{*, j}^{\top}\bX_{*, j}}\right)} 
\end{align*}  
Now from the sub-gaussianity of $X_j$ (with sub-gaussianity constant $\sigma_W$) applying Lemma \ref{lem:subg_second_moment}, we have probability going to $1$: 
$$
\max_{1 \le j \le n} \frac{\|\bX_{*, j}\|^2}{n} \le 3\sigma_W \,.
$$ 
Define the above event as $\Omega_n$. Using this we have: 
\begin{align*}
\bbP\left(\frac{1}{n}\left|\beps^{\top}P_{\bN_k}^{\perp}\bX_{*, j}\right| > t \right) & \le \bbP\left(\frac{1}{n}\left|\beps^{\top}P_{\bN_k}^{\perp}\bX_{*, j}\right| > t, \Omega_n \right) + o(1) \\
& = \bbE\left[\mathds{1}_{\Omega_n}\bbP\left(\frac{1}{n}\left|\beps^{\top}P_{\bN_k}^{\perp}\bX_{*, j}\right| > t \mid \sigma(X, Z, \eta, \cD_1)\right)\right] + o(1)\\
& \le 2\bbE\left[\mathds{1}_{\Omega_n}\exp{\left(-C\frac{nt^2}{\frac{1}{n}\bX_{*, j}^{\top}\bX_{*, j}}\right)} \right] + o(1) \\
& \le 2\exp{\left(-C\frac{nt^2}{3\sigma_W^2}\right)} + o(1) 
\end{align*}
Therefore an appropriate choice of $t$ yields: 
\begin{align}
\label{eq:final_bound_eps_d2} \frac{1}{n}\left\|\beps^{\top}P_{\bN_K}^{\top}\bX\right\|_\infty \lesssim_\bbP \sqrt{\frac{\log{p}}{n}} \,.
\end{align}
Now for the remainder term, recall that the remainder term $\bR$ consists of $\bR = \bR_1 + \bR_2$ where
\begin{align*}
\bR_{1, i} & = b(\eta_i) - b(\hat \eta_i) \\
\bR_{2, i} & = b(\hat \eta_i) - \bN_K(\hat \eta_i)^{\top}\omega_b \,.
\end{align*}
We now bound these two remainder terms separately. For the first remainder term: 
\begin{align}
\frac{1}{n}\left\|\bR_1^{\top}P_{\bN_K}^{\perp}\bX\right\|_\infty & = \max_{1 \le j \le p}\frac{1}{n}\left|\bR_1^{\top}P_{\bN_K}^{\perp}\bX_{*, j}\right| \notag \\
\label{eq:bound_rem_1_d2} & \le \sqrt{\frac{1}{n}\bR_1^{\top}\bR} \times \sqrt{\max_{1 \le j \le p}\frac{1}{n}\|\bX_{*, j}\|^2}
\end{align}
Again it follows from Lemma \ref{lem:subg_second_moment} that the second term of the above inequality is $O_p(1)$ based on the subgaussianity of $\bX_{*, j}$. For the first term, note that as $b$ is Lipschitz (Assumption \ref{assm:smooth_b}), we have: 
\begin{align*}
\frac{1}{n}\sum_{i=1}^n \bR_{1, i}^2 & = \frac{1}{n}\sum_{i=1}^n \left(b(\hat \eta_i) - b(\eta_i)\right)^2 \\
& \lesssim \frac{1}{n}\sum_{i=1}^n \left(\hat \eta_i - \eta_i\right)^2 \hspace{0.2in}[\text{As }b\text{ is Lipschitz}]\\
& \le \left\|\hat \gamma_n - \gamma_0\right\|^2 \times \frac{1}{n}\sum_{i=1}^n \left(Z_i^{\top}a_n\right)^2 \\
& \lesssim_\bbP \frac{s_\gamma \log{p}}{n} \,.
\end{align*}
where the last line follows from the fact that $\|\hat \gamma_n - \gamma_0\|^2 \lesssim_\bbP (s_\gamma \log{p})/n$ (from LASSO on $\cD_1$) and $(1/n)\sum_{i=1}^n \left(Z_i^{\top}a_n\right)^2 \lesssim_\bbP 1$ follows from subgaussianity of $Z_i$ along with Lemma \ref{lem:subg_second_moment}.  Using this bound in equation \eqref{eq:bound_rem_1_d2} we conclude: 
\begin{align}
\label{eq:final_bound_rem_1_d2} \frac{1}{n}\left\|\bR_1^{\top}P_{\bN_K}^{\perp}\bX\right\|_\infty & \lesssim_\bbP \sqrt{\frac{s_\gamma \log{p}}{n}} \,.
\end{align}
For the other remainder term, we expand $\bX_{*, j}$ as: 
$$
\bX_{*, j} = \bX_{*, j} - \check m_j(\hat{\mathbf{\eta}}) + \check m_j(\hat{\mathbf{\eta}}) =  \bX_{*, j} - \check m_j(\hat{\mathbf{\eta}}) + \bN_{K}(\hat \eta)w_j + \bR_j 
$$
This implies: 
$$
\bR_2^{\top}P_{\bN_K}^{\perp}\bX_{*, j} = \bR_2^{\top}P_{\bN_K}^{\perp}\left(\bX_{*, j} - \check m_j(\hat{\mathbf{\eta}})\right) + \bR_2^{\top}P_{\bN_K}^{\perp}\bR_j  \,.
$$
As both $\bR_2$ and $P_{\bN_K}^{\perp}$ are function of $\hat \eta$, we have from the sub-gaussian concentration bound: 
\begin{align*}
\bbP\left(\frac{1}{n}\left| \bR_2^{\top}P_{\bN_K}^{\perp}\left(\bX_{*, j} - \check m_j(\hat{\mathbf{\eta}})\right)\right| > t \mid \hat \eta\right) & \le 2\exp{\left(-C\frac{nt^2}{\frac{1}{n}\bR_2^{\top}P_{\bN_K}^{\perp}\bR_2}\right)} \\
& \le 2\exp{\left(-C\frac{nt^2}{\frac{1}{n}\bR_2^{\top}\bR_2}\right)}  \,.
\end{align*}
From the spline approximation error, $(1/n)\bR_2^{\top}\bR_2 \le K^{-2\upsilon}$ (where $K$ is the number of basis and $\upsilon$ is the smoothness index of $b$ (see Assumption \ref{assm:smooth_b}). Using this we have: 
$$
\bbP\left(\frac{1}{n}\left| \bR_2^{\top}P_{\bN_K}^{\perp}\left(\bX_{*, j} - \check m_j(\hat{\mathbf{\eta}})\right)\right| > t \mid \hat \eta\right) \le 2\exp{\left(-C\frac{nt^2}{K^{-2\upsilon}}\right)}
$$
which, upon unconditioning, taking union bound and choosing a suitable value of $t$ yields: 
$$
\max_{1 \le j \le p}\frac{1}{n}\left|\bR_2^{\top}P_{\bN_K}^{\perp}\left(\bX_{*, j} - \check m_j(\hat{\mathbf{\eta}})\right)\right| \lesssim_\bbP K^{-\upsilon}\sqrt{\frac{\log{p}}{n}} \,.
$$
For the other remainder term, i.e. $(1/n)\bR_2^{\top}P_{\bN_K}^{\perp}\bR_j$, an application of Cauchy-Schwarz inequality yields: 
$$
\frac{1}{n}\left|\bR_2^{\top}P_{\bN_K}^{\perp}\bR_j\right| \le \sqrt{\frac{1}{n}\bR_2^{\top}\bR_2} \times \sqrt{\frac{1}{n}\bR_j^{\top}\bR_j} \le K^{-2\upsilon} \,.
$$
Therefore we have: 
\begin{align}
\label{eq:final_bound_rem_2_d2}
\frac{1}{n}\left\|\bR_2^{\top}P_{\bN_K}^{\perp}\bX\right\|_\infty \lesssim_\bbP K^{-\upsilon}\sqrt{\frac{\log{p}}{n}} + K^{-2\upsilon} \,.
\end{align}
Combining the bounds in equation \eqref{eq:final_bound_eps_d2}, \eqref{eq:final_bound_rem_2_d2} and \eqref{eq:final_bound_rem_2_d2} we have: 
\begin{align}
\label{eq:final_bound_rem_d2}
\frac{1}{n}\left\|\left(\beps + \bR\right)^{\top}P_{\bN_K}^{\perp}\bX\right\|_{\infty} \lesssim_\bbP \sqrt{\frac{s_\gamma \log{p}}{n}} + K^{-2\upsilon} \asymp \lambda \,.
\end{align}
With this choice of $\lambda$, standard LASSO analysis completes the proof. 
\end{proof}
\noindent
Going back to the definition of $\hat \omega_b$ (equation \eqref{eq:def_omega_b_hat}), we expand the its estimation error as follows: 
\begin{align*}
\hat \omega_b - \omega_b & = \underbrace{\left(\frac{\bN_k^{\top}\bN_k}{n}\right)^{-1}\frac{\bN_k^{\top}\bX\left(\hat \beta - \beta_0\right)}{n}}_{T_1} + \underbrace{\left(\frac{\bN_k^{\top}\bN_k}{n}\right)^{-1}\frac{\bN_k^{\top}(\bR_1 + \bR_2)}{n}}_{T_2} \\
& \qquad \qquad \qquad \qquad \qquad \qquad \qquad + \underbrace{\left(\frac{\bN_k^{\top}\bN_k}{n}\right)^{-1}\frac{\bN_k^{\top}\beps}{n}}_{T_3}
\end{align*}
We next bound each $T_i$ separately. Same argument as for the fixed dimensional analysis (See the proof of Proposition \ref{thm:spline_consistency}) establishes: 
$$
\left\|\left(\frac{\bN_k^{\top}\bN_k}{n}\right)^{-1}\right\|_{op}  \lesssim_\bbP 1 \,.
$$
Therefore we can bound $T_1$ as: 
\begin{align*}
T_1 \lesssim_\bbP \left\|\frac{\bN_k^{\top}\bX\left(\hat \beta - \beta_0\right)}{n}\right\|_2 & \le  \frac{\sqrt{K}}{n} \max_{1 \le j \le k} \left|\bN_{K, *j}^{\top}\bX\left(\hat \beta - \beta_0\right)\right| \\
& \le  \frac{\sqrt{K}}{n} \left\|\hat \beta - \beta_0\right\|_1 \times  \max_{1 \le j \le k} \max_{1 \le l \le p} \left|\bN_{K, *j}^{\top}\bX_{*, l}\right| \\
& \lesssim_\bbP \sqrt{K} \left\|\hat \beta - \beta_0\right\|_1 \lesssim_\bbP \sqrt{K}s_\beta \left(\sqrt{\frac{s_\gamma \log{p}}{n}} + K^{-2\upsilon}\right)\,.
\end{align*}
For $T_2$ and $T_3$, the term containing residuals: 
\begin{align*}
T_2  \le \left(\frac{\bN_k^{\top}\bN_k}{n}\right)^{-1}\frac{\bN_k^{\top}(\bR_1 + \bR_2)}{n} & \le \left\|\left(\frac{\bN_k^{\top}\bN_k}{n}\right)^{-1/2}\right\|_{op}\left\|\frac{\bR_1 + \bR_2}{\sqrt{n}}\right\|_2  \\
& \lesssim_\bbP \left\|\frac{\bR_1 + \bR_2}{\sqrt{n}}\right\|_2   \\
& \le \sqrt{\frac{1}{n}\bR_1^{\top}\bR_1} + \sqrt{\frac{1}{n}\bR_2^{\top}\bR_2}  \\
& \lesssim_\bbP \left(K^{-\upsilon} + \sqrt{\frac{s_\gamma \log{p}}{n}}\right) \,.
\end{align*}
where the last rate inequality follows from the bounds on $\bR_1$ and $\bR_2$ established in the proof of Lemma \ref{lem:beta_est_d2}. For the error term we use the sub-gaussian bound using the fact that $\bbE[\eps \mid X, Z, \eta] = 0$. For this, note that for any vector $z \in \bbR^K$ we have: 
$$
\|z\| = \sup_{\|v\| = 1} |z^{\top}v| \,.
$$
Define $\cN_{1/2}(S^{K-1}$  to be $1/2$-covering number of the sphere in dimension $K$. Then we know $\cN_{1/2}(S^{K-1} \le 5^K$. Also we have: 
$$
\|z\| \le \sup_{v_1 \in \cN_{1/2}(S^{K-1})}|z^{\top}v_1| \le \frac12 \|z\| 
$$
which implies: 
$$
\|z\| \le 2 \sup_{v_1 \in \cN_{1/2}(S^{K-1})}|z^{\top}v_1| \,.
$$
Using this we have: 
\begin{align*}
& \bbP\left(\sup_{v_1 \in \cN_{1/2}(S^{K-1})}\left|v_1^{\top}\left(\bN_K^{\top}\bN_K\right)^{-1}\bN_K^{\top}\beps\right| > t\right) \\
& = \bbE_{\hat \eta}\left[\bbP\left(\sup_{v_1 \in \cN_{1/2}(S^{K-1})}\left|v_1^{\top}\left(\bN_K^{\top}\bN_K\right)^{-1}\bN_K^{\top}\beps\right| > t \mid \hat \eta \right) \right] \\
& \le \bbE_{\hat \eta}\left[\sum_{v_1 \in  \cN_{1/2}(S^{K-1})}\bbP\left(\left|v_1^{\top}\left(\bN_K^{\top}\bN_K\right)^{-1}\bN_K^{\top}\beps\right| > t \mid \hat \eta \right) \right] \\
& \le  \bbE_{\hat \eta}\left[\sum_{v_1 \in  \cN_{1/2}(S^{K-1})}2\exp{\left(-c\frac{nt^2}{v_1^{\top}\left(\frac{\bN_k^{\top}\bN_k}{n}\right)^{-1}v_1}\right)} \right]  \\
& \le \bbE_{\hat \eta}\left[2\exp{\left(K\log{5} -c\frac{nt^2}{\left\|\left(\frac{\bN_k^{\top}\bN_k}{n}\right)^{-1}\right\|_{op}}\right)}\right]
\end{align*}
This implies: 
$$
T_3 = \left(\frac{\bN_k^{\top}\bN_k}{n}\right)^{-1}\frac{\bN_k^{\top}\beps}{n} \lesssim_\bbP \sqrt{\frac{K}{n}} \,.
$$
Combining the bounds on $T_1, T_2, T_3$ we have: 
$$
\left\|\hat \omega_b - \omega_b\right\| \lesssim_\bbP \sqrt{K}s_\beta \left(\sqrt{\frac{s_\gamma \log{p}}{n}} + K^{-2\upsilon}\right) + K^{-\upsilon} + \sqrt{\frac{s_\gamma \log{p}}{n}} +  \sqrt{\frac{K}{n}}  \,.
$$
The above bound on the estimation error on $\omega_b$ translates to the estimation error of $b'$ as follows: 
\begin{align*}
\left|b'(x) - \hat b'(x)\right| & \le \left\|\nabla N_K(x)\right\|\left\|\hat \omega_b - \omega_b\right\| + K^{-(\upsilon - 1)} \\
& \lesssim_\bbP K^{3/2}\left[\sqrt{K}s_\beta \left(\sqrt{\frac{s_\gamma \log{p}}{n}} + K^{-2\upsilon}\right) + K^{-\upsilon} + \sqrt{\frac{s_\gamma \log{p}}{n}} +  \sqrt{\frac{K}{n}}\right] + K^{-(\upsilon - 1)} \\
& \lesssim_\bbP  K^2 s_\beta \sqrt{\frac{s_\gamma \log{p}}{n}} + K^{-(\upsilon - \frac{3}{2})} + K^{\frac{3}{2}}\sqrt{\frac{s_\gamma \log{p}}{n}} +  \frac{K^2}{\sqrt{n}} \\
& \lesssim_\bbP K^2s_\beta \sqrt{\frac{s_\gamma \log{p}}{n}} + K^{-\left(\upsilon - \frac32\right)} \,.
\end{align*}
Hence an optimal choice of $K$ would satisfy: 
\begin{align*}
K^2s_\beta \sqrt{\frac{s_\gamma \log{p}}{n}} \asymp K^{-\left(\upsilon - \frac32\right)} \implies \ & s_\beta \sqrt{\frac{s_\gamma \log{p}}{n}} \asymp K^{-(\upsilon + \frac12)} \\
\implies \  & \frac{s^2_\beta s_\gamma \log{p}}{n} \asymp K^{-(2\upsilon + 1)} \\
\implies & K \asymp \left(\frac{n}{s^2_\beta s_\gamma \log{p}}\right)^{\frac{1}{2\upsilon + 1}}
\end{align*}
Using this we conclude: 
$$
\sup_{|t| \le \tau} \left|b'(t) - \hat b'(t)\right| \lesssim_\bbP \left(\frac{n}{s^2_\beta s_\gamma \log{p}}\right)^{\frac{-(\upsilon - \frac32)}{2\upsilon + 1}} \,.
$$
This completes the proof. 
\end{proof}

\subsection{Estimation of $\alpha_0$ from $\cD_3$}
\label{sec:third_step_hd}
For notational simplicity, we here use $r_n (=K^{-\upsilon})$ to denote the B-spline approximation of $m_j$ and $\check m_j$ (see the definitions in Section \ref{sec:high_dim_rdd}) and $\dot r_n$ to denote the estimation error of $\hat b'$ obtained in Proposition \ref{prop:b_hat_est}, i.e. we write:  
\begin{enumerate}
    \item $\sup_{|t| \le \tau}\left|b'(t) - \hat b'(t)\right| \lesssim_\bbP \dot r_n$. 
    \item For $0 \le j \le 1+p_1+p_2$, $\sup_{|t| \le \tau}\left|m_j(t) - \bN_k(t)^{\top}\omega_j\right| \lesssim_\bbP r_n$. 
    \item For $0 \le j \le 1+p_1+p_2$, $\sup_{|t| \le \tau}\left|\check m_j(t) - \bN_k(t)^{\top}\check{\omega}_j\right| \lesssim_\bbP r_n$
\end{enumerate}
where $\omega_j$ and $\check \omega_j$ are the optimal projection vectors of $m_j$ and $\check m_j$ respectively on the space spanned by $K$ B-spline basis with respect to $\ell_\infty$ norm. Henceforth, we will work on the intersection of these events. 
First, Consider the LASSO regression of $\bS^{\perp}$ on $\breve \bW_{-1}^{\perp}$. By basic inequality we have: 
\begin{align*}
    & \frac{1}{2n}\left\|\bS^{\perp} - \breve \bW^{\perp}_{-1}\hat \theta_{-1, S}\right\|^2 + \lambda_1 \|\hat \theta_{-1, S}\|_1 \le \frac{1}{2n}\left\|\bS^{\perp} - \breve \bW^{\perp}_{-1} \theta^*_{S}\right\|^2 + \lambda_1 \|\theta^*_{S}\|_1
\end{align*}
Few algebraic manipulations (similar to that of standard LASSO analysis) yields: 
\begin{align}
\label{eq:s_lasso}
        & \frac{1}{2n}\left\|\breve \bW^{\perp}_{-1}\left(\hat \theta_{-1, S} - \theta^*_{S}\right)\right\|^2 + \lambda_1 \|\hat \theta_{-1, S}\|_1 \notag \\
    & \qquad \qquad \qquad \qquad \le \frac1n \left\|\left(\bS^{\perp} - \breve \bW^{\perp}_{-1}\theta^*_{S}\right)^{\top}\breve \bW^{\perp}_{-1}\right\|_\infty \left\|\hat \theta_{-1, S} - \theta^*_{S}\right\|_1 + \lambda_1 \|\theta^*_{S}\|_1
\end{align}
To find the optimal value of $\lambda_1$, we need to bound $(1/n) \|(\bS^{\perp} - \breve \bW^{\perp}_{-1}\theta^*_{S})^{\top}\breve \bW^{\perp}_{-1}\|_\infty$. Similarly, for the LASSO regression of $\bY$ on $\breve \bW_{-1}$, we have:
\begin{align}
\label{eq:y_lasso}
    & \frac{1}{2n}\left\|\breve \bW^{\perp}_{-1}\left(\hat \theta_{-1,Y} - \theta^*_Y\right)\right\|^2 + \lambda_0 \|\hat \theta_{-1,Y}\|_1 \notag \\
    & \qquad \qquad \qquad \qquad \le \frac1n \left\|\left(\bY^{\perp} - \breve \bW^{\perp}_{-1}\theta^*_Y\right)^{\top}\breve \bW^{\perp}_{-1}\right\|_\infty \left\|\hat \theta_{-1,Y} - \theta^*_Y\right\|_1 + \lambda_0 \|\theta^*_Y\|_1
\end{align}
and to obtain $\lambda_0$, we need to bound $(1/n)\|(\bY^{\perp} - \breve \bW^{\perp}_{-1}\theta^*_Y)^{\top}\breve \bW^{\perp}_{-1}\|_\infty$. Towards that direction, we need the following lemma: 
\begin{lemma}
\label{lem:sq_rate}
Define the random variable $\widehat{\tilde W}$ (and consequently the matrix $\widehat{\tilde \bW}$) and $W - \bbE[W \mid \hat \eta]$. Under our assumptions, we have for all $1 \le l \le 1+p_1+p_2$ we have with probability going to $1$:
$$
\max_{1 \le j \le 1+p_1+p_2}\frac{1}{n} \left\|\breve \bW^{\perp}_{*, j} - \widehat{\tilde \bW}_{*, j}\right\|^2 \lesssim_\bbP  \ \dot r_n^2 + \frac{s_\gamma \log{p}}{n} +\dot r_n \sqrt{\frac{\log{p}}{n}}
$$
$$
\max_{1 \le j \le 1+p_1+p_2}\frac{1}{n} \left\|\breve \bW^{\perp}_{*, j} - \tilde \bW_{*, j}\right\|^2 \lesssim_\bbP  \ \dot r_n^2 + \frac{s_\gamma \log{p}}{n} +\dot r_n \sqrt{\frac{\log{p}}{n}}\,,
$$
and 
$$
\max_{1 \le j, j' \le p} \frac{1}{n}\left|\left(\breve \bW_{*, j}^{\perp}\right)^{\top}\left(\breve \bW^{\perp}_{*, j'} - \tilde \bW_{*, j'}\right)\right| \lesssim_\bbP \ \dot r_n +  \sqrt{\frac{s_\gamma \log{p}}{n}} \,.
$$
\end{lemma}
\noindent
The proof of this lemma can be found in the Supplementary document. An immediate consequence of Lemma \ref{lem:sq_rate} is the following bound which will be used subsequently in this proof: 
\begin{align}
    & \max_{1 \le j \neq j' \le p} \frac{1}{n}\left|\left(\breve \bW^{\perp}_{*, j} - \tilde \bW_{*, j}\right)^{\top}\left(\breve \bW^{\perp}_{*, j'} - \tilde \bW_{*, j'}\right)\right| \notag  \\ 
    & \le \sqrt{\max_{1 \le j \le p}\frac{1}{n} \left\|\breve \bW^{\perp}_{*, j} - \tilde \bW_{*, j}\right\|^2} \times \sqrt{\max_{1 \le j' \le p}\frac{1}{n} \left\|\breve \bW^{\perp}_{*, j'} - \tilde \bW_{*, j'}\right\|^2} \notag \\
    \label{eq:cross_term} & \lesssim_\bbP  \ \dot r_n^2 + \frac{s_\gamma \log{p}}{n} + \dot r_n \sqrt{\frac{\log{p}}{n}} \,.
\end{align}
\newline
\noindent
Based on the bound obtained in Lemma \ref{lem:sq_rate}, the optimal choices for $\lambda_0$ and $\lambda_1$ are following:  
\begin{lemma}
\label{lem:choice_lambda}
Under our assumptions, we can choose $\lambda_0$ and $\lambda_1$ as:
$$
\lambda_1 \asymp  (1 + \|\theta^*_S\|_1)\left[\dot r_n +  \sqrt{\frac{s_\gamma \log{p}}{n}}\right] + \sqrt{\frac{s_\gamma \log{p}}{n}}\left(\log{\frac{n}{s_\gamma \log{p}}}\right)^{3/2}\,,
$$
$$
\lambda_0 \asymp   (1 + \|\theta^*_Y\|_1)\left[\dot r_n +  \sqrt{\frac{s_\gamma \log{p}}{n}}\right] + \sqrt{\frac{s_\gamma \log{p}}{n}}\left(\log{\frac{n}{s_\gamma \log{p}}}\right)^{3/2}  \,.
$$
where $\theta^*_S, \theta^*_Y$ are same as defined in Section \ref{sec:high_dim_rdd}. 
\end{lemma}
\noindent
The proof of Lemma \ref{lem:choice_lambda} is also presented in the Supplementary document. Another important ingredient in obtaining this LASSO-type bounds is the restricted eigenvalue assumption on the covariate matrix $\breve \bW^{\perp}$ which is presented in the next Proposition (proof is in Supplementary document): 
\begin{proposition}
\label{prop:RE}
Under our assumptions in Section \ref{sec:high_dim_rdd}, the matrix $\breve \bW^{\perp}_{-1}$ satisfies RE condition with high probability. 
\end{proposition}
We next obtain the estimation error of $\hat \theta_{-1, Y}$ and $\hat \theta_{-1, S}$ combining our findings from Lemma \ref{lem:choice_lambda} and Proposition \ref{prop:RE}. From \eqref{eq:s_lasso} we have: 
\begin{align*}
        & \frac{1}{2n}\left\|\breve \bW^{\perp}_{-1}\left(\hat \theta_{-1, S} - \theta^*_{S}\right)\right\|^2 + \lambda_1 \|\hat \theta_{-1, S}\|_1 \le \frac{\lambda_1}{2} \left\|\hat \theta_{-1, S} - \theta^*_{S}\right\|_1 + \lambda_1 \|\theta^*_{S}\|_1
\end{align*}
Therefore by the choice of $\lambda_1$ of Lemma \ref{lem:choice_lambda} and Proposition \ref{prop:RE} we have (via standard LASSO bound calculation): 
\begin{align*}
& \left\|\hat \theta_{-1, S} - \theta^*_{S}\right\|_1 \lesssim_\bbP \lambda_1 s_1 \,, \\
& \left\|\hat \theta_{-1, S} - \theta^*_{S}\right\|^2_2 \lesssim_\bbP \lambda^2_1 s_1 \,.
\end{align*}
Similarly, for $\hat \theta_{-1, Y}$, we have from equation \eqref{eq:y_lasso}: 
\begin{align*}
        & \frac{1}{2n}\left\|\breve \bW^{\perp}_{-1}\left(\hat \theta_{-1,Y} - \theta^*_{Y}\right)\right\|^2 + \lambda_0 \|\hat \theta_{-1,Y}\|_1 \le \frac{\lambda_0}{2} \left\|\hat \theta_{-1,Y} - \theta^*_{Y}\right\|_1 + \lambda_0 \|\theta^*_{Y}\|_1
\end{align*}
As before, the value of $\lambda_0$ from Lemma \ref{lem:choice_lambda} and Proposition \ref{prop:RE} yields:  
\begin{align*}
& \left\|\hat \theta_{-1, Y} - \left(\theta^*_{S}\alpha_0 + \theta^*_{0, -1}\right)\right\|_1 \lesssim_p \lambda_0 s_0 \,,  \\
    & \left\|\hat \theta_{-1, Y} - \left(\theta^*_{S}\alpha_0 + \theta^*_{0, -1}\right)\right\|^2_2 \lesssim_p \lambda_0^2 s_0 \,.
\end{align*}
From the above rates of the lasso estimates $\hat \theta_{-1, S}$ and $\hat \theta_{-1, Y}$ we can further conclude: 
\begin{align}
\left\|\theta_{0, -1} + \hat \theta_{-1, S}\alpha_0 - \hat \theta_{-1, Y}\right\|_1 & \le \left\|\theta_{0, -1} + \theta^*_{S}\alpha_0 - \hat \theta_{-1, Y}\right\|_1 + |\alpha_0|\left\|\theta^*_{S} - \hat \theta_{-1, S}\right\|_1 \notag \\
& \lesssim_\bbP |\alpha_0|\lambda_1 s_1 + \left\|\theta^*_{0, -1} + \hat \theta_{-1, S}\alpha_0 - \hat \theta_{-1, Y}\right\|_1 + \left\|\theta_{0, -1} - \theta^*_{0, -1}\right\|_1\notag  \\
\label{eq:extra_bound_1} & \lesssim_\bbP \lambda_0 s_0 + |\alpha_0|\lambda_1 s_1 + s_\gamma \sqrt{\frac{\log{p}}{n}}
\end{align}
and similarly: 
\begin{align}
    \label{eq:extra_bound_2} \left\|\theta^0_{-1} + \hat \theta_{-1, S}\alpha_0 - \hat \theta_{-1, Y}\right\|^2_2 & \lesssim_p \lambda^2_0 s_0 + \alpha^2_0 \lambda^2_1 s_1 + \frac{s_\gamma \log{p}}{n}\,.
\end{align}

%
\noindent

\noindent
Going back to equation \eqref{eq:debias_1}, we next show in the following Proposition that the common denominator of the three terms in the RHS stabilizes: 
\begin{proposition}
\label{prop:denom_conv}
Under our assumptions: 
$$
\frac{1}{n_3}\left\|\bS^{\perp} - \breve \bW^{\perp}\hat \theta_{-1, S}\right\|^2 = \bbE\left[\left(\tilde S - \tilde W_{-1}^{\top}\theta^*_{S}\right)^2 \mathds{1}_{|\eta| \le \tau}\right] + o_p(1) 
$$
\end{proposition}
\noindent
The proof of Proposition \ref{prop:denom_conv} can be found in the Appendix. We next show that the numerators of the first and the second term of the RHS of equation \eqref{eq:debias_1} are asymptotically negligible and the numerator of the third term contributes to asymptotic normality. 
\vspace{0.2in}
\\
{\bf Numerator of first term: }We start with the first term, which is $\sqrt{n_3}$ times: 
$$
\frac{1}{n_3}\left(\theta_{0, -1} + \hat \theta_{-1, S}\alpha_0 - \hat \theta_{-1, Y} \right)^{\top}\breve \bW_{-1}^{\perp^{\top}}\left(\bS^{\perp} - \breve \bW^{\perp}_{-1}\hat \theta_{-1, S}\right)
$$
We first expand it as follows: 
\begin{align*}
    & \frac{1}{n}\left(\theta_{0, -1} + \hat \theta_{-1, S}\alpha_0 - \hat \theta_{-1, Y} \right)^{\top}\breve \bW_{-1}^{\perp^{\top}}\left(\bS^{\perp} - \breve \bW^{\perp}_{-1}\hat \theta_{-1, S}\right) \\
    & \qquad \qquad = \underbrace{\frac{1}{n}\left(\theta_{0, -1} + \hat \theta_{-1, S}\alpha_0 - \hat \theta_{-1, Y} \right)^{\top}\breve \bW_{-1}^{\perp^{\top}}\left(\bS^{\perp} - \breve \bW^{\perp}_{-1}\theta^*_{S}\right)}_{T_1} \\
    & \qquad \qquad \qquad + \underbrace{\frac{1}{n}\left(\theta_{0, -1} + \hat \theta_{-1, S}\alpha_0 - \hat \theta_{-1, Y} \right)^{\top}\breve \bW_{-1}^{\perp^{\top}}\breve \bW_{-1}^{\perp}\left(\hat \theta_{-1, S} - \theta^*_{S}\right)}_{T_2} \,.
\end{align*}
To bound $T_1$ we use $\ell_1 - \ell_\infty$ bound: 
\begin{align*}
    & \frac{1}{n}\left(\theta_{0, -1} + \hat \theta_{-1, S}\alpha_0 - \hat \theta_{-1, Y} \right)^{\top}\breve \bW_{-1}^{\perp^{\top}}\left(\bS^{\perp} - \breve \bW^{\perp}_{-1}\theta^*_{S}\right) \\
    & \qquad \qquad \le \left\|\theta_{0, -1} + \hat \theta_{-1, S}\alpha_0 - \hat \theta_{-1, Y} \right\|_1 \times \frac1n \left\|\breve \bW_{-1}^{\top}\left(\bS^{\perp} - \breve \bW^{\perp}_{-1}\theta^*_{S}\right)\right\|_{\infty}
\end{align*}
We have already established bound on the first part in equation \eqref{eq:extra_bound_1} and the second part is bounded by $\lambda_1$ (see the proof of Lemma \ref{lem:choice_lambda}). Therefore we have with probability going to 1: 
\begin{align*}
T_1 & \lesssim_\bbP \lambda_1 \left(\lambda_0 s_0 + |\alpha_0| \lambda_1 s_1 + s_\gamma \sqrt{\frac{\log{p}}{n}}\right) \\
& \lesssim_\bbP \left(\lambda_0 \vee \lambda_1\right)^2 (s_0 \vee s_1) + \lambda_1 s_\gamma\sqrt{\frac{\log{p}}{n}} \,. 
\end{align*}
For $T_2$ we can use CS inequality to conclude:
\begin{align*}
    T_2 & \le \underbrace{\sqrt{\frac{1}{n}\left\| \breve \bW_{-1}^{\perp}\left(\hat \theta_{-1, S} - \theta^*_{S}\right)\right\|^2}}_{T_{21}} \times \underbrace{\sqrt{\frac{1}{n}\left\|\breve \bW_{-1}^{\perp}\left(\theta_{0, -1} + \hat \theta_{-1, S}\alpha_0 - \hat \theta_{-1, Y} \right)\right\|^2}}_{T_{22}}  
\end{align*}
The first term $T_{21}$ is the prediction error of Lasso when we regress $\bS^{\perp}$ on $\breve \bW^{\perp}_1$. Therefore from standard Lasso prediction error bound we have: 
$$
T_{21} \lesssim_\bbP \lambda_1 \sqrt{s_1} \,,
$$
For $T_{22}$ we need a bit more detailed calculation. First of all note that, $T_{22}$ can be further bounded as: 
\allowdisplaybreaks
\begin{align*}
T^2_{22} & = \frac{1}{n}\left\|\breve \bW_{-1}^{\perp}\left(\theta_{0, -1} + \hat \theta_{-1, S}\alpha_0 - \hat \theta_{-1, Y} \right)\right\|^2 \\
& \lesssim  \frac{1}{n}\left\|\breve \bW_{-1}^{\perp}\left(\theta_{0, -1} + \theta^0_{-1,1}\alpha_0 - \hat \theta_{-1, Y} \right)\right\|^2 +  \frac{\alpha_0^2 }{n}\left\|\breve \bW_{-1}^{\perp}\left(\hat \theta_{-1, S} - \theta^0_{-1, 1}\right)\right\|^2 \\
& \lesssim_\bbP \frac{1}{n}\left\|\breve \bW_{-1}^{\perp}\left(\theta^*_{0,-1} + \theta^0_{-1,1}\alpha_0 - \hat \theta_{-1, Y} \right)\right\|^2 +  \frac{1}{n}\left\|\breve \bW_{-1}^{\perp}\left(\theta^*_{0, -1} - \theta^0_{-1}\right)\right\|^2 + \lambda_1^2s_1 \\
& \lesssim_\bbP \lambda_0^2s_0 + \lambda_1^2s_1 + \frac1n \left\|P_{N_k}^{\perp}\tilde \bZ\left(\hat \gamma_n - \gamma_0\right)\right\|^2 \\
& \lesssim_\bbP \lambda_0^2s_0 + \lambda_1^2s_1 + \frac{s_\gamma \log{p}}{n}  \times \frac1n \sum_{i=1}^n \left(\tilde Z_i^{\top}a_n\right)^2 \\
& = \lambda_0^2s_0 + \lambda_1^2s_1 + \frac{s_\gamma \log{p}}{n}  \times \frac1n \sum_{i=1}^n \left(\hat b'(\hat \eta_i)\right)^2\left(Z_i^{\top}a_n\right)^2 \\
& = \lambda_0^2s_0 + \lambda_1^2s_1 + \frac{s_\gamma \log{p}}{n}  \times \left[\frac1n \sum_{i=1}^n \left(b'(\hat \eta_i)\right)^2\left(Z_i^{\top}a_n\right)^2 + \frac1n \sum_{i=1}^n \left(\hat b'(\hat \eta_i) - b'(\hat \eta_i)\right)^2\left(Z_i^{\top}a_n\right)^2\right] \\
& \lesssim_\bbP \lambda_0^2s_0 + \lambda_1^2s_1 + \frac{s_\gamma \log{p}}{n}  \times \left[1 + \dot  r_n\right] \,.
\end{align*}
where in the last inequality, we use the following facts: 
\begin{align*}
b'(\hat \eta_i) & \lesssim 1 \,, \hspace{0.2in} [\text{As }b' \text{ is bounded, Assumption \ref{assm:smooth_b}}] \\
\frac1n \sum_{i=1}^n \left(b'(\hat \eta_i)\right)^2\left(Z_i^{\top}a_n\right)^2 & \lesssim_\bbP 1 \,, \hspace{0.2in} [b' \text{ is bounded and Lemma \ref{lem:subg_second_moment}}] \\
\left\|\hat b' - b'\right\|_\infty & \lesssim_\bbP r'_n \hspace{0.2in} [\text{Proposition } \ref{prop:b_hat_est}] \,.
\end{align*}
Therefore we have: 
$$
T_{22} \lesssim_\bbP \lambda_0 \sqrt{s_0} + \lambda_1\sqrt{s_1} + \sqrt{\frac{s_\gamma \log{p}}{n}} \,.
$$
Taking products of the bounds on $T_{21}$ and $T_{22}$ we have: 
\begin{align*}
T_2 \le T_{21} \times T_{22} & \lesssim_\bbP \lambda_1\sqrt{s_1} \left(\lambda_0 \sqrt{s_0} + \lambda_1\sqrt{s_1} + \sqrt{\frac{s_\gamma \log{p}}{n}}\right) \\
& \lesssim_\bbP \left(\lambda_0 \vee \lambda_1\right)^2(s_0 \vee s_1) + \lambda_1 \sqrt{\frac{s_1 s_\gamma \log{p}}{n}} 
\end{align*}
Combining bounds on $T_1$ and $T_2$ we obtain that for the second term in equation \eqref{eq:debias_1}, with probability going to 1: 
\begin{align*}
& \frac{1}{n}\left(\theta^0_{-1} + \hat \theta_{-1, S}\alpha_0 - \hat \theta_{-1, Y} \right)^{\top}\breve \bW_{-1}^{\perp^{\top}}\left(\bS^{\perp} - \breve \bW^{\perp}_{-1}\hat \theta_{-1, S}\right) \\
& \qquad \qquad  \lesssim_\bbP \left(\lambda_0 \vee \lambda_1\right)^2 (s_0 \vee s_1) + \lambda_1 s_\gamma\sqrt{\frac{\log{p}}{n}} + \left(\lambda_0 \vee \lambda_1\right)^2(s_0 \vee s_1) + \lambda_1 \sqrt{\frac{s_1 s_\gamma \log{p}}{n}}   \\
& \qquad \qquad  \lesssim_\bbP \left(\lambda_0 \vee \lambda_1\right)^2(s_0 \vee s_1) + \lambda_1 \sqrt{\frac{s_\gamma (s_1 + s_\gamma) \log{p}}{n}} 
\end{align*}
which implies: 
\begin{equation}
\label{eq:bound_num_1}
\textbf{Numerator of Term 1 } \lesssim_\bbP \sqrt{n}\left[\left(\lambda_0 \vee \lambda_1\right)^2(s_0 \vee s_1) + \lambda_1 \sqrt{\frac{s_\gamma (s_1 + s_\gamma) \log{p}}{n}} \right] \,.
\end{equation}
For the above term to be asymptotically negligible we need: 
$$
\sqrt{n}\left[\left(\lambda_0 \vee \lambda_1\right)^2(s_0 \vee s_1) + \lambda_1 \sqrt{\frac{s_\gamma (s_1 + s_\gamma) \log{p}}{n}} \right] = o(1) \,.
$$
{\bf Numerator of the second term: }We next control $\sqrt{n}$ times the numerator of second term on the RHS of equation \eqref{eq:debias_1} which involves the remainder $\tilde \bR$. Recall that $\tilde \bR = \bR_1 + \bR_2 + \bR_3$ where: 
\begin{align*}
    \bR_1 & = b(\hat \eta) - \bN_k(\hat \eta)\omega_b \\
    \bR_2 & = \left(b'(\hat \eta) - \hat b'(\hat \eta)\right) \odot Z(\hat \gamma_n -\gamma_0) \\
    \bR_3 & = (\eta - \hat \eta)^2 \odot b''(\tilde \eta)
\end{align*}
Therefore the numerator of the second term can be written as: 
\begin{align*}
    \frac{1}{n}\tilde \bR^{\perp^{\top}}\left(\bS^{\perp} - \breve \bW^{\perp}_{-1}\hat \theta_{-1, S}\right) & = \underbrace{\frac{1}{n}\bR_1^{\perp^{\top}}\left(\bS^{\perp} - \breve \bW^{\perp}_{-1}\hat \theta_{-1, S}\right)}_{T_1} + \underbrace{\frac{1}{n}\bR_2^{\perp^{\top}}\left(\bS^{\perp} - \breve \bW^{\perp}_{-1}\hat \theta_{-1, S}\right)}_{T_2} \\
    & \qquad \qquad \qquad \qquad + \underbrace{\frac{1}{n}\bR_3^{\perp^{\top}}\left(\bS^{\perp} - \breve \bW^{\perp}_{-1}\hat \theta_{-1, S}\right)}_{T_3} \,.
\end{align*}
To bound $T_1$ note that $\bR_1$ is measurable function of $\hat \eta$. We further expand $T_1$ as: 
\begin{align*}
    T_1 & = \frac{1}{n}\bR_1^{\perp^{\top}}\left(\bS^{\perp} - \breve \bW^{\perp}_{-1}\hat \theta_{-1, S}\right) \\
    & = \underbrace{\frac{1}{n}\bR_1^{\perp^{\top}}\left(\hat{\tilde \bS}^{\perp} - \hat{\tilde \bW}^{\perp}_{-1}\theta^*_{S}\right)}_{T_{11}} + \underbrace{\frac{1}{n}\bR_1^{\perp^{\top}}\left(\bS^{\perp} - \hat{\tilde \bS}\right)}_{T_{12}} \\
    & \qquad \qquad \qquad + \underbrace{\frac{1}{n}\bR_1^{\perp^{\top}}\left(\breve \bW^{\perp} - \hat{\tilde \bW}_{-1}\right)\theta^*_{S}}_{T_{13}} + \underbrace{\frac{1}{n}\bR_1^{\perp^{\top}}\hat{\tilde \bW}_{-1}\left(\theta^*_{S} - \hat \theta_{-1, S}\right)}_{T_{14}} 
\end{align*}
Note that given $\hat \eta$, $T_{11}$ is a linear combination of centered subgaussian random variables. Therefore by subgaussian concentration inequality, we have with probabilities going to 1: 
$$
T_{11} \lesssim_\bbP \frac{r_n}{\sqrt{n}} \,.
$$
Further we have as $\bS^{\perp} - \hat{\tilde \bS} = P_{\bN_k}^{\perp}\check \bR_1 - P_{\bN_k}\check \nu_1$ we have with probability going to: 
$$
T_{12} = \frac{1}{n}\bR_1^{\perp^{\top}}\check \bR_1 \le r_n^2 \,.
$$
For $T_{13}$ we apply a similar analysis along with $\ell_1-\ell_\infty$ bound: 
\begin{align*}
T_{13} & \le \left\|\theta^*_{S}\right\|_1 \frac{1}{n}\left\|\bR_1^{\perp^{\top}}\left(\breve \bW^{\perp} - \hat{\tilde \bW}_{-1}\right) \right\|_\infty \\
& = \left\|\theta^*_{S}\right\|_1 \max_{2 \le j \le 2p}\frac{1}{n}\left|\bR_1^{\perp^{\top}}\left(\breve \bW_{*, j}^{\perp} - \hat{\tilde \bW}_{*, j}\right)\right| \\
& \le  \left\|\theta^*_{S}\right\|_1 \max_{2 \le j \le 2p}\frac{1}{n}\left[\left|\bR_1^{\perp^{\top}}\left(\breve \bW_{*, j}^{\perp} -  \bW_{*, j}^{\perp} \right)\right| +\left|\bR_1^{\perp^{\top}}\left( \bW_{*, j}^{\perp} - \hat{\tilde \bW}_{*, j}\right)\right|\right] \\
& \le  \left\|\theta^*_{S}\right\|_1 \max_{p+1 \le j \le 2p}\frac{1}{n}\left|\bR_1^{\perp^{\top}}\left(\breve \bW_{*, j}^{\perp} -  \bW_{*, j}^{\perp} \right)\right|  + \left\|\theta^*_{S}\right\|_1 \max_{2 \le j \le 2p}\frac{1}{n}\left|\bR_1^{\perp^{\top}}\check \bR_j \right| \\
& \lesssim \left\|\theta^*_{S}\right\|_1 \max_{p+1 \le j \le 2p} \sqrt{\frac{\bR_1^{\perp^{\top}}\bR_1}{n}} \sqrt{\frac{1}{n}\sum_i Z_{i,j}^2\left(\hat b(\hat \eta_i) - b(\eta_i)\right)^2} + \left\|\theta^*_{S}\right\|_1 r_n^2 \\
& \lesssim_\bbP \left\|\theta^*_{S}\right\|_1\left[r_n\left(\dot r_n + \sqrt{\frac{s_\gamma \log{p}}{n}}\right)\right] + \left\|\theta^*_{S}\right\|_1r_n^2 \\
& \lesssim_\bbP \left\|\theta^*_{S}\right\|_1 \left[r_n\left(\dot r_n + \sqrt{\frac{s_\gamma \log{p}}{n}}\right)\right] \,.
\end{align*}
And for the last term $T_{14}$ we have: 
\begin{align*}
    T_{14} & = \frac{1}{n}\bR_1^{\perp^{\top}}\hat{\tilde \bW}_{-1}\left(\theta^*_{S} - \hat \theta_{-1, S}\right) \\
    & \le \left\|\theta^*_{S} - \hat \theta_{-1, S}\right\|_1 \max_{1 \le j \le p} \frac{1}{n}\left| \bR_1^{\perp^{\top}}\hat{\tilde \bW}_{*, j}\right| \\
    & \lesssim \lambda_1 s_1  r_n \sqrt{\frac{\log{p}}{n}} \,.
\end{align*}
Combining the bounds on the different components of $T_1$ we have with probability going to 1: 
\begin{align*}
T_1 \lesssim_\bbP \frac{r_n}{\sqrt{n}} + r_n^2 +\left\|\theta^*_{S}\right\|_1 \left[r_n\left(\dot r_n + \sqrt{\frac{s_\gamma \log{p}}{n}}\right)\right] + \lambda_1 s_1  r_n \sqrt{\frac{\log{p}}{n}}\,.
\end{align*}
Now we consider the second remainder term $T_2$ which involves $\bR_2$. An easy bound on this term will be following:
\begin{align*}
T_2 = \frac{1}{n}\bR_2^{\perp^{\top}}\left(\bS^{\perp} - \breve \bW^{\perp}_{-1}\hat \theta_{-1, S}\right) & \le \sqrt{\frac{1}{n}\bR_2^{\top}\bR_2} \sqrt{\frac{1}{n}\left\|\bS^{\perp} - \breve \bW^{\perp}_{-1}\hat \theta_{-1, S}\right\|^2} \lesssim_\bbP \dot r_n \sqrt{\frac{s_\gamma \log{p}}{n}} \,.
\end{align*}
Finally, for the third term $T_3$ an easy CS bound is sufficient:
$$
T_3 = \frac{1}{n}\bR_3^{\perp^{\top}}\left(\bS^{\perp} - \breve \bW^{\perp}_{-1}\hat \theta_{-1, S}\right) \le \sqrt{\frac{1}{n}\bR_3^{\top}\bR_3} \sqrt{\frac{1}{n}\left\|\bS^{\perp} - \breve \bW^{\perp}_{-1}\hat \theta_{-1, S}\right\|^2} \lesssim \frac{s_\gamma \log{p}}{n} \,.
$$
Combining the bounds on $T_1, T_2, T_3$ we have with probability going to 1: 
\begin{align*}
T_1 + T_2 + T_3 & \lesssim_\bbP \frac{r_n}{\sqrt{n}} + r_n^2 +\left\|\theta^*_{S}\right\|_1 \left[r_n\left(\dot r_n + \sqrt{\frac{s_\gamma \log{p}}{n}}\right)\right] + \lambda_1 s_1  r_n \sqrt{\frac{\log{p}}{n}} \\
& \qquad \qquad +  \dot r_n \sqrt{\frac{s_\gamma \log{p}}{n}}  + \frac{s_\gamma \log{p}}{n}\,.
\end{align*}
which implies: 
\begin{align}
\textbf{Numerator 2nd} & \lesssim_\bbP \sqrt{n}\left[\frac{r_n}{\sqrt{n}} + r_n^2 +\left\|\theta^*_{S}\right\|_1 \left[r_n\left(\dot r_n + \sqrt{\frac{s_\gamma \log{p}}{n}}\right)\right] + \lambda_1 s_1  r_n \sqrt{\frac{\log{p}}{n}} \right. \notag \\
\label{eq:bound_num_2} & \qquad \qquad \left. +  \dot r_n \sqrt{\frac{s_\gamma \log{p}}{n}}  + \frac{s_\gamma \log{p}}{n}\right] \,.
\end{align}
\vspace{0.3in}
\\
{\bf Numerator of the third term: }Last but not the least, we need to analyze the numerator of the error term, i.e. third term on the RHS of equation \eqref{eq:debias_1} First we show that: 
$$
\frac{1}{\sqrt{n}}\beps^{\perp^{\top}}\left(\bS^{\perp} - \breve \bW^{\perp}_{-1}\hat \theta_{-1, S}\right) = \frac{1}{\sqrt{n}}\beps^{\top}\left(\tilde \bS - \tilde \bW_{-1}\theta^*_{S}\right) + o_p(1) \,.
$$
That the first term on the RHS is $O_p(1)$ is from the sub-exponential concentration inequality which will later be shown to be asymptotically normal. To show that the remainder is asymptotically negligible, first observe that: 
$$
\frac{1}{\sqrt{n}}\beps^{\perp^{\top}}\left(\bS^{\perp} - \breve \bW^{\perp}_{-1}\hat \theta_{-1, S}\right) = \frac{1}{\sqrt{n}}\beps^{\top}\left(\bS^{\perp} - \breve \bW^{\perp}_{-1}\hat \theta_{-1, S}\right)
$$
which follows from the idempotence of the projection matrix. So it is enough to show that: 
$$
\underbrace{\frac{1}{\sqrt{n}}\beps^{\top}\left(\bS^{\perp} - \tilde \bS\right)}_{T_1} + \underbrace{\frac{1}{\sqrt{n}}\beps^{\top}\left(\tilde \bW_{-1} - \breve \bW^{\perp}_{-1}\right)\theta^*_{S}}_{T_2} + \underbrace{\frac{1}{\sqrt{n}}\beps^{\top}\breve \bW_{-1}^{\perp}\left(\hat \theta_{-1, S} - \theta^*_{S}\right)}_{T_3} = o_p(1) \,.
$$
To bound $T_1$ we use the subgaussian concentration inequality along with the fact that $\bbE[\eps \mid X, Z, \eta] = 0$. Therefore the terms of $\beps^{\top}\left(\bS^{\perp} - \tilde \bS\right)$ are centered subgaussian random variables conditional on $(X, Z, \eta)$. So we have: 
$$
\bbP\left(\frac{1}{\sqrt{n}}\left|\beps^{\top}\left(\bS^{\perp} - \tilde \bS\right)\right| > t \mid \sigma\left(X, Z, \eta, \cD_1, \cD_2\right)\right) \le 2\exp{\left(-c\frac{t^2}{\sigma^2_\eps \frac{\left\|\bS^{\perp} - \tilde \bS \right\|^2}{n}}\right)}
$$
As we have already established in Lemma \ref{lem:sq_rate} that $\|\bS^{\perp} - \tilde \bS\|^2/n = o_p(1)$, by DCT we conclude that $T_1 = o_p(1)$. Similar subgaussian concentration for $T_2$ yields: 
$$
\bbP\left(\frac{1}{\sqrt{n}}\left|\beps^{\top}\left(\tilde \bW_{-1} - \breve \bW^{\perp}_{-1}\right)\theta^*_{S}\right| > t \mid \sigma\left(X, Z, \eta\right)\right) \le 2\exp{\left(-c\frac{t^2}{\sigma^2_\eps \frac{\left\|\left(\tilde \bW_{-1} - \breve \bW^{\perp}_{-1}\right)\theta^*_{S} \right\|^2}{n}}\right)}
$$
We have established in the proof of Proposition \ref{prop:denom_conv} that $\left\|\left(\tilde \bW_{-1} - \breve \bW^{\perp}_{-1}\right)\theta^*_{S}\right\|^2/n = o_p(1)$. Therefore, again by DCT we have $T_2 = o_p(1)$. 
\\\\
Finally for $T_3$, we first use $\ell_1 - \ell_\infty$ bound: 
\begin{align}
    \frac{1}{\sqrt{n}}\beps^{\top}\breve \bW_{-1}^{\perp}\left(\hat \theta_{-1, S} - \theta^*_{S}\right) & \le \left\|\hat \theta_{-1, S} - \theta^*_{S}\right\|_1 \frac{1}{\sqrt{n}}\left\|\beps^{\top}\breve \bW_{-1}^{\perp} \right\|_\infty \notag \\
    \label{eq:eps_break}& = \left\|\hat \theta_{-1, S} - \theta^*_{S}\right\|_1 \max_{1 \le j \le p} \frac{1}{\sqrt{n}}\left|\beps^{\top}\breve \bW_{*,j}^{\perp} \right|
\end{align}
The rest of the proof is purely technical. Recall the definition of $\tilde W_{j} = W_j - \bbE[W_j \mid \eta]$. Now in our matrix 
\begin{align*}
\Sigma & = \bbE_\eta\left[\left(W - \bbE[W \mid \eta]\right)\left(W - \bbE[W \mid \eta]\right)^{\top}\right] \\
& = \bbE_\eta\left[\var(W \mid \eta)\right] \\
& = \bbE_\eta\left[\var\left(\begin{pmatrix}S & X & b'(\eta) Z\end{pmatrix} \mid \eta\right)\right] \,.
\end{align*}
As per our Assumption \ref{assm:sigma} the eigenvalues of this matrix are bounded away from 0 and infinity. Therefore any diagonal entries are bounded in between $(C_{\min}, C_{\max})$ (the notations in Assumption \ref{assm:sigma}). So from the sub-exponential Bernstein's inequality we have for any constant $\upsilon > 0$: 
$$
\bbP\left(\frac{1}{n}\sum_i \tilde \bW_{i, j}^2 > (\upsilon+1) \Sigma_{j, j}\right) \le 2\exp{\left(-c\min\left\{\frac{n\upsilon^2\Sigma_{j, j}^2}{\sigma^2_W}, \frac{n\upsilon \Sigma_{j, j}}{\sigma_W}\right\}\right)}
$$
If we define the above event $\cA_j = \left\{(1/n)\sum_i \tilde \bW_{i, j}^2 > (\upsilon+1) \Sigma_{j, j}\right\}$ then we have: 
$$
\bbP\left(\cup_{j=1}^p \cA_j\right) \le 2\exp{\left(\log{p}-c\min\left\{\frac{n\upsilon^2C_{\max}^2}{\sigma^2_W}, \frac{n\upsilon C_{\max}}{\sigma_W}\right\}\right)}
$$
and the complement event is: 
$$
\bbA_{n, p} = \cap_j \cA_j^c = \left\{\frac{1}{n}\sum_i \tilde \bW_{i, j}^2 \le  (\upsilon+1) \Sigma_{j, j} \ \ \forall \ \  1 \le j \le p\right\} 
$$
Furthermore we have proved in the second part of Lemma \ref{lem:sq_rate} that with probability going to 1: 
$$
\max_{1 \le j \le p} \frac{1}{n} \left\|\breve \bW^{\perp}_{*, j} - \tilde \bW_{*, j}\right\|^2 \lesssim \left[r_n^2 + \left(\frac{\log{p}}{n} \vee r_n \sqrt{\frac{\log{p}}{n}}\right)\right] \vee \frac{s_\gamma \log{p}}{n}\,.
$$
Call that above event $\bbB_{n, p}$. On the event $\bbA_{n, p} \cap \bbB_{n, p}$, we have:
$$
\frac{1}{n}\breve \bW_{*, j}^{\perp^{\top}}\breve \bW_{*, j}^{\perp} \lesssim \left[r_n^2 + \left(\frac{\log{p}}{n} \vee r_n \sqrt{\frac{\log{p}}{n}}\right)\right] \vee \frac{s_\gamma \log{p}}{n} + (\upsilon+1) C_{\max} \lesssim (\upsilon+1) C_{\max} \,.
$$
Using the above findings we have: 
\begin{align*}
    \bbP\left(\max_{1 \le j \le p} \frac{1}{\sqrt{n}}\left|\beps^{\top}\breve \bW_{*,j}^{\perp} \right| > t\right) & = \bbP\left(\max_{1 \le j \le p} \frac{1}{\sqrt{n}}\left|\beps^{\top}\breve \bW_{*,j}^{\perp} \right| > t, \left(\bbA_{n, p} \cap \bbB_{n, p}\right)\right) + \bbP\left(\left(\bbA_{n, p} \cap \bbB_{n, p}\right)^c\right) \\
    & \le  \bbE\left[\bbP\left(\max_{1 \le j \le p} \frac{1}{\sqrt{n}}\left|\beps^{\top}\breve \bW_{*,j}^{\perp} \right| > t \mid \sigma(X, Z, \eta)\right)\mathds{1}_{\left(\bbA_{n, p} \cap \bbB_{n, p}\right)}\right] +\bbP\left(\left(\bbA_{n, p} \cap \bbB_{n, p}\right)^c\right) \\
    & = 2\sum_j \bbE \left[\exp{\left(-c\frac{t^2}{\sigma^2_\eps \frac{\left(\bW_{*,j}^{\perp}\right)^{\top}\bW_{*,j}^{\perp}}{n}}\right)}\mathds{1}_{\bbA_{n, p} \cap \bbB_{n, p}}\right] +  \bbP\left(\left(\bbA_{n, p} \cap \bbB_{n, p}\right)^c\right)  \\
    & \le 2\exp{\left(\log{p} - c\frac{t^2}{(\upsilon + 1)C_{\max}}\right)} + \underbrace{\bbP\left(\left(\bbA_{n, p} \cap \bbB_{n, p}\right)^c\right)}_{\to 0 \text{ as we proved }}
\end{align*}
Therefore we conclude that: 
$$
\max_{1 \le j \le p} \frac{1}{\sqrt{n}}\left|\beps^{\top}\breve \bW_{*,j}^{\perp} \right|  \lesssim_p \sqrt{\log{p}}
$$
which, along with equation \eqref{eq:eps_break} concludes with probability going to 1: 
$$
\frac{1}{\sqrt{n}}\beps^{\top}\breve \bW_{-1}^{\perp}\left(\hat \theta_{-1, S} - \theta^*_{S}\right) \lesssim \lambda_1 s_1 \sqrt{\log{p}}  = o_p(1)\,.
$$

\noindent
We next show that: 
\begin{align*}
 \frac{1}{\sqrt{n}}\beps^{\top}\left(\tilde \bS - \tilde \bW_{-1}\theta^*_{S}\right) & =  \frac{1}{\sqrt{n}} \sum_{i=1}^n \eps_i \left(\tilde \bS_i - \tilde \bW_{i, -1}^{\top}\theta^*_{S}\right)\mathds{1}_{|\hat \eta_i| \le \tau} \\
 & =  \frac{1}{\sqrt{n}} \sum_{i=1}^n \eps_i \left(\tilde \bS_i - \tilde \bW_{i, -1}^{\top}\theta^*_{S}\right)\mathds{1}_{| \eta_i| \le \tau} + o_p(1) \,.
\end{align*}
which we will achieve by simple Cauchy-Schwarz inequality. Note that conditional on $X, Z, \eta, \cD_1$ the term $\eps_i \left(\tilde \bS_i - \tilde \bW_{i, -1}^{\top}\theta^*_{S}\right)\left(\mathds{1}_{| \eta_i| \le \tau} - \mathds{1}_{|\hat \eta_i| \le \tau}\right)$ is a centered sub-gaussian random variable, where the sub-gaussianity follows from the sub-gaussianity of $\eps$. Therefore we need to show: 
$$
\frac{1}{n}\sum_i \left(\tilde \bS_i - \tilde \bW_{i, -1}^{\top}\theta^*_{S}\right)^2\left(\mathds{1}_{| \eta_i| \le \tau} - \mathds{1}_{|\hat \eta_i| \le \tau}\right)^2 = o_p(1) \,.
$$
Consequently it is enough to show: 
$$
\bbE\left[\left(\tilde S - \tilde W_{ -1}^{\top}\theta^*_{S}\right)^2\left(\mathds{1}_{| \eta| \le \tau} - \mathds{1}_{|\hat \eta| \le \tau}\right)^2 \mid \cD_1\right] = o_p(1) \,.
$$ 
From Holder inequality: 
\begin{align*}
& \bbE\left[\left(\tilde S - \tilde W_{ -1}^{\top}\theta^*_{S}\right)^2\left(\mathds{1}_{| \eta| \le \tau} - \mathds{1}_{|\hat \eta| \le \tau}\right)^2 \mid \cD_1\right] \\
& \le \left(\bbE\left[\left(\tilde S - \tilde W_{ -1}^{\top}\theta^*_{S}\right)^{2 + \delta} \right]\right)^{\frac{1}{1 + \frac{\delta}{2}}} \times \left(\bbP(\Delta_n \mid \cD_1)\right)^{\frac{\frac{\delta}{2}}{1 + \frac{\delta}{2}}}
\end{align*}
where $\Delta_n$ is the event that at-least one of the random variable $\eta$ or $\hat \eta$ is outside the interval $[-\tau, \tau]$. The first term in the above inequality is finite by Assumption \ref{assm:sigma} and therefore it is enough to show $\bbP(\Delta_n \mid \cD_1) = o_p(1)$ which follows from equation \eqref{eq:Delta_bound_1} and \eqref{eq:Delta_bound_2} in the proof of Lemma \ref{lem:choice_lambda} presented in the supplementary document. 
\\\\
\noindent
The final part of the proof is to establish asymptotic normality, which again follows from simple application of Lindeberg's central limit theorem along with Lyapounov's condition. Define a set of triangular array of random variables $\{\zeta_{n,i}\}_{i=1}^n$ as: 
$$
\zeta_{n, i} = \frac{\eps_i \left(\tilde \bS_i - \tilde \bW_{i, -1}^{\top}\theta^*_{S}\right)\mathds{1}_{| \eta_i| \le \tau}}{\sqrt{n}\sigma_{n, 1}}
$$ 
with $\sigma_{n, 1}$ being the standard deviation defined as:  
$$
\sigma_{n, 1} = \sqrt{\bbE\left[\eps^2\left(\tilde S - \tilde W_{ -1}^{\top}\theta^*_{S}\right)^2\left(\mathds{1}_{|\eta| \le \tau}\right) \right]} \,.
$$
Therefore, $\bbE[\zeta_{n, i}] = 0$ and $\sum_i \bbE[\zeta^2_{n, i}] = 1$. Furthermore from Assumption 1, we have: 
$$
\sum_i \bbE[\left|\zeta_{n, i}\right|^{2 + \delta}] = \frac{\bbE\left[\left|\eps\left(\tilde S - \tilde W_{ -1}^{\top}\theta^*_{S}\right)\right|^{2 + \delta}\left(\mathds{1}_{|\eta| \le \tau}\right) \right]}{n^{\delta}\sigma_{n, 1}^{2 + \delta}} \to 0 \,.
$$
Hence we conclude: 
$$
\frac{1}{\sigma_{n}\sqrt{n}} \beps^{\top}\left(\bS^{\perp} - \check \bW_{-1}\hat \theta_{-1, S}\right) \overset{\mathscr{L}}{\implies} \cN(0, 1) \,.
$$
\end{proof}

\subsection{Some sufficient conditions for normality}
\label{sec:discussion_sparsity_rate}
In the proof of Theorem \ref{thm:main_hd}, we need to ensure that RHS of equation \eqref{eq:bound_num_1} and \eqref{eq:bound_num_2} are $o(1)$ to ensure asymptotic normality of the debiased estimator. We here present some sufficient conditions in certain cases: 
\subsubsection{Case 1: }
Assume $s_0 = s_1 = s_\gamma = s_\beta \sim s$, i.e. all the sparsities are of similar order and $\|\theta^*_S\|_1 \sim \|\theta^*_Y\|_1 \sim \sqrt{s}$. Further assume that $\dot r_n \gg \sqrt{(s\log{p}/n)}$. We start with \eqref{eq:bound_num_1} which requires: 
\begin{align}
\label{eq:num_1_disc} \sqrt{n}\left[\left(\lambda_0 \vee \lambda_1\right)^2(s_0 \vee s_1) + \lambda_1 \sqrt{\frac{s_\gamma (s_1 + s_\gamma) \log{p}}{n}} \right] = o(1)
\end{align}
From Lemma \ref{lem:choice_lambda} we have under our setup: 
\begin{align*}
\lambda_1 \sim \lambda_0 \sim \sqrt{s}\left[\dot r_n + \sqrt{\frac{s\log{p}}{n}}\right] + \sqrt{\frac{s \log{p}}{n}}\left(\log{\frac{n}{s\log{p}}}\right)^{3/2} \,.
\end{align*}
Ignoring the log factor we have: 
$$
\lambda_1 \sim \lambda_0 \sim \sqrt{s}\dot r_n + \sqrt{\frac{s \log{p}}{n}} \,.
$$
Therefore the condition in equation \eqref{eq:num_1_disc} simplifies to: 
\begin{align*}
\sqrt{n}\left(\lambda_0 \vee \lambda_1\right)^2 (s_0 \vee s_1) & \sim \sqrt{n}s\left(s\dot r_n^2 + \frac{s\log{p}}{n}\right) \\
\sqrt{n}\lambda_1 \sqrt{\frac{s_\gamma(s_1 + s_\gamma)\log{p}}{n}} & \sim \sqrt{s\log{p}}\left(\sqrt{s}\dot r_n + \sqrt{\frac{s \log{p}}{n}}\right) \,.
\end{align*}
For both the terms on the RHS of the above equations to be $o(1)$ we need: 
\begin{align}
\label{eq:disc_1} \dot r_n & = o\left(\frac{1}{sn^{1/4} \wedge \frac{1}{s\sqrt{\log{p}}}}\right) \,, \\
\label{eq:disc_2} \frac{s\log{p}}{\sqrt{n}} &= o(1) \,.
\end{align}
Now for equation \eqref{eq:bound_num_2} to be asymptotically negligible, we need: 
\begin{align}
& \sqrt{n}\left[\frac{r_n}{\sqrt{n}} + r_n^2 +\left\|\theta^*_{S}\right\|_1 \left[r_n\left(\dot r_n + \sqrt{\frac{s_\gamma \log{p}}{n}}\right)\right] + \lambda_1 s_1  r_n \sqrt{\frac{\log{p}}{n}} \right. \notag \\
\label{eq:num_2_disc} & \qquad \qquad \left. +  \dot r_n \sqrt{\frac{s_\gamma \log{p}}{n}}  + \frac{s_\gamma \log{p}}{n}\right] = o(1) \,.
\end{align}
which in our above setup, reduces to: 
\begin{align*}
& r_n + \sqrt{n}r_n^2 + \sqrt{ns}r_n\dot r_n + sr_n \sqrt{\log{p}} + \dot r_n \sqrt{s\log{p}} + s^{3/2}\dot r_n r_n \sqrt{\log{p}} + \frac{s^{3/2}r_n\log{p}}{\sqrt{n}} + \frac{s\log{p}}{\sqrt{n}} = o(1) \,. 
\end{align*}
As described at the beginning of subsection \ref{sec:third_step_hd}, in most of the scenarios, we expect $r_n \ll \dot r_n$. Condition \eqref{eq:disc_1} implies that first seven summands of the above equation are o(1), whereas \eqref{eq:disc_2} implies the last summand is o(1). Finally, for the proof of Proposition \ref{prop:denom_conv}, \ref{prop:RE} and Lemma \ref{lem:subg_fourth_moment} we further need that: 
\begin{equation}
\label{eq:prop_cond}
(s_0 \vee s_1) \left\{\dot r_n^2 + \frac{s_\gamma \log{p}}{n} + \dot r_n \sqrt{\frac{\log{p}}{n}}\right\} = o(1) \,.
\end{equation}
which again holds under \eqref{eq:disc_1} and \eqref{eq:disc_2}. Therefore, in this setup, \eqref{eq:disc_1} and \eqref{eq:disc_2} are sufficient to ensure asymptotic normality.

\subsubsection{Case 2: }Now assume that $\dot r_n \ll \sqrt{(s\log{p})/n}$ whereas the other conditions remain same, i.e.  $s_0 = s_1 = s_\gamma = s_\beta \sim s$ and $\|\theta^*_S\|_1 \sim \|\theta^*_Y\|_1 \sim \sqrt{s}$. In this case the order of $\lambda_0$ and $\lambda_1$ becomes (ignoring the log factor): 
$$
\lambda_0 \sim \lambda_1 \sim s\sqrt{\frac{\log{p}}{n}} \,.
$$
Consequently, the condition of \eqref{eq:num_1_disc} simplifies to: 
$$
\sqrt{n}\left[\frac{s^2 \log{p}}{n} + \frac{s^{3/2}\log{p}}{n}\right] = o(1) \,.
$$
A sufficient condition for this is: 
\begin{equation}
\label{eq:disc_3}
\frac{s^2 \log{p}}{\sqrt{n}} = o(1) \,.
\end{equation}
In this setup, the bound in equation \eqref{eq:num_2_disc} reduces to: 
$$
r_n + \sqrt{n}r_n^2 + s \sqrt{n}r_n \dot r_n + s r_n \sqrt{\log{p}} + \frac{s^2 r_n \log{p}}{\sqrt{n}} + \dot r_n \sqrt{s\log{p}} + \frac{s\log{p}}{\sqrt{n}} \,.
$$
From the conditions $\dot r_n \ll \sqrt{(s\log{p})/n}$ and $r_n \ll \dot r_n$ and condition \eqref{eq:disc_3} it is immediate that the bound in the above display is $o(1)$. Furthermore, for the proof of Proposition \ref{prop:RE}, \ref{prop:denom_conv} and Lemma \ref{lem:subg_fourth_moment} we need to establish equation \eqref{eq:prop_cond}, which is also immediate from condition \ref{eq:disc_3}. Therefore, condition \ref{eq:disc_3} is sufficient to ensure asymptotic normality of the debiased estimator. 

\section{Proof of Theorem \ref{thm:main_thm}}
\label{sec:proof_main_thm}
\subsection{Proof of Step 1} 
\label{sec:s1_main}
\vspace{0.1in}
\noindent
First we decompose the matrix as follows: 
$$
\frac{\bW^{\top}\proj^{\perp}_{\tilde \bN_{K, a}}\bW}{n} = \frac{\bW_1^{\top}\proj^{\perp}_{\tilde \bN_K}\bW_1}{n} + \frac{\bW_2^{\top}\bW_2}{n}
$$
We show that: 
\begin{align}
\label{eq:step11} \frac{\bW_1^{\top}\proj^{\perp}_{\tilde \bN_K}\bW_1}{n} & \overset{P}{\to} \frac13 \bbE\left[\var\left(\begin{bmatrix} s & \bx & \bz b'(\eta)\end{bmatrix}\,\middle\vert\, \eta\right)\mathds{1}_{|\eta| \le \tau}\right]  \,, \\
\label{eq:step12} \frac{\bW_2^{\top}\bW_2}{n} & \overset{P}{\to} \frac13 \begin{bmatrix}
0 & 0 & 0 \\
0 & 0 & 0 \\
0 & 0 & \Sigma_Z 
\end{bmatrix}
\end{align}
where $\bW_1, \bW_2$ are as defined in Section \ref{sec:est_method}. Note that this implies: 
$$
\frac{\bW^{\top}\proj^{\perp}_{\tilde \bN_{K, a}}\bW}{n}  \overset{P}{\to}  \frac13 \bbE\left[\var\left(\begin{bmatrix} s & \bx & \bz b'(\eta)\end{bmatrix}\,\middle\vert\, \eta\right)\mathds{1}_{|\eta| \le \tau}\right] + \frac13 \begin{bmatrix}
0 & 0 & 0 \\
0 & 0 & 0 \\
0 & 0 & \Sigma_Z 
\end{bmatrix} = \frac13 \Omega_{\tau}
$$
which delivers the assertion of Step 1. 
\\\\
\noindent
Equation \eqref{eq:step12} follows immediately from an application of weak law of large numbers. For equation \eqref{eq:step11}, note that $\bW_1^{\top}\proj^{\perp}_{\tilde \bN_K}\bW_1/n$ is the norm of the residual of rows of $\bW_1$ upon projecting out the effect of $\bN_K$, which can be further decomposed as: 
\begin{align}
\frac{\bW_1^{\top}\proj^{\perp}_{\tilde \bN_K}\bW_1}{n} & = \frac{\bW_1^{*^{\top}}\proj^{\perp}_{\tilde \bN_K}\bW^*_1}{n} + 2\frac{\bW_1^{*^{\top}}\proj^{\perp}_{\tilde \bN_K}(\bW_1 - \bW^*_1)}{n} \notag \\
\label{eq:W1_decomp} & \qquad \qquad \qquad + \frac{(\bW_1 - \bW^*_1)^{\top}\proj^{\perp}_{\tilde \bN_K}(\bW_1 - \bW^*_1)}{n}
\end{align}
where $\bW^*_{1, i*} = \begin{bmatrix} S_i & X_i & b'(\eta_i)Z_i\end{bmatrix}$. Note that the only difference between $\bW_1$ and $\bW^*_1$ is in the last $p_2$ co-ordinates where we have replaced $\hat b'(\hat \eta_i)$ by $b'(\eta_i)$. We now show that: 
$$
\frac{\bW_1^{*^{\top}}\proj^{\perp}_{\tilde \bN_K}(\bW_{1}^* - \bW_1)}{n} \overset{P}{\longrightarrow} 0 \,.
$$
The other term $(\bW_1 - \bW^*_1)^{\top}\proj^{\perp}_{\tilde \bN_K}(\bW_1 - \bW^*_1)/n$ will consequently be $o_p(1)$ as it is a lower order term Fix $1 \le j, k \le 1+p_1+p_2$: 
\begin{align*}
\left|\left(\frac{\bW_1^{*^{\top}}\proj^{\perp}_{\tilde \bN_K}(\bW_{1}^* - \bW_1)}{n} \right)_{j, k}\right| & = \left|\frac{\left \langle \proj^{\perp}_{\tilde \bN_K}\bW^*_{1, *j}, \proj^{\perp}_{\tilde \bN_K}(\bW^*_1 - \bW_1)_{*k}\right \rangle}{n}\right| \\
& \le \sqrt{\frac{\|\bW^*_{1, *j}\|^2}{n}\frac{\|(\bW^*_1 - \bW_1)_{*k}\|^2}{n}}
\end{align*}
That $\|\bW^*_{1, *j}\|^2/n = O_p(1)$ follows from an immediate application of WLLN. To show that the other part is $o_p(1)$, note that $\|(\bW^*_1 - \bW_1)_{*k}\|^2/n = 0$ for $1 \le k \le p_1 + 1$. For $p_1 + 2 \le k \le p_1 + p_2$, define $\tilde k = k - (p_1 + 1)$. Then: 
\begin{align}
\frac{\|(\bW^*_1 - \bW_1)_{*k}\|^2}{n} & = \frac1n \sum_{i=1}^{n/3} \left(b'(\eta_i) - \hat b'(\hat \eta_i)\right)^2 Z_{i, \tilde k}^2 \mathds{1}_{|\hat \eta_i| \le \tau} \notag  \\
& \le 2\left[\frac1n \sum_{i=1}^{n/3} \left(b'(\eta_i) - b'(\hat \eta_i)\right)^2 Z_{i, \tilde k}^2 \mathds{1}_{|\hat \eta_i| \le \tau}  + \frac1n \sum_{i=1}^n \left(b'(\hat \eta_i) - \hat b'(\hat \eta_i)\right)^2 Z_{i, \tilde k}^2 \mathds{1}_{|\hat \eta_i| \le \tau} \right]  \notag \\
& \le \left[\frac2n \sum_{i=1}^{n/3} \left(\eta_i - \hat \eta_i\right)^2 (b''(\tilde \eta_i))^2Z_{i, \tilde k}^2 \mathds{1}_{|\hat \eta_i| \le \tau} + \sup_{|x| \le \tau}|b'(t) - \hat b'(t)|\frac1n \sum_{i=1}^n Z^2_{i, \tilde k}\right] \notag  \\
\label{eq:W*} & = O_p(n^{-1}) + O_p(n^{-2}) + o_p(1) = o_p(1) \,.
\end {align} 
That the first summand is $O_p(n^{-1})$ follows from the fact $b''$ is uniformly bounded (Assumption \ref{assm:b}) and $\hat \eta_i - \eta_i = O_p(n^{-1/2})$ and the second summand is $o_p(1)$ follows from  Proposition \ref{thm:spline_consistency}. We next show: 
\begin{equation}
\label{eq:wstar}
\frac{\bW_1^{*^{\top}}\proj^{\perp}_{\tilde \bN_K}\bW^*_1}{n} \overset{P}{\longrightarrow} \frac13 \Omega_{\tau}\,.
\end{equation}
Towards that direction, we first claim that $(\bW_1^{*^{\top}}\proj^{\perp}_{\tilde \bN_K}\bW^*_1/n) = O_p(1)$ . For any $1 \le l, m \le 1 + p_1 + p_2$, we have: 
\begin{align*}
\left|\frac{\bW_1^{*^{\top}}\proj^{\perp}_{\tilde \bN_K}\bW^*_1}{n}\right|_{l, m} & = \left|\frac{\langle \bW^*_{1, *l}, \proj_{\bN_K}\bW^*_{1, *m}\rangle}{n}\right| \\
& \le \left\|\frac{\bW^*_{1, *l}}{\sqrt{n}}\right\|\left\|\frac{\bW^*_{1, *m}}{\sqrt{n}}\right\|  = O_p(1) 
\end{align*} 
by WLLN. Setting $a_n = \log{n}/\sqrt{n}$, we next decompose this term into two further terms: 
\begin{align*}
\frac{\bW_1^{*^{\top}}\proj^{\perp}_{\tilde \bN_K}\bW^*_1}{n} & = \frac{\bW_1^{*^{\top}}\proj^{\perp}_{\tilde \bN_K}\bW^*_1}{n}\mathds{1}_{\|\hat \gamma_n - \gamma_0\| \le a_n} + \frac{\bW_1^{*^{\top}}\proj^{\perp}_{\tilde \bN_K}\bW^*_1}{n}\mathds{1}_{\|\hat \gamma_n - \gamma_0\| > a_n}
\end{align*}
As we have already established $(\bW_1^{*^{\top}}\proj^{\perp}_{\tilde \bN_K}\bW^*_1/n) = O_p(1)$, it is immediate that:
$$
\frac{\bW_1^{*^{\top}}\proj^{\perp}_{\tilde \bN_K}\bW^*_1}{n}\mathds{1}_{\|\hat \gamma_n - \gamma_0\| > a_n} = o_p(1) \,.
$$
Therefore, we need to establish the convergence of $(\bW_1^{*^{\top}}\proj^{\perp}_{\tilde \bN_K}\bW^*_1)/n\mathds{1}_{\|\hat \gamma_n - \gamma_0\| \le a_n}$. Define a function $g(a, t)$ and $V(a, t)$ as: 
\begin{align*}
g(a, t) & = \bbE\left[\begin{pmatrix} S & X & b'(\eta)Z \end{pmatrix} \,\middle\vert\, \eta + a^{\top}Z = t\right] \,, \\
V(a, t) & = \var\left[\begin{pmatrix} S & X & b'(\eta)Z \end{pmatrix} \,\middle\vert\, \eta + a^{\top}Z = t\right] 
\end{align*}
The following lemma characterizes some smoothness properties of the functions $g, V$: 
\begin{lemma}
\label{lem:g}
Under Assumptions \ref{assm:independence}-\ref{assm:moment}, the functions $g$ and $V$ are continuous. Moreover $g$ is continuously differentiable in both of its co-ordinates and consequently $g, \partial_a g$ and $\partial_t g$ are uniformly bounded on $\|a\| \le 1$ and $|t| \le \tau+1$.  
\end{lemma}
The proof of this lemma can be found in Section \ref{sec:additional_results}. Note that the definition of $g$ implies $g(\hat \gamma_n - \gamma_0, t) = \bbE\left[\bW^*_{1, i} \,\middle\vert\, \hat \eta = t, \cF(\cD_1)\right]$. As $g$ is a vector valued functions with range being a subset of $\bbR^{1+p_1 + p_2}$, we henceforth denote by $g_j$, the $j^{th}$ co-ordinate of $g$ for $1 \le j \le 1 + p_1 + p_2$. Now for each of the co-ordinates of $g$, we further define $\omega_{j, n, \infty}$ as the B-spline approximation vector of $g_j$, i.e.
$$
\omega_{j, n, \infty} = \argmin_{\omega} \sup_{|x| \le \tau} \left|g_j(\hat \gamma_n - \gamma_0, x) - \tilde N_K(x)^{\top}\omega\right|
$$
where $\tilde N_k$ is the scaled B-spline basis functions (for definition and brief discussion, see Section \ref{sec:spline_details}). We often drop the index $n$ from $\omega_{j, n, \infty}$ when there is no disambiguity. It is immediate from Theorem \ref{thm:spline_bias} (with $l=r=0$): 
\begin{equation}
\label{eq:g_spline_approx}
\sup_{|x| \le \tau} \left|g_j(\hat \gamma_n - \gamma_0, x) - \tilde N_K(x)^{\top}\omega_{j, \infty}\right| \lesssim \frac{2\tau}{K} \sup_{|x| \le \tau} \left|\partial_t g(\hat \gamma_n - \gamma_0, t)\right| = O_p(K^{-1}) \,.
\end{equation}
We define the matrix $\bG$ as $\bG_{i*} = \bbE\left[\bW^*_{1, i*} \,\middle\vert\, \hat \eta = \hat \eta_i, \cF(\cD_1)\right] = g(\hat \gamma_n - \gamma_0, \hat \eta_i)$ and the matrix $\bH$ as $\bH_{i,j} = \tilde N_K(\hat \eta_i)^{\top}\omega_{j, \infty}$. Using these matrices we expand the matrix under consideration as follows:  
\begin{align}
& \frac{\bW_1^{*^{\top}}\proj^{\perp}_{\tilde \bN_K}\bW^*_1}{n}\mathds{1}_{\|\hat \gamma_n - \gamma_0\| \le a_n}  \notag \\
& = \frac{(\bW^*_1 - \bG)^{\top}\proj^{\perp}_{\tilde \bN_K}(\bW^*_1 - \bG)}{n}\mathds{1}_{\|\hat \gamma_n - \gamma_0\| \le a_n} \notag \\
& \qquad \qquad + \frac{(\bG - \bH)^{\top}\proj^{\perp}_{\tilde \bN_K}(\bG - \bH)}{n}\mathds{1}_{\|\hat \gamma_n - \gamma_0\| \le a_n}\notag \\
& \qquad \qquad \qquad \qquad + 2\frac{(\bG - \bH)^{\top}\proj^{\perp}_{\tilde \bN_K}(\bW^*_1 - \bG)}{n}\mathds{1}_{\|\hat \gamma_n - \gamma_0\| \le a_n} \notag \\
& = \frac{(\bW^*_1 - \bG)^{\top}(\bW^*_1 - \bG)}{n}\mathds{1}_{\|\hat \gamma_n - \gamma_0\| \le a_n} - \frac{(\bW^*_1 - \bG)^{\top}\proj_{\bN_K}(\bW^*_1 - \bG)}{n}\mathds{1}_{\|\hat \gamma_n - \gamma_0\| \le a_n} \notag \\
& \qquad \qquad + \frac{(\bG - \bH)^{\top}\proj^{\perp}_{\tilde \bN_K}(\bG - \bH)}{n}\mathds{1}_{\|\hat \gamma_n - \gamma_0\| \le a_n} \notag \\
& \qquad \qquad \qquad \qquad + 2\frac{(\bG - \bH)^{\top}\proj^{\perp}_{\tilde \bN_K}(\bW^*_1 - \bG)}{n}\mathds{1}_{\|\hat \gamma_n - \gamma_0\| \le a_n} \notag \\
\label{eq:decomp_main_1} & := T _1 + T_2 + T_3 + T_4 
\end{align}
We first show that $[(\bW_1^* - \bG)^{\top}(\bW_1^* - \bG)/n]\mathds{1}(\|\hat \gamma_n -\gamma_0\| \le a_n)$ converges to some matrix. Towards that end, we further expand it as follows: 
\allowdisplaybreaks
\begin{align*}
 & \frac{(\bW_1^* - \bG)^{\top}(\bW_1^* - \bG)}{n}\mathds{1}_{\|\hat \gamma_n - \gamma_0\| \le a_n} \\
 & = \frac{1}{n}\sum_{i=1}^{n/3} (\bW^*_{1,i*} - g(\hat \gamma_n - \gamma_0, \hat \eta_i))(\bW^*_{1, i*} - g(\hat \gamma_n - \gamma_0, \hat \eta_i))^{\top}\mathds{1}_{|\hat \eta_i| \le \tau, \|\hat \gamma_n - \gamma_0\| \le a_n} \\
 & = \frac{1}{n}\sum_{i=1}^{n/3} (\bW^*_{1,i*} - g(0, \hat \eta_i))(\bW^*_{1,i*} - g(0, \hat \eta_i))^{\top}\mathds{1}_{|\hat \eta_i| \le \tau, \|\hat \gamma_n - \gamma_0\| \le a_n} \\
 & \qquad \qquad + \frac{2}{n}\sum_{i=1}^{n/3} (\bW^*_{1,i*} - g(0, \hat \eta_i))(g(0, \hat \eta_i) - g(\hat \gamma_n - \gamma_0, \hat \eta_i))^{\top}\mathds{1}_{|\hat \eta_i| \le \tau, \|\hat \gamma_n - \gamma_0\| \le a_n} \\
 & \qquad \qquad \qquad \qquad + \frac{1}{n}\sum_{i=1}^{n/3} (g(0, \hat \eta_i) - g(\hat \gamma_n - \gamma_0, \hat \eta_i))(g(0, \hat \eta_i) - g(\hat \gamma_n - \gamma_0, \hat \eta_i))^{\top}\mathds{1}_{|\hat \eta_i| \le \tau, \|\hat \gamma_n - \gamma_0\| \le a_n} \\
 & := T_{11} + T_{12} + T_{13}
\end{align*}
We now show that $T_{12} = o_p(1)$ and $T_{13} = o_p(1)$ follows immediately form there, as it is a lower order term. For  $T_{12}$ note that by Lemma \ref{lem:g}, the function $g(a, t)$ has continuous derivative with respect to $a$ which makes $g$ Lipschitz on the ball $\|a\| \le 1$. Hence we have: 
\begin{align*}
&\left\|  \frac{2}{n}\sum_{i=1}^{n/3} (\bW^*_{1,i*} - g(0, \hat \eta_i))(g(0, \hat \eta_i) - g(\hat \gamma_n - \gamma_0, \hat \eta_i))^{\top}\mathds{1}_{|\hat \eta_i| \le \tau, \|\hat \gamma_n - \gamma_0\| \le a_n} \right\|_F \\
& \le \frac{2}{n}\sum_{i=1}^{n/3} \left\|\bW^*_{1,i*} - g(0, \hat \eta_i)\right\| \left\|g(0, \hat \eta_i) - g_n(\hat \gamma_n - \gamma_0, \hat \eta_i)\right\| \mathds{1}_{|\hat \eta_i| \le \tau, \|\hat \gamma_n - \gamma_0\| \le a_n} \\
& \le \|\hat \gamma_n - \gamma_0\| \frac{2}{n}\sum_{i=1}^{n/3} \left\|\bW^*_{1,i*} - g(0, \hat \eta_i)\right\|\mathds{1}_{|\hat \eta_i| \le \tau}  = O_p(n^{-1/2}) \,.
\end{align*}
Now, to establish convergence of $T_{11}$ we further expand it as follows: 
\begin{align*}
& \frac{1}{n}\sum_{i=1}^{n/3} (\bW^*_{1,i*} - g(0, \hat \eta_i))(\bW^*_{1,i*} - g(0, \hat \eta_i))^{\top}\mathds{1}_{|\hat \eta_i| \le \tau, \|\hat \gamma_n - \gamma_0\| \le a_n}  \\
& =  \frac{1}{n}\sum_{i=1}^{n/3} (\bW^*_{1,i*} - g(0,  \eta_i))(\bW^*_{1,i*} - g(0,  \eta_i))^{\top}\mathds{1}_{|\hat \eta_i| \le \tau, \|\hat \gamma_n - \gamma_0\| \le a_n} \\
& \qquad \qquad +  \frac{2}{n}\sum_{i=1}^{n/3} (\bW^*_{1,i*} - g(0,  \eta_i))(g(0, \eta_i)- g(0, \hat \eta_i))^{\top}\mathds{1}_{|\hat \eta_i| \le \tau, \|\hat \gamma_n - \gamma_0\| \le a_n} \\
& \qquad \qquad \qquad \qquad  + \frac1n \sum_{i=1}^{n/3} (g(0, \eta_i)- g(0, \hat \eta_i))(g(0, \eta_i)- g(0, \hat \eta_i))^{\top}\mathds{1}_{|\hat \eta_i| \le \tau, \|\hat \gamma_n - \gamma_0\| \le a_n} \\
& := T_{111} + T_{112} + T_{113}
\end{align*}
From law of large numbers we first conclude: 
$$
T_{111} \overset{P}{\longrightarrow} \frac13 \bbE\left[\left(\bW^*_{1,1*} - \bbE(\bW^*_{1,1*} \,\middle\vert\, \eta)\right)\left(\bW^*_{1,1*} - \bbE(\bW^*_{1,1*} \,\middle\vert\,\eta)\right)^{\top}\mathds{1}_{|\eta| \le \tau}\right]
$$
where the factor $1/3$ comes due to data splitting. To complete the proof we show $T_{112} = o_p(1)$ and $T_{113} = o_p(1)$ follows immediately being a higher order term. We analyse $T_{112}$ as follows: 
\begin{align*}
\|T_{112}\| &  =  \left\|\frac{2}{n}\sum_{i=1}^{n/3} (\bW^*_{1,i*} - g(0,  \eta_i))(g(0, \eta_i)- g(0, \hat \eta_i))^{\top}\mathds{1}_{|\hat \eta_i| \le \tau, \|\hat \gamma_n - \gamma_0\| \le a_n}\right\| \\
& = \left\| \frac{2}{n}\sum_{i=1}^{n/3} (\bW^*_{1,i*} - g(0,  \eta_i))(g(0, \eta_i)- g(0, \hat \eta_i))^{\top}\mathds{1}_{|\hat \eta_i| \le \tau, \|\hat \gamma_n - \gamma_0\| \le a_n, |\hat \eta_i - \eta_i| \le 1}\right\| \\
& \qquad \qquad \qquad +   \left\|\frac{2}{n}\sum_{i=1}^{n/3} (\bW^*_{1,i*} - g(0,  \eta_i))(g(0, \eta_i)- g(0, \hat \eta_i))^{\top}\mathds{1}_{|\hat \eta_i| \le \tau, \|\hat \gamma_n - \gamma_0\| \le a_n, |\hat \eta_i - \eta_i| > 1}\right\| \\
& \le \frac2n \sum_i \left\|\bW^*_{1,i*}- g(0, \eta_i) \right\| \left| \eta_i - \hat \eta_i\right| \\
& \qquad \qquad \qquad + \frac2n \sum_i \left\|\bW^*_{1,i*} - g(0, \eta_i) \right\|\left\|g(0, \eta_i)- g(0, \hat \eta_i)\right\|\mathds{1}_{|\hat \eta_i| \le \tau, \|\hat \gamma_n - \gamma_0\| \le a_n, |\hat \eta_i - \eta_i| > 1} \\
& \le \|\hat \gamma_n - \gamma_0\|\frac2n \sum_i \left\|\bW^*_{1,i*} - g(0, \eta_i) \right\| \|Z_i\| \\
& \qquad \qquad \qquad + \frac2n \sum_i \left\|\bW^*_{1,i*} - g(0, \eta_i) \right\|\left\|g(0, \eta_i)- g(0, \hat \eta_i)\right\|\mathds{1}_{|\hat \eta_i| \le \tau, \|\hat \gamma_n - \gamma_0\| \le a_n, |\hat \eta_i - \eta_i| > 1} 
\end{align*}
That the first term is $o_p(1)$ is immediate. For the second term, note that: 
$$
\frac1n \sum_i  \left\|\bW^*_{1,i*} - g(0, \eta_i) \right\|\left\|g(0, \eta_i)- g(0, \hat \eta_i)\right\| = O_p(1)
$$  
and $\bbP\left(\left\|\hat \gamma_n - \gamma_0\right\| \le a_n, |\hat \eta_i - \eta_i| > 1\right)  \longrightarrow 0 \,.$ This finishes the proof of $T_1$, i.e. we have established: 
\begin{equation}
\label{eq:T1_matrix}
T_1 \overset{P}{\longrightarrow} \frac13 \bbE\left[\left(\bW^*_{1,1*} - \bbE(\bW^*_{1,1*} \,\middle\vert\,\eta)\right)\left(\bW^*_{1,1*} - \bbE(\bW^*_{1,1*} \,\middle\vert\,\eta)\right)^{\top}\mathds{1}_{|\eta| \le \tau}\right]
\end{equation}
\\\\
For $T_3$ in equation \eqref{eq:decomp_main_1} we have for any $1 \le j, k \le 1+p_1+p_2$: 
\begin{align}
& \left|\frac{(\bG - \bH)^{\top}\proj^{\perp}_{\tilde \bN_K}(\bG - \bH)}{n}\right|_{j, k}\mathds{1}_{\|\hat \gamma_n - \gamma_0\| \le a_n}  \notag \\
& \le \left|\frac{\langle (\bG - \bH)_{*j}, \proj^{\perp}_{\tilde \bN_K}(\bG - \bH)_{*k}\rangle}{n}\right| \notag \\
& \le \sqrt{\frac{\left\|(\bG - \bH)_{*j}\right\|^2}{n}\frac{\left\|(\bG - \bH)_{*k}\right\|^2}{n}} \notag \\
& \le \sup_{|x| \le \tau} \left|g_j(\hat \gamma_n - \gamma_0, t) - \tilde N_K(t)^{\top}\omega_{j, \infty}\right| \times \sup_{|x| \le \tau} \left|g_k(\hat \gamma_n - \gamma_0, t) - \tilde N_K(t)^{\top}\omega_{k, \infty}\right|  \notag \\
\label{eq:G-H}  & = O_p(K^{-2})= o_p(1) \hspace{0.2in} [\text{From equation }\ref{eq:g_spline_approx}]\,.
\end{align}

\noindent
Similarly for $T_4$ in equation \eqref{eq:decomp_main_1} and for any $1 \le j, k \le 1+p_1+p_2$: 
\allowdisplaybreaks
\begin{align}
& \left|\left(\frac{(\bG - \bH)^{\top}\proj^{\perp}_{\tilde \bN_K}(\bW^*_1 - \bG)}{n}\right)_{j, k}\right|\mathds{1}_{\|\hat \gamma_n - \gamma_0\| \le a_n}  \notag \\
& \le \left|\frac{\langle (\bG - \bH)_{*j}, \proj^{\perp}_{\tilde \bN_K}(\bW^*_1 - \bG)_{*k}\rangle}{n}\right| \notag \\
& \le  \sqrt{\frac{\left\|(\bG - \bH)_{*j}\right\|^2}{n}\frac{\left\|(\bW^*_1 - \bG)_{*k}\right\|^2}{n}} \notag \\
& \le \sup_{|x| \le \tau} \left|g_j(\hat \gamma_n - \gamma_0, t) - \tilde N_K(t)^{\top}\omega_{j, \infty}\right| \sqrt{\frac{\left\|(\bW^*_1 - \bG)_{*k}\right\|^2}{n}} \notag \\
\label{eq:t4_normality_s1} & = \sup_{|x| \le \tau} \left|g_j(\hat \gamma_n - \gamma_0, t) - \tilde N_K(t)^{\top}\omega_{j, \infty}\right| \times O_p(1) = o_p(1) \,.
\end{align}

\noindent
where $\left\|(\bW^*_1 - \bG)_{*k}\right\|^2/n = O_p(1)$ follows from law of large numbers and the uniform spline approximation error is $o_p(1)$ follows from equation \eqref{eq:g_spline_approx}. Finally, for $T_2$ in equation \eqref{eq:decomp_main_1}, first recall that $\bW^*_1$ consists of all the rows for which $|\hat \eta_i| \le \tau$. We can extend this matrix $\bW^{*, f}_1 \in \bbR^{(n, 1+p_1+p_2)}$ as: 
$$
\bW^{*, f}_{1, i*} = \begin{bmatrix}S_i & \bX_{i*} & b'(\eta_i)\bZ_{i*}\end{bmatrix}\mathds{1}_{|\hat \eta_i| \le \tau}. 
$$
The matrix $\bW^{*, f}_1$ is exactly $\bW^*_1$ appended with $0$'s in the rows where $|\hat \eta_i| > \tau$. Similarly, we can define $\bG^f$ with $\bG^f_{i*} = \bbE\left[\bW^{*, f}_{1, i*} \,\middle\vert\,\hat \eta = \hat \eta_i, \cF(\cD_1)\right]$ and $\tilde N^f_K$ as the basis matrix with $\tilde N^f_{K, i*} = \tilde N_k(\hat \eta_i) \mathds{1}_{|\hat \eta_i| \le \tau}$. It is easy to see that for any $1 \le j \le 1+p_1 +p_2$:
\begin{align*}
& \frac{(\bW^*_1 - \bG)_{*, j}^{\top}\proj_{\bN_K}(\bW^*_1 - \bG)_{*, j}}{n}\mathds{1}_{\|\hat \gamma_n - \gamma_0\| \le a_n} \\
& = \frac{(\bW^{*, f}_1 - \bG^f)_{*, j}^{\top}\proj_{\tilde \bN^f_K}(\bW^{*, f}_1 - \bG^f)_{*, j}}{n}\mathds{1}_{\|\hat \gamma_n - \gamma_0\| \le a_n}
\end{align*}
Hence we can bound $(j, j)^{th}$ term of $T_2$ as follows \footnote{We use Lemma \ref{lem:cond_prob} here to conclude that a sequence of non-negative random variables is $o_p(1)$ if their conditional expectation is $o_p(1)$.}: 
\begin{align}
& \bbE\left(\frac{(\bW^*_1 - \bG)_{*, j}^{\top}\proj_{\bN_K}(\bW^*_1 - \bG)_{*, j}}{n}\mathds{1}_{\|\hat \gamma_n - \gamma_0\| \le a_n} \,\middle\vert\,\cF(\cD_1) \right) \notag \\
& = \bbE\left(\frac{(\bW^{*, f}_1 - \bG^f)_{*, j}^{\top}\proj_{\tilde \bN^f_K}(\bW^{*, f}_1 - \bG^f)_{*, j}}{n}\mathds{1}_{\|\hat\gamma_n - \gamma_0\| \le a_n} \,\middle\vert\,\cF(\cD_1)\right) \notag \\
& = \tr\left(\bbE\left(\frac{(\bW^{*, f}_1 - \bG^f)_{*, j}^{\top}\proj_{\tilde \bN^f_K}(\bW^{*, f}_1 - \bG^f)_{*, j}}{n}\mathds{1}_{\|\hat\gamma_n - \gamma_0\| \le a_n}  \,\middle\vert\,\cF(\cD_1) \right) \right) \notag \\
& \le \frac{1}{n}\bbE\left(\tr\left(\proj_{\tilde \bN^f_K}\bbE\left((\bW^{*, f}_1 - \bG^f)_{*, j}(\bW^{*, f}_1 - \bG^f)_{*, j}^{\top} \,\middle\vert\,\cF(\hat \boeta,\cD_1) \right)\right)\right) \hspace{0.1in} [\hat \boeta = \{\hat \eta_i\}_{i=1}^{n/3} \in \cD_3]\notag \\
\label{eq:W*-G} & \le \frac{K+3}{n}\sup_{|t| \le \tau} \var\left(\begin{bmatrix}S & X & b'(\eta)Z \end{bmatrix}\,\middle\vert\,\hat \eta = t, \cF(\cD_1)\right)\mathds{1}_{\|\hat \gamma_n - \gamma_0\| \le a_n} = O_p\left(\frac{K}{n}\right) = o_p(1) \,.
\end{align}
where we use Lemma \ref{lem:trace_ineq} along with the fact that $\tr\left(\proj_{\tilde \bN_K}\right) = K+2$. The finiteness of the conditional variance follows from Lemma \ref{lem:g}. Combining our findings from equation \eqref{eq:T1_matrix}, \eqref{eq:G-H}, \eqref{eq:t4_normality_s1} and \eqref{eq:W*-G} we conclude \eqref{eq:step11} which along with \eqref{eq:step12} concludes: 
\begin{align*}
\left(\frac{\bW^{\top}\proj^{\perp}_{\bN_{K, a}}\bW}{n}\right) & \overset{P}{\longrightarrow}  \frac13 \bbE\left[\var\left(\begin{bmatrix} s & \bx & \bz b'(\eta)\end{bmatrix}\,\middle\vert\,\eta\right)\mathds{1}_{|\eta| \le \tau}\right]  + \frac13 \begin{bmatrix}
0 & 0 & 0 \\
0 & 0 & 0 \\
0 & 0 & \Sigma_Z 
\end{bmatrix} := \frac13 \Omega_{\tau} 
\end{align*}
As $\Omega_{\tau} \to \Omega_{\infty}$, using Assumption \ref{assm:eigen} we conclude that the minimum eigenvalue of $\Omega_{\tau}$ is positive.  
\\\\

\noindent
\subsection{Proof of Step 2 } 
\label{sec:s2_main}
\vspace{0.1in}
\noindent
Define $\cA_n$ to be generated by $\cD_1, \cD_2$ and $\{(X_i, Z_i, \eta_i)\}_{i=1}^{n/3}$ in $\cD_3$. We start with the following decomposition: 
$$
\frac{\bW^{\top} \proj^{\top}_{\tilde \bN_{K, a}}}{\sqrt{n}}\begin{pmatrix}
\beps \\ \eta
\end{pmatrix} = \frac{\bW_1^{\top}\proj^{\perp}_{\tilde \bN_K}\beps}{\sqrt{n}} +\frac{\bW_2^{\top}\eta}{\sqrt{n}}
$$
The \emph{asymptotic linear expansion} of the second summand is immediate: 
$$
\frac{\bW_2^{\top}\eta}{\sqrt{n}} = \frac{1}{\sqrt{n}}\sum_{i=1}^{n/3}\begin{bmatrix} 0 & 0 & \eta_iZ_i\end{bmatrix}
$$
Recall from the definition of $\bG$ that the $\bG_{i*} = g(\hat \gamma_n - \gamma_0, \hat \eta_i)$. Define another matrix $\bG^*$ as $\bG^*_i = g(0, \eta_i)$. For the first summand, we decompose it as follows: 
\begin{align}
\frac{\bW_1^{\top}\proj^{\perp}_{\tilde N_K}\eps}{\sqrt{n}} & =\frac{\bW_1^{^*{\top}}\proj^{\perp}_{\tilde N_K}\beps}{\sqrt{n}} +  \frac{(\bW_1 - \bW_1^*)^{\top}\proj^{\perp}_{\tilde N_K}\beps}{\sqrt{n}} \notag \\
& = \frac{(\bW_1^* - \bG^*)^{\top}\beps}{\sqrt{n}} +  \frac{(\bW_1 - \bW_1^*)^{\top}\proj^{\perp}_{\tilde N_K}\beps}{\sqrt{n}} \notag \\
& \hspace{8em} +  \frac{(\bW_1^* - \bG)^{\top}\proj_{\tilde N_K}\beps}{\sqrt{n}} + \frac{(\bG - \bH)^{\top}\proj^{\perp}_{\tilde N_K}\beps}{\sqrt{n}}  + \frac{( \bG^* - \bG)^{\top}\beps}{\sqrt{n}} \notag \\
\label{eq:weps} & := T_1 + T_2 + T_3 + T_4 + T_5
\end{align}
We show that $T_2, T_3, T_4, T_5$ are all $o_p(1)$. This will establish:
\begin{equation}
\label{eq:ale_1}
\frac{\bW_1^{\top}\proj^{\perp}_{\tilde N_K}\eps}{\sqrt{n}} =  \frac{(\bW_1^* - \bG^*)^{\top}\beps}{\sqrt{n}} + o_p(1) \,.
\end{equation}
For $T_2$ note that for any $p_1+2 \le j \le p_1 + p_2 + 1$: 
\begin{align*}
& \bbE\left(\left(\frac{(\bW_1 - \bW_1^*)_{*j}^{\top}\proj^{\perp}_{\tilde N_K}\beps}{\sqrt{n}}\right)^2 \,\middle\vert\, \cF(\cD_1, \cD_2)\right) \\
& = \frac1n \bbE\left[(\bW_1 - \bW_1^*)_{*j}^{\top}\proj^{\perp}_{\tilde N_K}\bbE\left(\beps \beps^{\top} \,\middle\vert\, \cA_n \right)\proj^{\perp}_{\tilde N_K}(\bW_1 - \bW_1^*)_{*j} \,\middle\vert\, \cF(\cD_1, \cD_2)\right] \\
& \le \sup_{\eta}\var\left(\eps \,\middle\vert\,\eta\right)  \times \bbE\left[\frac{(\bW_1 - \bW_1^*)_{*j}^{\top}(\bW_1 - \bW_1^*)_{*j}}{n} \,\middle\vert\, \cF(\cD_1, \cD_2)\right] \\
& = \sup_{\eta}\var\left(\eps \,\middle\vert\, \eta\right) \times o_p(1) \hspace{0.2cm} [\text{From equation }\eqref{eq:W*}]
\end{align*}
Now for $T_3$, for any $1 \le j \le 1+p_1+p_2$ by similar calculations:
\begin{align*}
& \bbE\left[\left(\frac{(\bW_1^* - \bG)^{\top}\proj_{\tilde N_K}\beps}{\sqrt{n}}\right)^2 \,\middle\vert\,\cA_n\right] \\
& \le \sup_{\eta}\var\left(\eps \,\middle\vert\,\eta\right) \times \frac{(\bW_1^* - \bG)^{\top}\proj_{\tilde N_K}(\bW_1^* - \bG)}{n} \\
& = \sup_{\eta}\var\left(\eps \,\middle\vert\,\eta\right) \times o_p(1) \hspace{0.2in} [\text{From equation }\eqref{eq:W*-G}]
\end{align*}   
For $T_4$, for $1 \le j \le 1+p_1 + p_2$: 
\begin{align*}
& \bbE\left[\left(\frac{(\bG - \bH)^{\top}\proj^{\perp}_{\tilde N_K}\beps}{\sqrt{n}}\right)^2 \,\middle\vert\,\cA_n\right] \\
& \le \sup_{\eta}\var\left(\eps \,\middle\vert\,\eta\right) \times \frac{(\bG - \bH)^{\top}(\bG - \bH)}{n} \\
& = \sup_{\eta}\var\left(\eps \,\middle\vert\,\eta\right) \times o_p(1) \hspace{0.2in} [\text{From equation }\eqref{eq:G-H}]
\end{align*}
Now for $T_5$, using similar technique we have for any $1 \le j \le p_1 + p_2 + 1$: 
\begin{align*}
& \bbE\left[\left( \frac{(\bG - \bG^*)_{*j}^{\top}\eps}{\sqrt{n}}\right)^2 \,\middle\vert\,\cA_n\right] \\
& \le \sup_{\eta}\var\left(\eps \,\middle\vert\,\eta\right) \times \frac{\|\bG - \bG^*\|^2}{n} \\
&  \le \sup_{\eta}\var\left(\eps \,\middle\vert\,\eta\right) \times \frac1n \sum_{i=1}^{n/3} \left(g_j(\hat \gamma_n - \gamma_0, \hat \eta_i) - g_j(0, \eta_i) \right)^2\mathds{1}_{\|\hat \eta_i\| \le \tau}  \\
&  \le \sup_{\eta}\var\left(\eps \,\middle\vert\,\eta\right) \times \left[\frac1n \sum_{i=1}^{n/3} \left(g_j(\hat \gamma_n - \gamma_0, \hat \eta_i) - g_j(0, \eta_i) \right)^2\mathds{1}_{\|\hat \eta_i\| \le \tau, \|\hat \gamma_n - \gamma_0\| \le a_n}  \right. \\ 
& \qquad \qquad \qquad \qquad \qquad \qquad + \left. \frac1n \sum_{i=1}^{n/3} \left(g_j(\hat \gamma_n - \gamma_0, \hat \eta_i) - g_j(0, \eta_i) \right)^2\mathds{1}_{\|\hat \eta_i\| \le \tau, \|\hat \gamma_n - \gamma_0\| > a_n} \right]
\end{align*}
The first term inside the square bracket is $o_p(1)$ from the boundedness of the partial derivatives of $g$ with respect to both $a$ and $t$. The second term inside the square bracket inside the square bracket is $o_p(1)$ because $(1/n)\sum_{i=1}^{n/3} \left(g_j(\hat \gamma_n - \gamma_0, \hat \eta_i) - g_j(0, \eta_i) \right)^2 =O_p(1)$ and $\bbP(\|\hat \gamma_n - \gamma_0\| > a_n) = o(1)$.
\\\\
\noindent
Finally we show that: 
\begin{align}
\frac{(\bW_1^* - \bG^*)^{\top}\beps}{\sqrt{n}} & = \frac{1}{\sqrt{n}}\sum_{i=1}^n \left(\bW^*_{1, i*} - \bG^*_{i*}\right)\eps_i\mathds{1}_{|\hat \eta_i|  \le \tau}  \notag \\
\label{eq:ale_2} & = \frac{1}{\sqrt{n}}\sum_{i=1}^n \left(\bW^*_{1, i*} - \bG^*_{i*}\right)\eps_i\mathds{1}_{|\eta_i|  \le \tau}  + o_p(1)
\end{align}
This along with equation \eqref{eq:ale_1} concludes: 
\begin{align*}
\frac{\bW_1^{\top}\proj^{\perp}_{\tilde N_K}\eps}{\sqrt{n}}  & = \frac{1}{\sqrt{n}}\sum_{i=1}^n \left(\bW^*_{1, i*} - \bG^*_{i*}\right)\eps_i\mathds{1}_{|\eta_i|  \le \tau}  + o_p(1) \\
& = \frac{1}{\sqrt{n}}\sum_{i=1}^n \left\{\begin{bmatrix} S_i & X_i & b'(\eta_i)Z_i \end{bmatrix} - \bbE\left[\begin{bmatrix} S_i & X_i & b'(\eta_i)Z_i \end{bmatrix} \,\middle\vert\,\eta = \eta_i\right]\right\}\eps_i \mathds{1}_{|\eta_i| \le \tau} + o_p(1)
\end{align*} 
Taking the function $\varphi$ as: 
\begin{align*}
\varphi(X_i, Z_i, \eta_i, \nu_i) & = \left\{\begin{bmatrix} S_i & X_i & b'(\eta_i)Z_i \end{bmatrix} - \bbE\left[\begin{bmatrix} S_i & X_i & b'(\eta_i)Z_i \end{bmatrix} \,\middle\vert\,\eta = \eta_i\right]\right\}\eps_i \mathds{1}_{|\eta_i| \le \tau} + \begin{bmatrix}0 & 0 & Z_i\eta_i \end{bmatrix}
\end{align*}
we conlcude the proof of Step 2. 
All it remains to prove is equation \eqref{eq:ale_2}. Define the event $\Delta_i$ as: 
$$
 \Delta_i = \{|\hat \eta_i| \le \tau \cap |\eta_i| > \tau\} \cup \{|\hat \eta_i| > \tau \cap |\eta_i| \le \tau\} \,.
$$
Now for any $1 \le j \le 1+p_1 + p_2$: 
\begin{align*}
&\bbE\left[\left(\frac{1}{\sqrt{n}}\sum_{i=1}^{n/3} \left(\bW^*_{1, i, j} - \bG^*_{i, j}\right)\eps_i\left(\mathds{1}_{|\hat \eta_i|  \le \tau} - \mathds{1}_{|\eta_i|  \le \tau} \right) \right)^2 \,\middle\vert\,\cF(\cD_1) \right] \\
&\bbE\left[\bbE\left[\left(\frac{1}{\sqrt{n}}\sum_{i=1}^{n/3} \left(\bW^*_{1, i, j} - \bG^*_{i, j}\right)\eps_i\left(\mathds{1}_{|\hat \eta_i|  \le \tau} - \mathds{1}_{|\eta_i|  \le \tau} \right)\right)^2 \,\middle\vert\,\cA_n\right] \,\middle\vert\,\cF(\cD_1)\right] \\
& \le \sup_{\eta}\var\left(\eps \,\middle\vert\,\eta\right) \times \frac{1}{n}\sum_{i=1}^{n/3} \bbE\left[\left(\bW^*_{1, i, j} - \bG^*_{i, j}\right)^2\mathds{1}_{ \Delta_i} \,\middle\vert\,\cF(\cD_1)\right] \\
& \le  \sup_{\eta}\var\left(\eps \,\middle\vert\,\eta \right) \times \frac{1}{n}\sum_{i=1}^{n/3} \sqrt{\bbE\left[\left(\bW^*_{1, i, j} - \bG^*_{i, j}\right)^4\right]\bbP\left( \Delta_i \vert \cF(\cD_1)\right)} \\
& \lesssim \frac1n \sum_{i=1}^{n/3} \sqrt{\bbP\left( \Delta_i \vert \cF(\cD_1) \right)}  = o_p(1)\,.
\end{align*} 
\\

\noindent
\subsection{Proof of Step 3 } 
\label{sec:s3_main}
\noindent
From the definition of $\varphi$ is Step 2 and the definition of $\Omega_{\tau}^*$, Step 3 immediately follows from a direct application of Central Limit theorem. 
\\

\noindent
\subsection{Proof of Step 4 }
\label{sec:s4_main}
\vspace{0.1in}
\noindent
In this subsection we prove that: 
$$
\frac{\bW^{\top}\proj^{\perp}_{\tilde \bN_{K, a}}}{\sqrt{n}}\begin{pmatrix} \bR \\ 0 \end{pmatrix} = \frac{\bW_1^{\top}\proj^{\perp}_{\tilde \bN_K}\bR}{\sqrt{n}} = o_p(1)  \,.
$$
From our discussion in Section \ref{sec:est_method}, the residual vector can be expressed as the sum of three residual term $\bR = \bR_1 + \bR_2 + \bR_3$, where 
$\bR_{1, i} = \frac12 (\hat \eta_i - \eta_i)^2 b''(\tilde \eta_i), \bR_{2, i} = (b'(\hat \eta_i) - \hat b'(\hat \eta_i))(\hat \eta_i - \eta_i)$ and $\bR_{3, i} = (b(\hat \eta_i) - \tilde N_K(\hat \eta_i)^{\top}\omega_{b, \infty, n})$. We show that each all these terms is asymptotically negligible. For the first term, we can write is as:
\begin{align*}
\bbE\left[\sum_{i=1}^{n_3} \bR_{1, i}^2 \,\middle\vert\,\cF(\cD_1)\right] & =   \frac14 \sum_{i=1}^{n/3}\bbE\left[(\hat \eta_i - \eta_i)^4 \left(b''(\tilde \eta_i)\right)^2\mathds{1}_{
|\hat \eta_i|\le \tau} \,\middle\vert\,\cF(\cD_1)\right]  \\
& \le  \frac14 \sum_{i=1}^{n/3}\bbE\left[(\hat \eta_i - \eta_i)^4 \left(b''(\tilde \eta_i)\right)^2 \,\middle\vert\,\cF(\cD_1)\right] \\
& \le \frac{1}{12} \|b''\|_{\infty}^2 \times n\|\hat \gamma_n - \gamma_0\|^4 \times \bbE[\|Z\|^4] = O_p(n^{-1}) \,.
\end{align*}
This implies $\|\bR_1\|_2 =\sqrt{\sum_{i=1}^{n_3}R_{1,i}^2} = O_p(n^{-1/2})$ and consequently we have for any $1 \le j \le 1 + d_1 + d_2$: 
\begin{align*}
\left|\frac{\bW_{1, *j}^{\top}\proj^{\perp}_{\tilde \bN_K}\bR_1}{\sqrt{n}}\right| & \le \frac{\bW_{1, *j}^{\top}\proj^{\perp}_{\tilde \bN_K}\bW_{1, *j}}{n} \|\bR_1\|_2  \\
& = O_p(1) \times O_p(n^{-1/2}) = o_p(1) \,.
\end{align*}
where $\bW_{1, *j}^{\top}\proj^{\perp}_{\tilde \bN_K}\bW_{1, *j}/n = O_p(1)$ has been proved in the proof of Step 1 (see equation \eqref{eq:step11}). For the second residual term, we have: 
\begin{align*}
\bbE\left[\sum_{i=1}^{n_3} R_{2, i}^2 \,\middle\vert\,\cF(\cD_1, \cD_2)\right] & = \bbE\left[\sum_{i=1}^{n/3} R_{2, i}^2\mathds{1}_{|\hat \eta_i| \le \tau} \,\middle\vert\,\cF(\cD_1, \cD_2)\right] \\  & = \sum_{i=1}^{n/3}\bbE\left[(b'(\hat \eta_i) - \hat b'(\hat \eta_i))^2(\hat \eta_i - \eta_i)^2\mathds{1}_{|\hat \eta_i| \le \tau} \,\middle\vert\,\cF(\cD_1, \cD_2)\right] \\
& \le \|\hat \gamma_n - \gamma_0\|^2 \sum_{i=1}^{n/3}\bbE\left[(b'(\hat \eta_i) - \hat b'(\hat \eta_i))^2\|Z_i\|^2\mathds{1}_{|\hat \eta_i| \le \tau} \,\middle\vert\,\cF(\cD_1, \cD_2)\right] \\
& \le \frac{n}{3}\|\hat \gamma_n - \gamma_0\|^2 \sup_{|t| \le \tau} \left|\hat b'(t) - b'(t)\right| \bbE\left[\|Z\|^2\right] 
\end{align*}
Now from Proposition \ref{thm:spline_consistency}, we have $\sup_{|t| \le \tau} \left|\hat b'(t) - b'(t)\right| = o_p(1)$ and from OLS properties we have: $ n\|\hat \gamma_n - \gamma_0\|^2 = O_p(1)$. Combining this, we conclude that $\bbE\left[\sum_{i=1}^{n_1} R_{2, i}^2\right] = o_p(1)$. Now we have: 
\begin{align*}
\left|\frac{\bW_{1, *j}^{\top}\proj^{\perp}_{\tilde \bN_K}\bR_2}{\sqrt{n}}\right| & \le \frac{W_{1, *j}^{\top}\proj^{\perp}_{\tilde \bN_K}W_{1, *j}}{n} \|\bR_2\|_2  \\
& = O_p(1) \times o_p(1) = o_p(1) \,.
\end{align*}
For the final residual term $\bR_3$ (residual obtained by approximating the mean function via B-spline basis) define $\bR_3^f$ to be the extended version of $\bR_3$ putting 0 in the places where $|\hat \eta_i| > \tau$, i.e. $\bR_3 \in \bbR^{n_1}$, whereas $\bR_3^f \in \bbR^{n/3}$. Using this we have: 
\begin{align*}
\frac{\bW_{1, *j}^{\top}\proj^{\perp}_{\tilde \bN_K}\bR_3}{\sqrt{n}} & = \frac{(\bW_{1, *j} - \bG_{*, j})^{\top}\proj^{\perp}_{\tilde \bN_K}\bR_3}{\sqrt{n}} + \frac{(\bG_{*,j}-\bH_{*, j})^{\top}\proj^{\perp}_{\tilde \bN_K}\bR_3}{\sqrt{n}}
\end{align*} 
where $\bG, \bH$ are same as defined in Subsection \ref{sec:s1_main} (just after equation \eqref{eq:g_spline_approx}). For the first summand above, we have: 
\begin{align*}
& \bbE\left[\left(\frac{(\bW_{1, *j} - \bG_{*j})^{\top}\proj^{\perp}_{\tilde N_K}\bR_3}{\sqrt{n}}\right)^2 \,\middle\vert\,\cF(\cD_1, \cD_2)\right] \\
& \bbE\left[\left(\frac{(\bW^f_{1, *j} - \bG^f_{*j})^{\top}\left(I - \proj_{\tilde N^f_K}\right)\bR^f_3}{\sqrt{n}}\right)^2 \,\middle\vert\,\cF(\cD_1, \cD_2)\right] \\
& = \frac1n \bbE\left(\bR_3^{f^{\top}}\proj^{\perp}_{\bN^f_K}(\bW^f_{1, *j} - \bG_{*j})(\bW^f_{1, *j} - \bG^f_{*j})^{\top}\proj^{\perp}_{\bN^f_K}\bR^f_3 \,\middle\vert\,\cF(\cD_1, \cD_2) \right) \\
& = \frac1n \bbE\left(\bR_3^{f^{\top}}\proj^{\perp}_{\tilde N^f_K}\bbE\left[(\bW^f_{1, *j} - \bG^f_{*j})(\bW^f_{1, j} - \bG^f_{*j})^{\top} \,\middle\vert\, \cF(\hat \boeta, \cD_1, \cD_2) \right]\proj^{\perp}_{\bN^f_K}\bR^f_3 \,\middle\vert\,\cF(\cD_1, \cD_2)\right) \\
& \le \sup_{|t| \le \tau} \var\left(\begin{bmatrix}S & X & b'(\eta)Z \end{bmatrix}\,\middle\vert\,\hat{\eta} = t, \cF(\cD_1)\right) \bbE\left[\|\bR_3\|^2/n \,\middle\vert\,\cF(\cD_1, \cD_2)\right] \\
& \le \sup_{|t| \le \tau} \var\left(\begin{bmatrix}S & X & b'(\eta)Z \end{bmatrix}\,\middle\vert\,\hat \eta = t, \cF(\cD_1)\right) \times \sup_{|t| \le \tau} \left|\hat b'(t) - b(t)\right|^2   \\
& = \sup_{|t| \le \tau} V(\hat \gamma_n - \gamma_0, t) \times \sup_{|t| \le \tau} \left|\hat b'(t) - b(t)\right|^2 = o_p(1) \,.
\end{align*}
where $\hat \boeta$ is $\{\hat \eta_i\}_{i=1}^{n/3}$ from $\cD_3$. The last line follows from continuity of $V(a, t)$ (Lemma \ref{lem:g}) and Proposition \ref{thm:spline_consistency}. 
\begin{align*}
\frac{(\bG_{*,j}-\bH_{*, j})^{\top}\proj^{\perp}_{\tilde \bN_K}\bR_3}{\sqrt{n}} & \le \frac{1}{\sqrt{n}}\left\|\bG_{*j} - \bH_{*j}\right\|\|\bR_3\| \\
& = \sqrt{n}\sqrt{\frac{\left\|\bG_{*j} - \bH_{*j}\right\|^2}{n}\frac{\|\bR_3\|^2}{n}} \\
& \le \sqrt{n} \sup_{|x| \le \tau}\left| g_j(\hat \gamma_n - \gamma_0, x) - \tilde N_K(x)^{\top}\omega_{j,\infty}\right| \times \sup_{|x| \le \tau}\left|b(x) - \tilde N_K(x)^{\top}\delta_{0, n}\right| \\
& \lesssim \sqrt{n}/K^4 = o(1) \,.
\end{align*}
The last approximation follows from Theorem \ref{thm:spline_bias} (with $l = r = 0$ for $g_j$ and $l=2, r = 0$ for $b$) along with Remark \ref{rem:K}, which completes the proof of the asymptotic negligibility of the residuals.

\section{Discussion on semi-parametric efficiency and proof of Theorem \ref{thm:sem_eff}}
\subsection{Some basics of semi-parametric efficiency calculations} 
\label{sec:semiparam_eff}
\vspace{0.1in}
\noindent
\noindent
We first present a brief sketch of our approach for the convenience of the general reader. Our proof is based on the techniques introduced in Section 3.3. of \cite{bickel1993efficient}, albeit we sketch the main ideas here for the convenience of the readers. Suppose $X_1, \dots, X_n \sim P_0 \in \cP$ with density function $p_0$. We will work with $\sqrt{p_0}$ instead of $p_0$ itself, as $s_0 := \sqrt{p_0}$ lies on the unit sphere of $L_2(\bbR)$ (with respect to Lebesgue measure) and gives rise to a nice Hilbert space. Suppose, we are interested in estimating a one dimensional functional $\theta(s_0)$, where $\theta: L_2(\bbR) \to \bbR$. As initially pointed out by Stein \cite{stein1956efficient}, estimating any one dimensional functional of some non-parametric component is at least as hard as estimating the functional by restricting oneself to an one-dimensional parametric sub-model that contains the true parameter. In general, any one dimensional smooth parametrization, i.e. a function from say $\varphi: (-t_0, t_0) \to L_2(\bbR)$ ($t_0 > 0$) with $\varphi(0) = s_0$ introduces a one dimensional parametric sub-model, which essentially is a curve on the unit sphere of $L_2(\bbR)$ passing through $s_0$. We restrict ourselves to \emph{regular parametrizations},  which are differentiable on $(-t_0, t_0)$ in the following sense: for any $|t| < t_0$, there exists some function $\dot s_{\varphi, t} \in L^2(\bbR)$ such that, 
$$
\lim_{h \to 0} \left\|\frac{\varphi(t+h) - \varphi(t)}{h} - \dot s_{\varphi, t}\right\|_{L^2(\bbR)} = 0\,.
$$
with $\|\dot s_{\varphi, t}\|_{L^2(\bbR)} > 0$. Let $\cG$ be the set of all such regular parametrizations. Under mild conditions, this derivative also coincides with the pointwise derivative of $\varphi(t)$ with respect to $t$, i.e. $\dot s_{\varphi, t}(x) = (d/dt) \varphi(t)(x)$. As those conditions are easily satisfied in our model, henceforth we use this fact in our derivations. Define the \emph{tangent set} of $\cP$ at $P_0$ as $\dot \cP_{P_0} = \{\dot s_{\varphi, 0}: \varphi \in \cG\}$ and the tangent space $T(P_0) = \overline{lin}(\dot \cP_{P_0})$, the closed linear subspace spanned by $\dot \cP_{P_0}$. 
\\\\

\noindent
We restrict the discussion to the functional $\theta$ that obeys the \emph{pathwise norm differentiabilty} condition, which asserts the existence of a bounded linear functional $\L: T(P_0) \to \bbR$ such that, for any $\varphi \in \cG$: 
$$
\L(\dot s_{\varphi, 0}) := \lim_{t \to 0} \frac{\theta(\varphi(t)) - \theta(\varphi(0))}{t}  \,.
$$
Now, for any fixed $\varphi \in \cG$, the collection $\{P_t\}_{|t| < t_0}$ is the one-dimensional regular parametric sub-model, where $P_t$ is the probability measure corresponding to $\varphi(t)$. Hence for this fixed $\varphi$, one may also view $\tilde \theta_\varphi$ as a function from $(-t_0, t_0) \mapsto \bbR$ via the identification $\tilde \theta_{\varphi}(t) = \theta(\varphi(t))$ and our parameter of interest is $\tilde \theta_\varphi(0)$. The information bound (henceforth denoted by IB) for estimating $\tilde \theta_\varphi(0)$ is: 
\allowdisplaybreaks
\begin{align*}
IB(\varphi) = \frac{\left(\tilde \theta_\varphi'(0)\right)^2}{\bbE\left[\left(\frac{d}{dt}\log{p_{t}(X)}\mid_{t = 0}\right)^2\right]} & = \frac{\left(\frac{d}{dt}\tilde \theta_\varphi(t) \mid_{t = 0}\right)^2}{\bbE\left[\left(\frac{d}{dt}\log{s^2_{t}(X)}\mid_{t = 0}\right)^2\right]} \\
& = \frac{\left(\frac{d}{dt}\theta(\varphi(t)) \mid_{t = 0}\right)^2}{4\bbE\left[\left(\frac{d}{dt}\log{s_{t}(X)}\mid_{t = 0}\right)^2\right]}  \\
& = \frac{\left(\L(\dot s_{\varphi, 0})\right)^2}{4\int \dot s_{\varphi, 0}(x)^2 \ dx} \\
& =  \frac{\left(\L(\dot s_{\varphi, 0})\right)^2}{\| \dot s_{\varphi, 0} \|^2_F} = \L^2 \left(\frac{\dot s_{\varphi, 0}}{\| \dot s_{\varphi, 0} \|_F}\right) \\
\end{align*}
where $\| \cdot \|_F = 2 \| \cdot \|_{L_2(\bbR)}$ is the Fisher norm (\cite{severini2001simplified}, \cite{wong1991maximum}) and the last equality follows from the fact that $\L$ is a bounded linear operator. 
The \emph{optimal asymptotic variance} (a term borrowed from \cite{van2000asymptotic}) for estimating $\theta(s_0)$ is defined as the supremum of all these Cramer-Rao lower bounds $IB(\varphi)$ over all regular one dimensional parametrization $\varphi \in \cG$, i.e.: 
\begin{align*}
\text{Optimal asymptotic variance} 
& = \sup_{\dot s_{\varphi, 0} \in T(P_0)} \L^2 \left(\frac{\dot s_{\varphi, 0}}{\| \dot s_{\varphi, 0} \|_F}\right) \\
& = \left(\sup_{\dot s_{\varphi, 0} \in T(P_0): \|\dot s_{\varphi, 0}\|_F = 1} \L \left(\dot s_{\varphi, 0}\right)\right)^2 =   \|\L\|_*^2 \,.
\end{align*}
where $\| \cdot \|_*$ is the dual norm of the functional $\L$ on $T(P_0)$ with respect to Fisher norm. As $\L$ is a bounded linear functional on the Hilbert space $T(p_0)$, by Reisz representation theorem, there exists some $s^* \in T(p_0)$ such that: 
$$
\L(\dot s_{\varphi, 0}) = \langle s^*, \dot s_{\varphi, 0} \rangle_F \ \ \forall \ \  \dot s_{\varphi, 0} \in T(P_0) \,. 
$$
where $\langle \cdot, \cdot \rangle_F$ is $2 \langle \cdot, \cdot \rangle_{L_2(\bbR)}$. This further implies $\|\L\|_* = \|s^*\|_{F}$ and consequently, the information bound corresponding to the hardest one dimensional parametric sub-model is: 
$$
\text{Optimal asymptotic variance}  = \|s^*\|_{F}^2 \,.
$$
Therefore the problem of estimating the efficient information bound boils down to finding the representer $s^*$ in the Tangent space $T(P_0)$. To summarize, the key steps are: 
\begin{enumerate}
\item First quantify $T(P_0)$ in the given model. 
\item Then find the  expression for $\L$ by differentiating $\theta(\varphi(t))$ with respect to $t$. 
\item Finally use the identity $\L(\dot s_{\varphi, 0}) = \langle s^*, \dot s_{\varphi, 0} \rangle_F$ for all $\dot s_{\varphi, 0} \in T(P_0)$ to find $s^*$. 
\end{enumerate}
A detailed proof of the efficiency of our estimator under normality is presented in Section \ref{sec:sem_eff}, but here we sketch the main idea to give readers a sense of the application of the above approach into our model. The log - likelihood function of our model, for any generic observation $(Y, Q, X, Z)$, can be written as: 
\begin{align*}
\ell(\vartheta) & = 2\log{s_0\left(Y \mid Q, X, Z\right)} + \log{s_0\left(Q \mid X, Z\right)} + \log{s_0(X, Z)} \\
& = 4 \left\{\log{\phi^{1/2}\left(Y - \alpha - X^{\top}\beta - b(Q - Z^{\top}\gamma)
\right)} + \log{s_{\eta}(Q - Z^{\top}\gamma)} + \log{s_{X, Z}(X, Z)}\right\} 
\end{align*}
where $\phi$ is the Gaussian density, $s_{\eta}$ and $s_{X, Z}$ are square-roots of the densities of $\eta$ and $(X,Z)$ respectively and $\vartheta$ is the collection of all unknown parameters, i.e. $\vartheta = (\alpha, \beta, \gamma, b, s_{\eta}, s_{X, Z})$. 
We are interested in the functional $\theta(\varphi(t)) = \alpha_{\varphi(t)}$ which implies that the derivative $\L(\dot \alpha_0) = \dot \alpha_{\varphi} = (d/dt)\alpha_{\varphi(t)} \mid_{t =0}$. Hence, the representer $s^* \in T(P_0)$ should satisfy: 
\begin{equation}
\label{eq:eff_identity}
\langle s^*, \dot s_{\varphi, 0} \rangle_F = \dot \alpha_{\gamma} \,,
\end{equation}
\emph{for all} $\varphi \in \cG$. As a consequence, the optimal asymptotic variance will be $\alpha^* = \langle s^*, s^* \rangle_F$. In the proof (Section \ref{sec:sem_eff}  of the Supplementary document), we use the identity \eqref{eq:eff_identity} for some suitably chosen $\varphi \in \cG$ (or equivalently $\dot s_{\varphi, 0} \in T(P_0)$) to obtain $\alpha^*$.

\subsection{Proof of Theorem \ref{thm:sem_eff}}
\label{sec:sem_eff}
\vspace{0.1in}
\noindent
The model we consider here is: 
\begin{align*}
Y_i & = \alpha_0 S_i + X_i^{\top}\beta_0 + b_0(\eta_i) + \epsilon_i \\
Q_i & = Z_i^{\top}\gamma_0 + \eta_i \,.
\end{align*}
where $S_i = \mathds{1}_{Q_i \ge 0}$ and $\epsilon \sim \cN(0, \tau^2)$. For simplicity we assume here $\tau^2 = 1$. An inspection to our proof immediately reveals that extension to general $\tau^2$ is straight forward. Statisticians observe $\{D_i \triangleq (Y_i, X_i, Z_i, Q_i)\}_{i=1}^n$ at stage $n$ of the experiment and hence the likelihood of the parameters $\theta \triangleq (\alpha, \beta, \gamma, b(\cdot), s_{\eta}, s_{X, Z})$ becomes: 
%
\begin{align*}
L\left(\theta \,\middle\vert\,D\right) & = \Pi_{i=1}^n \left[p\left(Y_i \mid Q_i, X_i, Z_i\right) \times p(Q_i \mid X_i, Z_i) \times p(X_i, Z_i)\right] 
\end{align*}
As we calculate the information bound, henceforth we only will deal with one observation and generically write: 
\begin{align*}
\ell(\theta) = \log{L(\theta)} & = \log{p_0(Y, Q, X, Z)} \\
& = \log{p_0\left(Y \mid Q, X, Z\right)} + \log{p_0\left(Q \mid X, Z\right)} + \log{p_0(X, Z)}  \\
& = 2\log{s_0\left(Y \mid Q, X, Z\right)} + \log{s_0\left(Q \mid X, Z\right)} + \log{s_0(X, Z)}
\end{align*} 
Now consider some parametrization $\gamma \in \cR$ where $\cR$ is the set of all regular parametric model as mentioned in Subsection \ref{sec:semiparam_eff} the derivative of log-likelihood along this curve at $t=0$ can be written as: 
\begin{align*}
S_{\gamma} & = \frac{d}{dt}\log{p_{\gamma(t)}(Y, Q, X, Z)}\vert_{t = 0} \\
& = 2\left[\frac{\dot s_{\gamma, 0}\left(Y \mid Q, X, Z\right)}{s_0\left(Y \mid Q, X, Z\right)} + \frac{\dot s_{\gamma, 0}\left(Q \mid X, Z\right)}{s_0\left(Q \mid X, Z\right)} + \frac{\dot s_{\gamma, 0}(X, Z)}{s_0(X, Z)}\right]
\end{align*}
and as a consequence, Fisher information for estimating $\alpha_0$ along this parametric submodel curve: 
\begin{align*}
I(\gamma) & = \left\|\dot s_{\gamma, 0}\right\|_F^2\\
& = \bbE\left(S^2(\gamma)\right) \\
& =  4 \bbE\left[\left(\frac{\dot s_{\gamma, 0}\left(Y \mid Q, X, Z\right)}{s_0\left(Y \mid Q, X, Z\right)}\right)^2 + \left(\frac{\dot s_{\gamma, 0}\left(Q \mid X, Z\right)}{s_0\left(Q \mid X, Z\right)} \right)^2 + \left(\frac{\dot s_{\gamma, 0}(X, Z)}{s_0(X, Z)}\right)^2\right] \\
& = 4 \left\{\bbE\left[\frac{\phi_1'(\epsilon)\left[-\dot{\alpha} S - X^{\top}\dot{\beta} - \dot{b}(\eta) + b_0'(\eta) Z^{\top}\dot{\gamma}\right]}{\phi_1(\epsilon)} \right]^2 \right. \\
& \qquad \qquad \qquad \qquad \left. + \bbE\left[\frac{\dot{s_{\eta}}(\eta) - s_{\eta}'(\eta)Z^{\top}\dot{\gamma}}{s_{\eta}(\eta)}\right]^2 + \bbE\left[ \frac{\dot{s}_{X, Z}(X, Z)}{s_{X, Z}(X, Z)}\right]^2 \right\} 
\end{align*}
where $\phi_1$ is the square root of the density of standard gaussian distribution, $s_{\eta}$ is the square root of the density of $\eta$ and $s_{X, Z}$ is the joint density of $(X, Z)$. The function $\dot s_{\eta}$ (resp. $\dot s_{X, Z}$) is defined as the $(d/dt) s_{\eta, \gamma(t)} \mid_{t=0}$ (resp. $(d/dt) s_{X, Z, \gamma(t)} \mid_{t=0}$). Similar definition holds for $\dot \alpha, \dot \beta, \dot \gamma$, where we omit the subscript $\gamma$ for notational simplicity. The function $s_{\eta}'$ here denotes the derivative of $s_{\eta, 0}(x)$ (true data generating density) with respect to $x$. Note that in the last equality we reparametrize the variable $(Y,Q, X, Z) \to (\eps, \eta, X, Z)$ which is bijective. The fisher inner product in $T(P_0)$ corresponding to two parametrization $\gamma_1, \gamma_2$ can be expressed as: 
\begin{align*}
\langle \dot s_{\gamma_1, 0}, \dot s_{\gamma_2, 0} \rangle_F & = 4\bbE\left[\left\{\frac{\dot s_{\gamma_1, 0}\left(Y \mid Q, X, Z\right)}{s_0\left(Y \mid Q, X, Z\right)} + \frac{\dot s_{\gamma_1, 0}\left(Q \mid X, Z\right)}{s_0\left(Q \mid X, Z\right)} + \frac{\dot s_{\gamma_1, 0}(X, Z)}{s_0(X, Z)}\right\} \right. \\
& \qquad \qquad \qquad \left. \times \left\{\frac{\dot s_{\gamma_2, 0}\left(Y \mid Q, X, Z\right)}{s_0\left(Y \mid Q, X, Z\right)} + \frac{\dot s_{\gamma_2, 0}\left(Q \mid X, Z\right)}{s_0\left(Q \mid X, Z\right)} + \frac{\dot s_{\gamma_2, 0}(X, Z)}{s_0(X, Z)}\right\} \right]\\
& = 4\bbE\left[\frac{\dot s_{\gamma_1, 0}\left(Y \mid Q, X, Z\right)}{s_0\left(Y \mid Q, X, Z\right)}\frac{\dot s_{\gamma_2, 0}\left(Y \mid Q, X, Z\right)}{s_0\left(Y \mid Q, X, Z\right)} \right. \\
& \hspace{10em} \left. + \frac{\dot s_{\gamma_1, 0}\left(Q \mid X, Z\right)}{s_0\left(Q \mid X, Z\right)}\frac{\dot s_{\gamma_2, 0}\left(Q \mid X, Z\right)}{s_0\left(Q \mid X, Z\right)} \right. \\
& \hspace{20em} \left. + \frac{\dot s_{\gamma_1, 0}(X, Z)}{s_0(X, Z)}\frac{\dot s_{\gamma_2, 0}(X, Z)}{s_0(X, Z)}  \right] \\
& = 4 \left\{\bbE\left[\left(\frac{\phi_1'(\eps)}{\phi_1(\eps)}\right)^2 \left\{-\dot{\alpha}^1 S - X^{\top}\dot{\beta}^1 - \dot{b}^1(\eta) + b_0'(\eta) Z^{\top}\dot{\gamma}^1\right\} \right. \right. \\
& \qquad \qquad \qquad \left. \left. \times \left\{-\dot{\alpha}^2 S - X^{\top}\dot{\beta}^2 - \dot{b}^2(\eta) + b_0'(\eta) Z^{\top}\dot{\gamma}^2\right\}\right]  \right. \\
& \qquad \qquad \qquad  \qquad \left. + \bbE\left[\left\{\frac{\dot{s_{\eta}}^1(\eta) - s_{\eta}'(\eta)Z^{\top}\dot{\gamma}^1}{s_{\eta}(\eta)}\right\}\left\{\frac{\dot{s_{\eta}}^2(\eta) - s_{\eta}'(\eta)Z^{\top}\dot{\gamma}^2}{s_{\eta}(\eta)}\right\}\right] \right. \\
& \qquad \qquad \qquad \qquad \qquad \left. + \bbE\left[\left\{\frac{\dot{s}_{X, Z}^1(X, Z)}{s_{X, Z}(X, Z)}\right\}\left\{\frac{\dot{s}_{X, Z}^2(X, Z)}{s_{X, Z}(X, Z)}\right\}\right] \right\} 
\end{align*}
where the superscript $i \in (1, 2)$ refers to the parametrization corresponding to $\gamma_i$. Our parameter of interest is $\theta(\gamma(t)) = \alpha_{\gamma(t)}$. Differentiating with respect to $t$ we obtain $L(\dot s_{\gamma, 0}) = \dot \alpha_{\gamma, 0} := \dot{\alpha}$. Hence we need to find the representer $s^*$ such that:
\begin{equation}
\label{eq:ip}
\langle s^*, \dot s_{\gamma, 0} \rangle_F = \dot \alpha\,.
\end{equation}
for all $\gamma \in \cR$. This further implies: 
\begin{align}
 \dot{\alpha} &= 4 \left\{\bbE\left[\left(\frac{\phi_1'(\epsilon)}{\phi_1(\epsilon)}\right)^2 \left\{-\dot{\alpha} S - X^{\top}\dot{\beta} - \dot{b}(\eta) + b_0'(\eta) Z^{\top}\dot{\gamma}\right\} \right. \right. \notag \\
 & \qquad \qquad \qquad \left. \left. \times \left\{-\alpha^* S - Z^{\top}\beta^* - b^*(\eta) + b_0'(\eta) Z^{\top}\gamma^*\right\}\right]  \right. \notag \\
& \qquad \qquad  \qquad \qquad \left. + \bbE\left[\left\{\frac{\dot{s_{\eta}}(\eta) - s_{\eta}'(\eta)Z^{\top}\dot{\gamma}}{s_{\eta}(\eta)}\right\}\left\{\frac{s_{\eta}^*(\eta) - s_{\eta}'(\eta)Z^{\top}\gamma^*}{s_{\eta}(\eta)}\right\}\right] \right. \notag \\
\label{eq:sem1}  & \qquad \qquad  \qquad \qquad  \qquad \left. + \bbE\left[\left\{\frac{\dot{s_{X, Z}}(X, Z)}{s_X(X, Z)}\right\}\left\{\frac{s_{X, Z}^*(X, Z)}{s_X(X, Z)}\right\}\right] \right\}
\end{align}
for all $\gamma \in \cR$ and the \emph{optimal asymptotic variance} in estimating $\alpha_0$ is: 
$$
\text{Optimal asymptotic variance } = \|s^*\|_F^2 = \alpha^* \,.
$$ 
In the rest of the analysis we use equation \eqref{eq:ip} repeatedly for different choices of $\dot s_{\gamma, 0}$ to obtain the value of $\alpha^*$. First, putting $\dot{\alpha} = \dot{\beta} = \dot{\gamma} = \dot{b} = \dot{s}_{\eta} = 0$ (as zero vector is always in $T(P_0)$) we obtain: 
$$\bbE\left[\left\{\frac{\dot{s}_{X, Z}(X, Z)}{s_{X, Z}(X, Z)}\right\}\left\{\frac{s_{X, Z}^*(X, Z)}{s_{X, Z}(X, Z)}\right\}\right] = 0 \ \ \forall \ \ \dot s_{X, Z} \,.$$
Hence we have $s^*_{X, Z} \equiv 0$. Thus we can modify equation \eqref{eq:sem1} to obtain: 
\begin{align}
\dot{\alpha} &= 4 \left\{\bbE\left[\left(\frac{\phi_1'(\eps)}{\phi_1(\eps)}\right)^2\left\{-\dot{\alpha} S - X^{\top}\dot{\beta} - \dot{b}(\eta) + b_0'(\eta) Z^{\top}\dot{\gamma}\right\} \right. \right. \notag \\
& \qquad \qquad \left. \left. \times \left\{-\alpha^* S - X^{\top}\beta^* - b^*(\eta_i) + b_0'(\eta) Z^{\top}\gamma^*\right\}\right]  \right. \notag \\
\label{eq:sem2} & \qquad \qquad  \qquad \left. + \bbE\left[\left\{\frac{\dot{s_{\eta}}(\eta) - s_{\eta}'(\eta)Z^{\top}\dot{\gamma}}{s_{\eta}(\eta)}\right\}\left\{\frac{s_{\eta}^*(\eta) - s_{\eta}'(\eta)Z^{\top}\gamma^*}{s_{\eta}(\eta)}\right\}\right] \right\}
\end{align}
Next we put $\dot{\alpha} = \dot{\beta} = \dot{\gamma}  = \dot{b} = 0$ in equation \eqref{eq:sem2} to obtain:
\begin{align*}
 \bbE\left[\left\{\frac{\dot{s_{\eta}}(\eta) }{s_{\eta}(\eta)}\right\}\left\{\frac{s_{\eta}^*(\eta) - s_{\eta}'(\eta)Z^{\top}\gamma^*}{s_{\eta}(\eta)}\right\}\right] = 0
\end{align*}
As $E(X) = 0$, we obtain from the above equation: 
\begin{align*}
 \bbE\left[\left\{\frac{\dot{s_{\eta}}(\eta) }{s_{\eta}(\eta)}\right\}\left\{\frac{s_{\eta}^*(\eta) - s_{\eta}'(\eta)Z^{\top}\gamma^*}{s_{\eta}(\eta)}\right\}\right] = 0
\end{align*}
As $E(Z^{\top}\gamma^*) = 0$, we can conclude that $s^*_{\eta}(\cdot) \equiv 0$.  So modifying equation \eqref{eq:sem2} we get the following equation: 
\begin{align*}
\dot{\alpha}  &= 4 \left\{\bbE\left[\left(\frac{\phi_1'(\eps)}{\phi_1(\eps)}\right)^2\left\{-\dot{\alpha} S - X^{\top}\dot{\beta} - \dot{b}(\eta) + b_0'(\eta) Z^{\top}\dot{\gamma}\right\} \right. \right. \\
& \qquad \qquad \left. \left. \times \left\{-\alpha^* S - X^{\top}\beta^* - b^*(\eta) + b_0'(\eta) Z^{\top}\gamma^*\right\}\right]  \right. \\
& \qquad \qquad  \qquad \left. + \dot{\gamma}^{\top} \bbE\left(\frac{s_{\eta}'(\eta)}{s_{\eta}(\eta)}\right)^2\Sigma_Z \gamma^* \right\} \notag
\end{align*}
Observe that: 
$$4\bbE\left\{ \left(\frac{\phi_1'(\eps)}{\phi_1(\eps)}\right)^2\right\} = \bbE\left\{ \left( \frac{d}{d\epsilon} \log{\phi(\epsilon)}\right)^2\right\} = 1 \,.$$
Hence defining $4 \bbE\left(\frac{s_{\eta}'(\eta)}{s_{\eta}(\eta)}\right)^2 = I_{\eta}$ the above equation becomes: 
\begin{align}
\label{eq:sem3}
\dot{\alpha} &= \left\{\bbE_{X, \eta}\left[\left\{-\dot{\alpha} S - X^{\top}\dot{\beta} - \dot{b}(\eta) + b_0'(\eta) Z^{\top}\dot{\gamma}\right\} \right. \right. \\
& \qquad \qquad \qquad \left. \left. \times \left\{-\alpha^* S - X^{\top}\beta^* - b^*(\eta) + b_0'(\eta) Z^{\top}\gamma^*\right\}\right]  + I_{\eta} \dot{\gamma}^{\top}\Sigma_Z \gamma^* \right\} \notag
\end{align}
Next we put $\dot{\alpha} = \dot{\beta} = \dot{\gamma} = 0$ in equation \eqref{eq:sem3} to get: 
\begin{align*}
\bbE\left[\dot{b}(\eta)\left\{\alpha^* S + X^{\top}\beta^* + b^*(\eta) - b_0'(\eta) Z^{\top}\gamma^*\right\}\right] = 0 \,.
\end{align*}
Again from the independence of $(X, Z)$ and $\eta$ we have: 
\begin{align*}
\bbE\left[\dot{b}(\eta)\left\{\alpha^* S + b^*(\eta)\right\}\right] = 0 \,.
\end{align*}
Hence, it is immediately clear the choice of $b^*(\eta) = -\alpha^* \bbE\left(S \,\middle\vert\,\eta\right)$. Using this we modify the equation \eqref{eq:sem3} as below: 
\begin{align}
\label{eq:sem4}
& \left\{\bbE\left[\left\{\dot{\alpha} S + X^{\top}\dot{\beta} - b_0'(\eta) Z^{\top}\dot{\gamma}\right\}\left\{\alpha^* (S - E(S \,\middle\vert\,\eta)) + X^{\top}\beta^* - b_0'(\eta) Z^{\top}\gamma^*\right\}\right]  \right. \\
& \qquad \qquad  \left. + I_{\eta} \dot{\gamma}^{\top} \Sigma_Z \gamma^* \right\} = \dot{\alpha} \notag
\end{align}
Next putting $\dot{\alpha} = \dot{\beta} = 0$ in equation \eqref{eq:sem4} we get: 
\begin{align}
& \bbE\left[\left\{-b_0'(\eta) Z^{\top}\dot{\gamma}\right\}\left\{\alpha^* (S - E(S \,\middle\vert\,\eta)) + X^{\top}\beta^* - b_0'(\eta) Z^{\top}\gamma^*\right\}\right] + I_{\eta}\dot{\gamma}^{\top} \Sigma_Z \gamma^* = 0 \notag \\
\end{align}
which further implies: 
\begin{align}
& \dot{\gamma}^{\top}\left[\alpha^* \bbE_{\eta}\left\{-b_0'(\eta)\bbE(ZS \,\middle\vert\,\eta)\right\}\right] + \bbE\left\{ -b'_0(\eta) \right\} \dot{\gamma}^{\top} \Sigma_{ZX} \beta^* \notag \\
& \hspace{15em}+ \left[\bbE\left\{(b_0'(\eta))^2\right\} + I_{\eta}\right]\dot{\gamma}^{\top}\Sigma_Z \gamma^* = 0 \notag \\ \notag \\
\implies & \alpha^* \dot{\gamma}^{\top}v_1 + c_1  \dot{\gamma}^{\top} \Sigma_{ZX} \beta^* + c_2 \dot{\gamma}^{\top}\Sigma_Z \gamma^* = 0 \notag \\ \notag \\
\label{eq:final_1}\implies & \alpha^* v_1 + c_1 \Sigma_{ZX} \beta^* + c_2 \Sigma_Z \gamma^* = 0 \,. \\\notag 
\end{align}
Here $v_1 =  \bbE_{\eta}\left\{-b_0'(\eta)\bbE(Z S \,\middle\vert\,\eta)\right\}$, $c_1 =  \bbE\left\{ -b'_0(\eta) \right\}$ and $c_2 = \left[\bbE\left\{(b_0'(\eta))^2\right\} + I_{\eta}\right]$. Now equation \eqref{eq:sem4} will be modified to: 
\begin{align}
\label{eq:sem5}
\notag \\
& \bbE\left[\left\{\dot{\alpha} S + X^{\top}\dot{\beta}\right\}\left\{\alpha^* (S - \bbE(S \,\middle\vert\,\eta)) + X^{\top}\beta^*- b_0'(\eta) Z^{\top}\gamma^* \right\}\right] = \dot{\alpha} 
\end{align} 
Putting $\dot{\alpha} = 0$ in equation \eqref{eq:sem5} we obtain: 
\begin{align}
& \bbE\left[\left\{X^{\top}\dot{\beta}\right\}\left\{\alpha^* (S - E(S \mid \eta)) + X^{\top}\beta^*- b_0'(\eta) Z^{\top}\gamma^* \right\}\right]  = 0 \notag \\ \notag \\
\implies & \dot{\beta}^{\top} \left[ \alpha^* \bbE\left(XS \right) \right] + \dot{\beta}^{\top} \Sigma_X \beta^* + \bbE(-b_0'(\eta)) \dot{\beta}^{\top}\Sigma_{XZ} \gamma^* = 0 \notag \\ \notag \\
\implies & \alpha^* \dot{\beta}^{\top}v_2 + \dot{\beta}^{\top} \Sigma_X \beta^* + c_1 \dot{\beta}^{\top} \Sigma_{XZ} \gamma^* = 0 \notag \\ \notag \\
\label{eq:final_2} \implies & \alpha^* v_2 + \Sigma_X \beta^*  + c_1 \Sigma_{XZ} \gamma^* = 0  \,.
\end{align}
where $v_2 = E(XS)$. Hence equation \eqref{eq:sem5} will be modified to: 
\begin{align}
\label{eq:sem6}
& \bbE\left[\left\{\dot{\alpha} S\right\}\left\{\alpha^* (S - E(S \mid \eta)) + X^{\top}\beta^*- b_0'(\eta) Z^{\top}\gamma^* \right\}\right]   = \dot{\alpha}
\end{align}
which implies: 
\begin{align}
& \bbE\left[s_i\left\{\alpha^* (S - E(S \mid \eta)) + X^{\top}\beta^*- b_0'(\eta) Z^{\top}\gamma^* \right\}\right]  = 1 \notag \\ \notag \\
\implies & \alpha^* \bbE_{\eta}\left\{\var(S \mid \eta)\right\} + \bbE(SX^{\top})\beta^* + \bbE_{\eta}\left\{-b_0'(\eta) \bbE(SZ^{\top}  \mid \eta_i)\right\}\gamma^* = 1 \notag \\ \notag \\
\label{eq:final_3} \implies & \alpha^* c_3 + v_2^{\top}\beta^* + v_1^{\top}\gamma^* = 1 \,.
\end{align} 
where $c_3 = \bbE_{\eta}\left\{\var(S \mid \eta)\right\}$. Finally we have three unknowns ($\alpha^*, \beta^*, \gamma^*$) and three equations (equation \eqref{eq:final_1}, \eqref{eq:final_2} and \eqref{eq:final_3}), which we solve to get the value of $\alpha^*$. For the convenience of the readers we write those equations here: 
\begin{align}
\label{eq:final_4}\alpha^* v_1 + c_1 \Sigma_{ZX} \beta^* + c_2 \Sigma_Z \gamma^* & = 0 \in \reals^{p_2}\,, \\
\label{eq:final_5}\alpha^* v_2 + \Sigma_{X} \beta^*  + c_1 \Sigma_{XZ} \gamma^* & = 0  \in \reals^{p_1}\,, \\
 \label{eq:final_6} \alpha^* c_3 + v_2^{\top}\beta^* + v_1^{\top}\gamma^* & = 1  \in \reals\,.
\end{align}
where $Z \in \reals^{p_2}$ and $X \in \reals^{p_1}$. These three equations can be written in a matrix form as following: 
\begin{equation*}
\begin{bmatrix}
c_3 & v_2^{\top} & v_1^{\top} \\
v_2 & \Sigma_X & c_1\Sigma_{XZ} \\
v_1 & c_1\Sigma_{ZX} & c_2\Sigma_Z 
\end{bmatrix} 
\begin{pmatrix}
\alpha^* \\
\beta^* \\
\gamma^* 
\end{pmatrix} = 
\begin{pmatrix}
1 \\
\mathbf{0} \\
\mathbf{0}
\end{pmatrix}
\end{equation*}
Hence we have: 
\begin{align}
\alpha^*  & = e_1^{\top}\begin{bmatrix}
c_3 & v_2^{\top} & v_1^{\top} \\
v_2 & \Sigma_X & c_1\Sigma_{XZ} \\
v_1 & c_1\Sigma_{ZX} & c_2\Sigma_Z 
\end{bmatrix} ^{-1}e_1 \notag \\
& = e_1^{\top}\begin{bmatrix}
c_3 & v_2^{\top} & -v_1^{\top} \\
v_2 & \Sigma_X & -c_1\Sigma_{XZ} \\
-v_1 & -c_1\Sigma_{ZX} & c_2\Sigma_Z 
\end{bmatrix} ^{-1}e_1 \notag\\
\label{eq:eff_info} & =  e_1^{\top}\begin{pmatrix}
\bbE\left(\var\left(S \,\middle\vert\,\eta\right)\right) & \bbE(SX^{\top}) & \bbE(SZ^{\top}b'(\eta))\\
\bbE(SX) & \Sigma_X & \bbE(XZ^{\top}b'(\eta))\\
\bbE(SZb'(\eta)) & \bbE(ZX^{\top}b'(\eta)) & \Sigma_Z(1 + (b'(\eta))^2)
\end{pmatrix}^{-1} e_1 \\
& = e_1^{\top}\Omega e_1 \,.
\end{align}

\begin{remark}
Note that if $b$ is a linear function, (which happens if $(\epsilon, \eta)$ is generated from bivariate normal with correlation $\rho$), then $b'$ is a constant function. Hence the second term in the expression of efficient information vanishes and we get: 
$$I_{eff} = \left[\left\{\bbE_{\eta}(\var(S\,\middle\vert\,\eta)) - \bbE(SX^{\top})\Sigma_X^{-1}\bbE(SX)\right\}  \right]$$
which is same as the efficient information of partial linear model. Hence, one may think the second term as the price we pay for non-linearity of $b$.  
\end{remark}

\section{Proof of Proposition \ref{thm:spline_consistency}}
\label{sec:proof_prop}
\vspace{0.1in}
\noindent
Recall that our model can be written as:  
\begin{align*}
    Y_i & = \alpha_0 S_i + X_i^{\top}\beta_0 + b(\eta_i) + \epsilon_i \\
    & = \alpha_0 S_i + X_i^{\top}\beta_0 + b(\hat \eta_i) + R_{1,i} + \epsilon_i
\end{align*}
where $R_{1,i} = b(\eta_i) - b(\hat \eta_i)$ is the residual in approximating $\eta_i$ by $\hat \eta_i$. For notational simplicity, we absorb $S_i$ into $X_i$ (and $\alpha_0$ into $\beta_0$) and write: 
\begin{equation}
\label{eq:spline_approx_1}
Y_i = X_i^{\top}\beta_0 + b(\hat \eta_i) + R_{1,i} + \epsilon_i
\end{equation}
As mentioned in Section \ref{sec:est_method}, we only need to estimate the derivative of the mean function $b$ on the interval $[-\tau, \tau]$ and consequently we consider only those observations for which $|\hat \eta_i| \le \tau$. We use scaled B-spline basis $\tilde N_K$ to estimate $b$ non-parametrically on the interval $[-\tau, \tau]$. For more details on the B-spline basis and its scaled version, see \ref{sec:spline_details}. Define a vector $\omega_{b, \infty, n}$ as: 
$$
\omega_{b, \infty, n} = \argmin_{\omega \in \bbR^{(K+2)}} \sup_{|x| \le \tau}\left|b(x) - \tilde N_K(x)^{\top}\omega\right|
$$
By applying Theorem \ref{thm:spline_bias} of Section \ref{sec:spline_details} we conclude:
\begin{align}
\label{eq:bound_b} \left\|b(x) - \tilde N_K(x)^{\top}\omega_{b, \infty, n}\right\|_{\infty, [-\tau, \tau]} & \le C\left(\frac{2\tau}{K}\right)^3\left\|b'''\right\|_{\infty, [-\tau, \tau]} \,,\\
\label{eq:bound_b'} \left\|b'(x) - \nabla \tilde N_K(x)^{\top}\omega_{b, \infty, n}\right\|_{\infty, [-\tau, \tau]} & \le C\left(\frac{2\tau}{K}\right)^2\left\|b'''\right\|_{\infty, [-\tau, \tau]}  \,.
\end{align}
were $\nabla \tilde N_k(x)$ is the vector of derivatives of the co-ordinate of the basis functions in $\tilde N_k(x)$. Using this spline approximation we further expand on equation \eqref{eq:spline_approx_1}: 
\begin{equation}
\label{eq:spline_approx_2}
Y_i = X_i^{\top}\beta_0 + \tilde N_K(\hat \eta_i)^{\top}\omega_{b, \infty, n} + R_{1i} + R_{2i} + \epsilon_i
\end{equation}
where $R_{2,i} = b(\hat \eta_i) -  \tilde N_K(\hat \eta_i)^{\top}\omega_{b, \infty, n}$ is the spline approximation error. To estimate $w_{b, \infty, n}$ we first estimate $\beta_0$ as: 
\begin{equation}
\label{eq:beta_estimate}
\hat \beta = \left(X^{\top}\proj^{\perp}_{\tilde \bN_K}X\right)^{-1}X^{\top}\proj^{\perp}_{\tilde \bN_K}Y \,.
\end{equation}
and then estimate $w_{b, \infty, n}$ as: 
\begin{equation}
\label{eq:omega_estimate}
\hat w_{b, \infty, n} =\left(\tilde \bN_K^{\top}\tilde \bN_K\right)^{-1}\tilde \bN_K^{\top}(Y - X\hat \beta) \,.
\end{equation}
and consequently we set $\hat b(x) = \tilde N_k(x)^{\top}\hat \omega_{b, \infty, n}$ and $\hat b'(x) = \nabla \tilde N_k(x)^{\top}\hat \omega_{b, \infty, n}$. The estimation error of $b'$ using $\hat b'$ is then bounded as follows: 
\begin{align}
\left|\hat b'(x) - b'(x)\right| & = \left|\nabla \tilde N_k(x)^{\top}\hat \omega_{b, \infty, n} - b'(x)\right| \notag \\
& \le \left|\nabla \tilde N_k(x)^{\top}\hat \omega_{b, \infty, n} - \nabla \tilde N_k(x)^{\top}\omega_{b, \infty, n}\right| + \left|\nabla \tilde N_k(x)^{\top}\omega_{b, \infty, n}  - b'(x)\right| \notag \\
& \le \left\| \nabla \tilde N_k(x) \right\| \left\|\hat \omega_{b, \infty, n} -  \omega_{b, \infty, n}\right\| + \sup_{|t| \le \tau} \left|\nabla \tilde N_k(t)^{\top}\omega_{b, \infty, n}  - b'(t)\right|  \notag \\
\label{eq:final_bound_1} & \lesssim  K\sqrt{K}\left\|\hat \omega_{b, \infty, n} -  \omega_{b, \infty, n}\right\|  +  C\left(\frac{2\tau}{K}\right)^2\left\|b'''\right\|_{\infty, [-\tau, \tau]}  
\end{align}
where the last inequality follows from the fact that $\|\nabla \tilde N_k(x)\| \lesssim K\sqrt{K}$ (see Lemma \ref{lem:spline_deriv_bound}) and equation \eqref{eq:bound_b'}. We now relate $\hat \omega_{b, \infty, n}$  to $\omega_{b, \infty, n}$ using equation \eqref{eq:spline_approx_2}: 
\begin{align}
\hat w_{b, \infty, n} &=\left(\tilde \bN_K^{\top}\tilde \bN_K\right)^{-1}\tilde \bN_K^{\top}(Y - X\hat \beta) \notag \\
&= w_{b, \infty, n} +\left(\tilde \bN_K^{\top}\tilde \bN_K\right)^{-1}\tilde \bN_K^{\top}X\left(\hat \beta - \beta_0\right) +\left(\tilde \bN_K^{\top}\tilde \bN_K\right)^{-1}\tilde \bN_K^{\top}\bR_1 \notag \\
& \hspace{6cm} +\left(\tilde \bN_K^{\top}\tilde \bN_K\right)^{-1}\tilde \bN_K^{\top}\bR_2 +\left(\tilde \bN_K^{\top}\tilde \bN_K\right)^{-1}\tilde \bN_K^{\top}\beps \notag \\
&= w_{b, \infty, n} + \left(\frac{\tilde \bN_K^{\top}\tilde \bN_K}{n}\right)^{-1}\frac{\tilde \bN_K^{\top}X\left(\hat \beta - \beta_0\right)}{n} + \left(\frac{\tilde \bN_K^{\top}\tilde \bN_K}{n}\right)^{-1}\frac{\tilde \bN_K^{\top}\bR_1}{n} \notag \\
& \hspace{6cm} + \left(\frac{\tilde \bN_K^{\top}\tilde \bN_K}{n}\right)^{-1}\frac{\tilde \bN_K^{\top}\bR_2}{n} + \left(\frac{\tilde \bN_K^{\top}\tilde \bN_K}{n}\right)^{-1}\frac{\tilde \bN_K^{\top}\beps}{n} \notag \\
& = w_{b, \infty, n} + \left(\frac{\tilde \bN_K^{\top}\tilde \bN_K}{n}\right)^{-1}\frac{\tilde \bN_K^{\top}X\left(\hat \beta - \beta_0\right)}{n} + \left(\frac{\tilde \bN_K^{\top}\tilde \bN_K}{n}\right)^{-1}\left[\frac{\tilde \bN_K^{\top}\bR_1}{n} + \frac{\tilde \bN_K^{\top}\bR_2}{n} + \frac{\tilde \bN_K^{\top}\beps}{n}\right] \notag \\
\label{eq:omega_hat_exp} & = w_{b, \infty, n} + T_1 + \left(\frac{\tilde \bN_K^{\top}\tilde \bN_K}{n}\right)^{-1}\left(T_2 + T_3 + T_4\right)
\end{align}
Rest of the proof is devoted to show $\left\|\hat \omega_{b, \infty, n} -  \omega_{b, \infty, n}\right\| = o_p\left(K^{-3/2}\right)$ via bounding $T_1, T_2, T_3$ and $T_4$.

\subsection{Bounding $T_1$ } 
\label{sec:t1_spline}
\vspace{0.1in}
\noindent
To bound $T_1$, we first bound $\|\hat \beta - \beta_0\|$ using the following Lemma: 
\begin{lemma}
\label{lem:beta_bound}
Under assumptions \ref{assm:independence}-\ref{assm:moment} we have $\| \hat \beta - \beta_0\| = o_p\left(K^{-3/2}\right) \,.$
\end{lemma} 
\noindent
The proof of Lemma \ref{lem:beta_bound} is similar to the proof of matrix convergence portion of the proof of Theorem \ref{thm:main_thm} and is deferred to Section \ref{sec:additional_results}. We now show that the operator norm of $(\tilde \bN_k^{\top}\tilde \bN_k/n)^{-1}$ is bounded above.  Using a conditional version of  Theorem \ref{thm:rudelson} of Section \ref{sec:spline_details} we have: 
$$
\bbE\left[\left\|\frac{\tilde \bN_k^{\top}\tilde \bN_k}{n} - \bbE\left(\tilde N_k(\hat \eta)\tilde N_k(\hat \eta)^{\top}\mathds{1}_{|\hat \eta| \le \tau} \,\middle\vert\,\cF(\cD_1)\right)\right\|_{op} \,\middle\vert\,\cF(\cD_1) \right] \le C\left( \frac{K\log{K}}{n} + \sqrt{\frac{K\log{K}}{n}}\right)
$$
As the bound on the right side does not depend on $\cF(\cD_1)$, we conclude, taking expectation on the both side: 
\begin{equation}
\label{eq:rudelson_bound}
\bbE\left[\left\|\frac{\tilde \bN_k^{\top}\tilde \bN_k}{n} - \bbE\left(\tilde N_k(\hat \eta)\tilde N_k(\hat \eta)^{\top}\mathds{1}_{|\hat \eta| \le \tau} \,\middle\vert\,\cF(\cD_1)\right)\right\|_{op} \right] \le C\left( \frac{K\log{K}}{n} + \sqrt{\frac{K\log{K}}{n}}\right)
\end{equation}
Note that we can write with $a_n = \log{n}/\sqrt{n}$: 
\begin{align*}
& \bbE\left(\tilde N_k(\hat \eta)\tilde N_k(\hat \eta)^{\top}\mathds{1}_{|\hat \eta| \le \tau} \,\middle\vert\,\cF(\cD_1)\right) \\
& = \bbE\left(\tilde N_k(\hat \eta)\tilde N_k(\hat \eta)^{\top}\mathds{1}_{|\hat \eta| \le \tau} \,\middle\vert\,\cF(\cD_1)\right) \mathds{1}_{\|\hat \gamma_n - \gamma_0\| \le a_n} \\
& \qquad \qquad \qquad + \bbE\left(\tilde N_k(\hat \eta)\tilde N_k(\hat \eta)^{\top}\mathds{1}_{|\hat \eta| \le \tau} \,\middle\vert\,\cF(\cD_1)\right) \mathds{1}_{\|\hat \gamma_n - \gamma_0\| > a_n} \,. 
\end{align*}
As both of the matrices on the right side are p.s.d., we conclude: 
\begin{align*}
& \lambda_{\min}\left(\bbE\left(\tilde N_k(\hat \eta)\tilde N_k(\hat \eta)^{\top}\mathds{1}_{|\hat \eta| \le \tau} \,\middle\vert\,\cF(\cD_1)\right)\right) \\
& \qquad \qquad \qquad \ge \lambda_{\min}\left(
\bbE\left(\tilde N_k(\hat \eta)\tilde N_k(\hat \eta)^{\top}\mathds{1}_{|\hat \eta| \le \tau} \,\middle\vert\, \cF(\cD_1)\right) \mathds{1}_{\|\hat \gamma_n - \gamma_0\| \le a_n}\right) \,.
\end{align*}
Choose $\delta > 0$ such that $\bbP(\|z\| \le \delta) > 0$. Now, from Theorem \ref{thm:eigen_spline}: 
\begin{align*}
& \hspace{-2em} \lambda_{\min}\left(
\bbE\left(\tilde N_k(\hat \eta)\tilde N_k(\hat \eta)^{\top}\mathds{1}_{|\hat \eta| \le \tau} \,\middle\vert\, \cF(\cD_1)\right) \mathds{1}_{\|\hat \gamma_n - \gamma_0\| \le 1}\right) \\
& \ge \kappa_4 \min_{|x| \le \tau} f_{\eta + (\gamma_0 - \hat \gamma_n)^{\top}Z}(x)\mathds{1}_{\|\hat \gamma_n - \gamma_0\| \le a_n} \\
& \ge \kappa_4 \min_{\substack{|x| \le \tau \\ \|a\| \le a_n}} f_{\eta + a^{\top}Z}(x)\mathds{1}_{\|\hat \gamma_n - \gamma_0\| \le a_n}  \\
& = \kappa_4 \min_{\substack{|x| \le \tau \\ \|a\| \le a_n}} \int_{\bbR^p}f_{\eta}(x - a^{\top}z)f_Z(z) \ dz  \mathds{1}_{\|\hat \gamma_n - \gamma_0\| \le a_n} \\
& \ge \kappa_4 \min_{\substack{|x| \le \tau \\ \|a\| \le a_n}} \int_{\|z\| \le \delta}f_{\eta}(x - a^{\top}z)f_Z(z) \ dz \mathds{1}_{\|\hat \gamma_n - \gamma_0\| \le a_n}  \\
& \ge \kappa_4\min_{|x| \le \tau+a_n \delta} f_\eta(x) P(\|Z\| \le \delta) \mathds{1}_{\|\hat \gamma_n - \gamma_0\| \le a_n}  \,. 
\end{align*}
Now for large $n$, $a_n \delta \le \xi$ (where $\xi$ is same as defined in (iv) of Assumption \ref{assm:b}) and hence for all large $n$: 
\begin{align}
\label{eq:min_eigval}
& \lambda_{\min}\left(
\bbE\left(\tilde N_k(\hat \eta)\tilde N_k(\hat \eta)^{\top}\mathds{1}_{|\hat \eta| \le \tau} \,\middle\vert\, \cF(\cD_1)\right) \mathds{1}_{\|\hat \gamma_n - \gamma_0\| \le 1}\right) \notag\\
& \qquad \qquad \qquad \ge \kappa_4\min_{|x| \le \tau+ \eps} f_\eta(x) P(\|Z\| \le \delta)\mathds{1}_{\|\hat \gamma_n - \gamma_0\| \le 1} \,.
\end{align}
From equation \eqref{eq:rudelson_bound} and \eqref{eq:min_eigval} we conclude: 
\begin{equation}
\label{eq:eigen_final_bound}
\left\|\left(\frac{\tilde \bN_k^{\top}\tilde \bN_k}{n}\right)^{-1}\right\|_{op} = \left(\lambda_{\min}\left(\frac{\tilde \bN_k^{\top}\tilde \bN_k}{n}\right)\right)^{-1} = O_p(1) \,.
\end{equation}
We next also provide a bound Now going back to $T_1$ in equation \eqref{eq:omega_hat_exp} we have: 
\begin{align*}
\|T_1\| & = \left\|\left(\frac{\tilde \bN_K^{\top}\tilde \bN_K}{n}\right)^{-1}\frac{\tilde \bN_K^{\top}X\left(\hat \beta - \beta_0\right)}{n}\right\|  \\
& \le \left\|\left(\frac{\tilde \bN_K^{\top}\tilde \bN_K}{n}\right)^{-1/2}\right\|_{op} \left\|\left(\frac{\tilde \bN_K^{\top}\tilde \bN_K}{n}\right)^{-1/2}\frac{\tilde \bN_K^{\top}X\left(\hat \beta - \beta_0\right)}{n}\right\| \\
& \lesssim_P \left\|\left(\frac{\tilde \bN_K^{\top}\tilde \bN_K}{n}\right)^{-1/2}\frac{\tilde \bN_K^{\top}X\left(\hat \beta - \beta_0\right)}{n}\right\| \\
& \le \left\|\frac{X(\hat \beta - \beta_0)}{\sqrt{n}}\right\| \\
& \le \|\hat \beta - \beta_0\| \sqrt{\lambda_{\max}(X^{\top}X/n)} \lesssim_P \|\hat \beta - \beta_0\| = o_p\left(K^{-3/2}\right) \hspace{0.2in} [\text{By Lemma } \ref{lem:beta_bound}]\,. 
\end{align*}

\subsection{Bounding $T_2$ }
\label{sec:t2_spline}
\noindent
The term $T_2$ can be bounded as: 
\begin{align*}
\|T_2\| = \left\|\frac{\tilde \bN_K^{\top}\bR_1}{n}\right\| & \le \left\|\frac{\bR_1}{\sqrt{n}}\right\| \lambda_{\max}\left(\frac{\tilde \bN_K^{\top}\tilde \bN_K}{n}\right)
\end{align*}
It was already proved in equation \eqref{eq:eigen_final_bound} that $\lambda_{\max}(\tilde \bN_K^{\top}\tilde \bN_K/n) = O_p(1)$. To control $\left\|\bR_1\right\|/\sqrt{n}$:   
\begin{align*}
\bbE\left(\frac{\|\bR_1\|^2}{n} \,\middle\vert\, \cF(\cD_1) \right) & = \bbE\left((b(\eta) - b(\hat \eta))^2 \,\middle\vert\,\cD_1\right) \\
& = \bbE\left(\left(b'(\eta)(\hat \eta - \eta) + (1/2)b''(\tilde \eta)(\hat \eta - \eta)^2\right)^2 \,\middle\vert\, \cF(\cD_1)\right) \\
& \le 2 \|\hat \gamma_n - \gamma_0\|^2 \bbE\left(\left(b'(\eta)\|Z\|\right)^2\right) + \frac12 \|b''\|_{\infty} \|\hat \gamma_n - \gamma_0\|^4 \bbE((\|Z\|^4)) \\
& = O_p(n^{-1}) \,.
\end{align*}
Hence $\|\bR_1/\sqrt{n}\| = O_p(n^{-1/2}) = o_p(K^{-3/2})$, where the last equality follows from Remark \ref{rem:K}.  
\noindent
\subsection{Bounding $T_3$ }
\label{sec:t3_spline}
\begin{align*}
\|T_3\| = \left\|\frac{\tilde \bN_K^{\top}\bR_2}{n}\right\| & \le \left\|\frac{\bR_2}{\sqrt{n}}\right\| \lambda_{\max}\left(\frac{\tilde \bN_K^{\top}\tilde \bN_K}{n}\right)
\end{align*}
With similar logic used in bounding $T_2$, all we need to bound $\|\bR_2\|/\sqrt{n}$. It is immediate that:
\begin{align*}
\frac{\|\bR_2\|}{\sqrt{n}} = \sqrt{\frac{\|\bR_2\|^2}{n}} & \le \sup_{|x| \le \tau} \left|b(t) - \tilde N_K(t)^{\top}\omega_{b, 2, \infty}\right| \\
& \le C\left(\frac{2\tau}{K}\right)^3 \|b'''\|_{\infty} = o_p\left(K^{-3/2}\right)\,.
\end{align*}

\subsection{Bounding $T_4$ }
\label{sec:t4_spline}
\begin{align*}
\bbE\left(\left\|\frac{\tilde \bN_K^{\top}\beps}{n}\right\|^2 \,\middle\vert\,\cF(\cD_1) \right) &= \frac{1}{n^2} \bbE\left[\beps^{\top}\tilde \bN_K \tilde \bN_K^{\top}\beps \,\middle\vert\,\cF(\cD_1) \right] \\
& = \frac{1}{n^2} \tr\left( \bbE\left[\beps \beps^{\top}\tilde \bN_K^{\top} \tilde \bN_K \,\middle\vert\,\cF(\cD_1) \right] \right) \\
& = \frac{1}{n^2} \tr\left( \bbE\left[\bbE\left(\beps \beps^{\top} \,\middle\vert\,\cF(\bZ, \boldsymbol\eta, \cD_1) \right)\tilde \bN_K^{\top}\tilde \bN_K \,\middle\vert\,\cF(\cD_1) \right] \right)  \\
& \le \frac{K+3}{n} \sup_{\eta} \var\left(\eps \,\middle\vert\,\eta\right) \lambda_{\max} \left(\bbE\left[\tilde N_K(\hat \eta)\tilde N_K(\hat \eta)^{\top}\mathds{1}_{|\hat \eta| \le \tau} \,\middle\vert\,\cF(\cD_1)\right]\right)  \\
& = O_p\left(\frac{K}{n}\right) = o_p\left(K^{-3/2}\right) \hspace{0.2in} [\text{Remark } \ref{rem:K}]\,.
\end{align*}
where $\boeta$ is $\{\eta_i\}_{i=1}^{n/3}$ in $\cD_2$ and the penultimate inequality follows from Lemma \ref{lem:trace_ineq}. These bounds established $\|\hat \omega_{b, n, \infty} - \omega_{b,   \infty}\| = o_p\left(K^{-3/2}\right)$ which completes the proof.

\section{Proof of supplementary lemmas}  
\label{sec:additional_results}

\subsection{Proof of Lemma \ref{lem:g}}
\noindent
The definition of function $g$ is as follows: 
$$
g(a, t) =\bbE\left[\begin{pmatrix} S & X & b'(\eta)Z \end{pmatrix} \,\middle\vert\,\eta + a^{\top}Z = t\right]
$$
Note that $g(a, t) \in \bbR^{1+p_1+p_2}$. Divide the components of $g$ in three parts as follows: 
\begin{enumerate}
\item $g_1(a, t) = \bbE\left[S \,\middle\vert\,\eta + a^{\top}Z = t\right]$. 
\item $g_2(a, t) = \bbE\left[X \,\middle\vert\,\eta + a^{\top}Z = t\right]$. 
\item $g_3(a, t) = \bbE\left[b'(\eta)Z \,\middle\vert\,\eta + a^{\top}Z = t\right]$.  
\end{enumerate}
If we prove the continuity of partial derivates of $g_1, g_2, g_3$ separately then we are done. We start with $g_1$. For fixed $a$ (i.e. we consider the partial derivative with respect to $t$): 
\begin{align*}
g_1(a, t) & = \bbE\left[S \,\middle\vert\,\eta + a^{\top}Z = t\right] \\
& = \frac{\int_{z: t > (a+ \gamma_0)^{\top}z} f_Z(z)f_{\eta}(t-a^{\top}z) \ dz}{f_{\eta + Z^{\top}a}(t)} \\
& \triangleq \frac{\int_{\Omega(t)} f_Z(z)f_{\eta}(t-a^{\top}z) \ dz}{f_{\eta + Z^{\top}a}(t)}
\end{align*}
By using Leibnitz rule for differentiating integral with varying domain, we can immediately conclude $\partial_t g(a, t)$ is continuously differentiable. The calculation for $\partial_a g(a, t)$ is similar and hence skipped for brevity. 
\\\\
\noindent
For $g_2$ define $h(Z) = E\left(X \,\middle\vert\,Z\right)$. Then we have: 
\begin{align*}
g_2(a,t) = \bbE\left[h(Z) \,\middle\vert\,\eta + a^{\top}Z = t\right] = \frac{\int h(z) f_Z(z)f_{\eta}(t - a^{\top}z) \ dz}{\int f_Z(z) f_{\eta}(t - a^{\top}z) \ dz}
\end{align*}
That $g_2$ is continuous both with respect to $a$ and $t$ follows from DCT and the fact that $E(\|h(Z)\|) < \infty$ (Assumption \ref{assm:moment}) and $\|f_{\eta}\|_{\infty}$ is finite (Assumption \ref{assm:b}). The differentiability and continuity of the derivative also follows from the fact that $f_{\eta}$ is differentiable  and $\bbE(\|Z\|\|h(Z)\|) < \infty$ as well as $\bbE(\|Z\|) < \infty$. Also note that differentiation under integral sign is allowed as $\bbE(\|h(Z)\|) < \infty$. 
\\\\
\noindent
Finally for $g_3$ define $h_1(Z) = b'(t - a^{\top}Z)Z$. Then we have: 
\begin{align*}
g_3(a, t) & = \bbE\left[h_1(Z) \,\middle\vert\,\eta + a^{\top}Z = t\right] 
\end{align*}
By the same logic as for $g_2$ (i.e. using $E(\|Z\||b'(t - a^{\top}Z)|) < \infty, E(\|Z\|^2|b'(t - a^{\top}Z)|) < \infty$ and $\|f_{\eta}\|_{\infty} < \infty$) our conclusion follows. Finally the continuity of $V(a, t)$ follows directly from the continuity of density of $\eta$ (Assumption \ref{assm:b}) and $(X, Z)$ (Assumption \ref{assm:moment}).

\subsection{Proof of Lemma \ref{lem:beta_bound}}
\noindent
Recall that $\beta_0$ in this proof is $(\alpha_0, \beta_0^{\top})^{\top}$ and $X = (S, X) \in \bbR^{n_2 \times (1+p_1)}$ as mentioned in the beginning of the proof of Proposition \ref{thm:spline_consistency}. From the definition of $\hat \beta$ (equation \eqref{eq:beta_estimate}):
\begin{align*}
\hat \beta & = \left(X^{\top}\proj^{\perp}_{\tilde \bN_K}X\right)^{-1}X^{\top}\proj^{\perp}_{\tilde \bN_K}Y \\
& = \left(\frac{X^{\top}\proj^{\perp}_{\tilde \bN_K}X}{n}\right)^{-1}\frac{X^{\top}\proj^{\perp}_{\tilde \bN_K}Y}{n} \\
& = \beta_0 +\left(\frac{X^{\top}\proj^{\perp}_{\tilde \bN_K}X}{n}\right)^{-1}\left[\frac{X^{\top}\proj^{\perp}_{\tilde \bN_K}\bR_1}{n}  +\frac{X^{\top}\proj^{\perp}_{\tilde \bN_K}\bR_2}{n} + \frac{X^{\top}\proj^{\perp}_{\tilde \bN_K}\beps}{n} \right]
\end{align*}
We divide the entire proof into few steps which we articulate below first: 
\begin{enumerate}
\item First we show the matrix $X^{\top}\proj^{\perp}_{\tilde \bN_K}X/n = O_p(1)$. More specifically we show that: 
$$
\frac{X^{\top}\proj^{\perp}_{\tilde \bN_K}X}{n} \overset{P}{\longrightarrow} 
\frac{1}{3}\bbE\left[(X - \bbE(X \,\middle\vert\,\eta)(X - \bbE(X \,\middle\vert\,\eta))^{\top})\mathds{1}_{|\eta| \le \tau}\right] \,.
$$
This along with Assumption \ref{assm:eigen}. implies that $(X^{\top}\proj^{\perp}_{\tilde \bN_K}X/n)^{-1} = O_p(1)$. 
\item Next we show that the residual terms are negligible: 
$$
\left(\frac{X^{\top}\proj^{\perp}_{\tilde \bN_K}X}{n}\right)^{-1}\left(\frac{X^{\top}\proj^{\perp}_{\tilde \bN_K}(\bR_1 + \bR_2)}{n}\right) = o_p(K^{-3/2}) \,. 
$$
\item Finally we show that $\bX^{\top}\proj^{\perp}_{\tilde \bN_K}\beps/n = o_p(K^{-3/2})$. This will complete the proof.  
\end{enumerate}
{\bf Proof of Step 1: }Recall the definition of $\bW^*_1$ from Subsection \ref{sec:s1_main}. It then follows immediately that: 
$$
X = \bW^*_1 \begin{bmatrix}e_1 & e_2 & \dots & e_{1+p_1} \end{bmatrix} := \bW^*_1 A
$$
where $e_i$ is the $i^{th}$ canonical basis of $\bbR^{(1+p_1+p_2)}$. Hence we have: 
$$
\frac{X^{\top}\proj^{\perp}_{\tilde \bN_K}X}{n} = A^{\top}\frac{\bW_1^{*^{\top}}\proj^{\perp}_{\tilde \bN_K}\bW_1^*}{n}A
$$
As established in Subsection \ref{sec:s1_main} (see equation \eqref{eq:wstar}): 
$$
\frac{\bW_1^{*^{\top}}\proj^{\perp}_{\tilde \bN_K}\bW_1^*}{n} \overset{P}{\to} \frac13 \Omega_{\tau} \,,
$$
we conclude from that: 
\begin{align}
\frac{X^{\top}\proj^{\perp}_{\tilde \bN_K}X}{n} & \overset{P}{\longrightarrow} \frac13 A^{\top}\Omega_{\tau}A \notag \\
\label{eq:s1_beta} & = \frac13 \bbE\left[\left(X - \bbE(X \,\middle\vert\,\eta)\right)\left(X - \bbE(X \,\middle\vert\,\eta)\right)^{\top}\mathds{1}_{|\eta| \le \tau}\right] \,.
\end{align}
\\\\
\noindent
{\bf Proof of Step 2: }
\label{sec:s2_beta}
For the first residual term observe that: 
\begin{align}
& \left\|\left(\frac{\bX^{\top}\proj^{\perp}_{\tilde \bN_K}\bX}{n}\right)^{-1}\frac{\bX^{\top}\proj^{\perp}_{\tilde \bN_K}\bR_1}{n}\right\| \notag \\
& \le \left\|\left(\frac{\bX^{\top}\proj^{\perp}_{\tilde \bN_K}\bX}{n}\right)^{-1/2} \right\|_{op}\left\|\left(\frac{\bX^{\top}\proj^{\perp}_{\tilde \bN_K}\bX}{n}\right)^{-1/2}\frac{\bX^{\top}\proj^{\perp}_{\tilde \bN_K}\bR_1}{n}\right\| \notag \\
\label{eq:s2_beta} & \le \left\|\left(\frac{\bX^{\top}\proj^{\perp}_{\tilde \bN_K}\bX}{n}\right)^{-1/2} \right\|_{op} \left\|\frac{\bR_1}{\sqrt{n}}\right\| = o_p(K^{-3/2}) \,.
\end{align}
The last equality follows from the fact that the first term of the above product is $O_p(1)$ and the second term $\|\bR_1\|/\sqrt{n}$ is $o_p(K^{-3/2})$ because: 
\begin{align*}
\frac{\bR_1^{\top}\bR_1}{n} & = \frac{1}{n}\sum_{i=1}^{n/3}\left(b(\eta_i) - b(\hat \eta_i)\right)^2 \mathds{1}_{|\hat \eta_i| \le \tau} \\
& \le \frac{2}{n}\sum_{i=1}^{n/3}\left(\eta_i - \hat \eta_i\right)^2(b'(\hat \eta_i))^2 \mathds{1}_{|\hat \eta_i| \le \tau} + \frac{1}{2n}\sum_{i=1}^{n/3}\left(\eta_i - \hat \eta_i\right)^4(b''(\tilde \eta_i))^2 \mathds{1}_{|\hat \eta_i| \le \tau} \\
& \le \|b'\|_{\infty, [-\tau, \tau]} \|\hat \gamma_n - \gamma_0\|^2 \frac2n \sum_{i=1}^{n/3}\|Z_i\|^2 + \|b''\|_{\infty} \|\hat \gamma_n - \gamma_0\|^4 \frac2n \sum_{i=1}^{n/3}\|Z_i\|^4 \\
& = O_p(n^{-1}) + O_p(n^{-2}) = o_p(K^{-3}) \hspace{0.2in} [\text{Remark }\ref{rem:K}]\,.
\end{align*}
For the other residual using the same calculation we first conclude: 
\begin{equation}
\label{eq:s3_beta}
\left\|\left(\frac{\bX^{\top}\proj^{\perp}_{\tilde \bN_K}\bX}{n}\right)^{-1}\frac{\bX^{\top}\proj^{\perp}_{\tilde \bN_K}\bR_2}{n}\right\|  \le \left\|\left(\frac{\bX^{\top}\proj^{\perp}_{\tilde \bN_K}\bX}{n}\right)^{-1/2} \right\|_{op} \left\|\frac{\bR_2}{\sqrt{n}}\right\| 
\end{equation}
and $\|\bR_2\|/\sqrt{n} = o_p(K^{-3/2})$ follows directly from equation \eqref{eq:bound_b}.  
\\\\
\noindent
{\bf Proof of Step 3: } 
\label{sec:s3_beta}
Finally we show that $\bX^{\top}\proj^{\perp}_{\tilde \bN_K}\beps/n = o_p(1)$ which completes the proof. Recall that, in Subsection \ref{sec:s1_main}, we show that the term $T_1 =$.  This immediately implies: 
$$
\frac{\bW_1^{*^{\top}}\proj^{\perp}_{\tilde \bN_K}\beps}{n} = O_p(n^{-1/2}) \,.
$$
which, in turn implies: 
\begin{equation}
\label{eq:s4_beta}
\frac{\bX^{\top}\proj^{\perp}_{\tilde \bN_K}\beps}{n} = A^{\top}\frac{\bW_1^{*^{\top}}\proj^{\perp}_{\tilde \bN_K}\beps}{n} A = O_p(n^{-1/2}) = o_p(K^{-3/2}) \hspace{0.1in} [\text{Remark }\ref{rem:K}]\,.
\end{equation}
Combining equation \eqref{eq:s1_beta}, \eqref{eq:s2_beta}, \eqref{eq:s3_beta} and \eqref{eq:s4_beta} we conclude $\|\hat \beta - \beta_0\| =o_p(K^{-3/2})$. 
\subsection{Some auxiliary lemmas} 
\label{sec:aux_lems}
In this section we present some auxiliary lemmas that are necessary to establish our main results. 
\begin{lemma}
\label{lem:trace_ineq}
Suppose $A$ is a p.s.d. matrix and $B$ is a symmetric matrix, then $\tr(AB) \le \lambda_{\max}(B) \tr(A)$. 
\end{lemma}
\begin{proof}
Note that $B - \lambda_{\max}(B)I \preceq 0$. Hence: 
$$
\tr\left(A(\lambda_{\max}(B)I - B)\right) = \lambda_{\max}(B) \tr(A) - \tr(AB) \ge 0 \,.
$$
as $A(\lambda_{\max}(B)I - B)$ is a p.s.d. matrix. 
\end{proof}
\begin{lemma}
\label{lem:cond_prob}
Suppose $\{X_n\}_{n \in \bbN}$ is a sequence of non-negative random variables and $\{\cF_n\}_{n \in \bbN}$ is a sequence of sigma fields. If $\bbE\left(X_n \mid \cF_n\right) = o_p(1)$, then $X_n = o_p(1)$. 
\end{lemma}
\begin{proof}
Fix $\eps > 0$. From a conditional version of Markov inequality, we have: 
$$
Y_n := \bbP\left(X_n > \eps \mid \cF_n\right) \le \frac{\bbE\left(X_n \mid \cF_n\right)}{\eps} = o_p(1)\,.
$$
Now as $\{Y_n\}_{n \in \bbN}$ is bounded sequence of random variables which converge to $0$ in probability, applying DCT we conclude: 
$$
\bbP\left(X_n > \eps\right) = \bbE(Y_n)  = o(1) \,.
$$
This completes the proof. 
\end{proof}
\section{Some preliminary discussion on B-spline basis} \
\label{sec:spline_details}
Recall that, we have mentioned in Section \ref{sec:est_method} of the main document that we use truncated B-spline basis to approximate both the unknown mean function $b(\eta) = \bbE(\nu \mid \eta)$ and its derivative. More specifically, we fit spline basis on $[-\tau_n, \tau_n]$ and then and define our estimator to be $0$ outside of it. Recall that the B-spline basis starts with $0^{th}$ order polynomial (i.e. a constant functions) and then is recursively defined for higher order polynomial. Let the knots inside $[-\tau_n, \tau_n]$ are: 
$$
-\tau_n = \xi_0 < \xi_1 < \xi_2 < \dots < \xi_K = \tau_n
$$ 
As we know from spline theory, the dimension of the space generated by spline basis of degree $s$ is $K + s$. When $s = 0$, (i.e. constant functions) we need $K$ basis functions which are defined as: 
$$
N_{i,0}(t) =
\begin{cases}
1 & \text{if } \xi_i \le t \le \xi_{i+1} \\
0 & \text{otherwise}
\end{cases} \,.
$$
for $0 \le i \le K-1$. Now we define the recursion, i.e. how we go to a collection of B-spline basis functions of degree $p$ degree from a collection of B-spline basis functions of degree $(p-1)$. Note that we need $K + p$ many basis functions of degree $p$ and each of the basis functions will be local in a sense that they only have support over $p+1$ many intervals (observe that the constant functions are only supported over one interval). For that we first append some knots at the both ends. For example, to go to $1$ degree polynomial basis from $0$-degree, we append two knots, one at the beginning and one at the very end. 
$$
-\tau_n = \xi_{-1} = \xi_0 < \xi_1 < \dots < \xi_K = \xi_{K+1} = \tau_n \,.
$$
Our $k+1$ basis functions are defined as: 
$$
N_{i,1}(t) =
\begin{cases}
\frac{t - \xi_i}{\xi_{i+1} - \xi_i}N_{i, 0}(t) + \frac{\xi_{i+2} - t}{\xi_{i+2} - \xi_{i+1}}N_{(i+1), 0}(t)  & \text{if } \xi_i \le t \le \xi_{i+2} \\
0 & \text{otherwise}
\end{cases} \,.
$$
for $-1 \le i \le K-1$. Note that when $i = -1$, $N_{-1, 0}$ does not exist. Hence for that case we forget this part and define: 
$$
N_{-1,1}(t) =
\begin{cases}
\frac{\xi_{1} - t}{\xi_{1} - \xi_{0}}N_{0, 0}(t)  & \text{if } \xi_{-1} \le t \le \xi_{1} \\
0 & \text{otherwise}
\end{cases} \,.
$$
and the for the last basis, i.e. $i = K-1$, $B_{K, 0}$ is not defined. Hence analogously we define: 
$$
N_{K-1,1}(t) =
\begin{cases}
\frac{t - \xi_{K-1}}{\xi_{K} - \xi_{K-1}}N_{K-1, 0}(t) & \text{if } \xi_{K-1} \le t \le \xi_{K} \\
0 & \text{otherwise}
\end{cases} \,.
$$
Now we will extend this pattern for any general degree $p$. For that we need to append $p$ knots at the both ends: 
 $$
-\tau_n = \xi_{-p} = \cdots = \xi_{-1} = \xi_0 < \xi_1 < \dots < \xi_K = \xi_{K+1} = \dots \xi_{k+p} = \tau_n \,.
$$
and the recursion is defined as: 
$$
N_{i,p}(t) =
\begin{cases}
\frac{t - \xi_i}{\xi_{i+p} - \xi_i}N_{i, (p-1)}(t) + \frac{\xi_{i+p+1} - t}{\xi_{i+p+1} - \xi_i}N_{(i+1), (p-1)}(t)  & \text{if } \xi_i \le t \le \xi_{i+p+1} \\
0 & \text{otherwise}
\end{cases} \,.
$$
for $-p \le i \le k-1$. Define a class of functions $\bbS_{p, k}$ is the linear combinations of all the functions of $p^{th}$ order B-spline basis $\{N_{i, p}\}_{i=-p}^{k-1}$, i.e.
$$
\bbS_{p, k} =\left\{f: f(x) = \sum_{i=-p}^{k-1} c_i N_{i, p}(x) \text{ with } c_{-p},\dots, c_{k-1} \in \bbR \right\} \,.
$$
The following theorem is Theorem 17 of Chapter 1 of \cite{kunoth2018splines} which provides an approximation error of a function (and its derivatives) with respect to B-spline basis: 

\begin{theorem}[Functional approximation using B-spline basis]
\label{thm:spline_bias}
For any $0 \le r \le l \le p$, if $\sup_{|x| \le \tau} |f(x)| <
\infty$ and $f$ is $(l+1)$ times differentiable with $(l+1)^{th}$ derivative also bounded on $[-\tau, \tau]$ then we have: 
$$
\inf_{s \in \bbS_{p, k}} \left\|\partial_r f - \partial_r s\right\|_{\infty} \le C \left(\frac{2\tau}{k}\right)^{(l + 1 -r)}\|f^{(l+1)}\|_{\infty} 
$$  
where the constant $C$ only depends on $p$, the order of the spline approximation. 
\end{theorem}
In our paper, we also need to control the behavior of the lower eigenvalue of the population matrix $\bbE\left(N_k(\hat \eta)N_K(\hat \eta)^{\top}\mathds{1}_{|\hat \eta| \le \tau} \mid \cD_1\right)$. For that reason, we use a scaled version of B-spline basis instead: 
$$
\tilde N_{i, p}(x) = \sqrt{\frac{K}{2\tau}}N_{i, p}(x) \,.
$$
and use the following theorem (see the theorem of section 3 of \cite{de1973quasi} or Theorem 11 of \cite{kunoth2018splines}): 
\begin{theorem}[Eigenvalue of spline matrix]
\label{thm:eigen_spline}
Define a matrix $G \in \bbR^{(k+p) \times (k+p)}$ such that: 
$$
G_{ij} = \int_{-\tau}^{\tau} \tilde N_{i, p}(x) \tilde N_{j, p}(x) \ dx 
$$
Then there exists a constant $\kappa_p >0$ only depending on $p$ such that: 
$$
\kappa_p \le \frac{x^{\top}Gx}{x^{\top}x} \le 1
$$
for all $x \in \bbR^{(k+p)}$. In particular if $X \sim F$ with density $f$ on $[-\tau, \tau]$ and if $0 < f_- \le f(x) \le f_+ < \infty $ then we have: 
$$
\kappa_p f_- \le \lambda_{\min}\left(\bbE\left(\tilde N_k(X)\tilde N_K(X)^{\top}\right)\right) \le \lambda_{\max}\left(\bbE\left(\tilde N_k(X)\tilde N_K(X)^{\top}\right)\right) \le f_+ \,.
$$
\end{theorem}
\begin{remark}
Note that the linear span of $\{N_{i, p}\}_{i=1}^{(k+p)}$ and $\{\tilde N_{i, p}\}_{i=1}^{(k+p)}$ is same as $\bbS_{p, k}$ because $N_{i, p}$ and $\tilde N_{i, p}$ only differs by a scaling constant. Hence Theorem \ref{thm:spline_bias} remains unaltered even if we use the scaled B-spline basis $\{\tilde N_{i, p}\}_{i=1}^{(k+p)}$. 
\end{remark}
Another result which is due of \cite{rudelson1999random}, is used in this paper to bound the minimum eigenvalue (e.g. the operator norm of the inverse) of the sample covariance matrix formed by $\{\tilde N_K(\hat \eta_i)\mathds{1}_{|\hat \eta_i| \le \tau}\}_{i=1}^n$ in terms of the population covariance matrix $\bbE\left[\tilde N_k(\hat \eta)\tilde N_k(\hat \eta)^{\top}\mathds{1}_{|\hat \eta| \le \tau}\right]$ is the following (see Lemma 6.2 of \cite{belloni2015some}): 
\begin{theorem}
\label{thm:rudelson}
Let $Q_1, \dots, Q_n$ be independent symmetric non-negative $k \times k$ matrix-valued random variable. Let $\bar Q = (1/n)\sum_i Q_i$ and $Q = E(\bar Q)$. If $\|Q_i\|_{op} \le M$ a.s. then we have: 
$$
\bbE\left\|\bar Q - Q\right\|_{op} \le C \left(\frac{M\log{k}}{n} + \sqrt{\frac{M\|Q\|_{op}\log{k}}{n}}\right) \,.
$$
for some absolute constant $C$. In particular if $Q_i = p_ip_i^{\top}$ for some random vector $p_i$ with $\|p_i\| \le \xi_k$ almost surely, then: 
$$
\bbE\left\|\bar Q - Q\right\|_{op} \le C \left(\frac{\xi_k^2\log{k}}{n} + \sqrt{\frac{\xi_k^2\|Q\|_{op}\log{k}}{n}}\right) \,.
$$
\end{theorem}

Lastly, as we are working on the derivative estimation of the mean function $b$, we need to have a bound on $\|\nabla \tilde N_{p, K}(x)\|$ for all $|x| \le \tau$, where $\nabla \tilde N_{p, K}(x)$ is the vector of derivatives of $\{\tilde N_{i, k}(x)\}_{i=1}^{k+p}$. Towards that end, we prove the following lemma; 
\begin{lemma}
\label{lem:spline_deriv_bound}
For all $|x| \le \tau$, we have: $\|\nabla \tilde N_{p, K}(x)\| \le C_p K\sqrt{K}$ for some constant $C_p$ depending only on the order of the spline basis $p$ and $\tau$. 
\end{lemma}

\begin{proof}
Note that, the unscaled B-spline $N_K(x)$ forms a partition of unity, i.e for any $x \in [-\tau, 
\tau]$, we have $\sum_{j=1}^{K+p} N_{j, k}(x) = 1$ and also, by definition, each $x$ only contributed to finitely many (at-most $p$ many) basis. Hence it is immediate that, for any $x \in [-\tau, \tau]$, $\|N_K(x)\| \lesssim 1$ and consequently $\|\tilde N_K(x)\| \lesssim \sqrt{K}$. Now, from Theorem 3 of \cite{kunoth2018splines} we have for any $1 \le j \le k+p$: 
$$
\frac{d}{dx} N_{j, p, K}(x) = \frac{K}{2\tau}\left(N_{j, p-1, K}(x) - N_{j+1, p-1, K}(x)\right)
$$
Hence: 
$$
\frac{d}{dx} \tilde N_{j, p, K}(x) = \frac{K\sqrt{K}}{\tau\sqrt{\tau}}\left(N_{j, p-1, K}(x) - N_{j+1, p-1, K}(x)\right)
$$
which along with with the partition of unity property of $\{N_{i, p}\}_{i=1}^{k+p}$ implies $\|\nabla \tilde N_{p, K}(x)\| \le C_p K\sqrt{K}$, which completes the proof of the theorem.  
\end{proof}

\section{Main algorithm}
\label{sec:main_algo}
In this Subsection, we present our estimation method of the treatment effect $\alpha_0$ detailed in Section \ref{sec:est_method} of the main document in an algorithmic format. 
\allowdisplaybreaks
\begin{algorithm}
\SetAlgoLined
 \begin{enumerate}
 \item Divide the whole data into three equal parts: $\cD = \cD_1 \cup \cD_2 \cup \cD_3$.
  \vspace{0.2in}
  \item Estimate $\gamma_0$ from $\cD_1$ by doing OLS regression of $Q$ on $Z$, i.e. set: 
  $$
  \hat \gamma_n = (Z^{\top}Z)^{-1}Z^{\top}Q \,.
  $$
\item Replace $\eta_i$ in equation \eqref{eq:new1} by $\hat \eta_i$ using $\hat \gamma_n$ obtained in the previous step (i.e. set $\hat \eta_i = Q_i - Z_i^{\top}\hat \gamma_n$ where $\{Q_i, Z_i\}$ are in $\cD_2$ and $\hat \gamma_n$ obtained in Step (b). Then estimate $b'$ from $\cD_2$ using equation \eqref{eq:new1} via spline estimation method. 
   \item Estimate $\alpha_0$ from $\cD_3$ using $\hat \gamma_n$ and $\hat b'$ estimated in the previous two steps. 
     \vspace{0.2in}
     \item Finally, do the above steps by rotating the datasets and combine them to gain efficiency.
 \end{enumerate}
 \label{algo:main}
\end{algorithm}

\section{Proofs of auxiliary lemmas of Theorem \ref{thm:main_hd}}
\label{sec:proof_aux_hd}
Lemma \ref{lem:sq_rate} involves the rate of estimation of the the deviation of $\check \bW$ from $\tilde \bW$. Note that, to establish the rate of convergence of a sequence of random variable, it is enough to establish the bound on the set of probability going to 1. More precisely suppose $\{A_n\}$ is a sequence of events such that $\lim_{n \to \infty} \bbP(A_n) = 1$. Then, if we have a sequence to random variable $\{X_n\}$ and we wish to show that $X_n \lesssim_\bbP t_n$, then it suffices to show that on $A_n$ as: 
$$
\limsup_{n \to \infty} \bbP\left(\frac{|X_n|}{t_n} > t\right) = \limsup_{n \to \infty} \bbP\left(\frac{|X_n|}{t_n} > t, A_n\right) \,.
$$
We will use this fact in our proof. More specifically, define the event $\Omega_n$ as: 
$$
\Omega_n = \left\{\|\hat \gamma_n -\gamma_0\| \le C\sqrt{\frac{s_\gamma \log{p}}{n}}\right\} \,.
$$
for some fixed large constant $C$. Then from the properties of Lasso, we have $\bbP(\Omega_n) \to 1$ as $n \to \infty$ (as this event if true as long as RE condition is satisfied, which approaches to 1 as $n$ goes to infinity.) Therefore, by our previous argument, whenever we want to establish rate of some terms, it is enough to establish the rate on the event $\Omega_n$. 

\subsection{Proof of Lemma \ref{lem:sq_rate}}
\begin{proof}
In this proof, we denote by $a_n = (\hat \gamma_n - \gamma_0)/\|\hat \gamma_n - \gamma_0\|$. We start with with bounding the distance from $\breve \bW^{\perp}$ to $\bW^{\perp}$. 
Recall that the only difference between $\breve \bW^{\perp}$ and $\bW^{\perp}$ is in the last $p_2$ columns where we replace the coefficients of $Z$ by $b'(\eta)$ and consequently the difference is $Z(\hat b'(\hat \eta) - b'(\eta))$. Therefore we have for any $p_1 + 2 \le j \le 1+p_1+p_2$: 
\begin{align}
\frac{1}{n}\left\|\breve \bW_{*, j}^{\perp} - \bW_{*, j}^{\perp}\right\|^2  & \le \frac{1}{n}\left\|\breve \bW_{*, j}^{\perp} - \bW_{*, j}^{\perp}\right\|^2 \notag \\
& = \frac{1}{n} \sum_{i=1}^n Z_{ij}^2(\hat b'(\hat \eta) - b'(\eta))^2 \notag \\
& \lesssim \frac{1}{n} \sum_{i=1}^n Z_{ij}^2(\hat b'(\hat \eta) - b'(\hat \eta))^2 + \frac{1}{n} \sum_{i=1}^n Z_{ij}^2(b'(\hat \eta) - b'(\eta))^2 \notag \\
& \lesssim  \frac{\dot r^2_n}{n} \sum_{i=1}^n Z_{i, j}^2 + \|b''\|_\infty^2 \sum_{i=1}^n Z_{i, j}^2(\hat \eta_i - \eta_i)^2  \notag \\
& \lesssim_{\bbP} \dot r^2_n + \left\|\hat \gamma_n - \gamma_0\right\|^2 \  \frac1n \sum_{i=1}^n Z_{i, j}^2\left(Z_i^{\top}a_n\right)^2 \notag \\
\label{eq:forgotten} & \lesssim_{\bbP} \dot r^2_n + \frac{s_\gamma \log{p}}{n} \,.
\end{align}
We next bound $\bW^{\perp}_{*, j} - \tilde \bW_{*, j}$. For notational simplicity, we define $p = 1 + p_1 + p_2$. Now for any $0 \le j \le p$ we have: 
\begin{align}
    \bW^{\perp}_{*, j} - \tilde \bW_{*, j} & = m_j(\boeta) - P_{\bN_k}\bW_{*, j} \notag \\
    & = m_j(\boeta) -  \check m_j(\hat \boeta) +  \check m_j(\hat \boeta) - P_{\bN_k}\left(\check m_j\left(\hat \boeta\right) + \check \nu_j\right) \notag \\
    & = \left[m_j(\boeta) -  \check m_j(\hat \boeta)\right] + \bN_k(\hat \boeta) \ \check \omega_j + \check \bR_j - P_{\bN_k}\left(\bN_k(\hat \boeta) \ \check\omega_j + \check \bR_j + \check \nu_j\right) \notag \\
    \label{eq:exp_1} & = \left[m_j(\boeta) -  \check m_j(\hat \boeta)\right] + P^{\perp}_{\bN_k}\check \bR_j - P_{\bN_k}\check \nu_j \,,
\end{align}
where $\check \bR_j = \check m_j(\hat \boeta) - \bN_K(\hat \boeta)\check \omega_j$ is the B-spline approximation error on $[-\tau, \tau]$ (recall that we are only using the observations in this interval) and both the notations $\widehat{\tilde W}_{*, j}$ and $\check \nu_j$ are used interchangeably to denote $W_j - \bbE[W_j \mid \hat \boeta] = W_j - \check m_j(\hat\boeta)$. Furthermore, note that: 
\begin{align}
\bW^{\perp}_{*, j} - \widehat{\tilde{\bW}}_{*, j} & = \check m_j(\hat \boeta) - P_{\bN_k}\bW_{*, j} \notag \\
\label{eq:exp_other} & = P^{\perp}_{\bN_k}\check \bR_j - P_{\bN_k}\check \nu_j 
\end{align}
So it is immediate that to prove the second part of the Lemma, we need to control three terms of the RHS of equation \eqref{eq:exp_1}, whereas for the first part we need to control the last two terms of the RHS of equation \eqref{eq:exp_1}. Therefore it is enough to control the three terms of the RHS to conclude both the parts of the Lemma. Going back to equation \eqref{eq:exp_1} we have using $(a+b+c)^2 \le 4(a^2 + b^2 + c^2)$: 
\begin{align}
    \max_{0 \le j \le 1+p_1+p_2}\frac{1}{n} \left\|\bW^{\perp}_{*, j} - \tilde \bW_{*, j}\right\|^2 & \le \underbrace{\frac{4}{n} \max_{0 \le j \le 1+p_1+p_2}\left\|m_j(\boeta) -  \check m_j(\hat \boeta)\right\|^2}_{T_1}  \notag \\
    \label{eq:exp_1_expanded} & \qquad \qquad +  \underbrace{4\max_{0 \le j \le 1+p_1+p_2} \frac{\check \bR_j^{\top}\check \bR_j}{n}}_{T_2} + \underbrace{4 \max_{0 \le j \le 1+p_1+p_2} \frac{\check \nu_j^{\top}P_{\bN_k}^{\perp}\check \nu_j}{n}}_{T_3} \,. 
\end{align} 
$T_2$ can be bounded uniformly from the functional approximation properties of B-spline basis: 
$$
T_2 = 4\max_{0 \le j \le 1+p_1+p_2} \frac{\check\bR_j^{\top}\check\bR_j}{n}  \lesssim r_n^2 \,.
$$
For $T_3$, by Assumption \ref{assm:sg} we have $\check \nu_j$ is a centered subgaussian vector conditionally on $\hat \eta$ with independent entries. Therefore applying Hanson-Wright inequality (\cite{rudelson2013hanson}) we have (with $A_n = P_{\bN_k}/n$): 
$$
\bbP\left(\left|\check \nu_j^{\top}A_n\check \nu_j - \bbE\left[\check \nu_j^{\top}A_n\check \nu_j \mid \hat \eta\right]\right| > t \mid \hat \eta\right) \le 2\exp\left\{\left(-c\left\{\frac{t^2}{\sigma_W^4\left\|A_n\right\|_{HS}^2}, \frac{t}{\sigma_W^2\|A_n\|_{op}}\right\}\right)\right\}
$$
From the properties of the projection matrix we have: $\left\|A_n\right\|_{HS}^2 = K/n^2$ and $\|A_n\|_{op} = 1/n$. Hence we have: 
$$
\bbP\left(\left|\check \nu_j^{\top}A_n\check \nu_j - \bbE\left[\nu_j^{\top}A_n\check \nu_j \mid \hat \eta\right]\right| > t \mid \hat \eta\right) \le 2\exp\left\{\left(-c\left\{\frac{n^2t^2}{\sigma_W^4K}, \frac{nt}{\sigma_W^2}\right\}\right)\right\}
$$
On the other hand, we have: 
$$
\max_{1 \le j \le p} \bbE\left[\frac{\check \nu_j^{\top}P_{\bN_k}\check \nu_j}{n}\right] \le \frac{K}{n} \sup_{\substack{1 \le j \le p \\ |t| \le \tau}}\var(W_{*, j} \mid \hat \eta = t) := \hat \var_{\sup} \times \frac{K}{n} \,.
$$
where by Assumption \ref{assm:sg}, the variance is finite. Hence for $t > \hat \var_{\sup} \times (K/n)$ we have: 
\begin{align*}
    \bbP\left(\max_{1 \le j \le p} \left|\hat \nu_j^{\top}A_n \hat \nu_j\right| > t\right) & = \sum_j \bbP\left(\left|\hat \nu_j^{\top}A_n \hat \nu_j\right| > t\right) \\
    & =  \sum_j \bbE_{\hat \eta}\left[\bbP\left(\left.\left|\hat \nu_j^{\top}A_n \hat \nu_j\right| > t \right\vert \hat \eta\right)\right] \\
    & \le  \sum_j \bbE_{\hat \eta}\left[\bbP\left(\left|\hat \nu_j^{\top}A_n \hat \nu_j - \bbE\left[\left.\hat \nu_j^{\top}A_n \hat \nu_j t \right| \hat \eta\right]\right| > t - \hat \var_{\sup} \frac{K}{n} \mid \hat \eta\right)\right] \\
    & \le 2\sum_j \bbE_{\hat \eta}\left[\exp\left\{\left(-c\left\{\frac{n^2\left(t -\hat \var_{\sup} \frac{K}{n} \right)^2}{\sigma_W^4k}, \frac{n\left(t -\hat \var_{\sup} \frac{K}{n} \right)}{\sigma_W^2}\right\}\right)\right\}\right] \\
    & = 2\exp\left\{\left(\log{p} -c\left\{\frac{n^2\left(t -\hat \var_{\sup} \frac{K}{n} \right)^2}{\sigma_W^4k}, \frac{n\left(t -\hat \var_{\sup} \frac{K}{n} \right)}{\sigma_W^2}\right\}\right)\right\}
\end{align*}
So this implies, with probability going to 1 we have with an appropriate choice of $t$: 
\begin{align}
T_3 = 4 \max_{1 \le j \le 1+p_1+p_2} \frac{\check \nu_j^{\top}P_{\bN_k}\check \nu_j}{n} & \lesssim_\bbP \frac{K}{n} + \sigma_W^2 \left(\frac{\log{p}}{n} \vee \sqrt{\frac{K\log{p}}{n^2}}\right) \notag \\
\label{eq:resid_etahat_bound} & \lesssim_\bbP \dot r^2_n + \sigma_W^2\left(\frac{\log{p}}{n} \vee \dot r_n \sqrt{\frac{\log{p}}{n}}\right) \,.
\end{align}
For $T_1$ in equation \eqref{eq:exp_1}, we use Assumption \ref{assm:smoothness_conditional}. Note that we have $m_j(\eta) = g_j(0, \eta)$ and $\check m_j(\hat \eta) = g_j(\hat \gamma_n - \gamma_0, \hat \eta)$. Hence we have: 
\begin{align}
    & \frac{4}{n} \max_{1 \le j \le 1+p_1+p_2}\left\|m_j(\boeta) -  \check m_j(\hat \boeta)\right\|^2 \notag \\
    & = 4 \max_{1 \le j \le 1+p_1+p_2} \frac{1}{n}\sum_{i=1}^n \left(\check m_j(\hat \eta_i) - m_j(\eta_i)\right)^2 \notag \\
    & = 4 \max_{1 \le j \le 1+p_1+p_2} \frac{1}{n}\sum_{i=1}^n \left(g_j(\hat \gamma_n - \gamma_0, \hat \eta) - g_j(\hat 0, \hat \eta_i) + g_j(0, \hat \eta) - g_j(0,  \eta_i)\right)^2 \notag \\
    & \le 8 \max_{1 \le j \le 1+p_1+p_2} \left[\frac{1}{n}\sum_{i=1}^n \left(g_j(\hat \gamma_n - \gamma_0, \hat \eta_i) - g_j(0, \hat \eta_i)\right)^2 + \frac{1}{n}\sum_{i=1}^n \left(g_j(0, \hat \eta_i) - g_j( 0,  \eta_i)\right)^2 \right] \notag\\
    & \le 8 (L_1 \vee L_2) \|\hat \gamma_n - \gamma_0\|^2\left(1 + \frac1n \sum_i \left(Z_i^{\top}a_n\right)^2\right) \hspace{0.2in} \left[\text{where } a_n = \frac{\hat \gamma_n - \gamma_0}{\|\hat \gamma_n - \gamma_0\|}\right] \notag \\
    \label{eq:m_diff_bound} & \lesssim \|\hat \gamma_n - \gamma_0\|^2 \lesssim_\bbP \ \frac{s_\gamma \log{p}}{n}  \hspace{0.2in} \left[Z \text{ is centered subgaussian with finite variance}\right] \,.
\end{align}
Combining the bounds on $T_i's$ we have for the first part of the Lemma: 
\begin{align*}
\max_{j}\frac{1}{n} \left\|\breve \bW^{\perp}_{*, j} - \widehat{\tilde \bW}_{*, j}\right\|^2 & \lesssim \max_{j}\frac{1}{n} \left\|\bW^{\perp}_{*, j} - \widehat{\tilde \bW}_{*, j}\right\|^2 + \max_j \frac{1}{n} \left\|\breve \bW^{\perp}_{*, j} -\bW^{\perp}_{*, j}\right\|^2 \\
& \lesssim_\bbP  \ \dot r_n^2 + \frac{s_\gamma \log{p}}{n} + \dot r_n \sqrt{\frac{\log{p}}{n}} \,.
\end{align*}
And for the second part of the Lemma: 
\begin{align*}
\max_{j}\frac{1}{n} \left\|\breve \bW^{\perp}_{*, j} - \tilde \bW_{*, j}\right\|^2 & \lesssim \max_{j}\frac{1}{n} \left\|\bW^{\perp}_{*, j} - \tilde \bW_{*, j}\right\|^2 + \max_j \frac{1}{n} \left\|\breve \bW^{\perp}_{*, j} -\bW^{\perp}_{*, j}\right\|^2  \\
& \lesssim_\bbP \ \dot r_n^2 + \frac{s_\gamma \log{p}}{n} +\dot r_n \sqrt{\frac{\log{p}}{n}} \,.
\end{align*}
This completes the proof of the first and the second part of the Lemma. 
\\\\
\noindent
For the third part: 
\begin{align*}
& \max_{1 \le j, j' \le p} \frac{1}{n}\left|\left(\breve \bW_{*, j}^{\perp}\right)^{\top}\left(\breve \bW^{\perp}_{*, j'} - \tilde \bW_{*, j'}\right)\right| \\
& \le \underbrace{\frac{1}{n}\max_{1 \le j, j' \le p}\left|\left(\breve \bW_{*, j}^{\perp} - \widehat{\tilde \bW}_{*, j}\right)^{\top}\left(\breve \bW^{\perp}_{*, j'} - \tilde \bW_{*, j'}\right)\right|}_{T_{1}} + \underbrace{\frac{1}{n}\max_{1 \le j, j' \le p}\left|\widehat{\tilde \bW}_{*, j}^{\top}\left(\breve \bW^{\perp}_{*, j'} - \tilde \bW_{*, j'}\right)\right|}_{T_{2}}
\end{align*}
To bound $T_{1}$ we use the previous parts of this Lemma and Cauchy-Schwarz inequality. which yields:   
\begin{align*}
    T_{1} & \le \sqrt{\frac{1}{n}\max_{1 \le j \le p}\left\|\left(\breve \bW_{*, j}^{\perp} - \widehat{\tilde \bW}_{*, j}\right)\right\|^2} \times \sqrt{\frac{1}{n}\max_{1 \le j' \le p}\left\|\left(\breve \bW_{*, j'}^{\perp} - \tilde \bW_{*, j'}\right)\right\|^2} \\
    & \lesssim_\bbP  \ \dot r_n^2 + \frac{s_\gamma \log{p}}{n} +\dot r_n \sqrt{\frac{\log{p}}{n}}  \,.
\end{align*}
where the bounds in the last inequality follows from the conclusions of the previous parts of this Lemma. To bound $T_{2}$ we first expand $\breve \bW^{\perp}_{*, j'} - \tilde \bW_{*, j'}$ as before: 
\begin{align*}
\breve \bW^{\perp}_{*, j'} - \tilde \bW_{*, j'} & = \breve \bW^{\perp}_{*, j'} - \bW^{\perp}_{*, j'} + \bW^{\perp}_{*, j'} - \tilde \bW_{*, j'} \\
& = \breve \bW^{\perp}_{*, j'} - \bW^{\perp}_{*, j'} + m_{j'}(\boeta) - \check m_{j'}(\hat \boeta) + P_{\bN_k}^{\perp}\check \bR_{j'} - P_{\bN_k} \check \nu_{j'} 
\end{align*}
Using the above expansion we have: 
\begin{align}
    T_{2} &= \frac{1}{n}\max_{1 \le j,j' \le p}\left|\widehat{\tilde \bW}_{*, j}^{\top}\left(\breve \bW^{\perp}_{*, j'} - \tilde \bW_{*, j'}\right)\right| \notag \\
    \label{eq:bound_T2} & \le \frac{1}{n}\max_{1 \le j, j' \le p}\left|\widehat{\tilde \bW}_{*, j}^{\top}\left(\breve \bW^{\perp}_{*, j'} - \bW^{\perp}_{*, j'} \right)\right| + \frac{1}{n}\max_{1 \le j, j' \le p}\left|\widehat{\tilde \bW}_{*, j}^{\top}\left(m_{j'}(\boeta) - \check m_{j'}(\hat \boeta)\right)\right| \notag \\
    & \qquad \qquad + \frac{1}{n}\max_{1 \le j, j' \le p}\left|\widehat{\tilde \bW}_{*, j}^{\top}P_{\bN_k}^{\perp}\check \bR_{j'}\right| + \frac{1}{n}\max_{1 \le j, j' \le p}\left|\widehat{\tilde \bW}_{*, j}^{\top}P_{\bN_k} \check \nu_{j'}\right|
\end{align}
The bound the first term of the RHS of the above equation, we have using Cauchy-Schwarz inequality: 
\begin{align*}
& \frac{1}{n}\max_{1 \le j, j' \le p}\left|\widehat{\tilde \bW}_{*, j}^{\top}\left(\breve \bW^{\perp}_{*, j'} - \bW^{\perp}_{*, j'} \right)\right| \\
& \le \underbrace{\sqrt{\max_j \frac1n \left\|\widehat{\tilde \bW}_{*, j}\right\|^2}}_{\lesssim_\bbP 1} \underbrace{\sqrt{\max_{j'} \frac1n \left\|\breve \bW^{\perp}_{*, j'} - \bW^{\perp}_{*, j'}\right\|^2}}_{\lesssim_\bbP \ \dot r_n + \sqrt{\frac{s_\gamma \log{p}}{n}}} \hspace*{0.2in} [\text{Equation }\eqref{eq:forgotten}] \\
& \lesssim_\bbP \ \dot r_n + \sqrt{\frac{s_\gamma \log{p}}{n}} \,.
\end{align*}
In the above inequality we have used the fact $\max_j (1/n)\left\|\widehat{\tilde \bW}_{*, j}\right\|^2 = O_p(1)$ which follows from the following application of triangle inequality: 
$$
\frac{1}{n}\max_{2 \le j \le p}\left\|\widehat{\tilde \bW}_{*, j}\right\|^2 = \frac{1}{n}\max_{2 \le j \le p}\left\|\tilde \bW_{*, j}\right\|^2  + \frac{1}{n}\max_{2 \le j \le p}\left\|\left(\tilde \bW_{*, j} - \widehat{\tilde \bW}_{*, j}\right)\right\|^2  \lesssim_\bbP 1 \,.
$$
as the first term is $O_p(1)$ due to bounded variance of $\tilde W_j$ (which follows from sub-gaussianity Assumption \ref{assm:sg}) and the second term is $o_p(1)$ from the previous parts of this Lemma. To bound the second term of equation \eqref{eq:bound_T2}: 
\begin{align*}
\frac{1}{n}\max_{1 \le j, j' \le p}\left|\widehat{\tilde \bW}_{*, j}^{\top}\left(m_{j'}(\boeta) - \check m_{j'}(\hat \boeta)\right)\right| & \lesssim \sqrt{\max_j \frac1n \left\|\widehat{\tilde \bW}_{*, j}\right\|^2} \sqrt{\max_{j'}\frac{1}{n}\sum_i \left(\check m_{j'}(\hat \eta_i) - m_{j'}(\eta_i)\right)^2} \\
& \lesssim_\bbP \sqrt{\frac{s_\gamma \log{p}}{n}} \,.
\end{align*}
where the last rate follows from the similar calculation as the first term of the RHS of equation \eqref{eq:exp_1_expanded}. To bound the third term in equation \eqref{eq:bound_T2} note that the entries of $\widehat{\tilde \bW}_{*, j}$ are centered and independent conditionally on $\hat \eta$. Therefore we have: 
\begin{align*}
    \bbP\left(\frac{1}{n}\max_{1 \le j, j' \le p}\left|\widehat{\tilde \bW}_{*, j}^{\top}P_{\bN_k}^{\perp}\check \bR_{j'}\right| > t \right) & \le \sum_{1 \le j,j' \le p} \bbP\left(\frac{1}{n}\left|\widehat{\tilde \bW}_{*, j}^{\top}P_{\bN_k}^{\perp}\check \bR_{j'}\right| > t \right) \\
    & \le \sum_{1 \le j,j' \le p} \bbE\left[\bbP\left(\frac{1}{n}\left|\widehat{\tilde \bW}_{*, j}^{\top}P_{\bN_k}^{\perp}\check \bR_{j'}\right| > t \mid \hat \boeta\right)\right] \\
    & \le \sum_{1 \le j,j' \le p} \bbE\left[2\exp{\left(-c\frac{nt^2}{\sigma_W \frac{\check \bR_{j'}^{\top}\check \bR_{j'}}{n}}\right)}\right] \\
    & \le 2\exp{\left(2\log{p}-c\frac{nt^2}{\sigma_W r_n^2}\right)}
\end{align*}
Therefore we have with probability going to 1 with an appropriate choice of $t$: 
$$
\frac{1}{n}\max_{1 \le j,j' \le p}\left|\widehat{\tilde \bW}_{*, j}^{\top}P_{\bN_k}^{\perp}\check \bR_{j'}\right| \lesssim_\bbP \  r_n\sqrt{\frac{\log{p}}{n}} \,.
$$
For the last term of equation \eqref{eq:bound_T2} note that: 
\begin{align*}
    \frac{1}{n}\max_{1 \le j,j' \le p}\left|\widehat{\tilde \bW}_{*, j}^{\top}P_{\bN_k} \check \nu_{j'}\right| & =  \frac{1}{n}\max_{1 \le j,j' \le p}\left|\check \nu_j^{\top}P_{\bN_k} \check \nu_{j'}\right| \hspace{0.2in} \left[\text{ As }\check \nu_j = \widehat{\tilde \bW}_{*, j}\right] \\
    & = \sqrt{\frac{1}{n}\max_j \left|\check \nu_j^{\top}P_{\bN_k} \check \nu_{j}\right|} \times \sqrt{\frac{1}{n}\max_{j'} \left|\check \nu_{j'}^{\top}P_{\bN_k} \check \nu_{j'}\right|} \\
    & \lesssim_\bbP \  \dot r_n^2 + \frac{\log{p}}{n} + \dot r_n \sqrt{\frac{\log{p}}{n}} \,.
\end{align*}
where the last inequality follows from the same calculation as for $\hat \nu_j$ in the first part of this Lemma. Combining the bounds for different components of $T_2$ we have, with probability going to $1$: 
$$
\max_{1 \le j, j' \le p} \frac{1}{n}\left|\left(\breve \bW_{*, j}^{\perp}\right)^{\top}\left(\breve \bW^{\perp}_{*, j'} - \tilde \bW_{*, j'}\right)\right| \lesssim_\bbP \  \dot r_n +  \sqrt{\frac{s_\gamma \log{p}}{n}} \,.
$$
This completes the proof the lemma. 
\end{proof}

\subsection{Proof of Lemma \ref{lem:choice_lambda}}
To obtain $\lambda_1$ we need to bound the $\ell_\infty$ norm on the first term of the RHS of equation \eqref{eq:s_lasso}. We can bound that term as: 
\begin{align}
    \frac1n \left\|\left(\bS^{\perp} - \breve \bW^{\perp}_{-1}\theta^*_S\right)^{\top}\breve \bW^{\perp}_{-1}\right\|_\infty  & \le \frac{1}{n}\max_{2 \le j \le p}\left|\left(\breve \bW_{*, j}^{\perp}\right)^{\top}\left(\bS^{\perp} - \breve \bW^{\perp}_{-1}\theta^*_S\right)\right| \notag \\
    & \le \underbrace{\frac{1}{n}\max_{2 \le j \le p}\left|\left(\breve \bW_{*, j}^{\perp}\right)^{\top}\left(\tilde \bS - \tilde \bW_{-1}\theta^*_S\right)\right|}_{T_1} + \underbrace{\frac{1}{n}\max_{2 \le j \le p}\left|\left(\breve \bW_{*, j}^{\perp}\right)^{\top}\left(\bS^{\perp} - \tilde \bS\right)\right|}_{T_2}  \notag \\
    \label{eq:basic_exp_lemma_lambda} & \qquad \qquad + \underbrace{\max_{2 \le j \le p} \frac{1}{n}\left|\left(\breve \bW_{*, j}^{\perp}\right)^{\top}\left(\breve \bW_{-1}^{\perp} - \tilde \bW_{-1}\right)\theta^*_S\right|}_{T_3}
\end{align} 
Observe that, we can bound $T_2$ via similar calculation we did to prove the third display of Lemma \ref{lem:sq_rate} as it is a special case for $j' = 1$. Therefore we have: 
$$
T_2 \lesssim_\bbP \  \dot r_n +  \sqrt{\frac{s_\gamma \log{p}}{n}} \,.
$$
We next bound $T_3$. First of all, an application $\ell_1$ - $\ell_\infty$ bound yields: 
\begin{align*}
    T_3 & = \max_{1 \le j \le p} \frac{1}{n}\left|\left(\breve \bW_{*, j}^{\perp}\right)^{\top}\left(\breve \bW^{\perp}_{-1} - \tilde \bW_{-1}\right)\theta^*_S\right| \\
    & \le \|\theta^*_S\|_1 \max_{1 \le j, j' \le p} \frac{1}{n}\left|\left(\breve \bW_{*, j}^{\perp}\right)^{\top}\left(\breve \bW^{\perp}_{*, j'} - \tilde \bW_{*, j'}\right)\right| \\
    & \lesssim_\bbP \|\theta^*_S\|_1 \left[\dot r_n +  \sqrt{\frac{s_\gamma \log{p}}{n}}\right] \,,
\end{align*}
where again the last inequality follows from third conclusion of Lemma \ref{lem:sq_rate}. 
To bound $T_1$, we further divide it into two terms: 
\begin{align}
    T_1 & = \frac{1}{n}\max_{2 \le j \le p}\left|\left(\breve \bW_{*, j}^{\perp}\right)^{\top}\left(\tilde \bS - \tilde \bW_{-1}\theta^*_S\right)\right| \notag \\
   \label{eq:centering_issue} & \le \underbrace{\frac{1}{n}\max_{2 \le j \le p}\left|\left(\breve \bW_{*, j}^{\perp} - \tilde \bW_{*, j}\right)^{\top}\left(\tilde \bS - \tilde \bW_{-1}\theta^*_S\right)\right|}_{T_{11}} + \underbrace{\frac{1}{n}\max_{2 \le j \le p}\left|\tilde \bW_{*, j}^{\top}\left(\tilde \bS - \tilde \bW_{-1}\theta^*_S\right)\right|}_{T_{12}}
\end{align}
To bound $T_{12}$ we use the first order condition on the definition of $\theta^*_S$. Recall that $\theta^*_S$ is defined as: 
$$
\theta^*_S = \argmin_\delta \bbE\left[\left(S - \tilde W_1^{\top}\delta\right)^2\mathds{1}_{|\eta| \le \tau}\right]
$$
Therefore we have: 
$$
\bbE\left[\tilde W_{-1}(S - \tilde W_{-1}^{\top}\theta^*_S)\mathds{1}_{|\eta| \le \tau}\right] = 0 \,.
$$
and consequently for any $2 \le j \le p$ and $1 \le i \le n$, the term $\tilde \bW_{i, j}(\tilde \bS_i - \tilde \bW_{i,-1}^{\top}\theta^*_S)\mathds{1}_{|\eta_i| \le \tau}$ is a centered subexponential and independent over $i$ with $\psi_1$ norm bounded by $\sigma^2_W\sqrt{1 + \|\theta^*_S\|^2}$. Note that $\|\theta^*_S\|$ is also bounded, which follows from the definition of $\theta^*_S$ (see equation \eqref{eq:def_theta_s}): 
\begin{align*}
\theta^*_S = \left(\Sigma_{\tau_{-1, -1}}\right)^{-1}\bbE[\tilde W_{-1}\tilde S\mathds{1}_{|\eta| \le \tau}] = \left(\Sigma_{\tau_{-1, -1}}\right)^{-1} \Sigma_{\tau_{-1, 1}} & = -\frac{1}{\left(\Sigma_\tau^{-1}\right)_{1, 1}}\left(\Sigma_\tau^{-1}\right)_{-1, 1} \\
& =  -\frac{1}{\left(\Sigma_\tau^{-1}\right)_{1, 1}} \left(\Sigma_\tau^{-1}e_1\right)_{-1}
\end{align*}
This implies: 
$$
\left\|\theta^*_S\right\|^2 = \frac{1}{\left(\Sigma_\tau^{-1}\right)_{1, 1}^2}\left\| \left(\Sigma_\tau^{-1}e_1\right)_{-1}\right\|^2 \le \frac{1}{\left(\Sigma_\tau^{-1}\right)_{1, 1}^2}\left\| \left(\Sigma_\tau^{-1}e_1\right)\right\|^2 \le \left(\frac{\lambda_{\max}(\Sigma_\tau^{-1})}{\lambda_{\min}(\Sigma_\tau^{-1})}\right)^2 \le \left(\frac{C_{\max}}{C_{\min}}\right)^2 
$$
where the last bound follows from Assumption \ref{assm:sigma}. Hence we have: 
$$
\left\|\tilde \bW_{i, j}(\tilde \bS_i - \tilde \bW_{i,-1}^{\top}\theta^*_S)\mathds{1}_{|\eta_i| \le \tau}\right\|_{\psi_1} \le \sigma^2_W\sqrt{1 +\left(\frac{C_{\max}}{C_{\min}}\right)^2 }
$$
We now bound $T_{12}$  as follows: 
\begin{align*}
    & \bbP\left(\frac{1}{n}\max_{1 \le j \le p}\left|\tilde \bW_{*, j}^{\top}\left(\tilde \bS - \tilde \bW_{-1}\theta^*_S\right)\right| > t, \Omega_n\right) \\
    & \le \sum_j \bbP\left(\frac{1}{n}\left|\tilde \bW_{*, j}^{\top}\left(\tilde \bS - \tilde \bW_{-1}\theta^*_S\right)\right| > t\right) \\
   & =  \sum_{j}  \bbP\left(\frac{1}{n}\left|\sum_i\tilde \bW_{i, j}\left(\tilde \bS_i - \tilde \bW_{-1, i}^{\top}\theta^*_S\right)\mathds{1}_{|\hat \eta_i| \le \tau}\right| > t\right) \\ 
   & \le  \sum_{j}  \bbP\left(\frac{1}{n}\left|\sum_i\tilde \bW_{i, j}\left(\tilde \bS_i - \tilde \bW_{-1, i}^{\top}\theta^*_S\right)\mathds{1}_{|\eta_i| \le \tau}\right| > \frac{t}{2}\right) \\
   & \qquad \qquad +  \sum_{j}  \bbP\left(\frac{1}{n}\left|\sum_i\tilde \bW_{i, j}\left(\tilde \bS_i - \tilde \bW_{-1, i}^{\top}\theta^*_S\right)\left(\mathds{1}_{|\hat \eta_i| \le \tau}- \mathds{1}_{| \eta_i| \le \tau}\right)\right| > \frac{t}{2}\right)  \\
   & \triangleq T_{121} + T_{122}
\end{align*}
Bounding $T_{121}$ is straightforward as we have already established $\tilde \bW_{i, j}\left(\tilde \bS_i - \tilde \bW_{-1, i}^{\top}\theta^*_S\right)\mathds{1}_{|\eta_i| \le t}$ is sub-exponential random variable, hence applying Bernstein's inequality we have with probability going to 1: 
\begin{align}
\label{eq:t12}
T_{121} \lesssim_\bbP \  \sqrt{\frac{\log{p}}{n}} \,.
\end{align}
We next bound $T_{122}$, which is trickier. Although the terms $\tilde \bW_{i, j}\left(\tilde \bS_i - \tilde \bW_{-1, i}^{\top}\theta^*_S\right)\left(\mathds{1}_{|\hat \eta_i| \le \tau}- \mathds{1}_{| \eta_i| \le \tau}\right)$ are sub-exponential by similar argument as before, they are not centered. Therefore, we further need to bound the expectation of the terms, i.e. $\bbE\left[\mathds{1}_{\Omega_n}\bbE\left[W_j\left(\tilde \bS - \tilde W_{-1}^{\top}\theta^*_S\right)\left(\mathds{1}_{|\hat \eta| \le \tau}- \mathds{1}_{| \eta| \le \tau}\right) \mid \cD_1\right]\right]$ (as we can always restrict ourselves on $\Omega_n$ to find the rate). Note that: 
\begin{align}
& \left|\bbE\left[\mathds{1}_{\Omega_n}\bbE\left[W_j\left(\tilde \bS - \tilde W_{-1}^{\top}\theta^*_S\right)\left(\mathds{1}_{|\hat \eta| \le \tau}- \mathds{1}_{| \eta| \le \tau}\right) \mid \cD_1\right]\right]\right| \notag \\
& \le \bbE\left[\mathds{1}_{\Omega_n}\left|\bbE\left[W_j\left(\tilde \bS - \tilde W_{-1}^{\top}\theta^*_S\right)\left(\mathds{1}_{|\hat \eta| \le \tau}- \mathds{1}_{| \eta| \le \tau}\right) \mid \cD_1\right]\right|\right]  \notag \\
& \le \left\|W_j\left(\tilde \bS - \tilde W_{-1}^{\top}\theta^*_S\right)\right\|_{\alpha}\bbE\left[\mathds{1}_{\Omega_n}\left(\bbE\left[\left|\mathds{1}_{|\hat \eta| \le \tau}- \mathds{1}_{| \eta| \le \tau}\right|^{\beta} \mid \cD_1\right]\right)^{\frac{1}{\beta}}\right] \notag \\
\label{eq:final_eq_new_1} & \lesssim \alpha \bbE\left[\mathds{1}_{\Omega_n}\left(\bbP\left(|\hat \eta| \le \tau, |\eta| > \tau \mid \cD_1\right) + \bbP\left(|\hat \eta| > \tau, |\eta| \le \tau \mid \cD_1\right)\right)^{\frac{1}{\beta}}\right]
\end{align}
We now analyze the probabilities inside the expectation. Fix $\rho > 0$.  
\begin{align}
\bbP\left(|\hat \eta| \le \tau, |\eta| > \tau \mid \cD_1\right)  & = \bbP\left(|\hat \eta| \le \tau, |\eta| > (1+ \rho)\tau \mid \cD_1\right) + \bbP\left(|\hat \eta| \le \tau, \tau < |\eta| \le (1 + \rho)\tau \mid \cD_1\right) \notag \\
& \le \bbP\left(\left|\hat \eta - \eta\right| > \rho \tau \mid \cD_1\right) + \bbP\left(\tau < |\eta| \le (1 + \rho)\tau \mid \cD_1\right) \notag \\
& \le \bbP\left(\left|Z^{\top}a_n\right| >  \frac{\rho \tau}{\|\hat \gamma_n - \gamma_0\|}\right) + \|f_\eta\|_\infty \rho \tau \notag \\
\label{eq:Delta_bound_1} & \le 2\exp{\left(-\frac{c \rho^2 \tau^2}{\sigma_W^2 \|\hat \gamma_n - \gamma_0\|_2^2}\right)} + \|f_\eta\|_\infty \rho \tau
\end{align}
Similar calculation also yields: 
\begin{align}
\label{eq:Delta_bound_2} \bbP\left(|\hat \eta| > \tau, |\eta| \le \tau \mid \cD_1\right) & \le  2\exp{\left(-\frac{c \rho^2 \tau^2}{\sigma_W^2 \|\hat \gamma_n - \gamma_0\|_2^2}\right)} + \|f_\eta\|_\infty \rho \tau \,.
\end{align}
Therefore we have from equation \eqref{eq:final_eq_new_1}: 
\begin{align*}
& \left|\bbE\left[\mathds{1}_{\Omega_n}\bbE\left[W_j\left(\tilde \bS - \tilde W_{-1}^{\top}\theta^*_S\right)\left(\mathds{1}_{|\hat \eta| \le \tau}- \mathds{1}_{| \eta| \le \tau}\right) \mid \cD_1\right]\right]\right| \\
&  \lesssim \alpha \bbE\left[\mathds{1}_{\Omega_n}\left(4\exp{\left(-\frac{c \rho^2 \tau^2}{\sigma_W^2 \|\hat \gamma_n - \gamma_0\|_2^2}\right)} + 2\|f_\eta\|_\infty \rho \tau\right)^{\frac{1}{\beta}}\right] \\
& = \alpha \bbE\left[\left(\mathds{1}_{\Omega_n}\left[4\exp{\left(-\frac{c \rho^2 \tau^2}{\sigma_W^2 \|\hat \gamma_n - \gamma_0\|_2^2}\right)} + 2\|f_\eta\|_\infty \rho \tau\right]\right)^{\frac{1}{\beta}}\right] \\
& \le \alpha \left(\left[4\exp{\left(-\frac{c \rho^2 \tau^2}{C\sigma_W^2 \frac{s_\gamma \log{p}}{n}}\right)} + 2\|f_\eta\|_\infty \rho \tau\right]\right)^{\frac{1}{\beta}}
\end{align*}
Choosing $\rho = (\sqrt{C}\sigma_W/\tau \sqrt{c})\sqrt{((s_\gamma \log{p})/n) \log{(n/(s_\gamma \log{p}))}}$  we have: 
\begin{align*}
& \alpha \bbE\left[\mathds{1}_{\Omega_n}\left(\bbP\left(|\hat \eta| \le \tau, |\eta| > \tau \mid \cD_1\right) + \bbP\left(|\hat \eta| > \tau, |\eta| \le \tau \mid \cD_1\right)\right)^{\frac{1}{\beta}}\right] \\
 & \lesssim \alpha \left(\frac{s_\gamma \log{p}}{n}\log{\left(\frac{n}{s_\gamma \log{p}}\right)}\right)^{\frac{1}{2\beta}} \\
& \triangleq \alpha t_n^{\frac{1}{\beta}} \hspace{0.2in} \left[t_n = \sqrt{\frac{s_\gamma \log{p}}{n}\log{\left(\frac{n}{s_\gamma \log{p}}\right)}}\right]\,.
\end{align*}
The above equation is true for all $\alpha, \beta > 1$ such that $1/\alpha + 1/\beta = 1$. Therefore minimizing the expression over $\alpha, \beta$ we have that the optimal choice of $\alpha$ is $\alpha = \log{(1/t_n)}$ which yields that: 
\begin{align*}
\left|\bbE\left[\mathds{1}_{\Omega_n}\bbE\left[W_j\left(\tilde \bS - \tilde W_{-1}^{\top}\theta^*_S\right)\left(\mathds{1}_{|\hat \eta| \le \tau}- \mathds{1}_{| \eta| \le \tau}\right) \mid \cD_1\right]\right]\right|  & \lesssim \log{\frac{1}{t_n}}t_n^{1- \frac{1}{\log{\frac{1}{t_n}}}} \\
& = t_n \log{\frac{1}{t_n}} \\
& \lesssim \sqrt{\frac{s_\gamma \log{p}}{n}}\left(\log{\frac{n}{s_\gamma \log{p}}}\right)^{3/2} \,.
\end{align*}
Using this bounds on the expectation of $\tilde \bW_{i, j}\left(\tilde \bS_i - \tilde \bW_{-1, i}^{\top}\theta^*_S\right)\left(\mathds{1}_{|\hat \eta_i| \le \tau}- \mathds{1}_{| \eta_i| \le \tau}\right)$ we have: 
$$
T_{122} \lesssim_\bbP \sqrt{\frac{s_\gamma \log{p}}{n}}\left(\log{\frac{n}{s_\gamma \log{p}}}\right)^{3/2}  + \sqrt{\frac{\log{p}}{n}} = \sqrt{\frac{s_\gamma \log{p}}{n}}\left(\log{\frac{n}{s_\gamma \log{p}}}\right)^{3/2}  \,.
$$
Combining the bound on $T_{121}$ and $T_{122}$ we have: 
$$
T_{12} \lesssim_\bbP \sqrt{\frac{s_\gamma \log{p}}{n}}\left(\log{\frac{n}{s_\gamma \log{p}}}\right)^{3/2} \,.
$$

\noindent
To bound $T_{11}$ not that from equation \eqref{eq:exp_1} we have: 
\begin{align*}
\breve \bW_{*, j}^{\perp} - \tilde \bW_{*, j} & = \breve \bW_{*, j}^{\perp} -  \bW^{\perp}_{*, j} +  \bW^{\perp}_{*, j} - \tilde \bW_{*, j} \\
& = \breve \bW_{*, j}^{\perp} -  \bW^{\perp}_{*, j}  + \left[m_j(\eta) -  \check m_j(\hat \eta)\right] + P^{\perp}_{\bN_k}\check \bR_j - P_{\bN_k}\check \nu_j
\end{align*}
Using this we have: 
\begin{align*}
    T_{11} & = \frac{1}{n}\max_{1 \le j \le p}\left|\left(\breve \bW_{*, j}^{\perp} - \tilde \bW_{*, j}\right)^{\top}\left(\tilde \bS - \tilde \bW_{-1}\theta^*_S\right)\right| \\
    & \le \underbrace{\frac{1}{n}\max_{1 \le j \le p}\left|\left(\breve \bW_{*, j}^{\perp} -  \bW^{\perp}_{*, j} \right)^{\top}\left(\tilde \bS - \tilde \bW_{-1}\theta^*_S\right)\right|}_{T_{111}}  + \underbrace{\frac{1}{n}\max_{1 \le j \le p}\left|\left(m_j(\eta) -  \check m_j(\hat \eta)\right)^{\top}\left(\tilde \bS - \tilde \bW_{-1}\theta^*_S\right)\right|}_{T_{112}} \\
    & \qquad \qquad + \underbrace{\frac{1}{n}\max_{1 \le j \le p}\left|\check \bR_j^{\top}P^{\perp}_{\bN_k}\left(\tilde \bS - \tilde \bW_{-1}\theta^*_S\right)\right|}_{T_{113}}  + \underbrace{\frac{1}{n}\max_{1 \le j \le p}\left|\check \nu_j^{\top}P_{\bN_k}\left(\tilde \bS - \tilde \bW_{-1}\theta^*_S\right)\right|}_{T_{114}}
\end{align*}
We start with bounding $T_{111}$. As mentioned previously, we have $\breve \bW_{*, j}^{\perp} -\bW^{\perp}_{*, j}$ is active for $p_1 + 2 \le j \le p_1 + p_2 + 1$ and for such a $j$: 
$$
\breve \bW_{i, j}^{\perp} -\bW^{\perp}_{i, j} = Z_{i,j}(\hat b'(\hat \eta_i) - b(\eta_i)) \,.
$$
Using this, we have: 
\begin{align*}
T_{111} & = \frac{1}{n}\max_{p_1 + 2 \le j \le p_1 + p_2 + 1}\left|\left(\breve \bW_{*, j}^{\perp} -  \bW^{\perp}_{*, j} \right)^{\top}\left(\tilde \bS - \tilde \bW_{-1}\theta^*_S\right)\right| \\
&\le \underbrace{\sqrt{\max_{p_1 + 2 \le j \le p_1 + p_2 + 1}\frac1n \sum_i Z_{i,j}^2\left(\hat b'(\hat \eta_i) - b'(\eta_i)\right)^2}}_{T_{1111}} \times \underbrace{\sqrt{\frac1n \left\|\tilde \bS - \tilde \bW_{-1}\theta^*_S\right\|^2}}_{T_{1112}}
\end{align*}
It follows from WLLN that $T_{1112} = O_p(1)$. For $T_{1111}$ similar calculation as of \eqref{eq:forgotten} yields: 
\begin{align*}
T_{1111} & = \sqrt{\max_{p_1 + 2 \le j \le p_1 + p_2 + 1}\frac1n \sum_i Z_{i,j}^2\left(\hat b'(\hat \eta_i) - b(\eta_i)\right)^2} \\
& \lesssim \sqrt{\max_{j}\frac1n \sum_i Z_{i,j}^2\left(\hat b'(\hat \eta_i) - b'(\hat \eta_i)\right)^2} + \sqrt{\max_{j}\frac1n \sum_i Z_{i,j}^2\left(b'(\hat \eta_i) - b'(\eta_i)\right)^2} \\
& \lesssim_\bbP \  \dot r_n \times \sqrt{\max_{j}\frac1n \sum_i Z_{i,j}^2} + \|b'\|_\infty \sqrt{\max_j \frac1n \sum_i Z_{i,j}^2\left(\hat \eta_i - \eta_i\right)^2} \\
& \lesssim_\bbP \  \dot r_n \times \sqrt{\max_{j}\frac1n \sum_i Z_{i,j}^2} + \|b'\|_\infty \sqrt{\frac{s_\gamma \log{p}}{n}}\sqrt{\max_j \frac1n \sum_i Z_{i,j}^2\left(Z_i^{\top}a_n\right)^2} \\
& \lesssim_\bbP \dot r_n + \sqrt{\frac{s_\gamma \log{p}}{n}} \,.
\end{align*}
Note that in the above analysis we have the used the facts: 
\begin{align*}
\max_{1 \le j \le p_2}\frac1n \sum_i Z_{i,j}^2 & = O_p(1) \hspace{0.2in} [\text{Lemma }\ref{lem:subg_second_moment}]\\
\max_{1 \le j \le p_2} \frac1n \sum_i Z_{i,j}^2\left(Z_i^{\top}a_n\right)^2 & = O_p(1) \hspace{0.2in}[\text{Lemma } \ref{lem:subg_fourth_moment}]\,.
\end{align*}
which follows from the subgaussianity of $Z$. Therefore the above bounds yield: 
$$
T_{111} \le T_{1111} \times T_{1112} \lesssim_\bbP \ \dot r_n + \sqrt{\frac{s_\gamma \log{p}}{n}} \,.
$$
To bound $T_{112}$, we use Cauchy-Schwarz inequality to show that with probability going to 1: 
\begin{align*}
    T_{112} & = \frac{1}{n}\max_{1 \le j \le p}\left|\left(m_j(\boeta) -  \check m_j(\hat \boeta)\right)^{\top}\left(\tilde \bS - \tilde \bW_{-1}\theta^*_S\right)\right| \\
    & \le \sqrt{\frac{1}{n}\max_{1 \le j \le p}\left\|m_j(\boeta) -  \check m_j(\hat \boeta) \right\|^2} \times \underbrace{\sqrt{\frac{1}{n}\left\|\tilde \bS - \tilde \bW_{-1}\theta^*_S\right\|^2}}_{O_p(1)} \\
    & \lesssim_\bbP \  \sqrt{\frac{s_\gamma \log{p_2}}{n}} \,.
\end{align*}
For $T_{113}$, we further sub-divide it as follows: 
\begin{align}
   T_{113}  & \le \underbrace{\frac{1}{n}\max_{1 \le j \le p}\left|\check \bR_j^{\top}P^{\perp}_{\bN_k}\left(\hat{\tilde \bS} - \hat{\tilde \bW}_{-1}
    \theta^*_S\right)\right|}_{T_{1131}}  + \underbrace{\frac{1}{n}\max_{1 \le j \le p}\left|\check \bR_j^{\top}P^{\perp}_{\bN_k}\left(\tilde \bS - \hat{\tilde \bS}\right)\right|}_{T_{1132}}  \notag \\
    \label{eq:exp_2}& \qquad \qquad \qquad + 
    \underbrace{\frac{1}{n}\max_{1 \le j \le p}\left|\check \bR_j^{\top}P^{\perp}_{\bN_k}\left(\tilde \bW_{-1} - \hat{\tilde \bW}_{-1}\right)
    \theta^*_S\right|}_{T_{1133}}
\end{align}
To bound $T_{1131}$, note that $\check \bR_j^{\top}P_{\hat \bN_k}^{\perp}$ is measurable with respect to $\hat \eta$, whereas $\hat{\tilde \bS} - \hat{\tilde \bW}_{-1}\theta^*_S$ has mean zero conditional on $\hat \eta$. Therefore, from subgaussian concentration inequality we have: 
\begin{align*}
    \bbP\left(\frac{1}{n}\max_{1 \le j \le p}\left|\check \bR_j^{\top}P_{\hat \bN_k}^{\perp}\left(\hat{\tilde \bS} - \hat{\tilde \bW}_{-1}\theta^*_S\right)\right| > t\right) & = \sum_{1 \le j \le p}\bbP\left(\frac{1}{n}\left|\check \bR_j^{\top}P_{\hat \bN_k}^{\perp}\left(\hat{\tilde \bS} - \hat{\tilde \bW}_{-1}\theta^*_S\right)\right| > t\right) \\
    & = \sum_{1 \le j \le p}\bbE_{\hat \eta}\left[\bbP\left(\frac{1}{n}\left|\check \bR_j^{\top}P_{\hat \bN_k}^{\perp}\left(\hat{\tilde \bS} - \hat{\tilde \bW}_{-1}\theta^*_S\right)\right| > t \mid \hat \eta\right)\right] \\
    & \le \sum_{1 \le j \le p}\bbE_{\hat \eta}\left[2\exp{\left(-c\frac{nt^2}{\sigma_W \frac{\check \bR_j^{\top}\check \bR_j}{n}}\right)}\right] \\
    & \le 2\exp{\left(\log{p}-c\frac{nt^2}{\sigma_W r_n^2}\right)}
\end{align*}
Therefore we have with probability going to $1$: 
$$
T_{1131} = \frac{1}{n}\max_{1 \le j \le p}\left|\check \bR_j^{\top}P_{\hat \bN_k}^{\perp}\left(\hat{\tilde \bS} - \hat{\tilde \bW}_{-1}\theta^*_S\right)\right| \lesssim_\bbP \  r_n\sqrt{\frac{\log{p}}{n}} \,.
$$
For $T_{1132}$ of equation \eqref{eq:exp_2} we apply CS inequality and use some already derived bounds to conclude that with probability going to 1: 
\begin{align*}
T_{1132} = \frac{1}{n}\max_{1 \le j \le p}\left|\check \bR_j^{\top}P^{\perp}_{\bN_k}\left(\tilde \bS - \hat{\tilde \bS}\right)\right| & \le \sqrt{\frac{1}{n}\max_{1 \le j \le p} \check \bR_j^{\top} \check \bR_j} \times \sqrt{\frac{1}{n}\left\|\tilde \bS - \hat{\tilde \bS}\right\|^2} \\
& \lesssim_\bbP \  r_n \sqrt{\frac{s_\gamma \log{p}}{n}} \,.
\end{align*}
For $T_{1133}$ we have: 
\begin{align*}
    \frac{1}{n}\max_{1 \le j \le p}\left|\check \bR_j^{\top}P^{\perp}_{\bN_k}\left(\tilde \bW_{-1} - \hat{\tilde \bW}_{-1}\right)
    \theta^*_S\right| & \le \|\theta^*_S\|_1 \frac{1}{n}\max_{1 \le j, j' \le p}\left|\check \bR_j^{\top}P^{\perp}_{\bN_k}\left(\tilde \bW_{*, j'} - \hat{\tilde \bW}_{*, j'}\right)\right| \\
    & \le \|\theta^*_S\|_1 \sqrt{\frac{1}{n}\max_{1 \le j \le p} \check \bR_j^{\top}\check \bR_j} \times \sqrt{\frac{1}{n}\max_{1 \le j' \le p}\left\|\tilde \bW_{*, j'} - \hat{\tilde \bW}_{*, j'}\right\|^2} \\
    & \lesssim_\bbP \  \|\theta^*_S\|_1 r_n \sqrt{\frac{s_\gamma \log{p}}{n}} \,.
\end{align*}
which concludes: 
$$
T_{1133}  \lesssim_{\bbP}  \|\theta^*_S\|_1 r_n \sqrt{\frac{s_\gamma \log{p}}{n}} \,.
$$
Combining the bounds on the different parts of $T_{113}$ we have: 
$$
T_{113} \lesssim_\bbP \ \left(1 \vee \|\theta^*_S\|_1\right)r_n \sqrt{\frac{s_\gamma \log{p}}{n}} \,.
$$
Finally, to bound $T_{114}$, we first expand it as we just did for the residual term $\check \bR_j$:    
\begin{align}
    T_{114} & = \frac{1}{n}\max_{1 \le j \le p}\left|\check \nu_j^{\top}P_{\bN_k}\left(\tilde \bS - \tilde \bW_{-1}\theta^*_S\right)\right|  \\
    & \le \underbrace{\frac{1}{n}\max_{1 \le j \le p}\left|\check \nu_j^{\top}P_{\bN_k}\left(\hat{\tilde \bS} - \hat{\tilde \bW}_{-1}
    \theta^*_S\right)\right|}_{T_{1141}}  + \underbrace{\frac{1}{n}\max_{1 \le j \le p}\left|\check \nu_j^{\top}P_{\bN_k}\left(\tilde \bS - \hat{\tilde \bS}\right)\right|}_{T_{1142}}  \notag \\
    \label{eq:exp_3}& \qquad \qquad \qquad + \underbrace{\frac{1}{n}\max_{1 \le j \le p}\left|\check \nu_j^{\top}P_{\bN_k}\left(\tilde \bW_{-1} - \hat{\tilde \bW}_{-1}\right)
    \theta^*_S\right|}_{T_{1143}} 
\end{align}
To bound $T_{1141}$, we use the same technique we use to bound $(\check \nu_j^{\top}P_{\bN_k}\check \nu_j)/n$ in the proof of Lemma \ref{lem:sq_rate}. Recall that we use Hanson-Wright inequality using the fact that $\check \nu_j$ has mean 0 conditional on $\hat \eta$ and $P_{\bN_k}$ is measurable with respect to $\hat \eta$. Therefore we have: 
\begin{align*}
    T_{1141} & = \frac{1}{n}\max_{1 \le j \le p}\left|\check \nu_j^{\top}P_{\bN_k}\left(\hat{\tilde \bS} - \hat{\tilde \bW}_{-1}
    \theta^*_S\right)\right| \\
    & \le \sqrt{\frac{1}{n}\max_{1 \le j \le p}\check \nu_j^{\top}P_{\bN_k}\check \nu_j} \times \sqrt{\frac{1}{n}\left(\hat{\tilde \bS} - \hat{\tilde \bW}_{-1}
    \theta^*_S\right)^{\top}P_{\bN_k}\left(\hat{\tilde \bS} - \hat{\tilde \bW}_{-1}
    \theta^*_S\right)} \\
    & \lesssim_\bbP \  \dot r_n^2 + \left(\frac{\log{p}}{n} \vee \dot r_n \sqrt{\frac{\log{p}}{n}}\right) \\
    & \lesssim_\bbP \ \left(\dot r_n + \sqrt{\frac{\log{p}}{n}}\right)^2 \,.
\end{align*}
For $T_{1142}$ of equation \eqref{eq:exp_3} we have: 
\begin{align*}
    \frac{1}{n}\max_{1 \le j \le p}\left|\check \nu_j^{\top}P_{\bN_k}\left(\tilde \bS - \hat{\tilde \bS}\right)\right| & \le \sqrt{\frac{1}{n}\max_{1 \le j \le p}\check \nu_j^{\top}P_{\bN_k}\check \nu_j} \times \sqrt{\frac{1}{n}\left\|\tilde \bS - \hat{\tilde \bS}\right\|^2} \\
    & \lesssim_\bbP \sqrt{\dot r_n^2 + \left(\frac{\log{p}}{n} \vee \dot r_n \sqrt{\frac{\log{p}}{n}}\right)} \times \sqrt{\frac{s_\gamma \log{p}}{n}} \\
    & \lesssim_\bbP \  \left(\dot r_n + \sqrt{\frac{\log{p}}{n}}\right)\sqrt{\frac{s_\gamma \log{p}}{n}} \,.
\end{align*}
Similarly for $T_{1143}$ of equation \eqref{eq:exp_3} we have: 
\begin{align*}
    \frac{1}{n}\max_{1 \le j \le p}\left|\check \nu_j^{\top}P_{\bN_k}\left(\tilde \bW_{-1} - \hat{\tilde \bW}_{-1}\right)
    \theta^*_S\right|  & \le \|\theta^*_S\|_1 \frac{1}{n}\max_{1 \le j, j' \le p}\left|\check \nu_j^{\top}P_{\bN_k}\left(\tilde \bW_{*, j'} - \hat{\tilde \bW}_{*, j'}\right)\right| \\
    & \le \|\theta^*_S\|_1 \sqrt{\frac{1}{n}\max_{1 \le j \le p}\check \nu_j^{\top}P_{\bN_k}\check \nu_j} \times \sqrt{\frac{1}{n}\max_{1 \le j' \le p}\left\|\tilde \bW_{*, j'} - \hat{\tilde \bW}_{*, j'}\right\|^2} \\
    & \le \|\theta^*_S\|_1\sqrt{\dot r_n^2 + \left(\frac{\log{p}}{n} \vee \dot r_n \sqrt{\frac{\log{p}}{n}}\right)} \times \sqrt{\frac{s_\gamma \log{p_2}}{n}} \\
    & \le \|\theta^*_S\|_1\left(\dot r_n + \sqrt{\frac{\log{p}}{n}}\right)\sqrt{\frac{s_\gamma \log{p_2}}{n}} \,.
\end{align*}
Aggregating the bound on the different parts of $T_{114}$ we have: 
$$
T_{114} \lesssim_\bbP \ \left(\dot r_n + \sqrt{\frac{\log{p}}{n}}\right)\left(\dot r_n + \sqrt{\frac{\log{p}}{n}} + (1 + \|\theta^*_S\|_1)\sqrt{\frac{s_\gamma \log{p}}{n}}\right) \,. 
%
$$
This completes the bounds for $T_{11}$, which yields: 
$$
T_{11} \lesssim_\bbP \ \dot r_n + \sqrt{\frac{s_\gamma \log{p}}{n}} + \left(1 \vee \|\theta^*_S\|_1\right)\left(\dot r_n\sqrt{\frac{s_\gamma \log{p}}{n}} \vee \sqrt{s_\gamma}\frac{\log{p}}{n}\right)\,.
$$
Combining the bounds for $T_{11}$ and $T_{12}$ we have: 
$$
T_{1} \lesssim_\bbP \ \dot r_n +  \sqrt{\frac{s_\gamma \log{p}}{n}}\left(\log{\frac{n}{s_\gamma \log{p}}}\right)^{3/2}  + \left(1 \vee \|\theta^*_S\|_1\right)\left(r_n\sqrt{\frac{s_\gamma \log{p}}{n}} \vee \sqrt{s_\gamma}\frac{\log{p}}{n}\right) \,.
$$
So far we have: 
\begin{align*}
    T_{1} & \lesssim_\bbP \ \dot r_n +  \sqrt{\frac{s_\gamma \log{p}}{n}}\left(\log{\frac{n}{s_\gamma \log{p}}}\right)^{3/2}  + \left(1 \vee \|\theta^*_S\|_1\right)\left(\dot r_n\sqrt{\frac{s_\gamma \log{p}}{n}} \vee \sqrt{s_\gamma}\frac{\log{p}}{n}\right) \,, \\
    T_2 & \lesssim_\bbP \  \left[\dot r_n +  \sqrt{\frac{s_\gamma \log{p}}{n}}\right]  \,, \\
    T_3 & \lesssim_\bbP \  (1 \vee \|\theta^*_S\|_1) \left[\dot r_n +  \sqrt{\frac{s_\gamma \log{p}}{n}}\right] \,.
\end{align*}
which implies: 
\begin{align*}
T_1 + T_2 + T_3  & \lesssim_\bbP \  (1 + \|\theta^*_S\|_1)\left[\dot r_n +   \sqrt{\frac{s_\gamma \log{p}}{n}}\left(\log{\frac{n}{s_\gamma \log{p}}}\right)^{3/2}  + \dot r_n\sqrt{\frac{s_\gamma \log{p}}{n}} \vee \sqrt{s_\gamma}\frac{\log{p}}{n}\right] \asymp \lambda_1 \\
& \lesssim_\bbP (1 + \|\theta^*_S\|_1)\left[\dot r_n +  \sqrt{\frac{s_\gamma \log{p}}{n}} \right] + \sqrt{\frac{s_\gamma \log{p}}{n}}\left(\log{\frac{n}{s_\gamma \log{p}}}\right)^{3/2}  \\
& \asymp \lambda_1 \,.
\end{align*}
where the last inequality follows from the fact that $\dot r_n\sqrt{s_\gamma \log{p}/n} \ll \dot r_n$ and $\sqrt{s_\gamma}(\log{p}/n) \ll \sqrt{(s_\gamma \log{p})/n}$. 
\\\\
\noindent
We next find the appropriate bound on $\lambda_0$. From equation \eqref{eq:y_lasso}, the value of $\lambda_0$ depends on the the bound on the $\ell_\infty$ norm of the first term of RHS. We start with the similar division as we used for $\lambda_1$ as follows: 
\begin{align*}
     \frac1n \left\|\left(\bY^{\perp} - \breve \bW^{\perp}_{-1}\theta^*_Y \right)^{\top}\breve \bW^{\perp}_{-1}\right\|_\infty  & \le \underbrace{\frac{1}{n}\max_{1 \le j \le p}\left|\left(\breve \bW_{*, j}^{\perp}\right)^{\top}\left(\tilde \bY - \tilde \bW_{-1}\theta^*_Y\right)\right|}_{T_1} + \underbrace{\frac{1}{n}\max_{1 \le j \le p}\left|\left(\breve \bW_{*, j}^{\perp}\right)^{\top}\left(\bY^{\perp} - \tilde \bY\right)\right|}_{T_2}  \\
    & \qquad \qquad + \underbrace{\max_{1 \le j \le p} \frac{1}{n}\left|\left(\breve \bW_{*, j}^{\perp}\right)^{\top}\left(\breve \bW_{-1}^{\perp} - \tilde \bW_{-1}\right)\theta^*_Y\right|}_{T_3}
\end{align*}
Observe that, $T_3$ here is identical to the term $T_3$ in equation \eqref{eq:basic_exp_lemma_lambda} in estimating $\lambda_1$ only $\theta^*_S$ replaced by $\theta^*_Y$. Therefore, from a direct approach to our previous calculations we conclude that with probability going to 1: 
$$
T_3 \lesssim_\bbP \  \|\theta^*_Y\|_1\left[\dot r_n + \sqrt{\frac{s_\gamma \log{p}}{n}}\right] \,.
$$
$T_2$ can also be bounded along the same line of argument used to prove the third display of Lemma \ref{lem:sq_rate}. 
An expansion of $\bY^{\perp} - \tilde \bY$ along the line of equation \eqref{eq:exp_1} yields: 
$$
\bY^{\perp} - \tilde \bY = m_0(\boeta) - \check m_0(\hat \boeta) + P_{\bN_k}^{\perp}\check \bR_0 - P_{\bN_k}\check \nu_0 \,.
$$
where $\check \nu_0 = \bY - \check m_0(\hat \boeta)$ is subgaussian with constant $\sigma_Y$. Thereby following the line of argument as of Lemma \ref{lem:sq_rate} we have: 
$$
\frac{1}{n}\left\|\bY^{\perp} - \tilde \bY\right\|^2 \lesssim_\bbP \frac{s_\gamma \log{p}}{n} + \left[\dot r_n^2 + \left(\frac{\log{p}}{n} \vee \dot r_n \sqrt{\frac{\log{p}}{n}}\right)\right]
$$
On the other hand for the cross term we expand it as: 
\begin{align*}
T_2 & \le \underbrace{\frac{1}{n}\max_{1 \le j \le p}\left|\left(\breve \bW_{*, j}^{\perp} - \widehat{\tilde \bW}_{*, j}\right)^{\top}\left(\bY^{\perp} - \tilde \bY\right)\right|}_{T_{21}} + \underbrace{\frac{1}{n}\max_{1 \le j \le p}\left|\widehat{\tilde \bW}_{*, j}^{\top}\left(\bY^{\perp} - \tilde \bY\right)\right|}_{T_{22}}
\end{align*}
For $T_{21}$, by CS inequality we have: 
\begin{align*}
    T_{21} & \le \sqrt{\frac{1}{n}\max_{1 \le j \le p}\left\|\left(\breve \bW_{*, j}^{\perp} - \widehat{\tilde \bW}_{*, j}\right)\right\|^2} \times \sqrt{\frac{1}{n}\left\|\left(\bY^{\perp} - \tilde \bY\right)\right\|^2} \\
    & \lesssim_\bbP \ \dot r_n^2 + \frac{s_\gamma \log{p}}{n} +\dot r_n \sqrt{\frac{\log{p}}{n}} \hspace{0.2in}[\text{Lemma }\ref{lem:sq_rate}]\,.
\end{align*}
To bound $T_{22}$: 
\begin{align}
    T_{22} &= \frac{1}{n}\max_{1 \le j \le p}\left|\widehat{\tilde \bW}_{*, j}^{\top}\left(\bY^{\perp} - \tilde \bY\right)\right| \notag \\
    \label{eq:bound_T2_Y} & \le \underbrace{\frac{1}{n}\max_{1 \le j \le p}\left|\widehat{\tilde \bW}_{*, j}^{\top}\left(m_{0}(\eta) - \check m_{0}(\hat \eta)\right)\right|}_{T_{221}} + \underbrace{\frac{1}{n}\max_{1 \le j \le p}\left|\widehat{\tilde \bW}_{*, j}^{\top}P_{\bN_k}^{\perp}\check \bR_{0}\right|}_{T_{222}} + \underbrace{\frac{1}{n}\max_{1 \le j \le p}\left|\widehat{\tilde \bW}_{*, j}^{\top}P_{\bN_k} \check \nu_{0}\right|}_{T_{223}}
\end{align}
To bound $T_{221}$ of equation \eqref{eq:bound_T2_Y}, first note that: 
$$
\frac{1}{n}\max_{1 \le j \le p}\left\|\widehat{\tilde \bW}_{*, j}\right\|^2 \lesssim \frac{1}{n}\max_{1 \le j \le p}\left\|\tilde \bW_{*, j}\right\|^2  + \frac{1}{n}\max_{1 \le j \le p}\left\|\left(\tilde \bW_{*, j} - \widehat{\tilde \bW}_{*, j}\right)\right\|^2  \lesssim_\bbP 1 \,.
$$
where the last inequality follows from the fact that $(1/n)\max_{1 \le j \le p}\|\tilde \bW_{*, j}\|^2 \lesssim_\bbP 1$ (as $\var(W_{-j})$ has uniformly bounded variances over $j$ and $W_{j}$'s are uniformly sub-gaussian) and as we have established in Lemma \ref{lem:sq_rate} $(1/n)\max_j \|(\tilde \bW_{*, j} - \widehat{\tilde \bW}_{*, j})\|^2  = o_p(1)$. 
Therefore we have: 
\begin{align*}
T_{221} & = \frac{1}{n}\max_{1 \le j \le p}\left|\widehat{\tilde \bW}_{*, j}^{\top}\left(m_{0}(\eta) - \check m_{0}(\hat \eta)\right)\right| \\
& \lesssim \underbrace{\sqrt{\frac{1}{n}\max_{1 \le j \le p}\left\|\widehat{\tilde \bW}_{*, j}\right\|^2}}_{\lesssim_\bbP 1} \times \underbrace{\sqrt{\frac{1}{n}\sum_i \left(\check m_{0}(\hat \eta_i) - m_{0}(\eta_i)\right)^2}}_{\lesssim_\bbP \sqrt{\frac{s_\gamma \log{p}}{n}}} \\
& \lesssim_\bbP \  \sqrt{\frac{s_\gamma \log{p}}{n}} \,.
\end{align*}
where the bound $(1/n)\|\check m_0(\hat \boeta) - m_0(\boeta)\|^2$ follows from equation \eqref{eq:m_diff_bound}. 

To bound $T_{222}$ of equation \eqref{eq:bound_T2_Y}, note that the entries of $\widehat{\tilde \bW}_{*, j}$ are centered and independent conditionally on $\hat \eta$. Therefore we have: 
\begin{align*}
    \bbP\left(\frac{1}{n}\max_{1 \le j \le p}\left|\widehat{\tilde \bW}_{*, j}^{\top}P_{\bN_k}^{\perp}\check\bR_{0}\right| > t \right) & \le \sum_{1 \le j \le p} \bbP\left(\frac{1}{n}\left|\widehat{\tilde \bW}_{*, j}^{\top}P_{\bN_k}^{\perp}\check\bR_{0}\right| > t \right) \\
    & \le \sum_{1 \le j \le p} \bbE\left[\bbP\left(\frac{1}{n}\left|\widehat{\tilde \bW}_{*, j}^{\top}P_{\bN_k}^{\perp}\check\bR_{0}\right| > t \mid \hat \eta\right)\right] \\
    & \le \sum_{1 \le j \le p} \bbE\left[2\exp{\left(-c\frac{nt^2}{\sigma_W \frac{\check\bR_{0}^{\top}\check\bR_{0}}{n}}\right)}\right] \\
    & \le 2\exp{\left(2\log{p}-c\frac{nt^2}{\sigma_W r_n^2}\right)}
\end{align*}
Using an appropriate choice of $t$ we conclude that with probability going to 1: 
$$
T_{222} = \frac{1}{n}\max_{1 \le j \le p}\left|\widehat{\tilde \bW}_{*, j}^{\top}P_{\bN_k}^{\perp}\check \bR_{0}\right| \lesssim_\bbP \  r_n\sqrt{\frac{\log{p}}{n}} \,.
$$
For the last term of equation \eqref{eq:bound_T2} note that: 
$$
T_{223} = \frac{1}{n}\max_{1 \le j \le p}\left|\widehat{\tilde \bW}_{*, j}^{\top}P_{\bN_k} \check \nu_{0}\right| =  \frac{1}{n}\max_{1 \le j \le p}\left|\check \nu_j^{\top}P_{\bN_k} \check \nu_{0}\right|
$$
as by definition $\check \nu_j = \widehat{\tilde \bW}_{*, j} = \bW_{*, j} - \bbE[W_j \mid \hat \boeta]$. Therefore an application of Cauchy-Schwarz inequality yields:
$$
T_{223} \le \sqrt{\frac{1}{n}\max_{1 \le j \le p}\left|\check \nu_j^{\top}P_{\bN_k} \check \nu_j\right|}\times \sqrt{\frac{1}{n}\left|\check \nu_0^{\top}P_{\bN_k} \check \nu_{0}\right|}
$$
Now we can bound both the terms of the above equation via similar calculation used to yield the rate in equation \eqref{eq:resid_etahat_bound}. Therefore we have: 
$$
T_{223} \lesssim_\bbP \ \dot r_n^2 + \left(\frac{\log{p}}{n} \vee \dot r_n \sqrt{\frac{\log{p}}{n}}\right) \,.
$$
Combining the bounds for different components of $T_{22}$ we have, with probability going to $1$: 
$$
T_{22} = \max_{1 \le j \le p} \frac{1}{n}\left|\left(\breve \bW_{*, j}^{\perp}\right)^{\top}\left(\bY^{\perp} -  \tilde \bY\right)\right| \lesssim_\bbP \left[\dot r_n^2 + \left(\frac{\log{p}}{n} \vee \dot r_n \sqrt{\frac{\log{p}}{n}}\right)\right] + \sqrt{\frac{s_\gamma \log{p}}{n}} \,.
$$
This implies we can bound $T_2$ as: 
\begin{align*}
T_2 & \lesssim_\bbP \  \left[\dot r_n^2 + \left(\frac{\log{p}}{n} \vee \dot r_n \sqrt{\frac{\log{p}}{n}}\right)\right] + \sqrt{\frac{s_\gamma \log{p}}{n}} \\
& \lesssim_\bbP \dot r_n^2 + \sqrt{\frac{s_\gamma \log{p}}{n}} \,.
\end{align*}
Combining the bounds on $T_2$ and $T_3$ we have: 
$$
T_2 + T_3 \lesssim_\bbP (1 \vee \|\theta^*_Y\|_1)\left[\dot r_n + \sqrt{\frac{s_\gamma \log{p}}{n}}\right] \,.
$$
Finally for $T_1$ we expand it as before: 
\begin{align*}
    T_1 & = \frac{1}{n}\max_{1 \le j \le p}\left|\left(\breve \bW_{*, j}^{\perp}\right)^{\top}\left(\tilde \bY - \tilde \bW_{-1}\theta^*_Y\right)\right| \\
    & \le \underbrace{\frac{1}{n}\max_{1 \le j \le p}\left|\left(\breve \bW_{*, j}^{\perp} - \tilde \bW_{*, j}\right)^{\top}\left(\tilde \bY - \tilde \bW_{-1}\theta^*_Y\right)\right|}_{T_{11}} + \underbrace{\frac{1}{n}\max_{1 \le j \le p}\left|\tilde \bW_{*, j}^{\top}\left(\tilde \bY - \tilde \bW_{-1}\theta^*_Y\right)\right|}_{T_{12}}
\end{align*}
To bound $T_{12}$ we have to adopt similar approach taken to bound $T_{12}$ of the analysis of $\lambda_1$ (equation \eqref{eq:centering_issue}). The reason is also same, i.e. the term $\tilde \bW_{i, j}(\tilde \bY_i - \tilde \bW_{i, -1}^{\top}\theta^*_Y)\mathds{1}(|\hat \eta_i| \le \tau)$ is not centered and consequently to apply Bernstein's inequality we need to bound the expectation of the above term. For brevity, we won't repeat the same calculation here as the same calculation will lead us to same conclusion: 
$$
T_{12} \lesssim_{\bbP}  \sqrt{\frac{s_\gamma \log{p}}{n}}\left(\log{\frac{n}{s_\gamma \log{p}}}\right)^{3/2} + \sqrt{\frac{\log{p}}{n}} \lesssim_\bbP   \sqrt{\frac{s_\gamma \log{p}}{n}}\left(\log{\frac{n}{s_\gamma \log{p}}}\right)^{3/2}\,.
$$
Analysis of $T_{11}$ is also similar to the $T_{11}$ term in the derivation of $\lambda_1$, with $\tilde \bS -   \tilde \bW_{-1}\theta^*_S$ is now replaced by $\tilde \bY - \tilde \bW_{-1}\theta^*_Y$. Further, as mentioned previously, we need to the careful about the subgaussian constant $\sigma_Y$ of $\tilde \bY - \tilde \bW_{-1}\theta^*_Y$. This yields that with probability going to $1$: 
$$
T_{11} \lesssim_\bbP \left(r_n^2 + \frac{\log{p}}{n}\right) + \left\{1 + (1 + \|\theta^*_Y\|_1)\left(r_n + \sqrt{\frac{\log{p}}{n}}\right)\right\}\sqrt{\frac{s_\gamma \log{p}}{n}}
$$
Hence we have: 
\begin{align*}
T_1 & = T_{11} + T_{12} \\
& \lesssim_\bbP  \bbP \left(r_n^2 + \frac{\log{p}}{n}\right) + \left\{\left(\log{\frac{n}{s_\gamma \log{p}}}\right)^{3/2} + (1 + \|\theta^*_Y\|_1)\left(r_n + \sqrt{\frac{\log{p}}{n}}\right)\right\}\sqrt{\frac{s_\gamma \log{p}}{n}} \\
& \lesssim_\bbP r_n^2 + \left\{\left(\log{\frac{n}{s_\gamma \log{p}}}\right)^{3/2} + (1 + \|\theta^*_Y\|_1)\left(r_n + \sqrt{\frac{\log{p}}{n}}\right)\right\}\sqrt{\frac{s_\gamma \log{p}}{n}}
\,.
\end{align*}
Finally, combining the bounds on $T_1, T_2, T_3$ we conclude: 
\begin{align*}
& \frac1n \left\|\left(\bY^{\perp} - \breve \bW^{\perp}_{-1}\theta^*_Y \right)^{\top}\breve \bW^{\perp}_{-1}\right\|_\infty  \\ 
& \qquad \qquad \lesssim_\bbP  (1 \vee \|\theta^*_Y\|_1)\left[\dot r_n + \sqrt{\frac{s_\gamma \log{p}}{n}}\right]  + \sqrt{\frac{s_\gamma \log{p}}{n}}\left(\log{\frac{n}{s_\gamma \log{p}}}\right)^{3/2} \\
&  \qquad \qquad \asymp \lambda_0 \,.
\end{align*}
This completes the proof of this lemma.

%

\subsection{Proof of Proposition \ref{prop:denom_conv}}
A basic expansion yields: 
\begin{align}
    \frac{1}{n_3}\left\|\bS^{\perp} -  \breve \bW^{\perp}_{-1}\hat \theta_{-1,1}\right\|^2 & = \frac{1}{n_3}\left\|\tilde \bS -  \tilde \bW_{-1}\theta^*_S\right\|^2 + \frac{1}{n_3}\left\|\tilde \bS - \bS^{\perp} \right\|^2 + \frac{1}{n_3}\left\|\breve \bW^{\perp}_{-1}\left(\hat \theta_{-1,1} - \theta^*_S \right)\right\|^2 \notag \\
    \label{eq:denom_exp_1} & \qquad \qquad \qquad + \frac{1}{n'_3}\left\|\left(
    \breve \bW^{\perp} - \tilde \bW\right)\theta^*_S\right\|^2
\end{align}
That the first term on the RHS of equation \eqref{eq:denom_exp_1} is $O_p(1)$ directly follows from assumption \ref{assm:sigma}. The bound on the second term was already established is the proof of Lemma \ref{lem:sq_rate} which ensures that this term is $o_p(1)$. The asymptotic negligibility of the third term directly follows from the prediction consistency of LASSO. Finally for the last term, define $\cS_1$ to be the set of active elements in $\theta^*_S$. By our assumption, we have $|\cS_1| \le s_1$. Hence we have:  
\begin{align*}
\frac{1}{n_3}\left\|\left(\breve \bW^{\perp} - \tilde \bW\right)\theta^*_S\right\|^2 = \frac{1}{n_3}\left\|\left(\breve \bW^{\perp}_{\cS_1} - \tilde \bW_{\cS_1}\right)\theta^*_S\right\|^2 & \le \lambda_{\max}\left(\frac{\left(\breve \bW^{\perp}_{\cS_1} - \tilde \bW_{\cS_1}\right)^{\top}\left(\breve \bW^{\perp}_{\cS_1} - \tilde \bW_{\cS_1}\right)}{n'_3}\right) \\
& \le \frac{1}{n_3}\tr\left(\left(\breve \bW^{\perp}_{\cS_1} - \tilde \bW_{\cS_1}\right)^{\top}\left(\breve \bW^{\perp}_{\cS_1} - \tilde \bW_{\cS_1}\right)\right) \\
& \le \frac{1}{n_3}\sum_{j \in \cS_1}\left\|\breve \bW^{\perp}_{*, j} - \tilde \bW_{*, j}\right\|^2  \\
& \lesssim_\bbP s_1 \left\{\dot r_n^2 + \frac{s_\gamma \log{p}}{n} + \dot r_n \sqrt{\frac{\log{p}}{n}}\right\} \,.
\end{align*}
where the last line follows from the second part of Lemma \ref{lem:sq_rate}. Now we have presented some sufficient condition in subsection \ref{sec:discussion_sparsity_rate} of the main document (especially \eqref{eq:disc_1}, \eqref{eq:disc_2}) under which the above bound in $o(1)$. Therefore under those sufficient condition, we have established that: 
\begin{align*}
 \frac{1}{n_3}\left\|\bS^{\perp} -  \breve \bW^{\perp}_{-1}\hat \theta_{-1,1}\right\|^2 & = \frac{1}{n_3}\left\|\tilde \bS -  \tilde \bW_{-1}\theta^*_S\right\|^2 + o_p(1) \\
 & = \frac{1}{n_3}\sum_{i=1}^{n/3} \left(\tilde S_i - \tilde W_{-1}^{\top}\theta^*_S\right)^2\mathds{1}_{|\hat \eta_i| \le \tau} \\
& = \frac{1}{n_3}\sum_{i=1}^{n/3} \left(\tilde S_i - \tilde W_{-1}^{\top}\theta^*_S\right)^2\mathds{1}_{|\eta| \le \tau} + o_p(1) \\
& = \bbE\left[\left(\tilde S - \tilde W_{-1}^{\top}\theta^*_S\right)^2 \mathds{1}_{|\eta| \le \tau}\right] 
\end{align*}
This completes the proof. 
%
%
%
%

\subsection{Proof of Proposition \ref{prop:RE}}
As per assumption \ref{assm:sparsity_RE}, we have some constants $\kappa, c > 0$ such that with high probability: 
$$
\inf_{\|\Delta_{S^c}\|_1 \le c\|\Delta_S\|_1} \frac{\frac{1}{n}\left\|\tilde \bW_{-1}\Delta\right\|^2}{\left\|\Delta\right\|^2} \ge \kappa \,.
$$
As we regress both $\bY$ and $\bS$ on $\breve \bW_{-1}^{\perp}$, the set $S$ here generically used to denote the active set of both $\theta^*_S$ and $\theta^*_Y$. We will show that the above inequality also holds for $\breve \bW_{-1}^{\perp}$ some $\kappa'$ (which can be taken as $\kappa/2$ for all large $n$) with high probabiltiy. Towards that end, first triangle inequality yields: 
$$
\left\|\breve \bW_{-1}^{\perp}\Delta\right\| \ge \left\|\tilde \bW_{-1}\Delta\right\| - \left\|\left(\breve \bW_{-1}^{\perp} - \tilde \bW_{-1}\right)\Delta\right\|
$$
We next show that with probability going to 1, 
$$
\sup_{\|\Delta_{S^c}\|_1 \le c\|\Delta_S\|_1} \frac{\frac1n \left\|\left(\breve \bW_{-1}^{\perp} - \tilde \bW_{-1}\right)\Delta\right\|^2}{\|\Delta\|^2} = o_p(1)
$$
which will complete the proof. Note that we have: 
$$
\left\|\left(\breve \bW_{-1}^{\perp} - \tilde \bW_{-1}\right)\Delta\right\|^2 = \sum_{j=1}^p \Delta^2_j \left\|\breve \bW_{*, j}^{\perp} - \tilde \bW_{*, j}\right\|^2 + \sum_{j \neq k} \Delta_j \Delta_k \left\langle \bW_{*, j}^{\perp} - \tilde \bW_{*, j}, \bW_{*, k}^{\perp} - \tilde \bW_{*, k} \right \rangle 
$$
which yields the following bound: 
\begin{align*}
    \frac{\frac1n \left\|\left(\breve \bW_{-1}^{\perp} - \tilde \bW_{-1}\right)\Delta\right\|^2}{\|\Delta\|^2} & \le \frac{1}{n} \max_j \left\|\breve \bW_{*, j}^{\perp} - \tilde \bW_{*, j}\right\|^2\left(1 +  \sup_{\|\Delta_{S^c}\|_1 \le c\|\Delta_S\|_1} \frac{\|\Delta\|_1^2}{\|\Delta\|^2}\right)
\end{align*}
For any $\Delta$ with $\|\Delta_{S^c}\|_1 \le c\|\Delta_S\|_1$ we have: 
\begin{align*}
    \|\Delta\|_1^2 & = \left(\|\Delta_S\|_1 + \|\Delta_{S^c}\|_1\right)^2 \\
    & \le (c+1)^2 \|\Delta_S\|^2_1 \\
    & \le (c+1)^2 \|\Delta_S\|^2_1  \\
    & \le s(c+1)^2\|\Delta_S\|^2 \\
    & \le s(c+1)^2\|\Delta\|^2 
\end{align*}
Hence, we have: 
$$
\sup_{\|\Delta_{S^c}\|_1 \le c\|\Delta_S\|_1} \frac{\frac1n \left\|\left(\breve \bW_{-1}^{\perp} - \tilde \bW_{-1}\right)\Delta\right\|^2}{\|\Delta\|^2} \le \frac{1}{n} \max_j \left\|\breve \bW_{*, j}^{\perp} - \tilde \bW_{*, j}\right\|^2 (1 \vee s)
$$
where $s = s_0 \vee s_1$. Using Lemma \ref{lem:sq_rate} we conclude with probability approaching to 1: 
$$
\sup_{\|\Delta_{S^c}\|_1 \le c\|\Delta_S\|_1} \frac{\frac1n \left\|\left(\breve \bW_{-1}^{\perp} - \tilde \bW_{-1}\right)\Delta\right\|^2}{\|\Delta\|^2} \lesssim ( s_0 \vee s_1)\left(\left[\dot r_n^2 + \left(\frac{\log{p}}{n} \vee \dot r_n \sqrt{\frac{\log{p}}{n}}\right)\right] \vee \frac{s_\gamma \log{p}}{n}\right)  
$$
Again, as in the case of for the proof of Proposition \ref{prop:denom_conv}, the above bound is $o(1)$ under certain sufficient condition as discussed in detail in subsection \ref{sec:discussion_sparsity_rate} of the main document. Under those sufficient conditions, we establish the RE condition. 

\section{Supplementary lemmas}
\begin{lemma}
\label{lem:prod_subexp}
Suppose $X, Y$ are two sub-exponential random variables, i.e. $\|X\|_{\psi_1}$ and $\|Y\|_{\psi_1}$ are finite. Then we have: 
$$
\|XY\|_{\psi_{1/2}} \le \|X\|_{\psi_1}\|Y\|_{\psi_1} \,.
$$
\end{lemma}
\begin{proof}
The proof follows along the similar line of arguments used in the proof of Lemma 2.7.7 of \cite{vershynin2018high}, which we present here for the sake of completeness. Without loss of generality assume that $\|X\|_{\psi_1} = \|Y\|_{\psi_1} = 1$ (otherwise we can always scale by it). 
\begin{align*}
\bbE\left[e^{\sqrt{|XY|}}\right] & \le \bbE\left[e^{\frac{|X| + |Y|}{2}}\right] \\
& = \bbE\left[e^{\frac{|X|}{2}}e^{\frac{|Y|}{2}}\right] \\
& \le \frac{1}{2}\bbE\left[e^{|X|} + e^{|Y|}\right] \\
& \le 2 \,.
\end{align*}
This completes the proof. 
\end{proof}

\begin{lemma}
\label{lem:subg_second_moment}
Suppose $X$ be a $n \times p$ matrix with i.i.d. rows and suppose each co-ordinates are centered sub-gaussian with sub-gaussian constant $\sigma$. Then we have: 
$$
\bbP\left(\max_{1 \le j \le p} \frac{1}{n} \sum_{i=1}^n X_{i, j}^2 > 3\sigma^2\right) \le 2\exp{\left[\log{p}-cn\right]} 
$$
for some constant $c$. 
\end{lemma}

\begin{proof}
From the sub-gaussianity of $X_{i, j}$ we have $\bbE[X_{i, j}^2] \le 2\sigma^2$. Furthermore we know $X_{i, j}^2$ is sub-exponential with $\|X_{i, j}^2\|_{\psi_1} = \|X_{i ,j}\|^2_{\psi_2} = \sigma^2$. Therefore using Bernstein inequality we have for any $1 \le j \le p$: 
\begin{align*}
& \bbP\left(\frac{1}{n} \sum_{i=1}^n X_{i, j}^2 - 2\sigma^2 > t\right) \\
& \le \bbP\left(\frac{1}{n} \sum_{i=1}^n X_{i, j}^2 - \bbE[X_j^2] > t + (2\sigma^2  - \bbE[X_j^2])\right) \\
& \le \bbP\left(\frac{1}{n} \sum_{i=1}^n X_{i, j}^2 - \bbE[X_j^2] > t \right) \\
& \le 2\exp{\left[-c\min\left(\frac{n^2 t^2}{n\sigma^4}, \frac{nt}{\sigma^2}\right) \right]} \\
& = 2\exp{\left[-c\min\left(\frac{n t^2}{\sigma^4}, \frac{nt}{\sigma^2}\right) \right]}
\end{align*}
Therefore, an application of union bound yields: 
$$
\bbP\left(\max_{1 \le j \le p}\frac{1}{n}\sum_i X_{i, j}^2 - 2\sigma^2 > t\right) \le 2\exp{\left[\log{p}-c\min\left(\frac{n t^2}{\sigma^4}, \frac{nt}{\sigma^2}\right) \right]}
$$ 
If we take $t = \sigma^2$, we have: 
$$
\bbP\left(\max_{1 \le j \le p}\frac{1}{n}\sum_i X_{i, j}^2 > 3\sigma^2 \right) \le 2\exp{\left[\log{p}-cn\right]} = o(1) 
$$
as long as $\log{p}/n \to 0$. 
\end{proof}

\begin{lemma}
\label{lem:subg_fourth_moment}
Suppose $X_1, \dots, X_n$ are i.i.d centered sub-gaussian random vector (with sub-gaussian constant $\sigma$) in dimension $p$. Then for any two vectors $a, b \in S^{p-1}$ we have: 
$$
\bbP\left(\max_{1 \le j \le p}\frac1n \sum_{i=1}^n \left(X_i^{\top}a\right)^2 \left(X_i^{\top}b\right)^2 > (2 + C_1)\mu\right) \le 2e^{\log{p} -\frac{\sqrt{n\mu}}{\log{(n+1)}}}
$$
for all large $n$, for some constant $C_1$ (involves $\sigma$) and $\mu$ as the mean of $\left(X_i^{\top}a\right)^2 \left(X_i^{\top}b\right)^2$. Moreover, $\mu$ is bounded by $\sigma^4$ and consequently we have: 
$$
max_{1 \le j \le p}\frac1n \sum_{i=1}^n \left(X_i^{\top}a\right)^2 \left(X_i^{\top}b\right)^2  = O_p(1)
$$ 
as long as $(\log{p}\log{n})/\sqrt{n} \to 0$. 
\end{lemma}

\begin{proof}
As $X_i$'s are sub-gaussian, we have: 
$$
\mu \triangleq \bbE\left[ \left(X_i^{\top}a\right)^2 \left(X_i^{\top}b\right)^2 \right] \le \sqrt{\bbE\left[\left(X_i^{\top}a\right)^4\right]\bbE\left[\left(X_i^{\top}b\right)^4\right]} \le 16 \sigma^2 \,.
$$
From Lemma \ref{lem:prod_subexp}, it is immediate that: 
$$
\left\| \left(X_i^{\top}a\right)^2 \left(X_i^{\top}b\right)^2 \right\|_{\psi_{1/2}} \le \left\| \left(X_i^{\top}a\right)^2\right\|_{\psi_1} \left\| \left(X_i^{\top}a\right)^2\right\|_{\psi_1} = \sigma^4 \,.
$$
For the rest of the proof, we use Theorem 3.2 of \cite{kuchibhotla2018moving} with $\alpha = 1/2$, which yields:
$$
\left|\frac{1}{n}\sum_{i=1}^n \left(X_i^{\top}a\right)^2 \left(X_i^{\top}b\right)^2 - \mu\right|_{\psi_{1/2}, L_n(1/2)} \le \frac{2e\sqrt{6} \sigma^4}{\sqrt{n}}  
$$
with
$$
L_n(1/2) \le C\sigma^8 \frac{\log^2{(n+1)}}{\sqrt{n}} \,.
$$
Therefore using the tail bound of \cite{kuchibhotla2018moving} (last display of page 8) we have: 
\begin{align*}
\bbP\left(\frac{1}{n}\sum_{i=1}^n \left(X_i^{\top}a\right)^2 \left(X_i^{\top}b\right)^2 - \mu > \frac{C_1}{\sqrt{n}}\left\{\sqrt{t} + \frac{t^2 \log^2{(n+1)}}{\sqrt{n}}\right\}\right) \le 2e^{-t} \,.
\end{align*}
Choosing $t = \sqrt{n\mu}/\log{(n+1)}$ we have: 
\begin{align*}
\bbP\left(\frac{1}{n}\sum_{i=1}^n \left(X_i^{\top}a\right)^2 \left(X_i^{\top}b\right)^2 > (2+C_1)\mu\right) \le 2e^{-\frac{\sqrt{n\mu}}{\log{(n+1)}}} \,.
\end{align*}
Therefore, a simple application of union bound yields: 
\begin{align*}
\bbP\left(\max_{1 \le j \le p}\frac{1}{n}\sum_{i=1}^n \left(X_i^{\top}a\right)^2 \left(X_i^{\top}b\right)^2 > (2+C_1)\mu\right) \le 2e^{\log{p} -\frac{\sqrt{n\mu}}{\log{(n+1)}}} \,.
\end{align*}
where the right hand side of the above equation in $o(1)$ as long as $(\log{p}\log{n})/\sqrt{n} \to 0$.
\end{proof}

\section{List of covariates}
\begin{table}[h]
\caption{Description of covariates of Turkey municipality voting data \label{tab:description_islamic}}
\fbox{%
\begin{tabular*}{0.9\textwidth}{*{2}{l}} 
vshr-islam 1994 & Islamic vote share 1994 \\
 \hline
partycount & Number of parties receiving votes 1994 \\
 \hline
lpop1994 & Log Population in 1994 \\
 \hline
merkezi & District center \\
 \hline
merkezp & Province center \\
 \hline
subbuyuk & Sub-metro center \\
 \hline
buyuk & Metro center \\
 \hline
ageshr19 & Population share below 19 in 2000 \\
 \hline
ageshr60 & Population share below 60 in 2000 \\
 \hline
sexr & Gender ratio in 2000 \\
 \hline
shhs & Household size in 2000 \\
 \hline
i89 & Indicator to Islamic mayor in 1989 \\
 \hline
 partycount & Number of parties in the election 1994 \\
 \hline 
 hischshr1520m & Share Men aged 15-20 with High School Education in 2000 
\end{tabular*}}
\end{table}

\begin{table}[h]
\caption{Table of covariates of GPA data \label{tab:covariates_LSO}}
\centering
\fbox{%
\begin{tabular*}{0.9\textwidth}{*{2}{l}}
Col. names & Explanation  \\
\hline 
hsgrade\_pct & High school grade in percentage \\
\hline 
totcredits\_year1 & Total credits taken in first year \\
\hline 
loc\_campus1 & Indicator whether the student in from Campus 1 \\
\hline 
loc\_campus2 & Indicator whether the student in from Campus 2 \\
\hline
male & Indicator of whether the student is male  \\
\hline
bpl\_north\_america & Whether the birth place in North America \\
\hline 
age\_at\_entry & Age of the student when they entered the college \\
\hline 
english & Indicator of whether the student is native english speaker \\
\end{tabular*}}
\end{table}

\bibliographystyle{plain}
\bibliography{mybib}

\end{document}